\documentclass[10pt,b5paper,openany]{book}

\usepackage{url}
\usepackage{color}
\usepackage{amssymb}
\usepackage{amsmath}
\usepackage{alltt}
\usepackage{verbatim}
\usepackage{stmaryrd}
\usepackage{mathtools}
\usepackage{wrapfig}
\usepackage{float}
\usepackage{graphicx}
\usepackage[all,2cell,cmtip,dvips,ps]{xy}
\usepackage{eucal}
\usepackage{colortbl}
\usepackage[bitstream-charter,expert,sfscaled=true]{mathdesign}
\usepackage{charter}

\usepackage{backref}
\renewcommand*{\backref}[1]{} 
\renewcommand*{\backrefalt}[4]{% 
\ifcase #1 % 
(Not cited.)% 
\or 
 \textbf {Cited} on page~#2.% 
\else 
 \textbf{Cited} on pages~#2.% 
\fi}

\usepackage{memhfixc}

\usepackage{tocbibind}

\usepackage{longtable}

\newcommand{\thesistitle}{}
\title{
{\fontsize{16.07}{15} \selectfont Coalgebraic Tools}\\
{\fontsize{16.07}{15} \selectfont for}\\
{\fontsize{16.07}{15} \selectfont Bisimilarity and Decorated Trace Semantics}\\
\bigskip
\bigskip
\bigskip
\normalsize
een wetenschappelijke proeve op het gebied\\
van de Natuurwetenschappen, Wiskunde en Informatica\\
\bigskip
\bigskip
\bigskip
\bigskip
Proefschrift\\
\bigskip
\bigskip
\bigskip
ter verkrijging van de graad van doctor\\
aan de Radboud Universiteit Nijmegen\\
op gezag van de rector magnificus prof. mr. S.C.J.J. Kortmann,\\
volgens besluit van het college van decanen\\
in het openbaar te verdedigen op maandag 16 december 2013\\
om 12.30 uur precies
\bigskip
\bigskip
\bigskip
\bigskip
door\\
\bigskip
\bigskip
\bigskip
Georgiana Caltais\\
\bigskip
\bigskip
geboren op 20 april 1984\\
te Suceava, Roemeni\"e
}
\makeindex
\usepackage{etex}
\usepackage{tikz}
\usetikzlibrary{arrows,automata}

\usepackage[Lenny]{fncychap}
\newcommand{\setmlength}[3]{\setlength{#1}{#3}\setlength{#1}{#2#1}}
\makeatletter
\ChNameVar{\Large\bfseries} % sets the style for name
\ChNumVar{\Large\bfseries} % sets the style for digit
\ChTitleVar{\LARGE\bfseries} % sets the style for title
\ChRuleWidth{4pt} % Set RW=4pt
\ChNameAsIs
\newlength{\@spacebeforechapterhead}
\setmlength{\@spacebeforechapterhead}{1.2}{2pt}
\newlength{\@spaceinchapterhead}
\setmlength{\@spaceinchapterhead}{1.2}{3.5pt}
\newlength{\@spaceafterchapterhead}
\setmlength{\@spaceafterchapterhead}{1.2}{80pt}
\renewcommand{\DOCH}{%
\CNV\FmN{\@chapapp}\space \CNoV\thechapter\par\nobreak
}
\renewcommand{\DOTI}[1]{%
   \vspace*{\@spacebeforechapterhead}%
   \parindent 0pt \Large\bfseries
   \interlinepenalty\@M
   \hrule
   \vspace*{\@spaceinchapterhead}%
   {\hfill \Large \bfseries #1}%
   \mbox{}\par
   \mbox{}\par
   \mbox{}\par
   \mbox{}\par
   \mbox{}\par
   \mbox{}\par
}

\makeatother

\mathcode`:="003A  % : ordinary symbol
\mathcode`;="003B  % ; ordinary symbol
\mathcode`?="003F  % ? ordinary symbol
\mathcode`|="026A  % | ordinary symbol
\mathcode`<="4268  % < abbreviates \langle
\mathcode`>="5269  % > abbreviates \rangle

\mathchardef\ls="213C    % less symbol (< used as \langle)
\mathchardef\gr="213E    % greater symbol (> used as \rangle)
\mathchardef\uparrow="0222  % adaptation mathmode of uparrow
\mathchardef\downarrow="0223  % adaptation mathmode of downarrow

\newcommand{\bb}[1]{\llbracket #1 \rrbracket}

\newcommand{\rules}[2]{\mbox{$\frac%
                      {\mbox{\small \rule[-5pt]{0pt}{14pt} $#1$}}
                      {\mbox{\small \rule[0pt]{0pt}{10pt}$#2$}}$}}

\DeclareSymbolFont{lasy}{U}{lasy}{m}{n}
\DeclareMathSymbol\myDiamond{\mathord}{lasy}{"33}
\newcommand{\myplus}{\mathbin{\rlap{$\myDiamond$}\hspace*{.01cm}\raisebox{.14ex}{$+$}}}

\date{}
\usepackage{graphics}
\usepackage[
left=2cm, right=2cm,
bottom=2cm, top=2.7cm
]{geometry}

\usepackage{eucal}% to make \mathcal fonts prettier, to get the old
\usepackage[bookmarks=true,
            pdftitle={\thesistitle.
                      PhD Thesis,
                      Radboud University Nijmegen, the Netherlands},
	    pdfauthor={Georgiana Caltais},
            pdfcreator={Georgiana Caltais},
	    citecolor=blue,
	    plainpages=false, % mandatory with book.sty
    pdfpagelabels,    % mandatory with book.sty
            pagebackref
            ]{hyperref}  % remember to use the latest version!!
\usepackage[amsmath,thmmarks,hyperref]{ntheorem}

\usepackage{multicol}%for the ipadissertations

\newtheorem{remark}{Remark}
\theoremseparator{.}
\theoremstyle{change}
\theorembodyfont{\normalfont} 
\theoremsymbol{\ensuremath{\clubsuit}} 
\newtheorem{definition}{\textsc{Definition}}[section]
\theoremstyle{change}
\theoremsymbol{}
\theoremheaderfont{\normalfont\bfseries} 
\theorembodyfont{\normalfont\itshape} 
\newtheorem{lemma}[definition]{\textsc{Lemma}}
\theoremstyle{change}
\theoremheaderfont{\normalfont\bfseries} 
\theorembodyfont{\normalfont\itshape} 
\newtheorem{theorem}[definition]{\textsc{Theorem}}
\theoremstyle{change}
\theoremheaderfont{\normalfont\bfseries} 
\theorembodyfont{\normalfont\itshape} 
\newtheorem{proposition}[definition]{\textsc{Proposition}}
\theoremstyle{change}
\theoremheaderfont{\normalfont\bfseries} 
\theorembodyfont{\normalfont\itshape} 
\newtheorem{corollary}[definition]{\textsc{Corollary}}
\theoremheaderfont{\normalfont\bfseries} 
\theorembodyfont{\upshape} 
\theoremsymbol{\ensuremath{\spadesuit}} 
\newtheorem{example}[definition]{\textsc{Example}}
\theorembodyfont{\upshape} 
\theoremheaderfont{\normalfont}
\theoremstyle{nonumberplain} 
\theoremsymbol{\ensuremath{\Box}} 
\newtheorem{proof}{\textsc{Proof}}

\usepackage[english,dutch]{babel}

\usepackage{listings} 
\lstdefinelanguage{NT}{ 
mathescape=true,
texcl=false, 
morekeywords=[1]{while, do, if, then, else, for, all},
morekeywords=[2]{return, skip},
morecomment=[s]{/*}{*/},
showstringspaces=false,
morestring=[b]",
morestring=[d],
tabsize=3,
extendedchars=false,
sensitive=true,
breaklines=false,
basicstyle=\ttfamily,
captionpos=b,
columns=[l]fixed,
identifierstyle={\color{black}},
keywordstyle=[1]{\color{dkviolet}},
keywordstyle=[2]{\color{dkred}},
stringstyle=\ttfamily,
commentstyle={\ttfamily\color{dkblue}},
literate=
    {true}{{{\color{dkgreen}$true$}}}3
    {false}{{{\color{dkgreen}$false$}}}4
}[keywords,comments,strings]

\lstnewenvironment{codeNT}{\lstset{language=NT}}{}

\renewcommand{\chaptermark}[1]{ \markboth{#1}{} }
\renewcommand{\sectionmark}[1]{ \markright{#1}{} }

\usepackage{afterpage}

\newcommand\HKC{\texttt{HKC}}
\newcommand\hkc[2]{\HKC\texttt{$(#1,#2)$}}
\newcommand\set[1]{\{#1\}}
\newcommand\dfa[1]{{%
  \newcommand\state[1]{##1}%
  \newcommand\fstate\overline%
  #1}}
\newcommand\nfa\dfa
\newcommand\prog\rightarrowtail

\newcommand\mst{\sqsubseteq_{mst}}
\newcommand\may{\sqsubseteq_{may}}

\newcommand\inclAnti{\subset \subset}

\newcommand\tauFailFin{\sqsubseteq_{\mathcal M}}
\newcommand\tauFailFind{\sqsubseteq_{{\mathcal M}_{2}}}
\newcommand\Fail{\emph{Fail}}
\newcommand\true{{\mathbf{tt}}}
\newcommand\false{{\mathbf{ff}}}
\newcommand{\I}{\textnormal{$\mathcal{I}$}}
\newcommand{\Powf}{\textnormal{${\mathscr P}_{\omega}$}}
\newcommand{\PD}{\mathscr{D}_{\omega}}
\newcommand{\PR}{{\mathcal{R}}_{p}}
\newcommand{\PF}{{\mathcal{F}}_{p}}
\newcommand{\PMF}{{\mathcal{MF}}_{p}}
\newcommand{\PT}{{\mathcal{T}}_{p}}
\newcommand{\PMT}{{\mathcal{MT}}_{p}}
\newcommand{\ZA}{0}
\newcommand{\ZP}{{\mathbf{0}}}
\newcommand{\Rs}{\textnormal{${\mathcal R}$}}
\newcommand{\Fs}{\textnormal{${\mathcal F}$}}
\newcommand{\CTs}{\textnormal{${\mathcal{CT}}$}}
\newcommand{\Ts}{\textnormal{${\mathcal T}$}}
\newcommand{\Tr}{\mathcal{T}}
\newcommand{\Ctr}{\mathcal{CT}}
\newcommand{\Rp}{\mathcal{R}}
\newcommand{\Fp}{\mathcal{F}}
\newcommand{\Pf}{\mathcal{PF}}
\newcommand{\Rtr}{\mathcal{RT}}
\newcommand{\Ftr}{\mathcal{FT}}

\newcommand{\ans}[1]{}%\textcolor{blue}{(Ans. #1)}}

\newcommand{\vdashInd}{\vdash_{\emph NDF}}

\newcommand{\bottom}{\mathop{\perp}}

\newcommand{\ExpS}{\mathsf{ExpStruct}}
\newcommand{\Fun}{\mathsf{Functor}}
\newcommand{\Fixpv}{\mathsf{FixpVar}}
\newcommand{\Alphb}{\mathsf{Alph}}
\newcommand{\Slt}{\mathsf{Slt}}
\newcommand{\AlphbN}{\mathsf{AlphName}}
\newcommand{\SltN}{\mathsf{SltName}}
\newcommand{\Ingred}{\mathsf{Ingredient}}
\newcommand{\IdtrlG}{\textnormal{$\id\triangleleft\mathscr{G}$}}
\newcommand{\BtrlG}{\textnormal{$\Bl\triangleleft\mathscr{G}$}}

\newcommand{\FtrlG}{\textnormal{$\mathscr{F}\triangleleft\mathscr{G}$}}
\newcommand{\GtrlG}{\textnormal{$\mathscr{G}\triangleleft\mathscr{G}$}}

\newcommand{\efr}[1]{{\fbox{$#1$}}}  %% freezing op for eqns
\newcommand{\RG}{\textnormal{${R}$}}
\newcommand{\RId}{\textnormal{${R_{\emph id}}$}}

\newcommand{\R}{\textnormal{${R}$}}
\newcommand{\B}{\textnormal{${B}$}}
\newcommand{\G}{\textnormal{${G}$}}

\newcommand{\Pow}{\textnormal{${\mathcal P}$}}
\newcommand{\algspec}{\textnormal{${\mathcal E}_{\itG}$}}
\newcommand{\algspecN}{\textnormal{${\mathcal E}$}}
\newcommand{\behspec}{\textnormal{${\mathcal B}_{\itG}$}}
\newcommand{\behspecACI}{\textnormal{${\mathcal B}_{\itG}^{\emph ACI}$}}

\newcommand{\ct}{\colon}

\newcommand{\cl}[1]{\emph{cl}(#1)}

\newcommand{\CIRC}{\textsf{CIRC}}

\newcommand{\code}[1]{{\fontsize{9}{10}\selectfont {#1}}}

\newcommand{\rTrans}[1]{\mathrel{{\overset{#1}{\Rightarrow}}}}

\newcommand{\eps}{\varepsilon}
\newcommand{\itB}{\mathsf{B}}

\newcommand{\itG}{\mathscr{G}}

\newcommand{\itF}{\mathscr{F}}
\newcommand{\itExp}{\mathsf{Exp}}

\newcommand\Gf{\mathscr G} %functor G
\newcommand\F{\mathcal F}  %functor F
\newcommand\Bl{{\mathsf B}} %join semilattice B
\newcommand\id{{\mathsf{Id}}}
\newcommand\pow{\mathcal{P}_{\!\!\omega}}
\newcommand\pf{{\emph NDF}}
\newcommand\Exp{{\mathsf{Exp}}}

\newcommand\E\varepsilon

\newcommand\emp{\underline\emptyset}

\newcommand\D{\mathscr{D}}
\newcommand\Rf{\mathscr{S}}
\newcommand\M{\mathscr{M}}
\newcommand\Pa{\mathscr{Q}}
\newcommand\N{\mathscr{N}}
\newcommand\Lf{\mathscr{L}}

\mathcode`:="003A  % : ordinary symbol
\mathcode`;="003B  % ; ordinary symbol
\mathcode`?="003F  % ? ordinary symbol
\mathcode`|="026A  % | ordinary symbol
\mathcode`<="4268  % < abbreviates \langle
\mathcode`>="5269  % > abbreviates \rangle

\mathchardef\ls="213C    % less symbol (< used as \langle)
\mathchardef\gr="213E    % greater symbol (> used as \rangle)
\begin{document}
\selectlanguage{english}

%Headers stuff

%\setlength{\headwidth}{\textwidth}
\pagestyle{headings}
%\makerunningwidth{myheadings}{\headwidth}
%\makepsmarks{myheadings}{\companionpshook}
%\makeevenhead{myheadings}{\normalfont\rmfamily\thepage}{}{\normalfont\itshape
%Chapter~\thechapter. \leftmark}{}
%\makeoddhead{myheadings}{\normalfont\itshape\rightmark}{}{\normalfont\rmfamily\thepage}
%\makeheadposition{myheadings}{flushright}{flushleft}{flushright}{flushleft}
%\pagestyle{myheadings}

%% chapter style for appendices, text comes on following page
%\setlength{\headwidth}{\textwidth}
%\makepagestyle{biblio}
%\makerunningwidth{biblio}{\headwidth}
%\makepsmarks{biblio}{\companionpshook}
%\makeevenhead{biblio}{\normalfont\rmfamily\thepage}{}{\normalfont\itshape\leftmark}{}
%\makeoddhead{biblio}{\normalfont\itshape\rightmark}{}{\normalfont\rmfamily\thepage}
%\makeheadposition{biblio}{flushright}{flushleft}{flushright}{flushleft}

\frontmatter

\pagestyle{empty}
%\vspace*{1.5cm}
%\noindent {\fontsize{19.416}{18.1216}\selectfont  \hfill \thesistitle}\\
\maketitle
%\vspace{.5cm}

%{\fontsize{12}{11.2}\selectfont  Georgiana Caltais}
%{\fontsize{15}{14}\selectfont  Georgiana Caltais}

%\newpage
%\include{titelblad_achter}
%
%
%\newpage
%\include{titelblad_voor}

\newpage
\begin{tabular}{ll}
\textbf{Promotoren:}
\\
\qquad
  Prof.~dr.~Jan~Rutten
\\
\qquad
	Prof.~dr.~Anna~Ing\'olfsd\'ottir &(H\'ask\'olinn \'i Reykjav\'ik, IJsland)
\\\\
\textbf{Copromotoren:}\\
\qquad
 Dr.~Alexandra~Silva
\\
\qquad
 Dr.~Marcello~Bonsangue &(Universiteit Leiden, Nederland)
\\\\
\textbf{ Manuscriptcommissie:}\\
  \qquad
Prof.~dr.~Luca~Aceto
&
(H\'ask\'olinn \'i Reykjav\'ik, IJsland)
\\ \qquad   
Prof.~dr.~Herman~Geuvers
& \\ \qquad
Dr.~Bas~Luttik
&
(Technische Universiteit Eindhoven, Nederland)
\\ \qquad 
Prof.~Ugo~Montanari
&
(Universit\'a di Pisa, Itali\"e)
\\ \qquad 
Dr.~Erik~de~Vink 
&
(Technische Universiteit Eindhoven, Nederland)
\end{tabular}

%\chapterstyle{chappell}
%\chapterstyle{thatcher}
\newpage

%\tableofcontents
%\addcontentsline{toc}{chapter}{Contents}
%\newpage
%\phantom{bla}

%% Remove all headers
%\setlength{\headrulewidth}{0pt}
%\lhead{}
%\rhead{}

%\newpage

%\include{preface}

%\pagestyle{thesisb}

%\include{Introduction}
%\include{Preliminaries}
%\include{automata}
%\include{mealy}
%\include{polynomial}
%\include{quantitative}
%\include{conclusions}

%\backmatter

%\cleardoublepage
%\pagestyle{biblio}

\bibliographystyle{alpha}
\cleardoublepage
\pagenumbering{roman}
\tableofcontents

\chapter{Acknowledgements}
\label{ch:ack}

First of all I want to thank, in alphabetical order, Alexandra Silva, Anna Ing\'olfsd\'ottir, Jan Rutten, Luca Aceto and Marcello Bonsangue who supervised my research during the last three years, inspired me in my work and provided me with the support I needed in order to complete this thesis. I truly appreciate their flexibility, understanding and willingness to initiate the collaboration between the two research groups at Radboud University and  Reykjavik University, making my PhD journey more interesting and prolific.

Many thanks to my collaborators that opened the doors for more applicability in my work: Filippo Bonchi and Damien Pous, co-authors of the results in Chapter 5, and Eugen-Ioan Goriac and Dorel Lucanu, co-authors of the contributions in Chapter 3. I am more than grateful to the members of the reading committee Bas Luttik, Erik de Vink, Herman Geuvers, Luca Aceto and Ugo Montanari for their comments and suggestions on this thesis. Moreover, I want to thank Wan Fokkink for planting the seeds of my interest in modelling and studying the semantics of concurrent reactive systems, during our rather fortuitous meeting in Reykjavik, in 2011.

I am very grateful to my friends in Reykjavik, for making my Icelandic experience brighter, warmer and more comfortable. Special credits go to the best Vikings ever Agnes, Ali, Andrea, Angelo, Claudio, Dario, Eugen, Gabriel, Ghiuseppe, Hamid, Hogni, M\u ad\u alina, Maria, Marijke, Matteo, Niccolo, Nils, Paulo, Pradipta, Robert, Rolanda, Stephan, Ute, Victor, Viky and all the Mjolnir crew, and to my dear Icelandic parents Valla and Skarpi. I want to express my appreciation to Andrei Manolescu and Marjan Sirjani for their support at the university. Whenever I traveled to the Netherlands I enjoyed the excellent company of my colleagues and friends Adam, Afrodita, Alexandra, Antonis, Darya, Fernando, Filippo, Francesco, Giannicola, Helle, Jonce, Jos, Joost, Neko, Michiel, Rita, Robbert, Sarah and Wojtek; they did a great job in cheering up my stays. I thank my Romanian friends Amalia, Anca, Andrei, Ciprian, Cozmin, Liliana, M\u ad\u alina, Marian, M\u ariuca, Mihaela, Mihai, Radu, Romina and Roxana for still remembering me, and making me feel welcome back home. Moreover, I want to thank my professors Gheorghe Marchitan, C\u at\u alin \c Tig\u aeru, and Adriana and Ovidiu Gheorghie\c s for awakening my passion for Mathematics and Computer Science, and for providing me with important advice when most needed.

I wish to express my gratitude to Goran Krist\'ofer for his constant encouragements, care, and uplifting discussions, that invaluably contributed to the successful writing of this thesis, and to my personal growth. Finally, I acknowledge my parents Floarea and Constantin for dedicating their lives to my support, nurturing and guidance, and for being next to me in my good and bad times. Everything I am, I owe them.

The work in this thesis was supported by the project ``Meta-theory of Algebraic Process Theories'' (no.~100014021) of the Icelandic Research Fund, Radboud University and a CWI Internship.

\mainmatter
\pagestyle{headings}
\renewcommand{\chaptermark}[1]{\markboth{Chapter~\thechapter.~ #1}{} }
\renewcommand{\sectionmark}[1]{ \markright{\thesection.~#1}{} }

%\pagestyle{biblio}
%\phantomsection
%\printindex
%\phantomsection
%\addcontentsline{toc}{chapter}{Summary}
%\include{summary}
%\phantomsection
%\addcontentsline{toc}{chapter}{Samenvatting}
%\include{samenvatting}
%\phantomsection
%\include{cv}
%\include{ipadissertations}

\pagenumbering{arabic}
\chapter{Introduction}

One of the research areas of great importance in Computer Science is the study of the semantics of concurrent reactive systems~\cite{Harel:1989:DRS:101969.101990}. 
These are systems that compute by interacting with their environment, and typically consist of several parallel components, which execute simultaneously and potentially communicate with each other.
Examples of such systems range from rather simple devices such as calculators and vending machines, to programs controlling mechanical devices such as cars, subways or spaceships. In light of their widespread deployment and complexity, the application of rigorous methods for the specification, design and reasoning on the behaviour of reactive systems has always been a great challenge.

One possible approach to formally handle reactive systems is to use a ``common language" for describing both the actual implementations and their specifications. When following this technique, checking whether an implementation and its specification describe the same behaviour reduces to proving some notion of equivalence/preorder between their corresponding descriptions over the chosen language. This procedure is also referred to as ``equivalence checking''.

Intuitively, we say that an implementation complies to its specification whenever the implementation displays only the behaviour allowed by the specification, and nothing more. However, it is important to notice that system verification can be performed at different levels of abstraction, with respect to non-determinism, for example, depending on the context of application. In this regard, we refer to a suite of semantics that are thoroughly studied throughout this thesis, namely: bisimilarity~\cite{Park81,milner89} -- the standard notion of behavioural equivalence in concurrency --, the spectrum of decorated trace semantics in van Glabbeek's work~\cite{Glabbeek01}, and must and may testing semantics~\cite{CleavelandH89,DBLP:journals/tcs/NicolaH84,Hennessy:1988:ATP:50497}.

Along time, different mathematical frameworks have been exploited for modelling reactive systems and their behaviours, and for deriving efficient verification algorithms for their computer-aided analysis. In the sequel, we provide a short overview on two of such ``dual'' frameworks: algebra~\cite{burris2012course,Hennessy:1988:ATP:50497} and coalgebra~\cite{Jacobs97atutorial,Rutten00}.

\section{Algebra}
\label{intro-algebra}
Algebraic process theories, or ``process algebras", have been successfully used as prototype specification languages for reactive systems. Typically, the definition of process algebras consists in providing a syntax and an operational semantics, usually given in terms of so-called Structural Operational Semantics (SOS) rules~\cite{Plotkin81astructural}.
Intuitively, SOS is a framework used for describing how programs compute step by step, by emphasising the corresponding state-transformations that occur after the execution of certain actions.
Once a desired notion of behavioural equivalence or preorder over processes is fixed, a corresponding sound (and ideally complete) axiomatisation is given. This way, one can establish the conformance of an implementation with its specification in an equational style, without generating the state space of processes, therefore potentially combatting the state explosion problem.
We hint, for example, to the works in~\cite{Aceto:1994:TSR:184662.184663} and \cite{DBLP:journals/jlp/BaetenV04}, where sound and complete axiomatisations for bisimilarity of systems complying to the GSOS~\cite{Bloom:1995:BCT:200836.200876} format and GSOS with termination, respectively, are provided.

Unfortunately, this approach has low flexibility as regards language modifications.
axiomatisations are usually shown sound and complete
by means of proof techniques that take into account the combinators of the language under consideration;
hence new syntactical constructs frequently impose new proofs (from scratch).
Consider, for instance, the work in~\cite{Aceto-GSOS} extending the results in~\cite{Aceto:1994:TSR:184662.184663} to the case of GSOS with predicates such as termination, divergence and convergence.
Even though syntactically trivial, the extension required the construction of a new axiomatisation that had to be proven sound and complete (which is often not a trivial task).

However, as soon as an axiomatisation is identified, the implementation of a verification tool based on equational reasoning is almost straightforward. We refer, for example, to the automated tool in~\cite{DBLP:conf/calco/AcetoCGI11} which can be used for reasoning on bisimilarity of systems complying to the extended GSOS format in~\cite{Aceto-GSOS}.

Moreover, in the algebraic setting, SOS rules can be used not only for specifying the behaviour of systems in an intuitive fashion, but also for imposing a series of (syntactic) constraints to guarantee that a certain notion of behavioural equivalence (or preorder) for systems satisfying the aforementioned restrictions is also a (pre)congruence. Semantics which are also (pre)congruences are important from the practical perspective as well. Intuitively, whenever a subcomponent of a system is replaced, showing the equivalence between the new ``upgraded'' system and the initial one, with respect to a notion of (pre)congruence, reduces to showing the equivalence between the two subsystems that have been interchanged. This way, the complexity of the verification procedure is obviously reduced. In this respect, we refer, for instance, to the GSOS~\cite{Bloom:1995:BCT:200836.200876} format which guarantees that bisimilarity is a congruence. In related work~\cite{Bloom:2004:PFD:963927.963929}, precongruence formats for decorated trace semantics~\cite{Glabbeek01} were established via modal characterisations of the corresponding preorders.

\section{Coalgebra}
\label{intro-coalgebra}
A possible representation of implementations and their specifications is in terms of state machines. These allow for a uniform manipulation of systems such as: streams~\cite{streams:rutten05}, (non)deterministic and probabilistic automata~\cite{Rabin:1959:FAD:1661907.1661909,DBLP:journals/siamcomp/Rabin80}, Moore~\cite{Moore56} and\linebreak Mealy~\cite{mea55} machines, and labelled transition systems~\cite{Keller:1976:FVP:360248.360251}.

Coalgebra~\cite{Jacobs97atutorial,Rutten00} is a recent unifying theory combining ideas from the mathematical theory of dynamical systems and from the theory of state-based computation, and has been successfully applied as a mathematical framework for the study of state-based systems.
%In close connection with the work in this thesis, we refer to: \cite{brs_lmcs} for a uniform coalgebraic handling of a large class of (nondeterministic) systems, \cite{Bonchi:2012:BAA:2340820.2340823,BONCHI:2012:HAL-00639716:4} for algorithms for reasoning on behavioural equivalence of systems in a coalgebraic setting, and \cite{rosu-lucanu-2009-calco} for an automated tool implementing coalgebraic techniques for verifying properties of infinite systems.
Intuitively, from the coalgebraic perspective, systems with (possibly) infinite behaviour are represented as \emph{black-box} machines analyzed only according to their observable behaviour. Mathematically, one can describe such a machine in terms of a coalgebra $(X, \delta \colon X \rightarrow \F(X))$ consisting of a set (of states) $X$, and a map $\delta$ encapsulating the corresponding behaviour based on a functor $\F$.
This map represents the set of observers, or destructors, allowing one to ``break'' (infinite) system behaviour into analyzable fragments.

Coalgebraic analysis on the behaviours of systems can be performed as follows. First, identify the appropriate functor associated with the class of systems under analysis. Then, %formulate the desired notion of behavioural equivalence or preorder in terms of observational indistinguishability and finally, 
reason on the corresponding notion of behavioural equivalence by coinduction~\cite{sangiorgi2011advanced}, a proof technique based on bisimulation,
already implemented in automated tools~\cite{BONCHI:2012:HAL-00639716:4,DBLP:conf/tacas/CranenGKSVWW13,Cleaveland:1993:CWS:151646.151648,DBLP:conf/tacas/GaravelLMS11,rosu-lucanu-2009-calco}.

All the systems mentioned above can be coalgebraically modelled in a uniform way, by simply varying the behaviour functor $\F$. For instance, for the case of streams ({\it i.e.}, infinite words) over an alphabet $A$, the functor $\F(X) = A \times X$ provides the head of the stream, which is an element of $A$, and its tail, which is again a stream. Labelled transition systems are intuitively defined by the functor $\F(X) = (\Pow X)^{A}$, which for an action labelled in $A$ returns the set of states that can be (non-deterministically) reached after executing that action. More interestingly, note that each functor induces a notion of behavioural equivalence~\cite{Rutten00}. For streams, for example, this coincides with stream equality, whereas for deterministic automata and labelled transition systems, the corresponding notions of behavioural equivalence are language equivalence and bisimilarity~\cite{Park81,milner89}, respectively.

As already stated, verification of systems can be performed at different levels of abstraction, depending on the context of application. The work in this thesis is closely related to the results in~\cite{fsttcs}. There it is shown how the generality and modularity of coalgebras can be exploited (via a coalgebraic subset construction) in order to uniformly reason about the behaviour of labelled transition systems in terms of trace, ready or failure equivalence~\cite{Glabbeek01}, rather than bisimilarity. Moreover, reasoning on the aforementioned equivalences follows ``for free'' by coinduction, and can be performed in a fully automated fashion using the tool in~\cite{BONCHI:2012:HAL-00639716:4}.

Even though the coalgebraic setting abstracts from the syntax in process description languages, its generality and uniformity enables also the interplay with syntax-based characterisations of systems. For example, we refer to the works in~\cite{Klin:2009:BMM:1512997.1513245,Turi:1997:TMO:788019.788864}, where bialgebraic frameworks for deriving congruence rule formats and proving compositionality of various kinds of semantics (such as bisimilarity and decorated trace semantics~\cite{Glabbeek01}) were provided based on the so-called \emph{distributive laws} of syntax over behaviour. From a simpler perspective, note that the dynamics of transition systems for process algebras can be coalgebraically characterised (in terms of states and transitions between states) according to the SOS rules expressing their behaviours.

\section{Aim and approach}
\label{intro-our-aim}

%As previously stated, the formal specification and verification of several types of systems and semantics, each of which suitable for use in different contexts of application, has represented a hot research topic.

Along the research lines mentioned so far, the \emph{aim} of our work is to exploit the strengths of the (co)algebraic framework in modelling reactive systems and reasoning on several types of associated semantics, in a uniform fashion. In particular, we are interested in handling notions of behavioural equivalence/preorder ranging from bisimilarity for systems that can be represented as non-deterministic coalgebras~\cite{brs_lmcs}, to decorated trace semantics for labelled transition systems and probabilistic systems, and testing semantics for labelled transition systems with internal behaviour.
Moreover, we aim at deriving a suite of corresponding verification algorithms suitable for implementation in automated tools.

\medskip
The \emph{approach} we adopt is based on the following steps.

\begin{itemize}
\item[$\bullet$]
First, we focus on the results in~\cite{BonsangueRS09} introducing a language of expressions for specifying a large class of systems that can be modelled as non-deterministic coalgebras, and a sound and complete axiomatisation for bisimilarity of such systems. The latter include, for example, streams, (non)deterministic automata, Mealy, Moore and labelled transition systems.
In~\cite{BonsangueRS09}, systems which are coalgebras of non-deterministic functors are described in a rather algebraic fashion, in terms of a language of expressions derived according to the functor of interest. Then, expressions are shown to have a coalgebraic structure, hence further enabling reasoning on their equivalence by coinduction.

In our approach, we exploit a combination of algebra and coalgebra, based on interplays such as constructors -- destructors, induction -- coinduction (both as definition and as proof principles), and congruence -- bisimilarity~\cite{Jacobs97atutorial}. Building on these associations and on the strength of coalgebras in deriving algorithms and tools for the automatic verification of systems, we
construct a decision procedure for the bisimilarity of generalised regular expressions~\cite{BonsangueRS09} (and therefore, of their corresponding non-deterministic systems). This is achieved by providing an algebraic specification for the coalgebra of expressions, and reducing coinduction to an entailment relation between this specification and a suitable set of equations.

The theory was implemented in CIRC~\cite{goguen-lin-rosu-2000-ase,rosu-lucanu-2009-calco} -- an automated theorem prover based on coinduction, successfully used for reasoning on properties of infinite data structures such as streams --, and can be tested online at:\linebreak
\url{http://goriac.info/tools/functorizer/}.

\item[$\bullet$]
Although bisimilarity~\cite{Park81,milner89} is the standard notion of behavioural equivalence in concurrency theory, considerable amount of work has been dedicated to the treatment of decorated trace semantics~\cite{Glabbeek01,GPS-Jou-Smolka}, and may and must testing semantics~\cite{CleavelandH89,DBLP:journals/tcs/NicolaH84,Hennessy:1988:ATP:50497}, for instance.

Studying semantics other than bisimilarity is not only an interesting research subject per se, but is also important from the applicability perspective.

For example, bisimilarity, which belongs to the class of the so-called ``branching time'' semantics, can be sometimes too fine for system verification. Therefore, coarser semantics such as the ``linear time'' semantics might be more appropriate. In this respect, we refer to the work in~\cite{vanGlabbeek:2001:BTS:377786.377855} for a survey on the aforementioned semantic equivalences (and preorders), and for a study on their context of application and advantages.

Semantics coarser than bisimilarity, for example, that are also\linebreak (pre)congruences, can play an important role in system reduction as well. Consider a scenario in which the correctness of concurrent systems is established according to a property expressed by a set of logical formulae.  It would be desirable to use a (pre)congruence preserving such a property for deriving a smaller (reduced) labelled transition system whose components are eventually checked for the aforementioned property. Hence, the coarser the (pre)congruence, the coarser the refinement of the original system.

We refer, for example, to the work in~\cite{Valmari-fail}, where it is shown that trace equivalence~\cite{Glabbeek01} is the weakest congruence preserving the property ``$P$ may ever execute action $a$'', whereas the so-called ``stable failure equivalence'' is the  weakest deadlock-preserving congruence with respect to any set of Basic Lotos~\cite{DBLP:journals/cn/BolognesiB87} operators containing parallel composition.
%For more insight on the applicability and advantages of using failure-based equivalences we refer to~\cite{Valmari-fail}.

We also hint to the work in~\cite{fsttcs}, where trace, failure and readiness semantics~\cite{Glabbeek01} were recovered in a coalgebraic setting by applying the generalised powerset construction~\cite{gen-pow}, which is reminiscent of the determinisation of non-deterministic automata.

Also of interest in concurrency, are must and may testing semantics \cite{DBLP:journals/tcs/NicolaH84,Hennessy:1988:ATP:50497}.
Unlike weak bisimilarity, must testing distinguishes between livelock and deadlock, for instance.
This can be useful in practice as, even though internal behaviour of systems do not provide any information to an external observer, it can be desirable to set apart infinite internal computations from the impossibility of performing any further move.
In~\cite{CleavelandH89}, an alternative characterisation of may and must testing semantics is based on sequences of observable actions processes can execute. Hence, it is of interest studying a possible connection with the approach in~\cite{fsttcs}, for a coalgebraic modelling of these semantics.

Motivated by these results and observations, as a second step we provide a uniform coalgebraic modelling of decorated trace, may and must testing semantics via the generalised powerset construction.

\item[$\bullet$] Last, but not least, we exploit the coalgebraic modelling of decorated trace and must testing semantics (which is more interesting than may testing semantics, as it is sensitive to the non-determinism of processes), and devise algorithms for reasoning on the corresponding equivalences and preorders.

Existing algorithms for the automated checking of these behavioural semantics over finite-state systems rely on the following idea. First, non-deterministic systems are transformed into the so-called (deterministic) ``acceptance graphs'', by applying a technique which is reminiscent of the determinisation of non-deterministic automata~\cite{Rabin:1959:FAD:1661907.1661909}. Then, reasoning on the aforementioned semantics on the original non-deterministic systems is reduced to the equivalent problem of reasoning on bisimilarity of the associated acceptance graphs.
We refer to ~\cite{cleaveland1993concurrency,DBLP:conf/cav/CleavelandS96,DBLP:journals/topnoc/CalzolaiNLT08} for examples of automated tools implementing such algorithms.

In our work, however, the coalgebraic setting enables the construction of verification algorithms which are not available for bisimilarity.
More precisely, we build an algorithm based on (Moore-) bisimulations (up-to)~\cite{BONCHI:2012:HAL-00639716:4,sangiorgi2011advanced,San98MFCS}, which follows as a consequence of the determinisation procedure previously mentioned.
Moreover, we provide a variation of Brzozowski's algorithm~\cite{Brzozowski}, by exploiting the abstract coalgebraic theory in~\cite{Bonchi:2012:BAA:2340820.2340823}.

Our approach is uniform and modular: once the ``recipe'' for handling failure semantics is established, the almost straightforward extension to non-deterministic systems with internal behaviour enabled shifting to must testing semantics. This is also a consequence of the fact that failure semantics coincides with must testing in the absence of divergence~\cite{CleavelandH89,DBLP:journals/acta/Nicola87}.
Furthermore, the algorithms for reasoning on failure semantics can be easily adapted also for other decorated trace semantics studied in this thesis.

Both the bisimulation-based and Brzozowski's minimisation techniques were implemented in an automated tool, and can be tested online at:\newline
\url{http://perso.ens-lyon.fr/damien.pous/brz/}.
\end{itemize}

\newpage
\section{Thesis outline}
\label{sec:thesis-struct}

We summarise the content and the main contributions of the thesis.

\bigskip
\textbf{Chapter 2} provides the basic definitions from coalgebra and recalls the generalised powerset construction, which we will use in our work.

\bigskip
\textbf{Chapter 3} presents an algorithm to decide whether two (generalised regular) expressions defining systems that can be modelled as non-deterministic coalgebras are bisimilar or not. The aforementioned expressions and an analogue of Kleene's theorem and Kleene algebra, were recently proposed by Silva, Bonsangue and Rutten in~\cite{brs_lmcs}.
Examples of systems we handle include infinite streams, deterministic automata, Mealy machines and labelled transition systems. The procedure is implemented in the automatic theorem prover CIRC, by reducing coinduction to an entailment relation between an algebraic specification and an appropriate set of equations.

The main contributions are summarised in the table below.
%\vspace{-15pt}
\[
\begin{tabular}{ll}
\hline
\multicolumn{2}{c}{A decision procedure for bisimilarity}\\
\hline
\hline
\textnormal{Algebraic modelling of expressions} & Figure~\ref{fig:CoVsAlg}\\[0.5ex]
\textnormal{Algebraic encoding of bisimilarity} & Corollary~\ref{cor:ii}\\[0.5ex]
\textnormal{Soundness}& Theorem~\ref{thm:soundnessCirc}\\[0.5ex]
\textnormal{Decision procedure} & Theorem~\ref{thm:decProc}\\
\hline
\hline
\end{tabular}
\]

%\vspace{-15pt}
\smallskip
This chapter is based on the following papers:\\
\emph{\cite{sbmf} Marcello M. Bonsangue, Georgiana Caltais, Eugen-Ioan Goriac, Dorel Lucanu, Jan J. M. M. Rutten, Alexandra Silva. A decision procedure for bisimilarity of generalised regular expressions. Proc. 13'th Brazilian Symposium on Formal Methods, 2011:226--241.}\\[0.5ex]
\emph{\cite{DBLP:journals/corr/abs-1303-1994} Marcello M. Bonsangue, Georgiana Caltais, Eugen-Ioan Goriac, Dorel Lucanu, Jan J. M. M. Rutten, Alexandra Silva. Automatic equivalence proofs for non-deterministic coalgebras. Science of Computer Programming, 2013:1324--1345.}

%\vspace{-20.5pt}
\bigskip
\textbf{Chapter 4} provides the coalgebraic handling of a series of semantics on transition systems in a uniform modular fashion, by employing the generalised powerset construction introduced by Silva, Bonchi, Bonsangue and Rutten in~\cite{gen-pow}.
As we shall see, this construction yields a notion of minimal representatives for (i) decorated trace equivalences for labelled transition systems (LTS's)~\cite{Keller:1976:FVP:360248.360251} and generative probabilistic systems (GPS's)~\cite{Glabbeek01,GPS-Jou-Smolka} and, (ii) must and may testing semantics for non-deterministic systems with internal behaviour~\cite{CleavelandH89,DBLP:journals/tcs/NicolaH84,Hennessy:1988:ATP:50497}.
As a consequence, reasoning on the aforementioned notions of behavioural equivalence/preorder can be performed in terms of (Moore-) bisimulations.
Moreover, we show how the spectrum of decorated trace semantics can be recovered from the coalgebraic modelling.

The main contributions are listed in the following table.

\[
\begin{tabular}{ll}
\hline
\multicolumn{2}{c}{Decorated traces and testing semantics coalgebraically}\\
\hline
\hline
\text{Correctness of the coalgebraic modelling of:}&\\[0.5ex]
\text{Ready \& failure semantics for LTS's}& Theorem~\ref{thm:eqiv-ready} \\[0.5ex]
\text{(Complete) trace semantics for LTS's}& Theorem~\ref{thm:eqiv-trace} \\[0.5ex]
\text{Possible-futures semantics for LTS's}& Theorem~\ref{thm:eqiv-poss-fut} \\[0.5ex]
\text{Ready \& failure trace semantics for LTS's}& Theorem~\ref{thm:eqiv-ready-tr} \\[0.5ex]
\text{Ready \& (maximal) failure semantics for GPS's}& Theorem~\ref{thm:equiv-r-f-mf} \\[0.5ex]
\text{(Maximal) trace semantics for GPS's}& Theorem~\ref{thm:equiv-trace} \\[0.5ex]
\text{May testing semantics}& Theorem~\ref{thm:may-tr} \\[0.5ex]
\text{Must testing semantics}& Theorem~\ref{thm:fail-mst-conn} \\[0.5ex]
\hline
\text{Recovering the spectrum}& Lemma~\ref{lm:rec-fail=ct}\\
&Lemma~\ref{lm:rec-ready-trace}\\
\hline
\hline
\end{tabular}
\]

\smallskip
This chapter is based on the papers:\\
\emph{\cite{dec-tr-MFPS12} Filippo Bonchi, Marcello Bonsangue, Georgiana Caltais, Jan Rutten, Alexandra Silva. Final semantics for decorated traces.
Electronic Notes in Theoretical Computer Science, 2012:73--86. Proc. Mathematical Foundations of Programming Semantics 2012. 
}

\bigskip
\textbf{Chapter 5} focuses on checking language equivalence (or inclusion) of finite automata. This is a classical problem in computer science, which has recently received a renewed interest and found novel and more effective solutions, such as the approaches based on antichains~\cite{tacas10,CAV06} or bisimulations up-to~\cite{BONCHI:2012:HAL-00639716:4,DBLP:conf/sofsem/RotBR13,sangiorgi2011advanced,San98MFCS}. Several notions of equivalence (or preorder) have been proposed for the analysis of concurrent systems. Some approaches reduce the problem of checking these equivalences to the problem of checking bisimilarity. In this chapter, we tackle this challenge differently, and propose to ``adapt'' algorithms for language semantics. 
More precisely, we introduce an analogue of Brzozowski's algorithm and \HKC\ -- an optimisation of Hopcroft and Karp's algorithm~\cite{hoka71} based on bisimulations up-to --, for checking must testing equivalence and preorder as well as failure equivalence. 
To achieve this transfer of technology (from language to must/failure semantics), we take a coalgebraic look at the problem at hand.

The table below summarises the main contributions of this chapter.
\[
\begin{tabular}{ll}
\hline
\multicolumn{2}{c}{Algorithms for decorated trace and must testing semantics}\\
\hline
\hline
\text{\HKC\ for failure semantics} & Sections~\ref{ssec:hkcFailure},~\ref{app:proofofcorrectness}\\
\text{Brzozowski for failure semantics} & Sections~\ref{sec:brzozowski-failure},~\ref{sec:proof:Brzozowski}\\
\text{\HKC\ for must testing semantics} & Sections~\ref{ssec:HKCmust},~\ref{app:proofofcorrectness}\\
\text{Brzozowski for must testing semantics} & Sections~\ref{sec:brzoz-must},~\ref{sec:proof:Brzozowski-must}\\
%\hline
%\text{Examples, tool and tests} & Sections~\ref{},~\ref{},~\ref{}\\
\hline
\hline
\end{tabular}
\]

\smallskip
This work is based on the paper:\\
\emph{Filippo Bonchi, Georgiana Caltais, Damien Pous, Alexandra Silva. Brzozowski's and Up-to Algorithms for Must Testing. To appear in volume 8301 of the Lecture Notes in Computer Science series.}

\section{Related work}
\label{sec:rel-work}

The contributions of the thesis stem from the underlying idea of formally specifying and verifying concurrent reactive systems in a uniform fashion, both in theory and practice, by exploiting the (co)algebraic framework.

\medskip
\emph{On the one hand}, we build our work based on previous results originating from the correspondence between regular expressions and finite deterministic automata (DFA's) -- two of the most basic structures in Computer Science --. Kleene's theorem~\cite{Kleene61} gives a fundamental correspondence between these two structures: each regular expression denotes a
language that can be recognised by a DFA and, conversely, the language accepted by a DFA
can be specified by a regular expression.
%Languages denoted by regular expressions are
%called regular. Two regular expressions are (language) equivalent if they denote the
%same regular language.
A sound and complete axiomatisation
(later refined by Kozen in~\cite{kozen91,kozen-nerode}) for proving the equivalence of regular expressions was introduced by Salomaa~\cite{salomaa}, and an extension for the case of LTS's modulo bisimilarity was derived by Milner in~\cite{milner}.

%Salomaa~\cite{salomaa} presented a sound and complete axiomatisation
%(later refined by Kozen in~\cite{kozen91,kozen-nerode}) for proving the equivalence of regular expressions.
%
%The above programme was applied by Milner in~\cite{milner} to
%process behaviours and labelled transition systems. Milner
%introduced a set of expressions for finite LTS's and proved an
%analogue of Kleene's Theorem: each expression denotes the behaviour
%of a finite LTS and, conversely, the behaviour of a finite LTS can
%be specified by an expression (modulo bisimilarity).
%%Milner called these behaviours \emph{regular behaviours}. Moreover,
%Milner also provided an axiomatisation for bisimilarity over his expressions, with the
%property that two expressions are provably equivalent if and only if
%they are bisimilar.

For coalgebras of a large class of functors, a language of regular
expressions, a corresponding generalisation of Kleene's theorem, and a
sound and complete axiomatisation for the associated notion of
behavioural equivalence were  introduced in~\cite{brs_lmcs}.
Both the language of expressions and their axiomatisation were derived, in a modular fashion, from the functor defining the type of the system.

%Algebra and related tools can be successfully used for reasoning on properties of systems.

One of the contributions of the thesis consists in a decision procedure for bisimilarity of generalised regular expressions in~\cite{brs_lmcs}, implemented in the coinductive theorem prover {\CIRC}~\cite{goguen-lin-rosu-2000-ase,rosu-lucanu-2009-calco}.
More explicitly, we derived an encoding of generalised regular expressions and their coalgebraic structure into CIRC-compatible constructs, and implemented a tool allowing this translation automatically, hence enabling the automated reasoning on bisimilarity of non-deterministic coalgebras. 

%
%More explicitly, we exploited an encoding of coalgebra into algebra by providing algebraic specifications that model both the language and the coalgebraic structure of generalised regular expressions, and defined an equational deduction relation used on the algebraic side for reasoning on the bisimilarity of such expressions.
%
%Moreover, we implemented a tool allowing the automated translation of both the algebraic specifications and the deduction relation into CIRC-compatible constructs, hence enabling the automated reasoning on bisimilarity of non-deterministic coalgebras.

We further mention some of the existing coalgebraic based tools for proving bisimilarity and the main differences with our tool. CoCasl~\cite{DBLP:conf/fase/HausmannMS05} and CCSL~\cite{ccsl} are tools that can generate proof obligations for 
theorem provers from coalgebraic specifications. In~\cite{DBLP:conf/fase/HausmannMS05} several tactics for interactive and 
automatic bisimulation building are implemented in Isabelle/HOL and are used 
to derive bisimilarities for translated specifications from CoCasl. The main difference between our tool and CoCasl or CCSL is that, given a functor, the tool derives a  specification language for which equivalence is decidable (that is, it is automatic and not interactive). CIRC~\cite{goguen-lin-rosu-2000-ase,rosu-lucanu-2009-calco}, on top of which the current tool is built, is based on hidden logic~\cite{rosu-thesis} and uses a partial decision procedure for proving bisimilarities via implicit construction of bisimulations. Our tool can be seen as an extension of CIRC to a fully automatic theorem prover for the class of non-deterministic coalgebras. We stress the fact that the focus of our work is on a language for which equivalence is decidable. Tools such as CoCasl, CCSL or {\CIRC} have a more expressive language, where one can, for instance, specify streams, which in our language could not be specified (intuitively, the streams we can specify in our language are eventually periodic). In those tools decidability of equivalence can, however, not be guaranteed.

\medskip
\emph{On the other hand}, we exploit the coalgebraic framework in order to provide a uniform handling of a suite of semantics, other than bisimilarity. More explicitly, we are interested in deriving coalgebraic characterisations and algorithms suitable for implementation for: decorated trace semantics in the context of LTS's and GPS's as introduced in~\cite{Glabbeek01,GPS-Jou-Smolka}, and testing semantics for LTS's with internal behaviour as given in~\cite{CleavelandH89}.

In the recent past, some of the decorated trace semantics in van Glabbeek's spectrum
have been cast in the coalgebraic framework. Notably,
trace semantics of LTS's was widely studied~\cite{HJS07,DBLP:journals/entcs/LenisaPW00,fsttcs} and, more recently, (complete) trace, ready and failure
semantics were recovered in~\cite{gen-pow} via a coalgebraic generalisation of the classical powerset construction~\cite{CancilaHL03,Lenisa19992,fsttcs}.
A coalgebraic characterisation of the spectrum was also attempted in~\cite{Monteiro08}
%, in a somewhat {\em ad hoc} fashion.
%Connections with these works are still to be explored.

Since the introduction of process calculi, a lot of research has also been
devoted to the analysis of testing semantics~\cite{DBLP:journals/tcs/NicolaH84}. Intuitively, with respect to a fixed set of tests, two systems are deemed to be equivalent if they pass exactly the same tests.

In~\cite{CleavelandH89}, a trace-based alternative characterisation of may and must testing was given.
Based on this approach, we provide a coalgebraic modelling of the aforementioned semantics via the generalised powerset construction.
Another coalgebraic outlook on must testing is presented in~\cite{DBLP:journals/tcs/BorealeG06} which introduces a fully abstract coalgebraic semantics for CSP. The main difference with our work consists in the fact that~\cite{DBLP:journals/tcs/BorealeG06} builds a coalgebra from the syntactic terms of CSP, while here we build a coalgebra starting from LTS's.
As a further coalgebraic approach to testing, it is worth mentioning test-suites~\cite{DBLP:journals/entcs/Klin04}, which tackle the semantics in van Glabbeek's spectrum~\cite{Glabbeek01}, but not must testing.

The problem of automatically reasoning on decorated trace and testing semantics of LTS's is an interesting research topic per se. One possible approach, which is reminiscent of the determinisation of non-deterministic automata, consists in deriving deterministic-like systems for which checking bisimilarity coincides with reasoning on the aforementioned semantics in the original LTS's. Several bisimulation-based algorithms are implemented in tools such as the ones in~\cite{cleaveland1993concurrency,DBLP:conf/cav/CleavelandS96,DBLP:journals/topnoc/CalzolaiNLT08}. We also refer to the more recent work in~\cite{BONCHI:2012:HAL-00639716:4}, where the determinised automata are related based on bisimulations up-to~\cite{sangiorgi2011advanced,San98MFCS}. The advantage of this procedure is that, in most cases, building the bisimulations up-to requires visiting only portions of the automata.
The partial exploration is also the key feature of the antichain algorithm~\cite{CAV06} for reasoning on language equivalence of non-deterministic finite automata.

The best-known algorithm for minimising LTS's with respect to bisimilarity is the so-called partition refinement~\cite{Kanellakis:1983:CEF:800221.806724,Paige:87:SIAM}, which is analogous to Hopcroft's minimisation algorithm~\cite{minimisation} for deterministic automata with respect to language equivalence.
Last, but not least, we refer to Brzozowski's minimisation algorithm~\cite{Brzozowski}, which has been provided with a coalgebraic understanding in~\cite{Bonchi:2012:BAA:2340820.2340823}.

Along this line of research, in Chapter~\ref{ch:algorithms} we introduce an analogue of Brzozowski's algorithm and an algorithm based on bisimulations up-to for failure and must testing semantics.

\chapter{Preliminaries}

In this chapter we recall the basic definitions for sets and coalgebras that are needed in the rest of the thesis. We also introduce the coalgebraic modelling of the (generalised) powerset construction. We assume the reader is familiar with basic notions from category theory. We refer the interested reader to~\cite{Rutten00} and~\cite{awodey} for more information on coalgebras and category theory, respectively.

\section{Sets}
\label{prelim:sets}

Let \textbf{Set} denote the category of sets (represented by capital letters $X, Y, \ldots$)
and functions (represented by lower case letters $f, g, \ldots$).
We write $Y^X$ for the family of functions from $X$ to $Y$ and $\pow(X)$ for the collection of finite subsets of a set $X$.
The product of two sets $X, Y$ is written as $X \times Y$
and has the projections functions $\pi_1$ and $\pi_2$:
$X
\mathrel{{\overset{\pi_1}{\longleftarrow}}}
X\times Y
\mathrel{{\overset{\pi_2}{\longrightarrow}}}
Y$.
We define
\(
X\myplus Y = X\uplus Y \uplus \{\bot,\top\}
\)
where $\uplus$ is the disjoint
union of sets, with injections
$
X \mathrel{{\overset{\kappa_1}{\longrightarrow}}} X
\uplus
Y \mathrel{{\overset{\kappa_2}{\longleftarrow}}} Y$.
Note that the set $X\myplus Y$ is different from the classical
coproduct of $X$ and $Y$ (which we shall denote by $X+Y$), because of the
two extra elements $\bot$ and
$\top$.
These extra elements are used to represent, respectively,
underspecification and inconsistency in the specification of some
systems.

For each of the operations defined above on sets, there is an analogous
one on functions. For the sake of brevity, we first introduce the notation $i \in \overline{1,n}$ as a shorthand for $i \in \{1, \ldots, n\}$. Let $f\colon X\to Y$, $f_1\colon X\to Y$ and $f_2\colon Z\to W$. We
define the following operations:
\[
\begin{array}{l@{\hspace{1.0cm}}l}
f_1\times f_2 \colon X\times Z \to Y\times W & f_1\myplus f_2 \colon X\myplus
Z \to Y\myplus W\\[.6ex]
(f_1\times f_2)(x,z) = <f_1(x),f_2(z)> &
(f_1\myplus f_2)(c) = c,\ c\in\{\bot,\top\}\\[.6ex]
 &
(f_1\myplus f_2)(\kappa_i(x)) = \kappa_i(f_i(x)),\ i\in\overline{1,2}\\[.6ex]
f^A \colon X^A \to Y^A  & \pow(f) \colon \pow(X) \to \pow(Y)\\
f^A (g) = f\circ g & \pow(f)(X_1) = \{y \in Y \mid f(x) = y, x\in X_1\}
\end{array}
\]

Note that in the definition above we are using the same symbols introduced for
the
operations on sets. It will always be clear from the context which
operation is being used.

\section{Coalgebras}
\label{prelim:coalg}

The examples handled throughout this thesis
live in the standard setting of sets and functions. We therefore define our formal frameworks for modelling and reasoning on behavioural equivalence of systems based on coalgebras of functors on $\mathbf{Set}$.

\begin{definition}[Coalgebra]
\label{def:coalgebra}
A \emph{coalgebra} is a pair $(S, f\colon
S\to \F(S))$, where $S$
is a set of states and $\F\colon \mathbf{Set}\to \mathbf{Set}$ is a functor.
\end{definition}

The functor $\F$, together with the
function $f$, determines the {\em transition structure} (or
dynamics) of the coalgebra~\cite{Rutten00}, also referred to as \emph{$\F$-coalgebra}.

A coalgebra $(S,f)$ is \emph{finite} if $S$ is a finite set.

\begin{definition}[Coalgebra homomorphism]
\label{def:coalg-hom}
A \emph{homomorphism} $h\colon(S,f) \rightarrow (T,g)$ from an $\F$-coalgebra $(S,f)$ to an $\F$-coalgebra $(T,g)$, is a function $h \colon S \rightarrow T$ making the following diagram commute:

\[
\xymatrix@C=1.2cm@R=.25cm{
S\ar[r]^{h}\ar[dd]_{f} & T\ar[dd]^{g} &\\
& & g \circ h = \F(h) \circ f\\
\F(S)\ar[r]_{\F(h)} & \F(T) &
}
\]
\end{definition}

\begin{definition}[Coalgebra isomorphism]
\label{def:coalg-iso}

A coalgebra homomorphism\linebreak $i \colon S \rightarrow T$ is a \emph{coalgebra isomorphism} if there exists a coalgebra homomorphism $j \colon T \rightarrow S$ such that $i\circ j={id}_{T}$ and $j\circ i={id}_{S}$.

\end{definition}

\begin{definition}[Final coalgebra]
\label{def:fin-coalg}
An $\F$-coalgebra $(\Omega, \omega)$ is \emph{final} if for any $\F$-coalgebra $(S,f)$ there exists a unique $\F$-coalgebra homomorphism 
\[\bb{-}\colon (S,f) \rightarrow (\Omega, \omega):\]

\[
\xymatrix@C=1.5cm@R=.25cm{
S\ar@{-->}[r]^{\bb{-}}\ar[dd]_{f} & \Omega\ar[dd]^{\omega} &\\
& & \omega \circ \bb{-} = \F(\bb{-}) \circ f\\
\F(S)\ar@{-->}[r]_{\F(\bb{-})} & \F(\Omega) &
}
\]
\end{definition}

Note that not all functors admit final coalgebras. However, it was shown in~\cite{Rutten00} that such coalgebras exist for the class of bounded functors~\cite{DBLP:journals/mscs/GummS02}. A functor $\F$ is \emph{bounded} if there are sets $B$ and $A$ and a surjective natural transformation from $B \times (-)^{A}$ to $\F$ (Theorem 4.7 in ~\cite{DBLP:journals/mscs/GummS02}). Moreover, final coalgebras, if they exist, are unique up to isomorphism.

Intuitively, a final $\F$-coalgebra $(\Omega, \omega)$ represents the universe of all possible \emph{behaviours} of $\F$-coalgebras $(S,f)$. The unique homomorphism $\bb{-}$ maps each element of $S$ to its behaviour. Using this mapping, behavioural equivalence can be defined as follows.

\begin{definition}[Behavioural equivalence]
\label{def:beh-equiv}
Let $\F$ be a functor that admits final coalgebras. For any two $\F$-coalgebras $(S,f)$ and $(T,g)$, $s \in S$ and $t \in T$ are \emph{behaviourally equivalent}, written $s \sim_{\F} t$, if and only if they have the same behaviour, that is:
\begin{equation}
\label{eq:beh-equiv}
s \sim_{\F} t \textnormal{  iff  } \bb{s} = \bb{t}.
\end{equation}
\end{definition}

Coalgebras provide a useful technique for proving behavioural equivalence, namely, bisimulation~\cite{DBLP:conf/ctcs/AczelM89}.

\begin{definition}[Bisimulation]
\label{def:bisimulation}
Let $(S,f)$ and $(T,g)$ be two $\F$-coalgebras. A relation $R \subseteq S \times T$ is a \emph{bisimulation} if there exists a map $\alpha \colon R \rightarrow \F(R)$ such that the projections $\pi_{1}\colon R \rightarrow S$ and $\pi_{2}\colon R \rightarrow T$ are coalgebra homomorphisms, \emph{i.e.}, they make the following diagram commute:
\[
\xymatrix@C=1.2cm@R=.8cm{
S\ar[d]_{f} & R\ar[d]_{ \alpha}\ar[l]_{\pi_{1}}\ar[r]^{\pi_{2}} & T\ar[d]^{g}\\
\F(S) & \F(R)\ar[l]^{\F(\pi_{1})}\ar[r]_{\F(\pi_{2})} & \F(T)
}
\]
\end{definition}

The following alternative definition of bisimulation, sometimes more appropriate for the proofs, was given in~\cite{DBLP:journals/iandc/HermidaJ98}:
a relation
$R \subseteq S\times T$ is a bisimulation
if and only if
\[
(s,t)\in R \Rightarrow (f(s), g(t))\in \overline \F(R)
\]
where $\overline \F(R)$ is defined as
\begin{equation}
\label{eq:bisim-alt-def}
\overline \F(R) = \{ (\F(\pi_1)(x),\F(\pi_2)(x)) \mid x \in \F(R) \}
\end{equation}

If two states are bisimilar, and a final coalgebra exists, then they are behaviourally equivalent.

In \cite{Rutten00}, it was shown that under certain conditions on $\F$ (which are met by all the functors considered in this thesis), bisimulations
are a \emph{sound and complete proof technique} for behavioural equivalence. Namely, by \textbf{coinduction} it holds that:
\begin{equation}
\label{eq:bisim-behEquiv}
s \sim_\F t\text{ iff there exists a bisimulation } R \text{ such that } s\,R\,t.
\end{equation}

For simplicity, we abuse the notation and write $s \sim_{\F} t$ whenever there exists a bisimulation relation containing $(s,t)$, and we call $\sim_{\F}$ the bisimilarity relation.

Note that different functors $\F$ induce different notions of behavioural equivalence. For the case of streams, deterministic automata, and finite labelled transition systems, for example, behavioural equivalence corresponds to stream equality, language equivalence and the standard notion of bisimilarity by Milner and Park~\cite{Park81,milner89}, respectively.

For more insight on the coalgebraic framework introduced in this section, we further provide the coalgebraic modelling of deterministic and Moore automata (extensively used in Chapter~\ref{ch:dec-trace-testing} and Chapter~\ref{ch:algorithms}).

\begin{example}
\label{eg:det-aut}
A \emph{deterministic automaton} (DA) is a pair
$(X,<o,t>)$, where $X$ is a (possibly infinite) set of states and $<o,t> \colon X \to
2\times X^A$ is a function with two components: $o$, the output
function, determines if a state $x$ is final ($o(x) = 1$) or not
($o(x) = 0$); and $t$, the transition function, returns for each letter $a$ in the input alphabet $A$ the next state. Note that here $2$ stands for the set with two elements $\set{0,1}$.

DA's are coalgebras for the functor
$\D(X) = 2\times X^A$.
The final coalgebra of this functor is
$(2^{A^*},<\epsilon, (-)_a>)$ where $2^{A^*}$ is the set of
languages over $A$ and $<\epsilon, (-)_a>$, given a language $L$,
determines whether or not the empty word $\eps$ is in the language
($\epsilon(L) =1$ or $\epsilon(L)=0$, respectively) and, for each input
letter $a$, returns the {\em derivative} of $L$: $L_a = \{ w \in A^*
\mid aw\in L\}$.

From any DA, there is a unique map $\bb{-}$ into $2^{A^*}$ which assigns
to each state its behaviour (that is, the language that the state
recognises)~\cite{DBLP:conf/concur/Rutten98}.
%\vspace*{-.4cm}
\[
\xymatrix@C=2cm@R=.5cm{X \ar@{-->}[r]^{\llbracket-\rrbracket}\ar[d]_{<o,t>} &
2^{A^*}\ar[d]^{<\epsilon, (-)_a>}\\
2\times X^A\ar@{-->}[r]_-{{id}\times \llbracket-\rrbracket^A} & 2\times (2^{A^*})^A }
\quad \quad \raisebox{-0.5cm}{$
\begin{array}{l}
\llbracket x\rrbracket(\eps) = o(x)\\
\llbracket x\rrbracket(aw) = \llbracket t(x)(a) \rrbracket(w)\\
\end{array}$}
\]

Behavioural equivalence for the functor $\D$ coincides with the classical language equivalence of automata: given a deterministic automaton $(S,\langle o,t\rangle)$, two states $x,y \in S$ are said to be \emph{language equivalent} if and only if they accept the same language.
\end{example}

We further provide the coalgebraic modelling of Moore automata -- a generalisation of DA's -- which, as we shall later see, will enable shifting from language equivalence to the context of decorated trace semantics.

\begin{example}
\label{eg:Moore}
\emph{Moore automata} with inputs in $A$
and outputs in $B$ are coalgebras for the functor
$\M(X) = B \times X^A$, that is pairs $(X,<o,t>)$ where $X$ is a set, $t\colon X \to X^A$ is the transition function (like for DA) and
$o\colon X \to B$ is the output function which maps every state to its output. Thus DA can be seen as a special case of Moore automata where
$B=2$.

The final coalgebra for $\M$ is
$(B^{A^*},<\epsilon, (-)_a>)$ where $B^{A^*}$ is the set of all functions $\varphi \colon A^* \to B$,
$\epsilon\colon B^{A^*}\to B$ maps each $\varphi$ into $\varphi(\varepsilon)$
and $(-)_a\colon B^{A^*} \to (B^{A^*})^A$ is defined for all $\varphi\in B^{A^*}$, $a\in A$ and $w\in A^*$ as $(\varphi)_a (w)= \varphi(aw)$.
\[
\xymatrix@C=2cm@R=.5cm{X \ar@{-->}[r]^{\llbracket-\rrbracket}\ar[d]_{<o,t>} &
B^{A^*}\ar[d]^{<\epsilon, (-)_a>}\\
B\times X^A\ar@{-->}[r]_-{{id}\times \llbracket-\rrbracket^A} & B\times (B^{A^*})^A }
\quad \quad \raisebox{-0.5cm}{$
\begin{array}{l}
\llbracket x\rrbracket(\eps) = o(x)\\
\llbracket x\rrbracket(aw) = \llbracket t(x)(a) \rrbracket(w)\\
\end{array}$}
\]

Hence, reasoning on behavioural equivalence of Moore automata reduces to checking equality of functions.
\end{example}

\section{The generalised powerset construction}
\label{prelim:gen-pow}

Sometimes, it is interesting to consider other equivalences than $\sim_{\F}$ for reasoning about $\F$-coalgebras. This is the case for non-deterministic automata (NDA's), for which language equivalence is often the intended semantics, instead of bisimilarity.
NDA's are coalgebras for the functor $\N(X) =2 \times (\Powf(X))^A$, where $\Powf$ stands for the finite powerset, and bisimilarity, which we denote by $\sim_\N$, is strictly included in language equivalence.
This can be achieved by applying the classical \emph{powerset construction}~\cite{Rabin:1959:FAD:1661907.1661909} for determinising non-deterministic automata, which can be briefly summarised as follows.

Consider an NDA, which is a coalgebra
\[
(X, \langle o, t\rangle \ct X \rightarrow 2 \times (\pow X)^{A})
\]
where (similarly to the case of DA's in Example~\ref{eg:det-aut}): $o$ is the output
function and determines if a state $x$ is final ($o(x) = 1$) or not
($o(x) = 0$), $t$ is the transition function returning for each
input letter $a$ the set of next states, and $2$ stands for the set with two elements $\set{0,1}$.

The powerset construction derives a DA
\[
(\Powf X, \langle o^{\sharp}, t^{\sharp} \rangle \ct \pow X \rightarrow 2 \times (\pow X)^{A}),\]
by associating to each state $x \in X$ of the NDA, a state $\set{x} \in \Powf X$. The new output and transition functions are:
\begin{equation}
\label{eq:op-DA}
\begin{array}{rcl}
o^{\sharp}(Y) & = & \bigsqcup\limits_{y \in Y} o(y)\\
t^{\sharp}(Y)(a) & = & \bigsqcup\limits_{y\in Y} t(y)(a)
\end{array}
\end{equation}
where $\bigsqcup$ is used to represent both the ``Boolean or'' and the set union. Intuitively, $\bigsqcup$ stands for the join operation corresponding to the semilattice with carrier $\set{0,1}$, and the one with carrier $\Powf S$, with $S\in {\bf Set}$, respectively.
The final coalgebra of the DA is the set of languages $2^{A^{*}}$ over $A$, and the semantic map
\[
\bb{-} \ct \Powf X \rightarrow 2^{A^{*}}
\]
associates to each $\set{x}$ the language $\bb{\set{x}}$ accepted by $x$, and is defined as introduced in Example~\ref{eg:det-aut} in the previous section.
Consequently, reasoning on language equivalence of two states $x_{1}$ and $x_{2}$ of an NDA reduces to identifying a bisimulation $R$ relating $\set{x_{1}}$ and $\set{x_{2}}$ in the corresponding DA:
\begin{equation}
\label{eq:lang-equiv-powerset}
\bb{\set{x_{1}}} = \bb{\set{x_{2}}} \textnormal{ iff } \set{x_{1}}\,\, R \,\, \set{x_{2}}.
\end{equation}
Based on these observations, we refer to the \emph{generalised powerset construction}~\cite{CancilaHL03,Lenisa19992,fsttcs} for coalgebras $f\colon X \to \F T(X)$ for a functor $\F$ and a
monad $T$.
Intuitively, this construction applies to the context of NDA's by simply instantiating $T$ with $\Powf$ and $\F$ with $2\times (-)^{A}$ (more details are provided later on in this section, in Example~\ref{ex:pow-constr-NDA}).

Monads are used to encompass computational effects such as non-determinism ($T(X) = \pow(X)$) or partiality ($T(X) = 1 + X$, where $1 = \{*\}$ stands for termination). They come equipped with two operations: unit ($\eta$) and multiplication ($\mu$). Intuitively, $\eta$  enables the embedding of any value into the monad structure, whereas $\mu$ allows to collapse several levels of computational effects.
For instance, the unit and multiplication of the powerset monad $T = (\pow, \eta, \mu)$ are defined as follows:
\begin{equation}
\label{eq:pow-unit-mult}
\begin{array}{l@{\hspace{1.5cm}}l}
\eta_{X}\colon X \rightarrow \pow X & \mu_{X}\colon \pow(\pow X) \rightarrow \pow X\\
\eta_{X}(x) = \{x\} & \mu_{X}(U) = \bigcup\limits_{\substack{S \in U}} S.
\end{array}
\end{equation}

We further provide an overview of the notions of a monad and algebras of a monad, and a series of intuitions for their integration into the context of the generalised powerset construction.

First recall that, given two functors $\F$ and $\Gf$ on $\textbf{Set}$,
a \emph{natural transformation} $\lambda \colon \F \Rightarrow \Gf$ is a family of functions $\lambda_{X} \colon \F(X) \rightarrow \Gf(X)$ such that, for all functions $f \colon X \rightarrow Y$, the following holds:
\[
\lambda_{Y} \circ \F(f) = \Gf(f) \circ \lambda_{X}.
\]

\begin{definition}[Monad]
\label{def:monad}
Let $T$ be a functor on $\textbf{Set}$. A \emph{monad} is a triple $(T, \eta, \mu)$ where $\eta \colon Id \Rightarrow T$ and $\mu \colon T^{2} \Rightarrow T$ are two natural transformations, called \emph{unit} and \emph{multiplication}, respectively, such that the following diagrams commute:

\[
\xymatrix@C=1.5cm@R=.7cm{
T\ar@{=>}[r]^{T\eta}\ar@{=>}[dr]_{id} & T^{2}\ar@{=>}[d]_{\mu} & T\ar@{=>}[l]_{\eta T}\ar@{=>}[dl]^{id} & T^{3}\ar@{=>}[r]^{T\mu}\ar@{=>}[d]_{\mu T} & T^{2}\ar@{=>}[d]^{\mu}\\
& T & & T^{2}\ar@{=>}[r]_{\mu} & T
}
\]
\end{definition}

\begin{definition}[Algebra of a monad]
\label{def:alg-of-monad}
An \emph{algebra} of a monad $(T, \eta, \mu)$, or a \emph{$T$-algebra}, is a pair $(X, h\colon T(X) \rightarrow X)$ satisfying the laws 
\[
\begin{array}{lr}
h \circ \eta = id & h \circ \mu = h \circ T h.
\end{array}
\]
\end{definition}

Intuitively, these laws show how to eliminate the computational effects by propagating the operation $h$ throughout the monadic structure.

For the case of the powerset monad defined in~(\ref{eq:pow-unit-mult}), for example, $T$-algebras are semilattices (with bottom).
Consider a join semilattice $(S, \bigsqcup)$ with $0$ the least element.
Showing that $S$ carries an algebra structure consists in proving that there exists 
$h \ct \Powf(S) \rightarrow S$
satisfying the laws in Definition~\ref{def:alg-of-monad}.
It is easy to check that by taking
\[
h(U) = \bigsqcup\limits_{u \in U} u,
\]
with $U \subseteq S$, we get the appropriate map.

The proof is as follows. Consider $u \in S$ and $\Psi \subseteq \Powf(\Powf S)$. Then:
\begin{align*}
(h \circ \eta)(u) & = h(\set{u})\\
&= u\\[0.5ex]
(h \circ \mu)(\Psi) &= h(\mu(\Psi))\\
&= h(\bigcup\limits_{U_{i} \in \Psi} U_{i})\\
&=\bigsqcup\limits_{\footnotesize\begin{array}{c}U_{i} \in \Psi\\ u_{j} \in U_{i} \end{array}} u_{j}\\[0.5ex]
%& = \begin{cases}
%h(\emptyset) & \text{ if } \Psi = \emptyset \text{ or } \Psi = \set{\emptyset}\\
%\bigsqcup\limits_{\footnotesize\begin{array}{c}S_{i} \in \Psi\\ s_{j} \in S_{i} \end{array}} s_{j} & \text{ otherwise}
%\end{cases}\\
%& = \begin{cases}
%0 & \text{ if } \Psi = \emptyset \text{ or } \Psi = \set{\emptyset}\\
%\bigsqcup\limits_{\footnotesize\begin{array}{c}S_{i} \in \Psi\\ s_{j} \in S_{i} \end{array}} s_{j} & \text{ otherwise}
%\end{cases}\\[0.5ex]
(h\circ \Powf h)(\Psi) & = h(\Powf h(\Psi))\\
&= h(\set{ h(U_{i}) \mid U_{i} \in \Psi})\\
&= h(\set{ \bigsqcup\limits_{u_{j} \in U_{i}} u_{j} \mid U_{i} \in \Psi})\\
&=\bigsqcup\limits_{\footnotesize\begin{array}{c}U_{i} \in \Psi\\ u_{j} \in U_{i} \end{array}} u_{j}
%& = \begin{cases}
%h(\emptyset) & \text{ if } \Psi = \emptyset\\
%h(\set{ h(S_{i}) \mid S_{i} \in \Psi}) & \text{ otherwise}
%\end{cases}\\
%& = \begin{cases}
%0 & \text{ if } \Psi = \emptyset\\
%h(\set{0}) & \text{ if } \Psi = \set{\emptyset}\\
%h(\set{ \bigsqcup\limits_{s_{j} \in S_{i}} s_{j} \mid S_{i} \in \Psi}) & \text{ otherwise}
%\end{cases}\\
%& = \begin{cases}
%0 & \text{ if } \Psi = \emptyset \text{ or } \Psi = \set{\emptyset}\\
%\bigsqcup\limits_{\footnotesize\begin{array}{c}S_{i} \in \Psi\\ s_{j} \in S_{i} \end{array}} s_{j} & \text{ otherwise}
%\end{cases}
\end{align*}

The first set of equalities is associated to the law $h \circ \eta = id$ in Definition~\ref{def:alg-of-monad} and, intuitively, states that eliminating the non-determinism from a singleton set $\set{u}$ consists in simply considering the value $u$.
The last two sets of equalities correspond to the law $h \circ \mu = h \circ T h$. Intuitively, they show that eliminating two levels of non-determinism captured within a set $\Psi \subseteq \Powf(\Powf S)$ can be performed in two different ways: (a) first flatten $\Psi$ and then return the join of the elements in the resulted set, or (b) first compute the joins $h(U_{i})$ of the elements of sets $U_{i} \in \Psi$ and then return the join of all such $h(U_{i})$'s.

\begin{definition}[Algebra homomorphism]
Let $(T, \eta, \mu)$ be a monad. A function\linebreak $f \colon X \rightarrow Y$ is a \emph{homomorphism} between two $T$-algebras $(X, h \colon T(X) \rightarrow X)$ and $(Y, g \colon T(Y) \rightarrow Y)$ if it makes the following diagram commute:
\[
\xymatrix@C=1cm@R=.6cm{
T(X)\ar[r]^{T(f)}\ar[d]_h & T(Y)\ar[d]^g\\
X\ar[r]_{f} & Y 
}
\]
\end{definition}

These are the key ingredients exploited in~\cite{gen-pow} in order to derive the {\bf generalised powerset construction} for coalgebras $f\ct X \rightarrow \F T(X)$ for a functor $\F$ and a monad $T$, with the proviso that $\F T(X)$ is a $T$-algebra, and $\F$ has a final coalgebra 
$(\Omega, \omega)$, as summarised in the following commuting diagram:
\begin{eqnarray}
\label{F-final}
\xymatrix@C=1.5cm@R=.7cm{
X \ar[d]_{f}\ar[r]^-{\eta} & T(X)\ar[dl]^-{f^\sharp}\ar@{-->}[r]^{\bb{-}} &  \Omega \ar[d]^{\omega}\\
\F T(X)\ar@{-->}[rr]_-{\F (\bb{-})}  && \F(\Omega)
}
\end{eqnarray}
We refer the interested reader to~\cite{gen-pow} where all the technical details are explored and many instances of the construction are shown.

At an intuitive level, the coalgebra $f\colon X \rightarrow \F T(X)$ is extended to $f^{\sharp} \colon T(X) \rightarrow \F T(X)$ which, for two elements $x_{1}, x_{2} \in X$, enables checking their ``$\F$-equivalence with respect to the monad $T$'' ($\eta(x_{1}) \sim_{\F} \eta(x_{2})$) rather than checking their $\F T$-equivalence.
Formally, assuming that $\F T(X)$ is a $T$-algebra, $f^\sharp$ is the unique algebra map between $(T(X),\mu)$ and $(\F T(X), h)$ (where $h$ is a given algebra structure on $\F T(X)$) such that 
\[ f^\sharp = h \circ Tf.
\]
%Moreover, one can show that, under certain additional conditions, also $\Omega$ (the final coalgebra of the functor $\F$) has an algebra structure and that $\bb{-}$ is also an algebra map~\cite{gen-pow}.

\begin{remark}
\label{rem:bisim-up-to}
Based on~(\ref{eq:beh-equiv}) and~(\ref{eq:bisim-behEquiv}), verifying $\F$-behavioural equivalence of two states $x_1, x_2$ in a coalgebra $(T(X), f^{\sharp})$ consists in identifying a bisimulation $R$ relating $\eta(x_{1})$ and $\eta(x_{2})$:
\begin{equation}
\label{eq:s-c-upto}
\llbracket \eta(x_{1}) \rrbracket = \llbracket \eta(x_{2}) \rrbracket \textnormal{  iff  } \eta(x_{1})\,\, R\,\, \eta(x_{2}).
\end{equation}
\end{remark}

\begin{example}
\label{ex:pow-constr-NDA}
Consider again the case of NDA's which are coalgebras
\[
(X, \langle o, t\rangle \ct X \rightarrow 2 \times (\pow X)^{A}),
\]
as introduced in the beginning of this section.
Observe that $\Pow (X)$ and $2 \cong \Pow(1)$ are (join) semilattices, which are algebras of the powerset monad (here $1$ stands for the singleton set $\{*\}$). Moreover, product and exponentiation preserve the algebra structure, hence guaranteeing that $2 \times (\Powf (X))^{A}$ is an algebra for $\Powf$ as well.

At this point it is easy to see that the generalised powerset construction applies to the context of NDA's by simply instantiating $T$ with $\Powf$ and $\F$ with $2\times (-)^{A}$.
It follows that the operation $\langle o,t\rangle$ of the NDA can be uniquely extended to $\langle o^{\sharp}, t^{\sharp}\rangle$, as in~(\ref{eq:op-DA}), in a deterministic setting.
The language recognised by a non-deterministic state $x$ can be defined by precomposing the unique morphism
$\llbracket - \rrbracket \colon \Powf X \to 2^{A^*}$ with the unit $\eta$ of $\Powf$. This enables reasoning on language equivalence of states of NDA's in terms of bisimulations, as in~(\ref{eq:lang-equiv-powerset}).
Recall that for the case of deterministic LTS's, language equivalence and bisimilarity coincide~\cite{DBLP:journals/tcs/Engelfriet85}.

As a last aspect, note that the set of languages $2^{A^*}$ can be provided a join semilattice structure by considering the union of languages as the binary operation and the empty language as the least element. It can be easily shown (by induction on words $w \in A^{*}$) that the semantic map $\bb{-}$ is a join semilattice homomorphism (or, equivalently, a $\Powf$-algebra homomorphism).

More generally, the semantic map $\bb{-}$ in~(\ref{F-final}) is a $T$-algebra homomorphism whenever there exists a distributive law $T \F \Rightarrow \F T$, which guarantees that the carrier $\Omega$ of the final coalgebra is a $T$-algebra as well (see Proposition 4 in~\cite{DBLP:conf/cmcs/JacobsSS12}).

\end{example}

%%%%%%%%%%%%%%%% begin SCP %%%%%%%%%%%%%%%%%%%%%%%%%
\chapter{Deciding bisimilarity}
\label{ch:dec-bisim}

%\newcounter{counter1}
%\setcounter{counter1}{1}
%\newcounter{counter2}
%\setcounter{counter2}{2}
%\newcounter{counter3}
%\setcounter{counter3}{3}

The results in this chapter are based on the work in~\cite{brs_lmcs}, where a language of regular expressions for specifying a large class of systems that can be modelled as non-deterministic coalgebras, and a sound and complete axiomatisation for the corresponding notions of behavioural equivalence were introduced.

Our contribution consists in a novel method for checking bisimilarity of generalised regular expressions using the coinductive theorem prover {\CIRC} \cite{goguen-lin-rosu-2000-ase,rosu-lucanu-2009-calco}.
The main novelty of the method lies in the generality of the systems it can handle; examples include streams of real numbers, Mealy machines and labelled transition systems. More precisely, our approach deals with systems that can be represented as locally finite coalgebras or, equivalently, coalgebras for which the smallest subcoalgebra generated by a state is finite~\cite{Rutten00}.

%(Co)algebra and related tools can be successfully used for reasoning on properties of systems. In this chapter, we present a novel method for checking for bisimilarity of generalised regular expressions modelling non-deterministic coalgebras~\cite{brs_lmcs}, using the coinductive theorem prover {\CIRC} \cite{goguen-lin-rosu-2000-ase,rosu-lucanu-2009-calco}.
%
%The main novelty of the method lies on the generality of the systems it can handle; examples include streams of real numbers, Mealy machines and labelled transition systems. More precisely, our approach handles systems that can be represented as locally finite coalgebras or, equivalently, coalgebras for which the smallest subcoalgebra generated by a point is finite~\cite{Rutten00}.

{\CIRC} is a metalanguage application implemented in Maude~\cite{DBLP:conf/maude/2007}, and its target is to prove properties over infinite data structures. It has been successfully used for checking the equivalence of programs, and trace equivalence and strong bisimilarity of processes.
%As a Maude tool, {\CIRC} inherits notable features such as matching modulo associativity, commutativity, unity and idempotency and fast term rewriting.
The tool may be tested online and downloaded from:\\
\url{https://fmse.info.uaic.ro/tools/Circ/}.

Determining whether two expressions are equivalent is important in order to be able to compare behavioural specifications.  
In the presence of a sound and complete axiomatisation one can determine equivalence using algebraic reasoning. 
A coalgebraic perspective on regular expressions has however provided a more operational/algorithmic way of checking equivalence: one constructs a bisimulation relation containing both expressions. The advantage of the bisimulation approach is that it enables automation since the steps of the construction are fairly mechanic and require almost no ingenuity.
%
%It has been shown that both problems are in PSPACE~\cite%{K08a,worthington}, but in practice bisimulation checking tends to be feasible.
We illustrate this with an example, to give the reader the feeling of the more algorithmic nature of bisimulation. We want to stress however that we are not underestimating the value of an algebraic treatment of regular expressions: on the contrary, as we will show later, the axiomatisation plays an important role in guaranteeing termination of the bisimulation construction and is therefore crucial for the main result of this chapter.

We show below a proof of the sliding rule: $a(ba)^* \equiv (ab)^*a$. The algebraic proof, using the rules and equations of Kleene algebra, needs to show the two containments
\[
{a(ba)^* \leq (ab)^*a} \qquad \text{ and } \qquad (ab)^*a \leq a(ba)^*
\]
and it requires some ingenuity in the choice of the equation applied in each step. We show the proof for the first inequality, the other follow a similar proof pattern.
\[
\begin{array}{@{}lcl@{\qquad}l}
&&a(ba)^* \leq (ab)^*a \\
&\Leftarrow& a + (ab)^*a (ba) \leq (ab)^*a &\text{right-star rule~\cite{kozen91}: $\begin{array}{ll}b + xa \leq x \Rightarrow\\ ba^{*} \leq x\end{array}$}\\
&\iff & (1+ (ab)^*ab)a \leq (ab)^*a &\text{associativity and distributivity}\\
&\iff & (ab)^*a \leq (ab)^*a &\text{right expansion rule: $1+r^*r = r^*$}
\end{array}
\]

For the coalgebraic proof, we build incrementally, and rather mechanically, a bisimulation relation containing the pair  $(a(ba)^*, (ab)^*a)$.  We start with the pair we want to prove equivalent and then we close the relation  with respect to syntactic language derivatives, also known as  {\em  Brzozowski derivatives}. In the current example, the bisimulation relation would contain three pairs: 
\[
R = \{ (a(ba)^*, (ab)^*a), ((ba)^*, b(ab)^*a+1), (0,0)\}
\]
where $1$ and $0$ are, respectively, the regular expressions denoting the language containing only the empty word and the empty language. In constructing this relation, no decisions were made, and hence the suitability of bisimulation construction as an automatic technique to prove equivalence of regular expressions. 

The main contributions of this chapter can be summarised as follows. We present a decision procedure to determine equivalence of generalised regular expressions, which specify behaviours of many types of transition systems, including Mealy machines, labelled transition systems and infinite streams.
%The valid expressions for each system are type-checked automatically in the tool.
We illustrate the decision procedure we devised by applying it to several examples. As a vehicle of implementation, we choose \CIRC, a coinductive theorem prover which has already been explored for the construction of bisimulations. To ease the implementation in \CIRC,  we present the algebraic specifications' counterpart of the coalgebraic framework of the generalised regular expressions mentioned above.
This enables us to automatically derive algebraic specifications that model the language of expressions, and to define an appropriate equational entailment relation which mimics our decision procedure for checking behavioural equivalence of expressions. The implementation of both the algebraic specification and the entailment relation in {\CIRC} allows for automatic reasoning on the equivalence of expressions.

\textit{Organisation of the chapter.} Section~\ref{SCP-sec:prelim} recalls the basic definitions of the language associated with a non-deterministic functor. Section~\ref{sec:dp} describes the decision procedure to check equivalence of regular expressions. Section~\ref{sec:algSpec} formulates the aforementioned language as an algebraic specification, which paves the way to implement in {\CIRC} the procedure to decide equivalence of expressions. The implementation of the decision procedure and its soundness are described in Section~\ref{sec:dec-proced}.
In Section~\ref{sec:caseStudy} we show, by means of several examples,
how one can check bisimilarity, using {\CIRC}.
In Section~\ref{SCP-sec:concl} we briefly wrap up the contributions of this chapter.

\section{generalised regular expressions}
\label{SCP-sec:prelim}

In this section we briefly recall the basic definitions in
\cite{brs_lmcs}.

\emph{Non-deterministic functors} are functors $\Gf \colon {\mathbf{Set}} \rightarrow {\mathbf{Set}}$
built inductively from the identity,
and constants, using $\times$,
$\myplus$, $(-)^A$ and $\pow$:

\begin{equation}\label{eq:fun-gram}
\pf\ni \Gf \,::\!=\, \id \mid \Bl \mid \Gf\myplus \Gf \mid \Gf\times
\Gf \mid \Gf^A \mid \pow \Gf
\end{equation}
where $\Bl$ is a finite join-semilattice and $A$ is a finite set.
Typical examples of such functors include $\Rf=\Bl\times \id$, $\M =
(\Bl\times \id)^A$,
$\D = 2 \times \id^A$, $\Pa = (1 \myplus \id)^A$,  $\N = 2 \times \pow(\id)^A$ and $\Lf = 1 \myplus \pow(\id)^A$. These
functors represent, respectively, the type of streams, Mealy,
deterministic, partial deterministic automata, non-deterministic automata and labelled transition systems with explicit termination. $\Rf$-bisimulation is stream equality, whereas $\D$-bisimulation coincides with language equivalence.

%\begin{remark}
%\ans{C.3.7}
%As stated in~\cite{brs_lmcs}, the use of join-semilattices for constant functors and the sum $\myplus$ instead of the ordinary product enabled the use of underspecification and inconsistency ({\emph i.e.}, $\bottom$ and $\top$, respectively) in the specification of systems, and moreover, has allowed the whole framework to be studied in the category ${\mathbf{Set}}$. Even though underspecification and inconsistency can be captured by a semilattice structure, and the axiomatisation provides the set of expressions with a join-semilattice structure (therefore allowing the work directly in the category of join-semilattices), remaining in the category ${\mathbf{Set}}$ was chosen for simplicity.
%\end{remark}

Next, we give the definition of  the ingredient relation, which
relates a non-deterministic functor $\Gf$ with its {\em ingredients}, {\em
i.e.}, the functors used in its inductive construction. We shall use
this relation later for typing our expressions.

\begin{definition}\label{def:ingred}

Let $\lhd\subseteq \pf\times \pf$ be the least reflexive and transitive relation on
non-deterministic functors
such that
\[
\begin{array}{lll}
\Gf_1\triangleleft \Gf_1\times \Gf_2 & \Gf_2\lhd \Gf_1\times \Gf_2 & \Gf_1\lhd \Gf_1\myplus \Gf_2\\
\Gf_2\lhd \Gf_1\myplus \Gf_2 &
\Gf\lhd \Gf^A & \Gf \lhd \pow\Gf.
\end{array}
\]
%
%\[
%\Gf_1\triangleleft \Gf_1\times \Gf_2,\ \ \ %\\
%\Gf_2\lhd \Gf_1\times \Gf_2,\ \ \ %\\
%\Gf_1\lhd \Gf_1\myplus \Gf_2,\ \ \ %\\
%\Gf_2\lhd \Gf_1\myplus \Gf_2,\ \ \ %\\
%\Gf\lhd \Gf^A, \ \ \ \Gf \lhd \pow\Gf.
%\]
\end{definition}
Throughout this chapter we use $\F \lhd \Gf$ as a shorthand
for $(\F,\Gf)\in \lhd$. If $\F \lhd \Gf$, then $\F$ is said to  be an
\emph{ingredient} of $\Gf$. For example, $2$, $\id$, $\id^A$ and $\D$
itself are all the ingredients of the deterministic automata functor
$\D$.

A language of expressions $\Exp_\Gf$ is associated with each non-deterministic functor $\Gf$:

\begin{definition}[Expressions]\label{def:expr}
Let $A$ be a finite set, $\Bl$ a finite join-semilattice and $X$ a set
of fixed-point variables. The set $\Exp$ of all {\em (generalised regular) expressions\/} is given
by the following grammar, where $a\in A$, $b\in \Bl$ and $x\in X$:
\begin{equation}\label{eq:grammar}
\begin{array}{lcl}
%\E &::\!=& \emp \mid x \mid \E \oplus \E \mid \mu x.\gamma
%   \mid b \mid l<\E> \mid r<\E> \mid l[\E] \mid r[\E] \mid a(\E) \mid \%{\E\}\\
\E &::\!=&  x \mid \E \oplus \E \mid \gamma\\
\end{array}
\end{equation}
where $\gamma$ is a {\em guarded expression} given by:
\begin{equation}\label{eq:grammar2}
\begin{array}{lcl}
\gamma &::\!=& \emp \mid \gamma \oplus \gamma \mid \mu
x.\gamma
   \mid b \mid l<\E> \mid r<\E> \mid l[\E] \mid r[\E] \mid a(\E) \mid \{\E\} \\
\end{array}
\end{equation}
\end{definition}
In the expression $\mu x.\gamma$, $\mu$ is a binder for
all the free occurrences of $x$ in $\gamma$. Variables that are not bound are
free. A {\em closed expression} is an expression
without free occurrences of fixed-point variables $x$. We denote the set of closed
expressions by $\Exp^c$.

The language of expressions for non-deterministic coalgebras
is a generalisation of the classical notion of regular expressions:
$\emp$, $\eps_1 \oplus \eps_2$ and $\mu x.\gamma$
play similar roles to the regular expressions denoting empty language, the union of languages and the
Kleene star.
Moreover, note that, not unexpectedly, in~\cite{brs_lmcs}, $\oplus$ was axiomatised as an associative, commutative and idempotent operator, with $\emp$ as a neutral element. 
The expressions $l\langle\eps\rangle$, $r\langle\eps\rangle$,
$l[\eps]$, $r[\eps]$, $a(\eps)$ and $\{\E \}$ specify the left and right-hand side
of products and sums, function application and singleton sets, respectively.

Next we present a type assignment system for associating
expressions to non-deterministic functors. This will allow us to associate with each functor
$\Gf$ the expressions $\E\in \Exp^c$ that are valid specifications of
$\Gf$-coalgebras.

\begin{definition}[Type system]\label{def:ts}
We define a typing relation $\mathbin\vdash\subseteq \Exp \times \pf
\times \pf $ that will associate an expression $\E$ with two non-deterministic functors $\F$ and $\Gf$, which are related by the
ingredient relation ($\F$ is an ingredient of $\Gf$). We shall write
$\vdash \E\colon \F\lhd \Gf$ (read ``$\E$ is of type $\F\lhd \Gf$'') for $(\E,\F,\Gf) \in \;\vdash$.  The rules that define $\vdash$ are the following:
\[
\renewcommand{\arraystretch}{0.5}
\begin{array}{@{}ccc@{}}
\rules{}{\vdash \emp \colon \F\lhd \Gf }&
\rules{}{\vdash b\colon \Bl\lhd \Gf}\,\, (b \in \Bl)&
\rules{}{\vdash x \colon \Gf\lhd \Gf }\,\, (x \in X)\\\\
\rules{\vdash\E\colon \Gf\lhd \Gf}
     {\vdash \mu x.\E \colon \Gf\lhd \Gf}&
\rules{\vdash \E_1 \colon \F\lhd \Gf\;\;\;\; \vdash\E_2\colon \F\lhd \Gf}{\vdash \E_1\oplus\E_2 \colon \F\lhd \Gf} &
\rules{\vdash \E \colon \Gf\lhd \Gf}
     {\vdash \E \colon \id\lhd \Gf}\\\\
     \rules{\vdash \E\colon \F_2\lhd \Gf}
     {\vdash r[\E] \colon \F_1\myplus \F_2\lhd \Gf}
     &
\rules{\vdash \E\colon \F\lhd \Gf}
     {\vdash a(\E) \colon \F^A\lhd \Gf}\,\, (a \in A)
&
\rules{\vdash \E\colon \F_1\lhd \Gf}
     {\vdash l<\E> \colon \F_1\times \F_2\lhd \Gf}\\\\
\rules{\vdash \E\colon \F_2\lhd \Gf}
     {\vdash r<\E> \colon \F_1\times \F_2\lhd \Gf}&
\rules{\vdash \E\colon \F_1\lhd \Gf}
     {\vdash l[\E] \colon \F_1\myplus \F_2\lhd \Gf} & \rules{\vdash \E\colon \F_1\lhd \Gf}{\vdash \{\E\} \colon \pow \F_1 \triangleleft \Gf
    }
\end{array}
\]
\end{definition}
We can now formally define the set of $\Gf$-expressions: well-typed
expressions associated with a non-deterministic functor $\Gf$.

\begin{definition}[$\Gf$-expressions]\label{def:g-expr}
Let $\Gf$ be a non-deterministic functor and $\F$ an ingredient of $\Gf$.
We define $\Exp_{\F\lhd \Gf}$ by:
\[
\Exp_{\F\lhd \Gf} = \{\E \in \Exp^c \mid\ \vdash\E \colon \F\lhd \Gf\}\,.
\]
We define the set $\Exp_\Gf$ of well-typed {\em
$\Gf$-expressions\/} by $\Exp_{\Gf\lhd \Gf}$.
\end{definition}

In \cite{brs_lmcs}, it was proved that
the set of $\Gf$-expressions for a given non-deterministic functor $\Gf$
has a coalgebraic structure:
\[
\delta_{\Gf} \colon {\Exp}_{\Gf} \to {\Gf}({\Exp}_{\Gf})
\]
More precisely, in \cite{brs_lmcs}, which we
refer to for the complete definition of $\delta_{\Gf}$, the authors defined a function
$
\delta_{\F \lhd \Gf} \colon {\Exp}_{\F\lhd \Gf} \to {\F}({\Exp}_{\Gf})
$
and then set $\delta_\Gf = \delta_{\Gf\lhd \Gf}$.

The coalgebraic structure on the set of expressions enabled the proof of a Kleene-like theorem:
\begin{theorem}[Theorems~3.12~and~3.14~in~\cite{brs_lmcs}]\label{thm:kleene}
Consider $\Gf$ a non-deterministic functor.
\begin{enumerate}
\item For any $\E\in \Exp_\Gf$, there exists a finite $\Gf$-coalgebra $(S,g)$  and $s\in S$ such that $\E\sim s$.
\item For every \ans{C.1.6, C.2.3} finite $\Gf$-coalgebra $(S,g)$  and $s\in S$ there exists an expression $\E_s\in \Exp_\Gf$ such that $\E_s\sim s$.
\end{enumerate}
\end{theorem}

In order to provide the reader with intuition regarding the notions presented above, we illustrate them with an example.
\begin{example}\label{ex:streams1}
Let us instantiate the definition of $\Gf$-expressions to the functor of
streams \ans{C.1.13} $\Rf = \Bl\times \id$ (the ingredients of this functor are $\Bl$, $\id$ and $\Rf$ itself).
Let $X$ be a set of
(recursion or) fixed-point variables. The set $\Exp_\Rf$ of {\em stream
expressions\/} is given by the set of closed, guarded expressions generated by the following BNF grammar. For $x \in X$: \ans{C.1.7}
\begin{equation}\label{eq:grammarStr}
\begin{array} {l@{\;}l@{\;}c@{\;}l}
\Exp_\Rf\ni &\E &::\!=&  \emp \mid \E \oplus \E \mid
        \mu x.\E \mid
x \mid l<\tau> \mid r<\E> \\
&\tau &::\!=& \emp \mid b \mid \tau\oplus\tau\\
\end{array}
\end{equation}
\end{example}
Intuitively, the expression $l\langle b \rangle$
is used to specify that the head of the stream is $b$, while
$r\langle \E \rangle$ specifies a stream whose tail behaves
as specified by $\E$.
For the two element join-semilattice $\Bl=\{0,1\}$ (with $\bot_\Bl=0$) examples of well-typed expressions include
$\emp$, $l<1>\oplus r<l<\emp>>$ and $\mu x. r<x> \oplus
l<1>$. The expressions $l[1]$, $l<1>\oplus 1$ and $\mu x. 1$ are examples of
non well-typed
expressions for $\Rf$, because the functor $\Rf$ does not involve $\myplus$, the
subexpressions in the sum have different type, and
recursion is not at the outermost level ($1$ has type $\Bl\lhd \Rf$),
respectively.

By applying the definition in \cite{brs_lmcs}, the coalgebra structure on expressions $\delta_\Rf$ is given by:
\[
\begin{array}{lcl}
\multicolumn{3}{l}{\delta_\Rf \colon \Exp_\Rf \to \Bl\times \Exp_\Rf}\\
\delta_\Rf(\emp) &=& <\bot_{\Bl}, \emp>\\
\delta_\Rf(\E_1 \oplus \E_2) &=& <b_1\vee b_2, \E_1' \oplus\E_2'>\text{ where } <b_i, \E_i'> = \delta_\Rf(\E_i), \ i\in\overline{1,2}\\
\delta_\Rf(\mu x.\E) &=&  \delta_\Rf(\E[\mu x.\E/x]) \\
\delta_\Rf(l<\tau> ) &=& <\delta_{\Bl\lhd \Rf}(\tau),\emp>\\
\delta_\Rf(r<\E> ) &=&  <\bot_\Bl,\E>\\
\delta_{\Bl\lhd \Rf}(\emp) &=& \bot_B\\
\delta_{\Bl\lhd \Rf}(b) &=& b\\
 \delta_{\Bl\lhd \Rf}(\tau\oplus\tau') &=&  \delta_{\Bl\lhd \Rf}(\tau)  \vee \delta_{\Bl\lhd \Rf}(\tau')\\
\end{array}
\]
The proof of Kleene's theorem provides algorithms to go from expressions to streams and
vice-versa. We illustrate it by means of examples.

Consider the following stream:
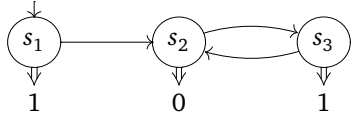
\[
\xymatrix@C=1.2cm@R=.25cm{
& & &\\
*++[o][F]{s_1} \ar@{<-}[u]\ar[r]\ar@{=>}[d] & *++[o][F]{s_2}\ar@/^/[r]\ar@{=>}[d] & *++[o][F]{s_3}\ar@/^/[l]\ar@{=>}[d]\\
1&0&1
}
\]
We draw the stream with an automata-like flavor. The transitions indicate the tail of the stream represented by a state and the output value the head. In a more traditional notation, the above automata represents the infinite stream $(1,0,1,0,1,0,1,\ldots)$. 

\ans{C.2.4} To compute expressions $\E_1$, $\E_2$ and $\E_3$ equivalent to $s_1$, $s_2$ and $s_3$ we associate with each state $s_i$ a variable $x_i$ and get the equations:
\[
\eps_1 = \mu x_1. l<1> \oplus r<x_2>\ \ \ \eps_2 = \mu x_2 .l<0> \oplus r<x_3>\ \ \
\eps_3 = \mu x_3. l<1> \oplus r<x_2>
\]
As our goal is to remove all the occurrences of free variables in our expressions, we proceed as follows.
First we substitute $x_2$ by $\eps_2$ in $\eps_1$, and $x_3$ by $\eps_3$ in $\eps_2$, and obtain the following expressions:
\[
\eps_1 = \mu x_1. l<1> \oplus r<\E_2>\ \ \
\eps_2 = \mu x_2 .l<0> \oplus r<\E_3>\ \ \
%\eps_3 = \mu x_3. l<1> \oplus r<x_2>
%\eps_3 = \mu x_3. l<1> \oplus r< \mu x_2 .l<0> \oplus r<x_3>>
\]
Note that at this point $\eps_1$ and $\eps_2$ already denote closed expressions. Therefore, as a last step, we replace $x_2$ in $\eps_3$ by $\eps_2$ and get the following closed expressions:
\[
\eps_1 = \mu x_1. l<1> \oplus r<\E_2>\ \ \
\eps_2 = \mu x_2 .l<0> \oplus r<\E_3>\ \ \
%\eps_3 = \mu x_3. l<1> \oplus r<x_2>
\eps_3 = \mu x_3. l<1> \oplus r< \mu x_2 .l<0> \oplus r<x_3>>
\]
satisfying, by construction, $\E_1\sim s_1$,  $\E_2\sim s_2$ and  $\E_3\sim s_3$.

For the converse construction, consider the expression $\E = (\mu x. r<x>) \oplus l<1>$. We construct an automaton by repeatedly applying the coalgebra structure on expressions $\delta_\Rf$, modulo associativity, commutativity and idempotency (ACI) of $\oplus$ in order to guarantee finiteness.

\ans{C.1.8} First, note that $\delta_\Rf(\mu x . r <x>) = \delta_\Rf(r<\mu x.r<x>>) = <\bottom_\Bl, \mu x . r <x>>$.
Applying the definition of $\delta_\Rf$ above, we have:
\[
\delta_\Rf(\E) = <1,  (\mu x. r<x>) \oplus \emp> \text{ and } \delta_\Rf((\mu x. r<x>) \oplus \emp) = <0, (\mu x. r<x>) \oplus \emp>
\]
which leads to the following stream (automaton):
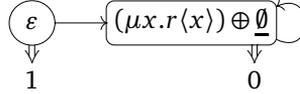
\[
\xymatrix@C=1.5cm@R0.25cm{
*++[o][F]{\E}\ar[r]\ar@{=>}[d]  &
*+[l][F-:<3pt>]{(\mu x. r<x>) \oplus \emp}\ar@(dr,ur)\ar@{=>}[d] \\
%*+[o][F]{z}\ar@(dr,ur) \ar@{=>}[d]
1&0
}
\]

At this point, we want to remark that the direct application of $\delta_\Rf$, without ACI, might generate infinite automata. Take, for instance, the expression $\E = \mu x. r<x\oplus x>$ . Note that  $\delta_\Rf(\mu x . r <x\oplus x>) = <0,\E\oplus \E>$, $\delta_\Rf(\E\oplus \E) = <0,(\E\oplus \E)\oplus (\E\oplus \E)>$, and so on. This would generate the infinite automaton
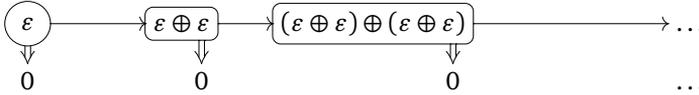
\[
\xymatrix@C=1.5cm@R0.25cm{
*++[o][F]{\E}\ar[r]\ar@{=>}[d]  &
*+[l][F-:<3pt>]{\E\oplus\E}\ar@{=>}[d] \ar[r] & *+[l][F-:<3pt>]{(\E\oplus \E)\oplus (\E\oplus \E)}\ar@{=>}[d] \ar[r] & \ldots \\
%*+[o][F]{z}\ar@(dr,ur) \ar@{=>}[d]
0&0&0& \ldots
}
\]
\noindent instead of the intended, simple and very finite, automaton
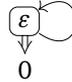
\[
\xymatrix@C=1.5cm@R0.25cm{
*+[l][F-:<3pt>]{\E}\ar@(dr,ur)\ar@{=>}[d] \\
%*+[o][F]{z}\ar@(dr,ur) \ar@{=>}[d]
0
}
\]
In order to guarantee finiteness, one needs to identify the expressions modulo ACI, as we will discuss further in this chapter. Moreover, the axiom $\E\oplus \emp \equiv \E$ could also be used in order to obtain smaller automata, but it is not crucial for termination. 

Streams will be often used as a
basic example to illustrate the definitions. It should be remarked
that the framework is general enough to include more complex examples,
such as deterministic automata, automata on guarded strings, Mealy machines and labelled transition systems.
The latter two will be used as examples in Section~\ref{sec:caseStudy}.

\section{Deciding equivalence of expressions}\label{sec:dp}

In this section, we briefly describe the decision procedure to determine whether two generalised regular expressions are equivalent or not. 

The key observation is that point $1.$ of Theorem~\ref{thm:kleene}  above guarantees that each expression in the language for a given system can always be associated with a \emph{finite} coalgebra. Given two expressions $\E_1$ and $\E_2$ in the language $\Exp_\Gf$ of a given functor $\Gf$ we can decide whether they are equivalent by constructing a \emph{finite} bisimulation between them. This is because the finite coalgebra generated from an expression contains precisely all states that one needs to construct the equivalence relation. Even though this might seem like a trivial observation, it has very concrete consequences: for (all well-typed) generalised regular expressions we can always either determine that they are bisimilar, and exhibit a proof in the form of a bisimulation, or conclude that they are not bisimilar and pinpoint the difference by showing why the bisimulation construction failed. Hence, we have a decision procedure for equivalence of generalised regular expressions.

We will give the reader a brief example on how the equivalence check works. Further examples, for different types of systems, including examples of non-equivalence, will appear in Section~\ref{sec:caseStudy}. 

We will show that the stream expressions 
\[\eps_1 = \mu x . r<x> \oplus l<0>\]
and
\[\eps_2 =  r< \mu x . r<x> \oplus l<0> > \oplus l<0>\]
are equivalent. In order to do that, we have to build a bisimulation relation  $R$ on expressions for the stream functor $\Rf$, defined above, such that  $(\eps_1, \eps_2) \in R$. We do this in the following way: we start by taking $R=\{(\E_1,\E_2)\}$ and we check whether this is already a bisimulation, by applying  $\delta_\Rf$ to each of the expressions and checking whether the expressions have the same output value and, moreover,  that no new pairs of expressions (modulo associativity, commutativity and idempotency, for more details see page~\pageref{just:ACI}) appear when taking transitions. Note that, for simplicity, we also use the sound axiom $\E\oplus \emp \equiv \E$. If new pairs of expressions appear we add them to $R$ and repeat the process. Intuitively, for this particular example, the transition structure can be depicted as in Figure~\ref{fig:streams-BC}.

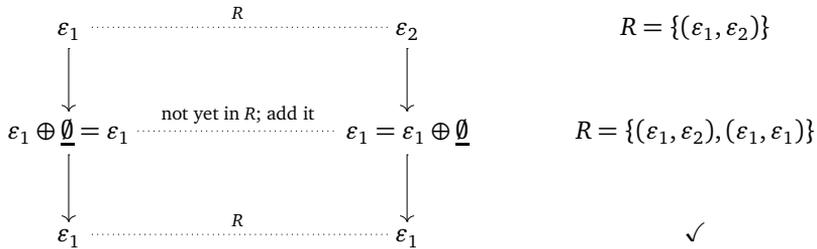
\begin{figure}[h]
\centering
%\begin{wrapfigure}[7]{r}{.4\textwidth }
$\xymatrix@C=1.2cm@R=0.8cm{
& {\eps_1}\ar[d]\ar@{..}[rr]^ R & & {\eps_2}\ar[d] & R = \{(\eps_1, \eps_2)\} \\
&{\eps_1 \oplus \emp = \eps_1 }  \ar[d]& &   {\eps_1 = \eps_1 \oplus \emp}\ar[d]\ar@{..}_{\text{not yet in }\R \text{; add it}} [ll]& R = \{(\eps_1, \eps_2), (\eps_1, \eps_1)\}\\
& {\eps_1}\ar@{..}[rr]^ R & & {\eps_1}& \checkmark
%*+[o][F]{\phantom{s_3}}\ar@(dr,ur)_{a|0}\ar@/^/[d]_{b|0} \\
%*+[o][F]{\phantom{s_4}}\ar@(dr,ur)_{b|0}\ar@(dl,ul)^{a|0} &
%*+[o][F]{s_2}\ar@(dr,ur)_{b|0}\ar@(dl,ul)^{a|0}
}$
\caption{Bisimulation construction}
\label{fig:streams-BC}
%\end{wrapfigure}
\end{figure}

In Figure~\ref{fig:streams-BC}, we omit the output values of the expressions, which are all $0$, and 
use the notation $\xymatrix@C=0.7cm@R=0.7cm{
 {\eps_1}\ar@{..}[r]^R&\eps_2}$ to denote $(\eps_1,\eps_2)\in R$.
Note that $R = \{(\eps_1, \eps_2), (\eps_2, \eps_2)\}$ is closed under transitions and is therefore a bisimulation. Hence, $\eps_1$ and $\eps_2$ are bisimilar and specify the same infinite stream (concretely, the stream with only zeros).

\section{An algebraic view on the coalgebra of expressions}
\label{sec:algSpec}

Recall that our goal is to reason about equality of generalised regular expressions in a fully automated manner. Obtaining this equality can be achieved in two distinct ways: either algebraically, reasoning with the axioms, or coalgebraically, by constructing a bisimulation relation. The latter, because of its algorithmic nature, is particularly suited for automation. Automatic constructions of bisimulations have been widely explored in {\CIRC} and we will use this tool to implement our algorithm. This section contains material that enables us to soundly use {\CIRC}. We want to stress however that the main result of this chapter is the description of a {\em decision procedure} to determine whether two expressions are equivalent or not. This procedure in turn could be implemented in any other suitable tool or even as a standalone application. Choosing {\CIRC} was natural for us, given the pre-existent work on bisimulation constructions.

In short, {\CIRC} is a behavioural extension of Maude~\cite{DBLP:conf/maude/2007} enabling the coinductive definition of infinite data structures by means of the so-called ``derivatives'' ($\delta_{\Gf}$ for the case of generalised regular expressions). The prover allows the (automated) reasoning on properties of such structures by coinduction (or bisimulation construction). The coinductive definitions are fed to {\CIRC} in the shape of algebraic specifications which are closely related to the original mathematical representations. Once a proof obligation is set, {\CIRC} starts the proving mechanism which repeatedly applies the derivatives, and (potentially) stops when a bisimulation containing the initial obligation is reached. For more insight on {\CIRC} we refer to~\cite{goguen-lin-rosu-2000-ase,rosu-lucanu-2009-calco} and Section~\ref{sec:dec-proced}.

In Section~\ref{sec:dec-proced}, we show that the process of generating the $\Gf$-coalgebras associated with expressions by repeatedly applying $\delta_{\Gf}$ and normalising the expressions obtained at each step is closely related to the proving mechanism already existent in {\CIRC}.

In Section~\ref{SCP-sec:prelim}, we have introduced a (theoretical) framework which, given a functor $\G$, allows
for the uniform derivation of 1) a language $\itExp_\itG$ for specifying behaviours of
$\itG$-systems, and 2) a coalgebraic structure on $\itExp_\itG$, which provides an
operational semantics to the set of expressions.
In this context, given that {\CIRC} is based on algebraic specifications, we need two things in order to reach our final goal:
\begin{itemize}\itemsep2pt
\item extend and adapt the framework of Section~\ref{SCP-sec:prelim} in order to enable the implementation of a tool which allows the automatic derivation of \emph{algebraic specifications} that model 1) and 2) above, to deliver to {\CIRC};
\item provide a decision procedure, implemented in {\CIRC} based on an \emph{equational entailment relation}, in order to check
bisimilarity of expressions.
\end{itemize}

\begin{comment}
We now have a (theoretical) framework which, given a functor $\itG$, allows
for the uniform derivation of 1) a language $\itExp_\itG$ for specifying behaviours of
$\itG$-systems, and 2) a coalgebraic structure on $\itExp_\itG$, which provides an
operational semantics to the set of expressions.
In the rest of the paper, we will extend and adapt the framework of the previous
section in order to:
\begin{itemize}\itemsep2pt
\item enable the implementation of a tool which allows for
the automatic derivation of 1) and 2) above;
\item \ans{C.1.10} enable
automatic reasoning on equivalence of specifications;
the process will be performed in a fully automated manner, first by using the aforementioned tool in order to derive the language of expressions and their operational semantics, and then by applying the coinductive prover {\CIRC}~\cite{rosu-lucanu-2009-calco} on the resulted specification, to complete the actual reasoning.
\end{itemize}
{\CIRC}\ is based on algebraic specifications and, therefore, to reach
our final goal we need two things:
\begin{itemize}\itemsep2pt
\item \emph{algebraic specifications} that model both the
language and the coalgebraic structure of expressions associated with non-deterministic
functors to provide to \CIRC;
\item a decision procedure, implemented in {\CIRC} based on an \emph{equational entailment relation}, in order to check for
the bisimilarity of expressions.
\end{itemize}
\end{comment}

In the rest of this chapter we will present the algebraic setting for reasoning on bisimilarity of generalised regular expressions.
A brief overview on the parallel between the coalgebraic concepts in~\cite{brs_lmcs} and their algebraic correspondents introduced in this section is provided later, in Figure~\ref{fig:CoVsAlg}.

An {\em algebraic specification} is a triple
${\mathcal E} = (S,\Sigma,E)$, where $S$ is a set of {\em sorts}, $\Sigma$ is a
{\em $S$-sorted signature} and $E$ is a set of
{\em conditional equations} of the form
$(\forall X)\,t=t'{ ~\textit{if}~}(\bigwedge_{i\in I} u_i=v_i)$, where
$t$, $t'$, $u_i$, and $v_i$ ($i\in I$ -- a set of indices for the conditions) are $\Sigma$-terms with variables in $X$.
We say that the \textit{sort of the equation} is $s$
whenever $t, t' \in {\mathcal T}_{\Sigma, s}({X})$.
Here, ${\mathcal T}_{\Sigma, s}({X})$ denotes the set of terms of sort $s$ of the
$\Sigma$-algebra freely generated by $X$.
If $I=\emptyset$ then the equation is \emph{unconditional} and may
be written as $(\forall X)\,t=t'$.

Let $\vdash$ be the
{\em equational entailment (deduction) relation} defined as in \cite{Goguen92order-sortedalgebra}.
We write $\mathcal E \vdash e$ whenever
equation $e$ is deducible from the equations $E$ in $\mathcal E$ by reflexivity, symmetry, transitivity, congruence or (conditional) substitutivity ({\emph{i.e.}}, whenever $E \vdash e$).

\begin{comment}
We extend $\mathcal E$ by adding the freezing operation
$\efr{-} \ct s \rightarrow {\sf Frozen}$ for each sort $s \in \Sigma$,
where $\sf Frozen$ is a fresh sort.
By $\efr{t}$ we represent the \textit{frozen} form of a $\Sigma$-term $t$,
and by
$\efr{e}$ a \textit{frozen equation} of the shape
$(\forall X)\, \efr{t} = \efr{t'} \textnormal{ if } c$.
The need for the frozen operator
will become clear in  Example~\ref{ex:streams}: without it the
congruence rule could be applied freely leading to the derivation of
untrue equations.
\ans{C.1.11, C.2.5}
More details regarding the definition of the entailment relation $\vdash$ over frozen equations are provided in Section~\ref{sec:dec-proced}, following the line in~\cite{rosu-lucanu-2009-calco}.
\end{comment}

\medskip

The algebraic specification of generalised regular expressions is built on top of definitions based on grammars in Backus-Naur form (BNF), such as (\ref{eq:fun-gram}) and (\ref{eq:grammar}). Next we introduce the general technique for transforming BNF notations into algebraic specifications.

\medskip
The general rule used for translating definitions based on BNF grammars into algebraic specifications is as follows: each syntactical category and vocabulary is considered as a sort and each production is considered as a constructor operation or a subsort relation.

For instance, according to the grammar (\ref{eq:fun-gram}) of non-deterministic functors, we have a sort $\SltN$ -- representing the vocabulary of join-semilattices $\itB$, a sort $\AlphbN$ -- for the vocabulary of the alphabets $A$, a sort $\Fun$ -- associated with the syntactical category of the non-deterministic functors $\Gf$, a subsort relation
$\SltN \ls \Fun$ representing the production $\Gf \,::\!\!= \itB$, and constructor operations for the other productions.

\ans{C.1.14} Generally, each production $A \,::\!\!= {\emph{rhs}}$ gives rise to a constructor $({\emph{rhs}}) \rightarrow (A)$, the direction of the arrow being reversed.
For instance, for grammar (\ref{eq:fun-gram}), the production $\Gf \,::\!\!= \id$ is represented by a constant (nullary operation) $\id : {}\to  \Fun$, and
the sum construction by the binary operation 
\[\_\myplus\!\_ \,\,:\,  \Fun~\Fun \to \Fun.\]

\begin{remark}
\ans{C.3.16}
Note that the above mechanism for translating BNF grammars into algebraic specifications makes use of subsort relations for representing productions such as $\Gf\,::\!\!=\Bl$. This is because {\CIRC} works with order-sorted algebras, and we want to keep the algebraic specifications of non-deterministic functors as close as possible to their implementation in {\CIRC}.
%However, order-sorted algebras can be reduced to many-sorted algebras~\cite{Goguen92order-sortedalgebra}, where a subsort relation $s \ls s'$ is modeled by an inclusion operation $c_{s,s'}\,:\, s \rightarrow s'$.
%This way, even if we use order-sorted algebras, we remain in the framework of circular coinduction.
\end{remark}

\medskip
The algebraic specifications of coalgebras of generalised regular expressions are defined in a modular fashion, based on the specifications of:
\begin{itemize}\itemsep0pt
\item non-deterministic functors ($\Gf$);
\item generalised regular expressions ($\eps \in \Exp_\Gf$);
\item ``transition" functions ($\delta_\Gf$);
\item ``structured" expressions ($\sigma \in \F(\Exp_\Gf)$, for all $\F$ ingredients of $\Gf$).
\end{itemize}

Moreover, recall that for a non-deterministic functor $\Gf$, bisimilarity of $\Gf$-expressions is decided based on the relation lifting $\overline{\Gf}$ over ``structured" expressions in $\Gf(\Exp_\Gf)$ (see~(\ref{eq:bisim-alt-def}) in Section~\ref{prelim:coalg}). Therefore, the deduction relation $\vdash$ has to be extended to allow a restricted contextual reasoning over ``structured" expressions in $\F(\Exp_\Gf)$, for all ingredients $\F$ of $\Gf$.

The aforementioned algebraic specifications and the extension of $\vdash$ are modelled as follows.

\medskip
The algebraic specification
of a non-deterministic functor $\Gf$ includes:
\begin{itemize}
\item the translation of the BNF grammar~(\ref{eq:fun-gram}), as presented above;
\item the specification of the functor ingredients, given by a sort $\Ingred$ and a constructor $\_\triangleleft\_{} \ct \Fun~\Fun \to \Ingred$ (according to Definition~\ref{def:ingred});
%\item the operations $\emph owner \ct \Alphb \to \AlphbN$ used to model the relationship between the elements of an alphabet and the name of the alphabet, and $\emph owner \ct \Slt \to \SltN$ for the relationship between the elements and the name of a semilattice (we will see later that these operations are semantically equivalent to the membership relation);
\item the specification of each alphabet $A=\{a_1,\ldots,a_n\}$ occurring in the definition of $\Gf$: this consists of a subsort $A \ls \Alphb$, a constant $a_i:{}\to A$ for $i\in\overline{1,n}$, and a distinguished constant $A$ of sort $\AlphbN$ used to refer the alphabet in the definition of the functor;
\item the specification of each semilattice  $\itB=(\{b_1,\ldots,b_n\},\lor,\bot_B)$ occurring in the definition of $\Gf$: this consists of a subsort $\itB \ls \Slt$, a constant $b_i:{}\to \itB$ for $i\in\overline{1,n}$, a distinguished constant $\itB$ of sort $\SltN$ used to refer the corresponding semilattice in the definition of the functor, and the equations defining $\lor$ and $\bottom_B$ (this should be one of the $b_i$'s);
\item an equation defining $\Gf$ (as a functor expression).
\end{itemize}

\medskip
The algebraic specification of generalised regular expressions consists of:
\begin{itemize}
\item (according to the BNF grammar
in Definition~\ref{def:expr}) a sort $\Exp$ representing expressions $\eps$, $\Fixpv$ the sort for the vocabulary of the fixed-point variables, and $\Slt$ the sort for the elements of semilattices. Moreover, we consider constructor operations for all the productions. For example, the production $\eps\, ::\!\!= \eps \oplus \eps$ is represented by an operation $\_\oplus\_\,\ct\Exp\,\,\Exp \to \Exp$, \ans{C.4.3} and $\eps\, ::\!\!= \mu x . \gamma$ is represented by $\mu\_.\_\,\ct\Fixpv\,\,\Exp\to\Exp$. (We chose not to provide any restriction to guarantee that $\gamma$ is a guarded expression, at this stage in the definition of $\mu\_.\_$. However, guards can be easily checked by pattern matching, according to the grammars in Definition~\ref{def:expr});

\item the specification of the substitution of a fixed-point variable with an expression, given by an operation
$\textbf{\_[\_\,/\_]} \ct \Exp~\Exp~\Fixpv \to \Exp$ and a set of equations, one for each constructor. \ans{C.4.3} For example, the equations associated with $\emp$ and $\oplus$ are: $\emp[\eps/x] = \emp$, and respectively, $(\eps_1 \oplus \eps_2)[\eps/x] = (\eps_1[\eps/x]) \oplus (\eps_2[\eps/x])$, where $\eps, \eps_1, \eps_2$ are $\Gf$-expressions and $x$ is a fixed-point variable;

\item the specification of the type-checking relation in Definition~\ref{def:ts},
given by an operation %\linebreak
$\_\,\ct\!\_ \,\,\ct \Exp~\Ingred\to {\mathsf{Bool}}$ and an equation for each inference rule defining this relation. For example the rule
\[
\rules{\vdash \E_1 \colon \F\lhd \Gf\;\;\;\; \vdash\E_2\colon \F\lhd \Gf}{\vdash \E_1\oplus\E_2 \colon \F\lhd \Gf}
\]
is represented by the equation
$
\E_1 \oplus \E_2 \ct \FtrlG =
\E_1 \ct \FtrlG \land \E_2 \ct \FtrlG
$. The type-checking operator is used in order to verify whether the expressions checked for equivalence are well-typed (Definition~\ref{def:g-expr}). Moreover, note that for the consistency of notation, algebraically we write
$\eps \ct \FtrlG$ to represent expressions $\eps$ of type $\FtrlG$.
\end{itemize}

\medskip
The algebraic specification of $\delta_\Gf$ consists of:
\begin{itemize}
\item the specification of the coalgebra of $\Gf$-expressions $\delta_\Gf$ given by three operations
\[
\begin{array}{l}
\delta\_(\_) \ct \Ingred~\Exp\to \ExpS\\
\emph Empty \ct \Ingred \to \ExpS\\
\emph Plus\_(\_,\_) \ct \Ingred~\ExpS~\ExpS\to \ExpS;
\end{array}
\]

\item equations describing the definitions of these operations as in~\cite{brs_lmcs}.
\end{itemize}

\medskip
As mentioned above, the set of $\Gf$-expressions is provided with a coalgebraic structure given by the function $\delta_{\Gf} \colon {\Exp}_{\Gf} \to {\Gf}({\Exp}_{\Gf})$, where ${\Gf}({\Exp}_{\Gf})$ can be understood as the set of expressions with structure given by $\Gf$ (and its ingredients). The set of structured expressions is defined by the following grammar:
\begin{equation}\label{eq:struct-expr}
\sigma \;::\!=\; \eps\mid b \mid \langle\sigma,\sigma\rangle\mid k_1(\sigma)\mid
               k_2(\sigma)\mid\bot\mid\top\mid\lambda {.}(a,\FtrlG,\sigma) \mid \{\sigma\}
\end{equation}
where $\eps \in \Exp_\Gf$ and $b\in \Bl$. The typing rules below give precise meaning to these expressions. Note that $\bot,\top$ are two expressions coming from $\Gf = \Gf_1 \myplus \Gf_2$, used to denote underspecification and overspecification, respectively.

The associated algebraic specification includes:
\begin{itemize}
\item a sort $\ExpS$ representing expressions $\sigma$ (from {$\F(\Exp_\Gf)$}, with $\FtrlG$), and one operation for each production in the BNF grammar~(\ref{eq:struct-expr}). Note that the construction $\lambda {.}(a,\FtrlG,\sigma)$ has as coalgebraic correspondent a function
$f \in \itF^{A}(\Exp_{\Gf})$, and is defined by cases as follows:

\[\lambda {.}(a,\FtrlG,\sigma)(a') = \textnormal{ if } (a=a') \textnormal{ then } \sigma \textnormal{ else } {\emph{Empty}}_{\FtrlG};\]

\item the extension of the type-checking relation to structured expressions, defined by:
\\[0.5ex]
$
\begin{array}{@{\hspace{-0.1cm}}l@{\hspace{0.3cm}}l}
\dfrac{\vdash b\ct \BtrlG}
       {\vdash b\in \Bl({\Exp}\,\Gf)}
&
\dfrac{\vdash \eps\ct \IdtrlG}
       {\vdash \eps\in {\id}({\Exp}\,\Gf)}
\\[3ex]
\dfrac{}
       {\vdash \bot\in \F_1{\myplus}\F_2({\Exp}\,\Gf)}
&
\dfrac{}
       {\vdash \top\in \F_1{\myplus}\F_2({\Exp}\,\Gf)}
\\[3ex]
\dfrac{\vdash \sigma\in \F_i({\Exp}\,\Gf)}
       {\vdash k_i(\sigma)\in \F_1{\myplus}\F_2({\Exp}\,\Gf)}~i\in\overline{1,2}
&
\dfrac{\vdash \sigma_1\in \F_i({\Exp}\,\Gf)\,\,\,\,\,\vdash \sigma_2\in \F_i({\Exp}\,\Gf)}
       {\vdash \langle\sigma_1,\sigma_2\rangle\in
\F_1{\times}\F_2({\Exp}\,\Gf)}\\[3ex]
\dfrac{\vdash \sigma\in \F({\Exp}\,\Gf),~a\in A}
       {\vdash \lambda.(a,\FtrlG,\sigma)\in \F^A({\Exp}\,\Gf)}
&
\dfrac{\vdash \sigma\in \F({\Exp}\,\Gf)}
       {\vdash \{\sigma\}\in \pow \F({\Exp}\,\Gf)}
\end{array}
$\\[1.5ex]
and specified by an operation $\_\in\!\_({\Exp}\,\_) \,\ct \ExpS~\Fun~\Fun\to {\mathsf{Bool}}$
(where we used a mix-fix notation) and an equation for
each of the above inference rules.
For example, the first rule has associated the equation $b\in \Bl({\Exp\,\Gf}) = b\ct \BtrlG$.
For consistency of notation, we write
$\sigma \in \F({\Exp}_\Gf)$ to denote that $\sigma$ is an element
of $\F({\Exp}_\Gf)$.
\end{itemize}

\begin{remark}
In terms of membership equational logic (MEL)~\cite{Bouhoula-Jouannaud-Meseguer00},
both $\FtrlG$ and\linebreak $\F({\Exp}\,\Gf)$ can be thought of as being sorts and,
for example, $\eps\ct\FtrlG$ as a membership assertion. Even if MEL is an elegant theory,
we prefer not to use it here because this implies the dynamic declaration of sorts and a set of assertions for such a sort.  The above approach is generic and therefore more flexible.
\end{remark}

\medskip
As previously hinted at the beginning of this section, in order to algebraically reason on bisimilarity of $\Gf$-expressions in {\CIRC}, one has to extend the deduction relation $\vdash$ to allow a restricted contextual reasoning on expressions in $\F(\Exp_\Gf)$, for all ingredients $\F$ of a non-deterministic functor $\Gf$. We call the extended entailment $\vdashInd$.

The aforementioned restriction refers to inhibiting the use of congruence during equational reasoning, in order to guarantee the soundness of {\CIRC} proofs. This is realised by means of a \emph{freezing operator}, which intuitively behaves as a wrapper on the expressions checked for equivalence, by changing their sort to a fresh sort ${\mathsf{Frozen}}$. This way, the hypotheses collected during a {\CIRC} proof session cannot be used freely in contextual reasoning, hence preventing the derivation of untrue equations (as illustrated in Example~\ref{ex:streams}).

We further show how the freezing mechanism is implemented in our algebraic setting, and define $\vdashInd$.

Let ${\mathcal E}$ be an algebraic specification. We extend $\mathcal E$ by adding the freezing operation
$\efr{-} \ct s \rightarrow {\mathsf{Frozen}}$ for each sort $s \in \Sigma$,
where $\mathsf{Frozen}$ is a fresh sort.
By $\efr{t}$ we represent the \textit{frozen} form of a $\Sigma$-term $t$,
and by
$\efr{e}$ a \textit{frozen equation} of the shape
$(\forall X)\, \efr{t\phantom{\!'}} = \efr{t'} \textnormal{ if } c$. Note that, according to~\cite{rosu-lucanu-2009-calco}, conditions $c$ need not to be frozen, as their (so-called visible) sort does not allow their collection into the set of {\CIRC} hypotheses. The entailment relation $\vdash$ is defined over frozen equations following~\cite{rosu-lucanu-2009-calco}; \ans{C.1.11, C.2.5} more details are provided in Section~\ref{sec:dec-proced}.

Recall that a relation
$\R \subseteq \Exp_\Gf \times \Exp_\Gf$ is a bisimulation if and only if $(s, t) \in \R \Rightarrow (\delta_{\GtrlG}(s), \delta_{\GtrlG}(t))\in \overline \Gf(\R)$. Here, $\overline{\Gf}(\R) \subseteq \Gf(\Exp_\Gf) \times \Gf(\Exp_\Gf)$ is the lifting of the relation $\R \subseteq \Exp_\Gf \times \Exp_\Gf$, defined as 

$$\overline \Gf(\R) = \{ (\Gf(\pi_1)(x),\Gf(\pi_2)(x)) \mid x \in \Gf(\R) \}\ .$$

 So, intuitively, reasoning on bisimilarity of two expressions $(\eps, \eps')$ in $\R$ reduces to checking whether the application of $\delta_\Gf$ maps them into $\overline{\Gf}(\R)$.
 
Therefore, checking whether a pair $(s^{\delta}, t^{\delta})$ is in $\overline{\Gf}(\R)$ consists in checking, for example for the case of $\Gf = \Gf_1 \times \Gf_2$, whether $(s_{1}^{\delta}, t_{1}^{\delta}) \in \overline{\Gf_1}(\R)$ and $(s_{2}^{\delta}, t_{2}^{\delta}) \in \overline{\Gf_2}(\R)$, where $s^{\delta} = <s^{\delta}_{1}, s^{\delta}_{2}>$ and $t^{\delta} = <t_{1}^{\delta}, t_{2}^{\delta}>$.
In an algebraic setting, this would reduce to building an algebraic specification $\algspecN$ and defining an entailment relation $\vdashInd$ such that one can infer $\algspecN \vdashInd \efr{<s^{\delta}_1, s^{\delta}_2>} = \efr{<t^{\delta}_1, t^{\delta}_2>}$ (this is the algebraic correspondent we consider for $(<s^{\delta}_1, s^{\delta}_2>, <t^{\delta}_1, t^{\delta}_2>) \in \overline{\Gf}(\R)$) by showing $\algspecN \vdashInd \efr{s^{\delta}_1} = \efr{t^{\delta}_1}$ (or $(s^{\delta}_1, t^{\delta}_1) \in \overline{\Gf_1}(\R)$) and 
$\algspecN \vdashInd \efr{s^{\delta}_2} = \efr{t^{\delta}_2}$ (or $(s^{\delta}_2, t^{\delta}_2) \in \overline{\Gf_2}(\R)$).
We hint that the aforementioned algebraic specification $\algspecN$ consists of $\algspec$ and a set of frozen equations (see Corollary~\ref{cor:ii}).

The entailment relation $\vdashInd$ for reasoning on bisimilarity of $\Gf$-expressions is based on the definition of $\overline{\Gf}$.
%
\begin{comment}
Recall from Section~\ref{SCP-sec:prelim} that a relation
$\R \subseteq \Exp_\GtrlG \times \Exp_\GtrlG$ is a bisimulation
if and only if
$(s, t) \in \R \Rightarrow <\delta_{\GtrlG}(s), \delta_{\GtrlG}(t)>\in \overline \Gf(\R)$.
In order to enable the algebraic framework to decide
bisimilarity of $\itG$-expressions, we define
a new entailment relation for non-deterministic functors
$\vdashInd$
(the definitions of $\overline{\mathscr G}$ and $\vdashInd$ are closely related).
\end{comment}
%
\begin{definition}
\label{def:PF}
The entailment relation $\vdashInd$ is the extension of $\vdash$ with the following inference rules,
which allow a  restricted contextual reasoning over the frozen equations of structured expressions:
\begin{equation}
\frac{\algspec  \vdashInd \efr{\sigma_1 \phantom{\hspace{-2.5ex}\sigma_{1}'}} = \efr{\sigma'_1} \,\,\,\,\,\,
\algspec  \vdashInd \efr{\sigma_2\phantom{\hspace{-2.5ex}\sigma_{2}'}} = \efr{\sigma_2'}}
{\algspec  \vdashInd \efr{\langle \sigma_1, \sigma_2\rangle} =
\efr{\langle \sigma_1', \sigma_2'\rangle}}
\label{rl:times}
\end{equation}
\vspace{-1ex}
\begin{equation}
\frac{\algspec  \vdashInd \efr{\sigma\phantom{\hspace{-2.1ex}\sigma'}} = \efr{\sigma'}}
{\algspec  \vdashInd \efr{k_i(\sigma)} = \efr{k_i(\sigma')}}~i \in \overline{1,2}
\label{rl:plus}
\end{equation}
\begin{equation}
\frac{\algspec  \vdashInd \efr{f(a)} = \efr{g(a)}\, , ~{\emph for~all~} a \in {\emph A}}
{\algspec  \vdashInd \efr{f} = \efr{g\phantom{\hspace{-1.1ex}f}}}\label{rl:expo}
\end{equation}
\begin{equation}
\frac{\algspec  \vdashInd \efr{\sigma_{i_1}\phantom{\hspace{-3.0ex}\sigma'_{j_1}}} = \efr{\sigma'_{j_1}}\,,
\ldots,\,
\algspec  \vdashInd \efr{\sigma_{i_k}\phantom{\hspace{-3.4ex}\sigma'_{j_k}}} = \efr{\sigma'_{j_k}}}
{\algspec  \vdashInd \efr{\{\sigma_1, \ldots, \sigma_n\}} = \efr{\{\sigma'_1, \ldots, \sigma_m'\}}}{\small
%\begin{array}{l} \bigcup_{l=1}^k i_l = \overline{1,n}\\  \bigcup_{l=1}^k j_l = \overline{1,m}\end{array}
\begin{array}{l} \{i_1,\ldots, i_k\} = \{1,\ldots, n\}\\  \{j_1,\ldots, j_k\} = \{1,\ldots, m\}\end{array}
}\label{rl:pow}
\end{equation}
\begin{comment}
\begin{equation}
\frac{{\algspec \vdashInd \efr{\sigma_{i_1}} = \efr{\sigma'_{j_1}}\,,
\ldots,\,
\algspec \cup E \vdashInd \efr{\sigma_{i_k}} = \efr{\sigma'_{j_k}} \atop 
\bigcup_{l=\overline{1,k}}\{i_l\}=\{1,\ldots,n\},\,
\bigcup_{l=\overline{1,k}}\{j_l\}=\{1,\ldots,m\}}}
{\algspec \cup E \vdashInd \efr{\bigcup_{i=\overline{1,n}}\{\sigma_i\}} = \efr{\bigcup_{j=\overline{1,m}}\{\sigma'_j\}}}\label{rl:pow}
\end{equation}
\end{comment}

\end{definition}

\begin{remark}
Note that the extension of the entailment relation $\vdash$ to $\vdashInd$ implies that\linebreak
${\algspec  \vdash e} \textnormal{\emph{ iff }} {\algspec  \vdashInd e}$ holds, for any equation $e$ of shape $\,\efr{\eps_1} = \efr{\eps_2}$ or $\eps_1 = \eps_2$, with $\eps_1, \eps_2$ non-structured expressions. Below, we will use the notation $\algspec  \vdashInd E$, where $E$ is a set of possibly frozen equations, to denote $\forall_{e\in E}\cdot\algspec  \vdashInd e$. 
\end{remark}

It is interesting to recall the relation lifting for the powerset functor which is encoded in the last rule of Definition~\ref{def:PF}. A pair $(U,V)$ is in  $\overline{\pow\Gf}(\R)$ if and only if for every $u\in U$ there exists a $v\in V$ such that $(u,v)$ belongs to $\overline\Gf(\R)$ and, conversely, for every $v\in V$, there exists a $u\in U$ such that $(u,v)$ belongs to $\overline\Gf(\R)$.

\begin{remark}
\ans{C.2.9}
As already hinted (and proved in Corollary~\ref{cor:ii}), reasoning on bisimilarity of expressions in a binary relation $\R \subseteq \Exp_\Gf \times \Exp_\Gf$ reduces to showing that 
$\efr{\delta_\Gf(s)} = \efr{\delta_\Gf(t)}$ is a $\vdashInd$-consequence, for all $(s,t) \in \R$.
The equational proof is performed in a ``top-down" fashion, by reasoning on the subsequent equalities between the components of the corresponding structured expression $\delta_\Gf(s)$, $\delta_\Gf(t)$ in an inductive manner. This is realised by applying the inverted rules (\ref{rl:times})--(\ref{rl:pow}).

Moreover, note that rule (\ref{rl:pow}) is not invertible in the usual sense; rather any statement matching the form of the conclusion can only be proved by some instance of the rule.
\end{remark}

%The following theorem and corollary correspond
%to the equivalences $(\roman{counter1})$, and respectively $(\roman{counter2})$, in Fig.~\ref{fig:CoVsAlg}.

We will further formalise the connection between the inductive definition of $\overline{\Gf}$ (on the coalgebraic side)
and $\vdashInd$ (on the algebraic side) in Theorem~\ref{thm:i}, hence enabling the
definition of bisimulations in algebraic terms,
in Corollary~\ref{cor:ii}.

\begin{remark}\label{rem:eqRed}
Equations in $\algspec$ (built as previously described in this section) are used in the equational reasoning only for reducing terms of shape ${\mathsf{op}}(t_1, \ldots, t_n)$ according to the definition of the operation ${\mathsf{op}}$. For the simplicity of the proofs of Theorem~\ref{thm:i} and Corollary~\ref{cor:ii}, whenever we write ${\mathsf{op}}(t_1, \ldots, t_n)$, we refer to the associated term reduced according to the definition of ${\mathsf{op}}$.
\end{remark}

First we introduce some notational conventions.
Let $\Gf$ be a non-deterministic functor and $\R \subseteq \Exp_{\Gf} \times \Exp_{\Gf}$.
We write:
\begin{itemize}\itemsep1pt
\item $\RId$ to denote the set $\RG \cup \{(\eps, \eps) \mid\algspec \vdash \eps \ct \GtrlG = {\emph true}\}$;
\item $\cl{\R}$ for the closure of $\R$ under transitivity, symmetry and reflexivity;
\item $\efr{R}$ to represent the set $\bigcup_{e \in R}\{\efr{e}\}$;
(application of the freezing operator to all elements of $\R$)
\item $\delta_{\GtrlG}(\eps=\eps')$ to represent the equation $\delta_{\GtrlG}(\eps)=\delta_{\GtrlG}(\eps')$;

\item $\algspec \cup \efr{\R}$ as a shorthand for
$(S, \Sigma, E \cup \{\efr{\E\phantom{\hspace{-1.5ex}\E'}} = \efr{\E'} \mid (\E, \E') \in \R \})$, where $\algspec = (S, \Sigma, E)$;

\item \ans{C.2.6} $(\sigma, \sigma') \in  \overline\Gf(\RG)$ as a shorthand for: $(\sigma, \sigma')$ is among the enumerated elements of a set $S$ explicitly constructed as an enumeration of the finite set $\overline\Gf(\RG)$ (in the algebraic setting, $\overline\Gf(\RG)$ is a subset of ${\mathcal T}_{\Sigma, \ExpS\!\!} \times {\mathcal T}_{\Sigma, \ExpS\!\!}$ and $\algspec \vdash \overline\Gf(\RG) = S$).

%\item $<\sigma,\sigma'> \in  \overline\Gf(\RG)$ is a shorthand
%for: $(\sigma,\sigma')$ is an element of the set $S$,
%where $\algspec \vdash \overline\Gf(\RG) = S$
%(here, $\overline\Gf(\RG) \subseteq
%{\mathcal T}_{\Sigma, \ExpS\!\!} \times {\mathcal T}_{\Sigma, \ExpS\!\!}$)
\end{itemize}

\begin{theorem}\label{thm:i}
Consider a non-deterministic functor $\Gf$. Let $\F$ be an ingredient of $\Gf$, $\RG$ a binary relation on
the set of $\Gf$-expressions,
and $\sigma, \sigma' \in {\F}(\Exp_{\itG})$.
\begin{itemize}\itemsep2pt
\item[a)] If $\Gf$ is not a constant functor, then
$(\sigma, \sigma') \in \overline{\F}(\cl{\RId})$ iff
$\algspec \cup \efr{\RG} \vdashInd \efr{\sigma} = \efr{\sigma'}$;

\item[b)] If $\Gf$ is a constant functor $\Bl$, then
$(\sigma, \sigma') \in \overline{\Bl}(\cl{\RId})$ iff
$\algspec \vdashInd \efr{\sigma\phantom{\hspace{-1.7ex}\sigma'}} = \efr{\sigma'}$.
\end{itemize}
\end{theorem}

\noindent
In order to prove Theorem~\ref{thm:i}.$a)$ we introduce the following lemma:

\begin{lemma}\label{lm:idir}
Consider $\Gf$ a non-deterministic functor and
$\RG$ a binary relation on
the set of $\Gf$-expressions.
If $(\eps, \eps') \in \cl{\RId}$ then
$\algspec \cup \efr{\RG} \vdashInd \efr{\eps\phantom{\hspace{-1.7ex}\eps'}} = \efr{\eps'}$.
\end{lemma}

\begin{proof}
\ans{C.2.7}
The proof is trivial, as equality is reflexive, symmetric and transitive.
\qed
\end{proof}

We are now ready to prove Theorem~\ref{thm:i}.
\begin{proof}[Theorem~\ref{thm:i}]
\verb##
\begin{itemize}
\item Proof of Theorem~\ref{thm:i}.$a)$.\\
\begin{itemize}\itemsep10pt
\item {$``\Rightarrow"$.}
The proof is by induction on the structure of $\itF$.\\
\textit{Base case}:
\begin{itemize}\itemsep6pt
\item $\itF = \itB$. It follows that $(\sigma,\sigma')$ is of shape
$(b, b)$ where $b\in \itB$, therefore
$\algspec \cup \efr{\RG} \vdashInd \efr{b} = \efr{b}$ holds
by reflexivity.

\item $\itF = \id$. In this case $(\sigma, \sigma') \in cl(\RId) = \overline{\id}(cl(\RId))$,
so the result follows immediately by Lemma~\ref{lm:idir}.
\end{itemize}
\textit{Induction step}:
\begin{itemize}\itemsep6pt
\item ${\itF} = {\itF}_1 \times {\itF}_2$. Obviously,
$\sigma = < \sigma_1, \sigma_2>$ and $\sigma' = < \sigma'_1, \sigma'_2>$,
where $(\sigma_1, \sigma'_1) \in \overline{{\itF}_1}(cl(\RId))$
and $(\sigma_2, \sigma'_2) \in \overline{{\itF}_2}(cl(\RId))$.
Therefore, by the induction hypothesis, both
$\algspec \cup \efr{\RG} \vdashInd \efr{\sigma_1\phantom{\hspace{-2.3ex}\sigma'_{1}}} = \efr{\sigma'_1}$
and
$\algspec \cup \efr{\RG} \vdashInd \efr{\sigma_2\phantom{\hspace{-2.3ex}\sigma'_{2}}} = \efr{\sigma'_2}$ hold.
Hence, according to the definition of $\vdashInd$ (see~(\ref{rl:times})),
we conclude that
$\algspec \cup \efr{\RG} \vdashInd \efr{<\sigma_1, \sigma_2>} = \efr{<\sigma'_1, \sigma'_2>}$ holds.

\item The cases
${\itF} = {\itF}_1 \myplus {\itF}_2$, ${\itF} = {\itF}_1^{\emph A}$ and ${\itF} = \pow{\itF'}$
are handled in a similar way.
\end{itemize}

\item {$``\Leftarrow"$.}
We proceed also by induction on the structure of $\itF$. Moreover, recall that the observations in Remark~\ref{rem:eqRed} hold (for each of the subsequent cases).
\\
\textit{Base case}:
\begin{itemize}\itemsep6pt
\item $\itF = \Bl$. In this case $(\sigma, \sigma')$ is of shape
$(b, b')$, where $b, b'$ are two elements of the
semilattice $\Bl$.
Also, recall that $\Gf \not = \Bl$, therefore, the equations (of type $\GtrlG \not = \F(\Exp_\Gf)$) in $\R$ are not involved in the equational reasoning. We deduce that $\efr{b\phantom{\hspace{-1.3ex}b'}} = \efr{b'}$
is proved by reflexivity, hence $(b, b') = (b, b) \in \overline{\Bl}(\cl{\RId})$.

\item $\itF = \id$. Note that for this case, $\sigma, \sigma'$ are expressions of the same type with the expressions in $\R$. We further identify two possibilities:
\begin{itemize}\itemsep2pt
\item $\efr{\sigma\phantom{\hspace{-1.7ex}\sigma'}} = \efr{\sigma'}$ is proved by reflexivity.
Therefore $(\sigma, \sigma') \in \{(\eps, \eps) \mid \eps : \GtrlG\} \subseteq
\RId \subseteq \cl{\RId} = \overline{\id}(\cl{\RId})$.

\item The equations in
$\efr{\RG}$ are used in the equational reasoning
$\algspec \cup \efr{\RG} \vdashInd \efr{\sigma\phantom{\hspace{-1.7ex}\sigma'}} = \efr{\sigma'}\,$.
In addition, the freezing operator inhibits contextual reasoning, therefore $\efr{\sigma\phantom{\hspace{-1.7ex}\sigma'}} = \efr{\sigma'}$ is proved according to the equations in $\efr{\R}$, based on the symmetry and transitivity of $\vdashInd$.
In other words, $(\sigma, \sigma') \in \cl{\RId} = \overline{\id}(\cl{\RId})$.

\end{itemize}
\end{itemize}
\textit{Induction step}:
\begin{itemize}\itemsep6pt
\item ${\itF} = {\itF}_1 \times {\itF}_2$.
Obviously, due to their type, the equations in $\R$ are not involved in the equational reasoning.
Therefore, $\algspec \cup \efr{\RG} \vdashInd \efr{<\sigma_1, \sigma_2>} = \efr{<\sigma'_1, \sigma'_2>}$ is a consequence of the inverted rule~(\ref{rl:times}). More explicitly, it follows that 
$\algspec \cup \efr{\RG} \vdashInd \efr{\sigma_1\phantom{\hspace{-2.3ex}\sigma'_{1}}} = \efr{\sigma'_1}$
and $\algspec \cup \efr{\RG} \vdashInd \efr{\sigma_2\phantom{\hspace{-2.3ex}\sigma'_{2}}} = \efr{\sigma'_2}$ must hold.
By the induction hypothesis, we deduce that
$(\sigma_1, \sigma'_1) \in \overline{\F}_1(\cl{\RId})$
and $(\sigma_2, \sigma'_2) \in \overline{\F}_2(\cl{\RId})$.
So by the definition of $\overline{{\itF}_1 \times {\itF}_2}$
we conclude that
$(\langle \sigma_1, \sigma_2 \rangle, \langle \sigma'_1, \sigma'_2 \rangle) =
(\sigma, \sigma') \in \overline{{\itF}_1 \times {\itF}_2}(\RG)$.

\item The cases
${\itF} = {\itF}_1 \myplus {\itF}_2$, ${\itF} = ({\itF}_1)^{\emph A}$
and ${\itF} = \pow{\itF'}$
follow a similar reasoning.
\end{itemize}
\end{itemize}

\item Proof of Theorem~\ref{thm:i}.$b)$. It follows immediately by the definition of $\overline{\Bl}$ and Remark~\ref{rem:eqRed}.
\end{itemize}
\qed
\end{proof}

\begin{comment}
\begin{proof}
The proof is by induction on the structure of $\itF$.
%Both implications are proved based on the definitions of
%$\overline{\itF}$ and $\vdashInd$.
Take, for example, the direct implication $``\Rightarrow"$. The base case
$\F = \Bl$ holds by the reflexivity of $\vdashInd$.
The case $\F = \id$ follows immediately according to an auxiliary
result stating that
if $(\eps, \eps') \in \cl{\RId}$ then
$\algspec \cup \efr{\RG} \vdashInd \efr{\eps} = \efr{\eps'}$.
Inductive steps hold by the rules~(\ref{rl:times}), (\ref{rl:plus}) and (\ref{rl:expo}),
defining $\vdashInd$.
A similar reasoning is used for proving $``\Leftarrow"$.
%(The whole proof of this theorem is provided in Appendix~\ref{anx:thmi}.)
\qed
\end{proof}
\end{comment}

\begin{remark}
\label{rmk:theorem-cases}
For a more intuitive justification on the distinction of constant/non-constant functor in Theorem~\ref{thm:i},
note that in {\CIRC}, proof obligations $\efr{\eps} = \efr{\eps'}$ of a type (sort) that serves as ``base case'' in the co-recursive definitions are not collected as hypotheses during a proof session. Hence, in the context of $\Gf$-expressions, whenever $\Gf = \Bl$, the hypotheses set $\RG$ is empty. Consequently, a corresponding obligation $\efr{\eps} = \efr{\eps'}$ of type $\Bl$ is proved only according to the equations in $\algspec$, by applying transitivity, symmetry and reflexivity.
\end{remark}

\begin{corollary}\label{cor:ii}
Let $\Gf$ be a non-deterministic functor and 
$\RG$ a binary relation on the set of $\Gf$-expressions.
\begin{itemize}\itemsep2pt
\item[a)] If $\Gf$ is a non-constant functor, then
$\cl{\RId}$ is a bisimulation iff
$\algspec \cup \efr{\RG} \vdashInd \efr{\delta_{\GtrlG}(\RG)}$;

\item[b)] If $\Gf$ is a constant functor $\Bl$, then
$\cl{\RId}$ is a bisimulation iff
$\algspec \vdashInd \efr{\delta_{\GtrlG}(\RG)}$.
\end{itemize}
\end{corollary}

\begin{proof}
\verb##
\begin{itemize}\itemsep2pt
\item Proof of Corollary~\ref{cor:ii}.$a)$. We reason as follows:
\begin{align*}
&\cl{\RId} \text{ is a bisimulation }\\
\Leftrightarrow~&  (\forall (\eps, \eps') \in \cl{\RId}).((\delta_{\GtrlG}(\eps),
\delta_{\GtrlG}(\eps')) \in \overline{\mathscr G}(\cl{\RId})\\
\Leftrightarrow~& \algspec \cup \efr{\RG} \vdashInd \efr{\delta_{\GtrlG}(\cl{\RId})}\tag{{Thm.~\ref{thm:i}}}\\
\Leftrightarrow~&\algspec \cup \efr{\RG} \vdashInd \efr{\delta_{\GtrlG}(\RG)}\tag{$\cl{\RId}, \vdashInd$}
\end{align*}
%\[
%\begin{array}{lclr}
%&&\cl{\RId} \text{ is a bisimulation } \\[2ex]&\Leftrightarrow& (\forall (\eps, \eps') \in \cl{\RId}).((\delta_{\GtrlG}(\eps),
%\delta_{\GtrlG}(\eps')) \in \overline{\mathscr G}(\cl{\RId}) & \\[2ex]&\Leftrightarrow& \algspec \cup \efr{\RG} \vdashInd \efr{\delta_{\GtrlG}(\cl{\RId})} & {(\textnormal{Thm.~\ref{thm:i}})}\\[2ex]
%&\Leftrightarrow&
%\algspec \cup \efr{\RG} \vdashInd \efr{\delta_{\GtrlG}(\RG)}&{(\cl{\RId}, \vdashInd)}
%\end{array}
%\]

\item Proof of Corollary~\ref{cor:ii}.$b)$. It follows immediately by the definition of bisimulation relations and according to the observations in Remark~\ref{rem:eqRed}.
\end{itemize}
\qed
\end{proof}

%In Figure~\ref{fig:CoVsAlg} we briefly illustrate the algebraic encoding of the coalgebraic setting in~\cite{brs_lmcs}.

In Figure~\ref{fig:CoVsAlg} we briefly summarise the results of the current section, namely, the algebraic encoding of the coalgebraic setting presented in~\cite{brs_lmcs}.

\begin{figure}[h]
\centering
%\hspace{-10pt}
\renewcommand{\arraystretch}{1.3}
\begin{tabular}{|c|cr|}
\hline
coalgebraic & algebraic &\\
\hline\hline

$\vdash \eps \ct \FtrlG$ & $\algspec \vdash \eps \ct \FtrlG = {\emph true}$&\\
\hline

${\Exp_{\FtrlG}}$
&
\multicolumn{2}{|c|}{$\{\eps \in {\mathcal T}_{\Sigma, \Exp\!\!} \mid \algspec \vdash \eps \ct \FtrlG = {\emph true}\}$}\\
\hline

${\Exp_{\Gf}}$
&
\multicolumn{2}{|c|}{$\{\eps \in {\mathcal T}_{\Sigma, \Exp\!\!} \mid \algspec \vdash \eps \ct \GtrlG = {\emph true}\}$}\\
\hline

${\F(\Exp_\Gf)}$ &
\multicolumn{2}{|c|}{$\{\sigma \in {\mathcal T}_{\Sigma, \ExpS\!\!} \mid
\algspec \vdash \sigma\in \F({\Exp}\,\Gf) = {\emph true} \}$}\\
\hline

$\delta_{\FtrlG} \ct {\Exp_{\FtrlG}} \rightarrow {\F(\Exp_\Gf)}$ &
\multicolumn{2}{|c|}{
$\delta \_ ( \_ ) \ct {\Ingred~\Exp \rightarrow \ExpS}$}\\
\hline

&
$
\begin{array}{c}
\algspec \vdash \sigma\in \F({\Exp}\,\Gf) = {\emph true},\\
\algspec \vdash \sigma'\in \F({\Exp}\,\Gf) = {\emph true}
\end{array}
$
&\\

$(\sigma, \sigma') \in \overline{\F}(\cl{\RId})$ &
$\algspec \cup \efr{\R} \vdashInd \efr{\sigma\phantom{\hspace{-1.8ex}\sigma'}} = \efr{\sigma'}$\,\, if $\Gf \not= \Bl$ &\\

& or &\\

& $\algspec \vdashInd \efr{\sigma\phantom{\hspace{-1.8ex}\sigma'}} = \efr{\sigma'}$\,\, if $\Gf = \Bl$ & (Thm.~\ref{thm:i})\\

\hline

%$\cl{\RId}$ is a bisimulation &
%$\algspec \cup \efr{\R} \vdashInd \efr{\delta_{\GtrlG}(\cl{\RId})}$ & $(\roman{counter2})$
%\\
%\hline

$\cl{\RId}$ is a bisimulation &
$\algspec \cup \efr{\R} \vdashInd \efr{\delta_{\GtrlG}(\R)}$\,\, if $\Gf \not= \Bl$&\\

& or &\\

& $\algspec \vdashInd \efr{\delta_{\GtrlG}(\R)}$\,\, if $\Gf = \Bl$&(Cor.~\ref{cor:ii})\\
\hline
\end{tabular}
\caption{Non-deterministic functors - coalgebraic vs. algebraic approach}
\label{fig:CoVsAlg}
\end{figure}

\section{Deciding bisimilarity in \CIRC}
\label{sec:dec-proced}

We next describe how the coinductive theorem prover {\CIRC}~\cite{lucanu-etal-2009-calco} can be used to implement the decision procedure for the bisimilarity of generalised regular expressions, which we discussed above.

{\CIRC} can be seen as an extension of Maude
with behavioural features and its implementation
is derived from that of Full-Maude.
%The proving power of the circular coinduction is extended with special
%contexts~\cite{DBLP:conf/icfem/LucanuR09}, generalisations and
%simplification rules~\cite{synasc09}.
In order to use the prover, one needs to provide a specification (a {\CIRC} theory)
and a set of goals.
%Here we use the latest version of {\CIRC} described in~\cite{GoriacLR10}.
A {\CIRC} theory ${\mathcal B}=(S, (\Sigma,\Delta),(E,{\mathcal I}))$ consists of an algebraic specification $(S,\Sigma,E)$, a set $\Delta$ of \emph{derivatives}, and a set $\mathcal I$ of equational interpolants, which are expressions of the form $e \Rightarrow \{ e_i \mid i \in I \}$ where $e$ and $e_i$ are equations.
\ans{C.3.26}
The intuition for this type of expressions is simple: $e$ holds whenever for any $i$ in $I$ the equation $e_i$ holds. In other words, to prove $E \vdash e$ one can chose to instead prove $E \vdash \{e_i \mid i \in I\}$.  For the particular case of non-deterministic functors, we use equational interpolants to extend the initial entailment relation in a consistent way with rules~(\ref{rl:times})--(\ref{rl:pow}).
(For more information on equational interpolants see \cite{acca}). A derivative $\delta\in\Delta$ is a $\Sigma$-term containing a special variable ${*}{:}s$ \ans{C.1.20} ({\emph i.e.}, a $\Sigma$-context), where $s$ is the sort of the variable $*$. If $e$ is an equation $t=t'$ with $t$ and $t'$ of sort $s$, then $\delta[e]$ is $\delta[t/{*}{:}s]=\delta[t'/{*}{:}s]$.
We call this type of equation a \emph{derivable equation}. The other equations are \emph{non-derivable}. We write $\delta[\R]$ to represent $\{\delta[e] \mid e \in \R\}$, where $\R$ is a set of derivable equations, and $\Delta[e]$ for the set $\{\delta[e]\mid\delta\in\Delta{\textnormal{~appropriate~for~}}e\}$.

\ans{C1.11} Moreover, note that {\CIRC} works with an extension of the entailment relation $\vdash$ over frozen equations (introduced in Section~\ref{sec:algSpec}), with two more axioms, as in~\cite{rosu-lucanu-2009-calco}:
\begin{equation}
\label{A1}
E \cup \R \vdash \efr{e}\,\,\,\, \textnormal{\emph{iff}}\,\,\, E \vdash e
\end{equation}
\begin{equation}
\label{A2}
E \cup \R \vdash \G\,\,\,\, \emph{implies}\,\,\, E \cup \delta[\R] \vdash \delta[\G] \textnormal{\emph{ for each }} \delta \in \Delta 
\end{equation}

Above, $E$ ranges over unfrozen equations, $e$ over non-derivable unfrozen equations, and $\R,\G$ over derivable frozen equations.

\begin{remark}
Note that the new entailment $\vdashInd$ extended over frozen equations (in Definition~\ref{def:PF}) satisfies the assumptions~(\ref{A1}) and~(\ref{A2}).
\end{remark}

{\CIRC} implements the coinductive proof system
given in~\cite{rosu-lucanu-2009-calco} using a set of reduction rules
of the form
$({\mathcal B}, {F}, {G}) \Rightarrow ({\mathcal B}, {F'}, {G'})$, where
${\mathcal B}$ represents a specification,
${F}$ is the coinductive hypothesis (a set of frozen equations) and
${G}$ is the current set of goals. The freezing operator is defined
as described in Section~\ref{sec:algSpec}.
Here is a brief description of these rules:
\begin{itemize}\itemsep 5pt
\item[]
%The rule
\textsf{[Done]}:
$({\mathcal B}, { F}, \{\})
\Rightarrow
\cdot
$\\
%}
Whenever the set of goals is empty, the system terminates with success.
\item[]
%The rule
%\centerline{
\textsf{[Reduce]}:
$({\mathcal B}, { F}, { G} \cup \{\efr{e}\})
\Rightarrow
({\mathcal B}, { F}, { G})
\textit{\;if\;} {\mathcal B} \cup { F} \vdash \efr{e}
$\\
%}
If the current goal is a $\vdash$-consequence
of ${\mathcal B} \cup { F}$ then $\efr{e}$ is removed from
the set of goals.
\item[]
%The rule
%\centerline{
\textsf{[Derive]}:
$({\mathcal B}, { F}, { G} \cup \{\efr{e}\})
\Rightarrow
({\mathcal B}, { F} \cup \{\efr{e}\}, { G} \cup \efr{\Delta[e]})$
$\textit{if\;} {\mathcal B} \cup { F}\, {\not \vdash}\,\, \efr{e} $\\
%}
When the current goal $e$ is derivable and it is not a
$\vdash$-consequence, it is added
to the hypothesis and its derivatives to the set of goals.
\ans{C.1.21}

\item[]
\textsf{[Simplify]}:
$({\mathcal B}, { F}, { G} \cup \{\efr{\theta(e)}\})
\Rightarrow
({\mathcal B}, { F}, { G} \cup \{\efr{\theta(e_i)} \mid i \in I\})$
\\
\verb##{\hspace{38pt}}
$\textit{\;if\;}  e  \Rightarrow \{ e_i \mid i \in I \}$
is an equational interpolant
from the\\
\verb##{\hspace{41pt}}
specification and $\theta \ct X \rightarrow {\mathcal T}_{\Sigma}(Y)$ is a substitution. \ans{C.2.12}
\item[]
\textsf{[Fail]}:
$({\mathcal B}, { F}, { G} \cup \{\efr{e}\})
\Rightarrow
{\emph{failure}} \textit{\;if\;} {\mathcal B} \cup { F} \,\, {\not \vdash}\, \efr{e} \,\land\,
e$ {\emph{is non-derivable}}\\
%}
This rule stops the reduction process  with failure whenever the current goal
$e$ is non-derivable and is not a $\vdash$-consequence of ${\mathcal B} \cup { F}$.
\end{itemize}

It is worth noting that there is a strong connection between a {\CIRC} proof and the construction of a bisimulation relation. We illustrate this fact and the importance of the freezing operator with a simple example.

\begin{example}\label{ex:streams}
Consider the case of infinite streams. The set $\itB^\omega$ of infinite streams over a set $\itB$ is the final coalgebra of the functor $\Rf =  \itB \times \id$, with a coalgebra structure given by {\emph hd} and {\emph tl}, the functions that return the head and the tail of the stream, respectively. Our purpose is to prove that $0^\infty = (00)^\infty$. Let $z$ and $zz$ represent the stream on the left-hand side and, respectively, on the right-hand side. These streams are defined by the equations: ${\emph hd}(z) = 0, {\emph tl}(z) = z, {\emph hd}(zz) = 0, {\emph tl}(zz) = 0:zz$. 
Note that equations over $\itB$ like ${\emph hd}(z) = 0$ are not derivable and equations over streams like ${\emph tl}(z) = z$ are derivable.

In Figure~\ref{fig:parallel} we present the correlation between the {\CIRC} proof and the construction of the bisimulation relation. Note how {\CIRC} collects the elements of the bisimulation as frozen hypotheses.

\begin{figure}[h]
\centering
%\hspace{-11pt}
\renewcommand{\arraystretch}{1.9}
\begin{tabular}{|c|c|}
\hline
{\CIRC} proof & Bisimulation construction \\
\hline\hline

\raisebox{-7pt}{\code{(add goal z = zz .)}} &
$\xymatrix@R=.2cm@C=1cm{
*+[o][F]{z}\ar@(dr,ur) \ar@{=>}[d] &
*+[o][F]{zz}\ar@/^/[r] \ar@{=>}[d]&
*+[o][F]{(zz)'}\ar@/^/[l] \ar@{=>}[d]\\
0 & 0 & 0
}$\\
\hline
$({\mathcal B}, \emptyset, \{\efr{z} = \efr{zz}\})$ &
${ F} = \emptyset; \,\, z \sim zz~?$
\\
\hline

\raisebox{-4pt}{$
\mathrel{{\overset{\textsf{[Derive]}}{\longrightarrow}}}
\left({\mathcal B}, \{\efr{z} = \efr{zz}\},
\left\{{\efr{{\emph hd}(z)} = \efr{{\emph hd}(zz)}} \atop
{\efr{{\emph tl}(z)} = \efr{{\emph tl}(zz)}}\right\}\right)
$}
&
\raisebox{-4pt}
{$
{ F} = \{(z, zz)\};
{
\,\, z \mathrel{{\overset{0}{\longrightarrow}}} z
\atop
\,\, zz \mathrel{{\overset{0}{\longrightarrow}}} (zz)'
}
$}
\\
\hline
$
\mathrel{{\overset{\textsf{[Reduce]}}{\longrightarrow}}}
({\mathcal B}, \{\efr{z} = \efr{zz}\},
\{\efr{z} = \efr{0 : zz}\})
$
&
{
$
{ F} = \{(z, zz)\};
\,\, z \sim (zz)'~?
$
}
\\
\hline

\raisebox{-4pt}{$
\mathrel{{\overset{\textsf{[Derive]}}{\longrightarrow}}}
\left({\mathcal B},
\left\{{\efr{z} = \efr{zz}} \atop
{\efr{z} = \efr{0 : zz}}\right\}
,
\left\{{\efr{{\emph hd}(z)} = \efr{{\emph hd}(0:zz)}} \atop
{\efr{{\emph tl}(z)} = \efr{{\emph tl}(0:zz)}}\right\}\right)
$}
&
\raisebox{-3pt}{$
{ F} = \{(z, zz), (z, (zz)')\};
{
\,\, z \mathrel{{\overset{0}{\longrightarrow}}} z
\atop
\,\, (zz)' \mathrel{{\overset{0}{\longrightarrow}}} zz
}
$}
\\
\hline
\raisebox{-3pt}{$\mathrel{{\overset{\textsf{[Reduce]}}{\longrightarrow}}}
\left({\mathcal B},
\left\{{\efr{z} = \efr{zz}} \atop
{\efr{z} = \efr{0 : zz}}\right\}
,
\{\}\right)
$}
&
$
{ F} = \{(z, zz), (z, (zz)')\}~\checkmark
$
\\
\hline
\end{tabular}
\caption{Parallel between a {\CIRC} proof and the bisimulation construction}%\vspace{-4ex}
\label{fig:parallel}
\end{figure}

Let us analyze what would happen if the freezing operator $\efr{-}\,$ was not used. Suppose the circular coinduction algorithm would add the equation $z = zz$ in its unfrozen form to the hypotheses. After applying the derivatives we obtain the goals
${\emph hd}(z) = {\emph hd}(zz), {\emph tl}(z) = {\emph tl}(zz)$.
At this point, the prover could use the freshly added equation \ans{C.1.22} $z = zz$, and according to the congruence rule, both goals would be proven directly, though we would still be in the process of showing that the hypothesis holds. By following a similar reasoning, we could \ans{C.1.23} also prove that $0^\infty = 1^\infty$! In order to avoid these situations, the hypotheses are frozen, ({\emph{i.e.}}, their sort is changed from {$\mathsf{Stream}$} to {$\mathsf{Frozen}$}) and this stops the application of the congruence rule, forcing the application of the derivatives according to their definition in the specification. Therefore, the use of the freezing operator is vital for the soundness of circular coinduction.

\end{example}

Next, we focus on using {\CIRC} for automatically reasoning on the equivalence of $\Gf$-expressions. As we will show, the implementation of both
the algebraic specifications associated with non-deterministic functors
and the equational entailment relation described
in Section~\ref{sec:algSpec} is immediate.
Given a non-deterministic functor $\Gf$, we define a {\CIRC} theory
$\behspec=(S, (\Sigma,\Delta),(E,{\mathcal I}))$ as follows:\\[-4ex]
\begin{itemize}
\item $(S,\Sigma,E)$ is $\algspec$
\item $\Delta=\{\delta_{\GtrlG}(*{:}\Exp)\}${, so the only derivable equations are those of sort $\mathsf{Exp}$. As we have already seen for the example of streams, equations of sort $\mathsf{Slt}$ must not be derivable. Since we have the subsort relation $\mathsf{Slt} \ls  \mathsf{Exp}$, we avoid the application of the derivative $\delta_{\GtrlG}(*{:}\Exp)$ over equations of sort $\mathsf{Slt}$ by means of an interpolant (see below).}
\item $\mathcal I$ consists of the following equational interpolants \ans{C.1.25}, whose role is to replace current proof obligations 
%with new ones:%[-4ex]
{over non-trivial structures with simpler ones:}
\begin{align}
\langle  \sigma_1, \sigma_2  \rangle  =  \langle  \sigma'_1, \sigma'_2  \rangle \,\, \Rightarrow \,\, &
\,\{ \sigma_1 = \sigma'_1,\,\, \sigma_2 = \sigma'_2 \} \label{srl:times}\\
k_i(\sigma) = k_i(\sigma') \,\, \Rightarrow \,\, & \,\{\sigma = \sigma'\} \label{srl:plus}\\
 f = g \,\, \Rightarrow \,\, & \,\{ f(a) = g(a)\mid a \in A\} \label{srl:expo}\\
 \cup_{i\in\overline{1,n}}\{\sigma_i\} = \cup_{j\in\overline{1,m}}\{\sigma'_j\}
\,\, \Rightarrow \,\, & \,\{ \land_{i\in\overline{1,n}}(\lor_{j\in\overline{1,m}}\, \sigma_i = \sigma'_j)\notag\\
& \hspace{4.5pt}\land_{j\in\overline{1,m}}(\lor_{i\in\overline{1,n}}\, \sigma_i = \sigma'_j)\}
\label{srl:pow}
\end{align}
{
together with an equational interpolant
\begin{align}
& t=t' \,\, \Rightarrow \,\, \{t\simeq t'= \mathsf{true}\}\label{srl:slt}
\end{align}
where $\simeq$ is the equality predicate equationally defined over the sort $\mathsf{Slt}$. The last interpolant transforms the equations of sort $\mathsf{Slt}$ from derivable (because of the subsort relation $ \mathsf{Slt}\ls \mathsf{Exp}$) into non-derivable and equivalent ones.\\ 
}
\end{itemize}

%\vspace{-1ex}

The interpolants 
%(\ref{srl:times}), (\ref{srl:plus}), (\ref{srl:expo}) and (\ref{srl:pow}) 
{(\ref{srl:times}--\ref{srl:slt})}
in $\mathcal I$ extend the entailment relation $\vdashInd$ (introduced in Definition~\ref{def:PF}) as follows:%[-4ex]
\[
\dfrac{E\vdashInd \{e_i\mid i\in I\}}
         {E\vdashInd e}~{\mathrm if~}e\Rightarrow\{e_i\mid i\in I\}{\mathrm~in~}{\mathcal I}
\]

\begin{theorem}[Soundness]\label{thm:soundnessCirc}
Let $\Gf$ be a non-deterministic functor, and $\G$ a binary relation on the
set of $\mathscr G$-expressions.\\
If $({\behspec}, { F}_0 = \emptyset, G_0 = \efr{\G}) \rTrans{*}
({\behspec}, { F}_n, G_n = \emptyset)$
using \textsf{[Reduce]}, \textsf{[Derive]} and \textsf{[Simplify]},
%and ${\mathcal F} = \cup_{i = 0 \ldots n}{\mathcal F}_i$,
then $\G \subseteq \sim_{\Gf}$.
%$F = \cl{{\mathcal F}^{\neg \square}_{\emph id}}$
%is a bisimulation relation and $\RG \subseteq F$.
\end{theorem}

\begin{proof}
The idea of the proof is to find a bisimulation relation
$\widetilde{ F}$ s.t. $G
\subseteq \widetilde{ F}$.\\
First let ${F}$ be the set of hypotheses (or derived goals) collected during the proof session.
We distinguish between two cases:

\begin{itemize}\itemsep6pt
\item[a)] $\Gf = \Bl$.
For this case, the set of expressions in $\G$ is given by the following grammar:
\begin{equation}
\eps\,::= \emp \mid b \mid \eps \oplus \eps \mid \mu x . \eps\,.
\label{eq:const}
\end{equation}

%\begin{todo}
%\textcolor{blue}
{
Note that the goals $\eps = \eps'$ in $\G$ are proven
\begin{enumerate}
\item either according to \textsf{[Simplify]}, applied in the context of the equational interpolant~(\ref{srl:slt}). If this is the case, then $\eps = \eps'$ holds by reflexivity, therefore
\begin{equation}\label{eq:s}
\behspec \vdashInd \efr{\delta_{\Bl \triangleleft \Bl}(\eps)} = \efr{\delta_{\Bl \triangleleft \Bl}(\eps')}
\end{equation}
also holds;
\item or after the application of \textsf{[Derive]}, case in which $\behspec \cup \efr{ F} \vdashInd \efr{\delta_{\Bl \triangleleft \Bl}(\eps)} = \efr{\delta_{\Bl \triangleleft \Bl}(\eps')}$ holds. Note that $\delta_{\Bl \triangleleft \Bl}(\eps)$ and $\delta_{\Bl \triangleleft \Bl}(\eps')$ are reduced to $b$, respectively $b' \in \Bl$, according to~(\ref{eq:const}) and the definition of $\delta_{\Bl \triangleleft \Bl}$. Consequently, the non-derivable (due to the subsort relation $\Bl\ls \mathsf{Slt}$) goal $\efr{b} = \efr{b'}$ holds by reflexivity, so the following is a sound statement:
\begin{equation}\label{eq:ss}
\behspec \vdashInd \efr{\delta_{\Bl \triangleleft \Bl}(\eps)} = \efr{\delta_{\Bl \triangleleft \Bl}(\eps')}.
\end{equation}
\end{enumerate}
Based on~(\ref{eq:s}),~(\ref{eq:ss}) and Corollary~\ref{cor:ii}.b), we conclude that $\widetilde{ F} = {\emph cl}({\G_{\emph id}})$ is a bisimulation, hence $\G \subseteq {\emph cl}({\G_{\emph id}}) \subseteq {\sim_\Gf}$.
}
%\end{todo}

\item[b)] $\Gf \not = \Bl$.
Based on the reduction rules implemented in {\CIRC},
it is quite easy to see that the initial set of goals $\G$ is a
$\vdashInd$-consequence of $\behspec \cup \efr{{F}}$.
%,where $\mathcal{F}$ is the set of hypotheses collected during a proof session.
In other words, $\G \subseteq \cl{{{F}}_{\emph id}}$.
So, if we anticipate a bit, we should show that
$\widetilde{ F}=\cl{{{F}}_{\emph id}}$ is a bisimulation,
\textit{i.e.}, according to Corollary~\ref{cor:ii},
$\behspec \cup \efr{{F}} \vdashInd \efr{\delta_{\GtrlG}({F})}$.
This is achieved by proving that
$\behspec \cup \efr{{F}} \vdashInd G_i\, (i \in \overline{0,n})$
(note that $\efr{\delta_{\GtrlG}({F})} \subseteq \bigcup_{i \in \overline{0,n}} \G_i$,
according to \textsf[Derive]).
The proof is by induction on $j$, where $n-j$ is the current
proof step, and by
case analysis on the {\CIRC} reduction rules applied at each step.
%(The full proof can be found in Appendix~\ref{anx:soundnessCirc}).

We further provide a sketch of the proof.\\
The \emph{base case} $j = n$ follows immediately, as $\behspec \cup \efr{{F}} \vdashInd \G_{n} = \emptyset$.\\
For the \emph{induction step} we proceed as follows. Let $\efr{e} \in \G_{j}$. If $\efr{e} \in \G_{{j+1}}$ then $\behspec \cup \efr{{F}} \vdashInd \efr{e}$ by the induction hypothesis. If $\efr{e} \not \in \G_{{j + 1}}$ then, for example, if \textsf{[Reduce]} was applied then it holds that $\behspec \cup { F_{j}} \vdashInd \efr{e}$. Recall that ${ F_{j}} \subseteq \efr{ F}$, so $\behspec \cup \efr{ F} \vdashInd \efr{e}$ also holds. The result follows in a similar fashion for the application of \textsf{[Derive]} or \textsf{[Simplify]}.
\end{itemize}
\qed
\end{proof}

\begin{remark}
The soundness of the proof system we describe in this chapter does not follow directly from Theorem 3 in \cite{rosu-lucanu-2009-calco}.
This is due to the fact that
we do not have an experiment-based definition of bisimilarity.
So, even though the mechanism we use for proving
$\behspec \cup \efr{{F}} \vdashInd \efr{\delta_{\GtrlG}({F})}$ (for the case $\Gf \not = \Bl$)
is similar to the one described in
\cite{rosu-lucanu-2009-calco},
the current soundness proof is conceived in terms of bisimulations
(and not experiments).
\end{remark}

\begin{remark}
The entailment relation $\vdashInd$ that {\CIRC} uses
for checking the equivalence of generalised regular expressions
is an instantiation of the parametric entailment
relation $\vdash$ from the proof system in \cite{rosu-lucanu-2009-calco}.
This approach allows {\CIRC} to reason automatically on a large class
of systems which can be modelled as non-deterministic coalgebras.
\end{remark}

As already stated, our final goal is to use {\CIRC} as a
decision procedure for the bisimilarity of generalised regular expressions.
That is, whenever provided a set of expressions, the prover stops
with a yes/no answer with respect to their equivalence.
In this context,
an important aspect is that the sub-coalgebra generated
by an expression $\eps \in {\Exp_{\Gf}}$ by repeatedly applying
$\delta_{\Gf}$ is, in general, infinite. Take for example
the non-deterministic functor \ans{C.1.26} $\Rf = \itB \times \id$ associated with infinite streams, and
consider the property
\(
\mu x . \emp \oplus r\langle x \rangle =
\mu x . r\langle x \rangle
\).
In order to prove this, {\CIRC} builds an infinite proof sequence
by repeatedly applying  $\delta_{\Rf}$ as follows:
\begin{center}
\begin{tabular}{rcl}
$\delta_{\Rf}(\mu x . \emp \oplus r\langle x \rangle)$ &
$=$ &
$\delta_{\Rf}(\mu x . r\langle x \rangle)$\\
& $\downarrow$ &\\
$\langle 0, \emp \oplus (\mu x . \emp \oplus r\langle x \rangle)\rangle$ &
$=$ &
$\langle 0, \mu x . r\langle x \rangle \rangle$\\[2ex]

$\delta_{\Rf}(\emp \oplus (\mu x . \emp \oplus r\langle x \rangle))$ &
$=$ &
$\delta_{\Rf}(\mu x . r\langle x \rangle)$\\
& $\downarrow$ &\\
$\langle 0, \emp \oplus \emp \oplus (\mu x . \emp \oplus r\langle x \rangle)\rangle$ &
$=$ &
$\langle 0, \mu x . r\langle x \rangle \rangle$ [\ldots\!]
\end{tabular}
\end{center}
\label{just:ACI}
In this case, the prover never stops. We observed in Section~\ref{sec:dp} that Theorem~\ref{thm:kleene} guarantees we can associate a finite coalgebra to a certain expression. In the proof of the aforementioned theorem, which is presented in \cite{brs_lmcs}, it is shown 
that the axioms for associativity, commutativity and idempotency (ACI) of $\oplus$
guarantee finiteness of the generated sub-coalgebra (note that these axioms have also been proven
sound with respect to bisimulation).
ACI properties can easily be specified in {\CIRC}
as the prover is an extension of Maude, which has a
powerful matching modulo ACUI (ACI plus unity) capability.
The idempotence is given by the equation $\eps \oplus \eps = \eps$, and
the commutativity and associativity are specified as attributes
of $\oplus$. It is interesting to remark that for the powerset functor termination is guaranteed without the axioms, because the coalgebra structure on the expressions for the powerset functor already includes ACI (since $\pow(\Exp)$ is itself a join-semilattice).

\begin{theorem}\label{thm:decProc}
Let $\G$ be a set of proof obligations over generalised regular expressions.
{\CIRC} can be used as a decision procedure for the equivalences in $\G$,
that is, it can \ans{C.1.27} decide whenever a goal $(\E_1, \E_2) \in \G$ is a true or false equality.
\end{theorem}

\begin{proof}
\ans{C.3.31}
Note that as proven in~\cite{brs_lmcs}, the ACI axioms for $\oplus$ guarantee that $\delta_\Gf$ is applied for a finite number of times in the generation of the sub-coalgebra associated with a $\Gf$-expression. Therefore, it straightforwardly follows that by implementing the ACI axioms in {\CIRC} (as attributes of $\oplus$), the set of new goals obtained by applying $\delta_\Gf$ is finite.
\begin{comment}
The result is a consequence of the fact that by implementing
the ACI axioms in {\CIRC},
the set of new goals obtained by repeatedly applying the derivative $\delta$
is finite.
\end{comment}
In these circumstances, whenever {\CIRC} stops according to the
reduction rule \textsf{[Done]},
the initial proof obligations are bisimilar. On the other hand, whenever it terminates
with \textsf{[Fail]}, the goals are not bisimilar.
%(The whole proof can be found in Appendix~\ref{anx:termination}.)
\qed
\end{proof}

\section{A {\CIRC}-based Tool}
\label{sec:caseStudy}

We have implemented a tool that, when provided with a functor $\Gf$,
automatically generates a specification for {\CIRC} which can then be used in order to automatically check whether two $\Gf$-expressions are bisimilar.

The tool is implemented as a metalanguage application in Maude. It can be downloaded from the address \url{http://goriac.info/tools/functorizer/}. In order to start the tool, one needs to launch Maude along with the extension Full-Maude and load the downloaded file using the command \code{in functorizer.maude .}

The general use case consists in providing the join-semilattices, the alphabets and the expressions. After these steps, the tool automatically checks if the provided expressions are guarded, closed and correctly typed. If this check succeeds, then it outputs a specification that can be further processed by {\CIRC}. In the end, the prover outputs either the bisimulation, if the expressions are equivalent, or a negative answer, otherwise.

We present two case studies in order to emphasise the high degree of generality for the types of systems we can handle, and show how the tool is used. 

\begin{example}
\label{eg:mealy}
%\textcolor{blue}{[More complex Mealy machines?]}
We consider the case of Mealy machines, which are coalgebras
for the functor $(\itB \times \id)^A$. 

Formally, a Mealy machine is a pair $(S,\alpha)$ consisting of a set
$S$ of states and a transition function $\alpha\colon S\to (\itB \times
S)^A$, which for each state $s\in S$ and input $a\in A$ associates an
output value $b$ and a next state $s'$. Typically, we write
$\alpha(s)(a) = (b,s') \Leftrightarrow
\xymatrix{*+[o][F]{s}\ar[r]^{a|b} &*+[o][F]{s'}}$.

\medskip
In this example and in what follows we will consider for the output the two-value join-semilatice  $\itB = \{0,1\}$ (with $\bottom_\Bl = 0$) and for the input alphabet $A = \{a,b\}$. 
The expressions for Mealy machines are given by the grammar:
\[
\begin{array}{rl}
E_{\phantom{0}} &::\!= \emp \mid x \mid E \oplus E \mid \mu x . E_{2} \mid a(r<E>) \mid b(r<E>) \mid a(l<E_1>) \mid b(l<E_1>) \\
E_{1} &::\!= \emp \mid E_1 \oplus E_1 \mid 0 \mid 1\\
E_{2} &::\!= \emp \mid E_2 \oplus E_2 \mid \mu x . E_{2} \mid a(r<E>) \mid b(r<E>) \mid a(l<E_1>) \mid b(l<E_1) \\
\end{array}
\]

Intuitively, an expression of shape $a(l<E_1>)$ specifies a state that for an input $a$ has an output value specified by $E_1$. For example, the expression $a(l<1>)$ specifies a state that for input $a$ outputs $1$, whereas in the case of $a(l<\emp>)$ the output is $0$. An expression of shape $a(r<E>)$ specifies a state that for a certain input $a$ has a transition to a new state represented by $E$. For example, the expression $\mu x. a(r<x>)$ states that for input $a$, the machine will perform a ``$a$-loop" transition, whereas $a(r<\emp>)$ states that for input $a$ there is a transition to the state denoted by $\emp$. It is interesting to note that a state will only be fully specified in what concerns transitions and output (for a given input $a$ if both $a(l<E_1>)$ and $a(r<E>)$ appear in the expression (combined by $\oplus$). In the case only transition (respectively, output) are specified, the underspecification is solved by setting the target state (respectively, output) to $\emp$ (respectively, $\bot_B = 0$). 
\end{example}

Next, to provide the reader with intuition, we will explain how one can reason on the bisimilarity of two simple expressions, by constructing bisimulation relations. Later on, we show how {\CIRC} can be used in conjunction with our tool in order to act as a decision procedure when checking equivalence of two expressions, in a fully automated manner.

%% THIS IS WRONG!! Melay machines have TOTAL functions, the transitions always have to be defined.
%Consider the following Mealy machines:
%
%\begin{figure}[h]
%\centering
%%\begin{wrapfigure}[7]{r}{.4\textwidth }
%$\xymatrix@C=1.9cm@R=0.5cm{
%*+[o][F]{s_1}\ar@(dr,ur)_{a|0} &
%*+[o][F]{s_2}
%}$
%\caption{Mealy machines: $s_1 \sim s_2$}
%\label{fig:mealy-a-noTrans}
%%\end{wrapfigure}
%\end{figure}
%
%First note that both $s_1$ and $s_2$ provide no observable behaviour. The output displayed by $s_1$ is $0$ -- the value $\bottom_\Bl$ has been set to. Hence, a transition (which for Mealy machines depends on both the input values in $A$ and the output values in $\Bl$) $\xymatrix{*+[o][F]{s}\ar[r]^{a|0} &*+[o][F]{s'}}$ is similar to the \emph{silent step} in process algebra, so no behaviour is observed externally. Obviously, $s_2$ displays no output as it performs no transitions.

%In terms of observable output, $s_1$ and $s_2$ behave the same, so we anticipate that $s_1$ and $s_2$ are bisimilar. 

We will start with the expressions $\eps_1 = \mu x . a(r<x>)$ and $\eps_2 = \emp$. We have to build a bisimulation relation  $\R$ on $\Gf$-expressions, such that  $(\eps_1, \eps_2) \in \R$. We do this in the following way: we start by taking $\R=\{(\E_1,\E_2)\}$ and we check whether this is already a bisimulation, by considering the output values and transitions and check whether no new expressions appear in this process. If new pairs of expressions appear we add them to $\R$ and repeat the process. Intuitively, this can be represented as follows:

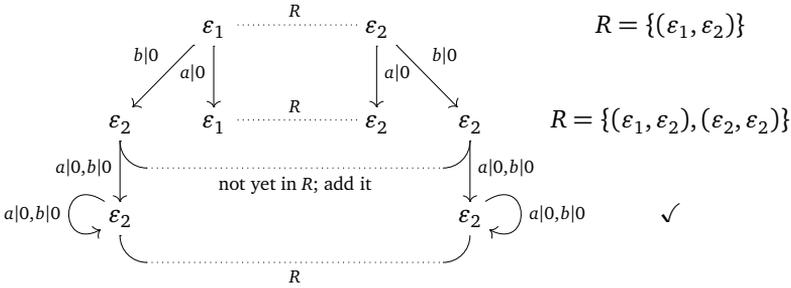
\begin{figure}[h]
\centering
%\begin{wrapfigure}[7]{r}{.4\textwidth }
$\xymatrix@C=0.7cm@R=0.7cm{
& {\eps_1}\ar[d]_{a|0}\ar@{..}[rr]^\R\ar[dl]_{b|0} & & {\eps_2}\ar[d]^{a|0}\ar[dr]^{b|0} & & \R = \{(\eps_1, \eps_2)\} \\
{\eps_2 }\ar[d]_{a|0,b|0} & {\eps_1}\ar@{..}[rr]^\R & & {\eps_2} & {\eps_2}\ar@{..} `d_l[llll] `[llll]^{\text{not yet in }\R \text{; add it}} [llll] \ar[d]^{a|0,b|0}& \R = \{(\eps_1, \eps_2), (\eps_2, \eps_2)\}\\
{\eps_2}\ar@(ul,dl)_{a|0,b|0}& & & & {\eps_2}\ar@(ur,dr)^{a|0,b|0} \ar@{..} `d_l[llll] `[llll]^\R [llll] & \checkmark
%*+[o][F]{\phantom{s_3}}\ar@(dr,ur)_{a|0}\ar@/^/[d]_{b|0} \\
%*+[o][F]{\phantom{s_4}}\ar@(dr,ur)_{b|0}\ar@(dl,ul)^{a|0} &
%*+[o][F]{s_2}\ar@(dr,ur)_{b|0}\ar@(dl,ul)^{a|0}
}$
\caption{Bisimulation construction}
\label{fig:mealy-a-noTrans-BC}
%\end{wrapfigure}
\end{figure}

In the figure above, and as before, we use the notation $\xymatrix@C=0.7cm@R=0.7cm{
 {\eps_1}\ar@{-}[r]^\R&\eps_2}$ to denote $(\eps_1,\eps_2)\in\R$.
As illustrated in Figure~\ref{fig:mealy-a-noTrans-BC}, $\R = \{(\eps_1, \eps_2), (\eps_2, \eps_2)\}$ is closed under transitions and is therefore a bisimulation. Hence, $\eps_1 \sim_\Gf \eps_2$.

The proved equality $\emp = \mu x . a(r< x >)$ might seem unexpected, if the reader is familiar with labelled transition systems. The equality is sound because these are expressions specifying behaviour of a Mealy machine and, semantically, both denote the function that for every non-empty word outputs $0$ (the semantics of Mealy machines is given by functions $B^{A^+}$, intuitively one can think of these expressions as both denoting the empty language). This is visible if one draws the automata corresponding to both expressions (say, for simplicity, the alphabet is $A=\{a\}$):
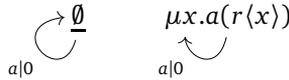
\[
\xymatrix{\emp\ar@(d,l)^{a|0} &  \mu x . a(r< x >) \ar@(d,l)^{a|0}}
\]
Note that (i) the $\emp$ expression for Mealy machines is mapped with $\delta$ to a function that for input $a$ gives $<0, \emp>$, which represents a state with an $a$-loop to itself and output $0$; (ii) the second expression specifies explicitly an $a$-loop to itself and it also has output $0$, since no output value is explicitly defined.  
Now, also note that similar expressions for labelled transition systems (LTS's), or coalgebras of the functor $\pow(-)^A$, would not be bisimilar since one would have an a-transition and the other one not. This is because the $\emp$ expression for LTS's really denotes a deadlock state. In operational terms they would be converted to the systems
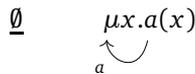
\[
\xymatrix{\emp &  \mu x . a( x ) \ar@(d,l)^a}
\] 
which now have an obvious difference in behaviour. 

By performing a similar reasoning as in the example above one can show that the expressions $\eps_1 = \mu x . a(r<x>) \oplus b(r<x>)$ and $\eps_2 = \mu x . a(r<x>)$ are bisimilar, and the bisimulation relation is built as illustrated in Figure~\ref{fig:mealy-ab-trans-BC}:

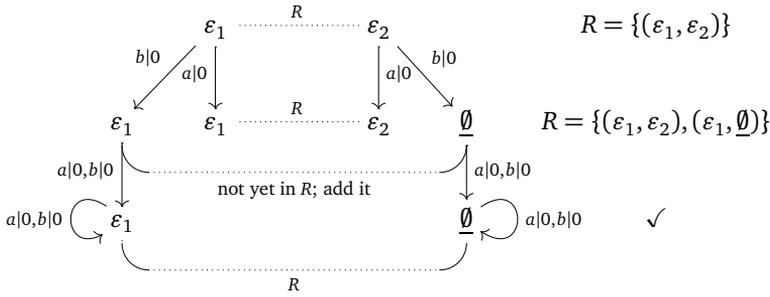
\begin{figure}[H]
\centering
%\begin{wrapfigure}[7]{r}{.4\textwidth }
$\xymatrix@C=0.7cm@R=0.7cm{
& {\eps_1}\ar[d]_{a|0}\ar@{..}[rr]^\R\ar[dl]_{b|0} & & {\eps_2}\ar[d]^{a|0}\ar[dr]^{b|0} & & \R = \{(\eps_1, \eps_2)\} \\
{\eps_1 }\ar[d]_{a|0,b|0} & {\eps_1}\ar@{..}[rr]^\R & & {\eps_2} & {\emp}\ar@{..} `d_l[llll] `[llll]^{\text{not yet in }\R \text{; add it}} [llll] \ar[d]^{a|0,b|0}& \R = \{(\eps_1, \eps_2), (\eps_1, \emp)\}\\
{\eps_1}\ar@(ul,dl)_{a|0,b|0}& & & & {\emp}\ar@(ur,dr)^{a|0,b|0} \ar@{..} `d_l[llll] `[llll]^\R [llll] & \checkmark
%*+[o][F]{\phantom{s_3}}\ar@(dr,ur)_{a|0}\ar@/^/[d]_{b|0} \\
%*+[o][F]{\phantom{s_4}}\ar@(dr,ur)_{b|0}\ar@(dl,ul)^{a|0} &
%*+[o][F]{s_2}\ar@(dr,ur)_{b|0}\ar@(dl,ul)^{a|0}
}$
\caption{Bisimulation construction}
\label{fig:mealy-ab-trans-BC}
%\end{wrapfigure}
\end{figure}

Let us further consider the Mealy machine depicted in Figure~\ref{fig:mealy1}, where all states are bisimilar.

\begin{figure}[H]
\centering
%\begin{wrapfigure}[7]{r}{.4\textwidth }
$\xymatrix@C=1.9cm@R=0.5cm{
*+[o][F]{s_1}\ar@/^/[r]_{a|0}\ar@/_/[d]^{b|1} &
*+[o][F]{\phantom{s_3}}\ar@(dr,ur)_{a|0}\ar@/^/[d]_{b|1} \\
*+[o][F]{\phantom{s_4}}\ar@(dr,ur)_{b|1}\ar@(dl,ul)^{a|0} &
*+[o][F]{s_2}\ar@(dr,ur)_{b|1}\ar@(dl,ul)^{a|0}
}$
\caption{Mealy machine: $s_1 \sim s_2$}
\label{fig:mealy1}
%\end{wrapfigure}
\end{figure}
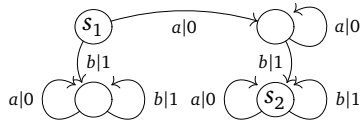

We show how to check the equivalence of two expression characterising the states $s_1$ and $s_2$, in a fully automated manner, using {\CIRC}. These expressions are
\ans{C.1.30}
$\eps_1 = \mu x . b(l<1>) \oplus b(r<\eps_2>) \oplus a(\mu y . a(r<y>) \oplus b(r<\eps_2>) \oplus b(l<1>))$ and $\eps_2 = \mu x . b(l<1>) \oplus b(r<x>) \oplus a(r<x>)$, respectively.

In order to check bisimilarity of $\eps_1$ and $\eps_2$ we
load the tool and define the semilattice $\Bl = \{0,1\}$ and the alphabet
${\emph A} = \{a, b\}$:
\\[0.5ex]
\code{
(jslt B is 0 1
  bottom 0 .
  0 v 0 = 0 .
  0 v 1 = 1 .
  1 v 1 = 1 .
endjslt)}\\
\code{(alph A is a b endalph)}

\smallskip

We provide the functor $\Gf$ using the command \code{(functor (B x Id)}\verb#^#\code{A .)}.
The command \code{(set goal ... .)} specifies the goal we want to prove:
\begin{alltt}
\fontsize{9}{10}
\selectfont(set goal
 \verb#\#mu X:FixpVar . b(l<1>) (+) a(l<0>) (+) b(r<X:FixpVar>) (+)
                 a(r<X:FixpVar>) =
 \verb#\#mu X:FixpVar . b(l<1>) (+) b(<\verb#\#mu X:FixpVar . b(l<1>) (+)
                 b(r<X:FixpVar>) (+) a(r<X:FixpVar>)>) (+)
                 a(\verb#\#mu Y:FixpVar . a(r<Y:FixpVar>) (+)
                 b(<\verb#\#mu X:FixpVar . b(l<1>) (+) a(l<0>) (+)
                 b(r<X:FixpVar>) (+) a(r<X:FixpVar>)>) (+) b(l<1>)) .)
\end{alltt}

\smallskip
In order to generate the {\CIRC} specification
we use the command \code{(generate coalgebra .)}.
Next we need to load {\CIRC} along with the resulting specification
and start the proof engine using the command
\code{(coinduction .)}.

As already shown, behind the scenes, {\CIRC} builds a bisimulation relation
that includes the initial goal.
The proof succeeds and the output consists of (a subset of) this bisimulation:
\begin{alltt}
\fontsize{9}{10}\selectfont{}Proof succeeded.
  Number of derived goals: 2
  Number of proving steps performed: 50
  Maximum number of proving steps is set to: 256

Proved properties:
- phi (+) (\verb#\#mu X . a(l<0>) (+) a(r<X>) (+) b(l<1>) (+) b(r<X>)) =
  phi (+) (\verb#\#mu Y . a(r<Y>) (+) b(l<1>) (+)
  b(r<\verb#\#mu X . a(l<0>) (+) a(r<X>) (+) b(l<1>)(+)b(r<X>)>))

- \verb#\#mu X . a(l<0>) (+) a(r<X>) (+) b(l<1>) (+) b(r<X>) =
  \verb#\#mu Z . a(r<\verb#\#mu Y . a(r<Y>) (+) b(l<1>) (+)
          b(r<\verb#\#mu X . a(l<0>) (+) a(r<X>) (+) b(l<1>) (+) b(r<X>)>)>) (+)
          b(l<1>) (+) b(r<\verb#\#mu X . a(l<0>) (+) a(r<X>) (+)
          b(l<1>) (+) b(r<X>)>)
 \end{alltt}

For the ease of understanding, here we printed a readable version of the proved properties. In Section~\ref{sec:code}, however, we show that internally each expression is brought to a canonical form by renaming the variables.
Moreover, note that in our tool, $\emp$ is represented by the constant \code{phi}. All the examples provided in the current section make use of this convention.

As previously mentioned, {\CIRC} is also able to detect when two
expressions are not equivalent. Take, for instance, the expressions
$\mu x . a(l<0>) \oplus a(r<a(l<1>) \oplus a(r<x>)>)$ and $a(l<0>) \oplus a(r<a(r< \mu x . a(r<x>) \oplus a(l<0>)>) \oplus a(l<1>)>)$, characterising the states $s_1$ and $s_3$ from the Mealy machines in Figure~\ref{fig:mealy2}. After following some steps similar to the ones previously enumerated, the proof fails and the output message is
\code{Visible goal [...] failed during coinduction}.

%\vspace{-5ex}
\begin{figure}[H]
\centering
$\xymatrix@C=1cm{
*+[o][F]{s_1}\ar@/^1pc/[r]^{a|0} &
*+[o][F]{s_2}\ar@/^1pc/[l]_{a|1} &
*+[o][F]{s_3}\ar@/^/[r]^{a|0} &
*+[o][F]{s_4}\ar@/^/[r]^{a|1} &
*+[o][F]{s_5}\ar@(dr,ur)_{a|0} & \\
}$
\caption{Mealy machines: $s_1 \not\sim s_3$}
\label{fig:mealy2}
\end{figure}
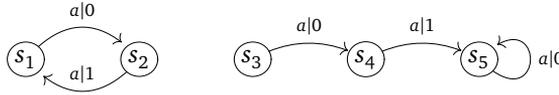

\begin{example}
\label{eg:ccs}

%\cite{brs_lmcs}

Next we show how to check strong bisimilarity of
non-deterministic processes of a non-trivial CCS-like language
with termination, deadlock, and divergence, as studied in \cite{Aceto:1992:TDD:147508.147527}. A process is a guarded, closed term defined by the following grammar:

\newcommand{\success}{\checkmark}
\newcommand{\empsum}{\textnormal{$\Omega$}}

\begin{eqnarray}\label{eq:ccs-gram}
P &::\!=& \success \mid \delta \mid \empsum \mid  a.P \mid P + P \mid x \mid \mu x . P
\end{eqnarray}

where:

\begin{itemize}
\item $\success$ is the constant for successful termination,
\item $\delta$ denotes deadlock,
\item $\empsum$ is the divergent computation (\emph{i.e.}, the undefined process),
\item $a.P$\, is the process executing the action $a$ and then continuing as the process $P$, for any action $a$ from a given set $A$,%not containing $\delta, \tick$ and $0$;
\item $P_1 + P_2$ is the non-deterministic process behaving as either $P_1$ or $P_2$, and
\item $\mu x . P$ is the recursive process $P[\mu x . P / x]$.
%\item $a$ is used as a shorthand notation for the process $a.\success$.
\end{itemize}

In~\cite{brs_lmcs} is is shown that, up to strong bisimilarity, the above syntax of processes is equivalent to the canonical set of (guarded, closed) regular expressions derived for the functor $1 \myplus {\mathcal P}_{\omega}(\id)^A$,

\newcommand{\unu}{1}

\[
\begin{array}{lcl}
E   & ::= & \emptyset \mid E \oplus E \mid x \mid \mu x.E \mid l[E_1] \mid r[E_2]\\
E_1 & ::= & \emptyset \mid  E_1 \oplus E_1 \mid \unu \\
E_2 & ::= & \emptyset \mid  E_2 \oplus E_2 \mid a(E_3)\\
E_3 & ::= & \emptyset \mid  E_3 \oplus E_3 \mid \{ E\}
\end{array}
\]

The translation map $(-)^\dagger$ from processes to expressions is defined by
induction on the structure of the process:
\[
\begin{array}{lcl@{\hspace{1.3cm}}lcl}
(\success)^\dagger &=& l[\unu] & (a.P)^\dagger &=& r[a(\{P^\dagger\})]\\
(\delta)^\dagger &=& r[\emptyset] & (P_1+P_2)^\dagger &=&  (P_1)^\dagger \oplus (P_2)^\dagger\\
(\empsum)^\dagger &=& \emptyset & (\mu x. P)^\dagger &=& \mu x. P^\dagger\\
x^\dagger  &=& x \,.
\end{array}
\]

Consider now two processes $P$ and $Q$ over the alphabet $A = \{ a, b \}$:
\[
\begin{array}{lcl}
P & = & \mu x. (a.x + a.P_1 + b.b.\success + b.(\delta + \empsum))\\
Q & = & \mu z. (a.z + b.(\delta + b.\success) + b.\delta)
\end{array}
\]
where $P_1 = \mu y. (a.(y+\delta)+b.\delta+b.(\delta+b.\success)+\delta)$.
Graphically, the two processes can be represented by the following labelled transition systems (for simplicity we omit annotating states with information regarding the satisfiability of successful termination, divergence, and deadlock):

\begin{comment}
\begin{figure}[H]
\centering
$\xymatrix@C=.7cm@R=.7cm{
  *+[o][F]{\delta} &
  *+[o][F]{P}\ar@(dl,ul)_{a}\ar[r]^{b}\ar[d]_{a}\ar[dr]^{b} &
  *+[o][F]{\phantom{A}}\ar@{~>}[r]\ar@{~>}[dr] &
  *+[o][F]{\delta} &
  &
  *+[o][F]{Q}\ar@(dr,ur)^{a}\ar[d]_{b}\ar[dr]^{b} &
  *+[o][F]{\delta} \\
%
  *+[o][F]{\phantom{A}}\ar@{~>}[u] &
  *+[o][F]{P_1}\ar@(dr,ur)^{a}\ar[l]_{b}\ar[d]_{b} &
  *+[o][F]{\phantom{A}}\ar[dr]^{b} &
  *+[o][F]{\empsum} &
  &
  *+[o][F]{\phantom{A}}\ar[d]_{b}\ar@{~>}[dr] &
  *+[o][F]{\phantom{A}}\ar@{~>}[u] \\
%
  *+[o][F]{\delta}&
  *+[o][F]{\phantom{A}}\ar@{~>}[l]\ar[r]^{b} &
  *+[o][F]{\success} &
  *+[o][F]{\success} &
  &
  *+[o][F]{\success} &
  *+[o][F]{\delta} \\
}$
\caption{Non-deterministic processes: $P \sim Q$}
\label{fig:pa}
\end{figure}
\end{comment}

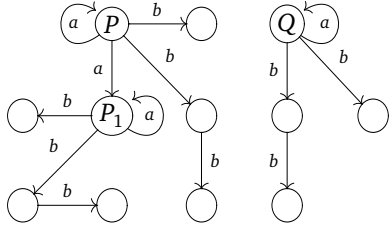
\begin{figure}[H]
\centering
$\xymatrix@C=.7cm@R=.7cm{
&
*+[o][F]{P}\ar@(dl,ul)_{a}\ar[r]^{b}\ar[d]_{a}\ar[dr]^{b} &
*+[o][F]{\phantom{A} } &
*+[o][F]{Q}\ar@(dr,ur)^{a}\ar[d]_{b}\ar[dr]^{b}\\
*+[o][F]{\phantom{A}} &
*+[o][F]{P_1}\ar@(dr,ur)^{a}\ar[l]_{b}\ar[dl]_{b} &
*+[o][F]{\phantom{A}}\ar[d]^{b} &
*+[o][F]{\phantom{A}}\ar[d]_{b} &
*+[o][F]{\phantom{A}}  \\
*+[o][F]{\phantom{A}}\ar[r]^{b} &
*+[o][F]{\phantom{A}} &
*+[o][F]{\phantom{A}} &
*+[o][F]{\phantom{A}} \\
}$
\caption{Non-deterministic processes: $Q \sim P$}
\label{fig:pa}
\end{figure}

We want to check if the process $P$ is strongly bisimilar to the process $Q$. By using the above translation, process $P$ is represented by the expression
\[
\begin{array}{l@{}l@{}l}
\mu x .
(&r[a(\{\mu y. (&r[a(\{ y \oplus r[\emptyset]\})] \oplus
                 r[b(\{r[\emptyset]\})] \oplus\\
 &              &r[b(\{r[\emptyset] \oplus r[b(\{l[\unu]\})] \})]\oplus
                r[\emptyset]
               )
   \})] \oplus \\
 & \multicolumn{2}{l}{ r[a(\{x\})] \oplus
   r[b(\{r[b(\{ l[\unu]\})] \})] \oplus
   r[b(\{r[\emptyset] \oplus \emptyset\})]
)}\\[1ex]
\end{array}
\]
whereas process $Q$ is represented by the expression
\[
\begin{array}{l@{}l@{}l}
\mu z .
 (&r[a(\{z\})] \oplus
        r[b(\{r[\emptyset] \oplus r[b(\{l[\unu]\})] \})] \oplus
        r[b(\{r[\emptyset]\})]
       ).
\end{array}
\]

%{\CIRC} can prove the equivalence of these two expressions after deriving 15 extra goals.

In order to use the tool, one needs to specify the semilattice, % (consider 1 standing for *),
the alphabet, the functor, and the goal in a manner similar to the one previously presented:
\\[0.5ex]
\code{
(jslt B is 1
  bottom 1 .
  1 v 1 = 1 .
endjslt)}\\
\code{(alph A is a b endalph)}\\
\code{(functor B + (P Id)}\verb#^#\code{A .)}
\begin{alltt}
\fontsize{9}{10}\selectfont(set goal \verb#\#mu X:FixpVar .
          r[ a( \{ X:FixpVar \} ) ] (+)
          r[ a( \{ \verb#\#mu Y:FixpVar .
                  r[ a( \{ Y:FixpVar (+) r[ phi ] \} ) ] (+)
                  r[ b( \{ r[ phi ] \} ) ] (+)
                  r[ b( \{ r[ phi ] (+) r[ b( \{ l[ 1 ] \} ) ] \} ) ] (+)
                  r[ phi ]
              \} )
          ] (+)
          r[ b( \{ r[ b( \{ l[ 1 ] \} ) ] \} ) ] (+)
          r[ b( \{ r[ phi ] (+) phi \} ) ]
          =
          \verb#\#mu Z:FixpVar .
          r[ a( \{ Z:FixpVar \} ) ] (+)
          r[ b( \{ r[ phi ] (+) r[ b( \{ l[ 1 ] \} ) ] \} ) ] (+)
          r[ b( \{ r[ phi ] \} ) ]  .)
\end{alltt}

For the generated specification {\CIRC} terminates and outputs a positive result:

\begin{alltt}
\fontsize{9}{10}\selectfont{}Proof succeeded.
  Number of derived goals: 15
  Number of proving steps performed: 58
  Maximum number of proving steps is set to: 256

Proved properties:
- r[phi] (+) (\verb#\#mu Y. r[phi] (+) r[a(\{r[phi] (+) Y\})] (+) r[b(\{r[phi]\})]
  (+) r[b(\{r[phi] (+) r[b(\{l[1]\})]\})])
  =
  \verb#\#mu Z. r[a(\{Z\})] (+) r[b(\{r[phi]\})] (+) r[b(\{r[phi] (+) r[b(\{l[1]\})]\})]
- r[b(\{l[1]\})] = r[phi] (+) r[b(\{l[1]\})]
- \verb#\#mu Y. r[phi] (+) r[a(\{r[phi] (+) Y\})] (+) r[b(\{r[phi]\})] (+)
  r[b(\{r[phi] (+) r[b(\{l[1]\})]\})]
  =
  \verb#\#mu Z. r[a(\{Z\})] (+) r[b(\{r[phi]\})] (+) r[b(\{r[phi] (+) r[b(\{l[1]\})]\})]
- \verb#\#mu X. r[a(\{X\})] (+) r[a(\{\verb#\#mu Y. r[phi] (+) r[a(\{r[phi] (+) Y\})] (+)
  r[b(\{r[phi]\})] (+) r[b(\{r[phi] (+) r[b(\{l[1]\})]\})]\})] (+)
  r[b(\{r[phi] + phi\})] (+) r[b(\{r[b(\{l[1]\})]\})]
  =
  \verb#\#mu Z. r[a(\{Z\})] (+) r[b(\{r[phi]\})] (+) r[b(\{r[phi] (+) r[b(\{l[1]\})]\})]
\end{alltt}

\end{example}

\subsection{Implementation}
\label{sec:code}

In this section we present details on the implementation of the algebraic specification given in Section~\ref{sec:algSpec}, based on the examples from Section~\ref{sec:caseStudy}.

In order to generate the algebraic specifications for {\CIRC} when provided a functor and two expressions, we used the Maude system \cite{DBLP:conf/maude/2007}. We choose it for its suitability for performing equational and rewriting logic based computations, and because of its reflective properties allowing for the development of advanced metalanguage applications. As the technical aspects on how to work at the meta-level are beyond the scope of this paper, we refrain from presenting them and show, instead, what the generated specifications consist of.

Most of the algebraic specifications from Section~\ref{sec:algSpec} have a straightforward implementation in Maude. Consider, for instance, the case of Mealy machines presented in Example~\ref{eg:mealy}. The generated grammars for functors (\ref{eq:fun-gram}) and expressions (Definition~\ref{def:expr}) are coded as:
\begin{alltt}
\fontsize{9}{10}\selectfont{}sort Functor .                          sorts Exp ExpStruct Alph Slt .
sorts AlphName SltName .                subsort Exp < ExpStruct .
subsort SltName < Functor .             enum A is a b . enum B is 0 1 .
                                        subsort A < Alph .
op A : -> AlphName .                    subsort B < Slt .
op B : -> SltName .                     
op G : -> Functor .                     op _`(+`)_ : Exp Exp -> Exp .
op Id : -> Functor .                    op _`(_`) : Alph Exp -> Exp .
op _+_ : Functor Functor -> Functor .   op \verb#\#mu_._ : FixpVar Exp -> Exp .
op _^_ : Functor AlphName -> Functor .  ops l<_> r<_> : Exp -> Exp .
op _x_ : Functor Functor -> Functor .   op phi : -> Exp .

                          eq G = (B x Id) ^ A .
\end{alltt}

Most of the syntactical constructs are Maude-specific: \code{sorts} and \code{subsort} declare the sorts we work with and, respectively, the relations between them; \code{op} declares operators; \code{eq} declares equations (the equation in our case defines the shape of the functor \code{G}). The only {\CIRC}-specific construct, \code{enum}, is syntactic sugar for declaring enumerable sorts, \emph{i.e.}, sorts that consist only of the specified constants. As a side note, if brackets (\code{(}, \code{[}, \code{\{})  are used in the declaration of an operation, then they must be preceded by a backquote (\code{`}).

\begin{comment}
As mentioned in Section~\ref{SCP-sec:prelim}, in order to guarantee the finiteness of our procedure, one needs to consider the ACI axioms for \code{(+)}. By turning on the axiomatisation flag using the command \code{(axioms on .)}, the following code is generated:
\end{comment}

As mentioned in Section~\ref{SCP-sec:prelim}, in order to guarantee the finiteness of our procedure, one needs to include the ACI axioms for \code{(+)}. \ans{C.1.5} Moreover, we have observed that the unity axiom for \code{(+)} plays an important role in decreasing the number of states generated by the repeated application of $\delta_{\Gf}$, therefore improving the overall time performance of the tool.
\ans{C.2.23} For example, the number of rewritings {\CIRC} performed in order to prove the bisimilarity of $\E_1$ and $\E_2$ in Figure~\ref{fig:mealy-ab-trans-BC} was halved when the unity axiom was used.

By turning on the axiomatisation flag using the command \code{(axioms on .)}, the following code is generated:

\begin{alltt}
\fontsize{9}{10}\selectfont{}op _`(+`)_ : Exp Exp -> Exp [assoc comm] .
eq E:Exp (+) E:Exp = E:Exp .
eq E:Exp (+) phi = E:Exp .
\end{alltt}

It is an obvious question why not to add other axioms to the tool, since the unity axiom has improved performance. At this stage we have not studied in detail how much adding other axioms would help. It is in any case a trade-off on how many extra axioms one should include, which will get the automaton produced from an expression closer to the minimal automaton, and how much time the tool will take to reduce the expressions in each step modulo the axioms. For classical regular expressions, there is an interesting empirical study on this~\cite{derivatives-jfp09}. We leave it as future work to carry on a similar study for our expressions and axioms.

\begin{comment}
\noindent
Note that we also add the unity axiom as we observed that it significantly improves the time performance of the tool.
\end{comment}

The process of substituting fixed-point variables has a natural implementation. We present the equations handling the basic expressions $\emp$ and $x$, and the operation \code{(+)}:
\begin{alltt}
\fontsize{9}{10}\selectfont{}op _`[_/_`] : Exp Exp FixpVar -> Exp .
eq phi [ E:Exp / X:FixpVar ] = phi .
ceq Y:FixpVar [ E:Exp / X:FixpVar ] = E:Exp if (X:FixpVar == Y:FixpVar) .
eq Y:FixpVar [ E:Exp / X:FixpVar ] = Y:FixpVar [owise] .
eq (E1:Exp (+) E2:Exp) [ E:Exp / X:FixpVar ] = 
   (E1:Exp [E:Exp / X:FixpVar]) (+) (E2:Exp [E:Exp / X:FixpVar]) .
\end{alltt}
  
In order to avoid matching problems and to overcome the fact that in Maude one cannot handle an equation that has fresh variables in its right-hand-side (\emph{i.e.}, they do not appear in the left-hand-side), we replace expression variables with parameterised constants: \code{op var : Nat -> FixpVar .} The operation that obtains this canonical form has an inductive definition on the structure of the given expression and makes use of the substitution operation presented above. For this reason, the bisimulation {\CIRC} builds contains parameterised constants instead of the user declared variables. The property proved in Example~\ref{eg:ccs} is, therefore, written as:

\begin{alltt}
\fontsize{9}{10}\selectfont{}\verb#\#mu var(2) . r[a(\{var(2)\})] (+) r[a(\{\verb#\#mu var(1) . r[phi] (+)
r[a(\{r[phi] (+) var(1)\})] (+) r[b(\{r[phi]\})] (+) r[b(\{r[phi] (+)
r[b(\{l[1]\})]\})]\})] (+) r[b(\{r[phi] (+) phi\})] (+) r[b(\{r[b(\{l[1]\})]\})]
=
\verb#\#mu var(1) . r[a(\{var(1)\})] (+) r[b(\{r[phi]\})] (+)
r[b(\{r[phi] (+) r[b(\{l[1]\})]\})]
\end{alltt}

The most important part of the algebraic specification consists of the equations defining the operations $\delta\_(\_)$, $\emph{Plus}\_(\_,\_)$, and $\emph{Empty}$. Most of these equations are implemented as presented in \cite{brs_lmcs}. The only difficulties we encountered were for the exponentiation case, as Maude does not handle higher-order functions. Without entering into details, as a workaround, we introduced a new sort \code{Function < ExpStruct} and an operation \code{${\mathbf \textbackslash}$. : ExpoCase Alph Functor ExpStruct -> Function} in order to emulate function-passing. The first argument is used to memorize the origin where the exponentiation ingredient is encountered: $\delta$, $\emph{Plus}$, or $\emph{Empty}$. Its purpose is purely technical -- we use it in order to avoid some internal matching problems. The other three parameters are those of the structured expression $\lambda {.}(a,\FtrlG,\sigma)$ presented in Section~\ref{sec:algSpec}: a letter in the alphabet, an ingredient, and some other structured expression.

Another thing worth describing is the way we enable {\CIRC} to prove equivalences when the powerset functor occurs. Namely, we present how interpolant (\ref{srl:pow}) is implemented. Recall that we want to show that two sets of expressions are equivalent, which means that for each expression in the first set there must be an equivalent one in the second set and vice-versa.

\newcommand{\largeusq}{{\fontsize{12}{13}\selectfont\_}}
In order to handle sets of structured expressions we introduce a new sort, \code{ExpStructSet} as a supersort for \code{ExpStruct}. We also consider the set separator \code{\largeusq,\largeusq \,\,: ExpStructSet ExpStructSet -> ExpStructSet [assoc,comm]}, the empty set \code{emptyS : -> ExpStructSet}, and the set wrapping operation \code{\{\largeusq\} : ExpStructSet -> ExpStruct}. In order to mimic universal quantification over a set, we use a special constant referred to as token ``\code{[/]}''. In what follows, we consider two variables of sort \code{ExpStructSet}: \code{ES} and \code{ES'}, and two variables of sort \code{ExpStructSet}: \code{ESS} and \code{ESS'}. We now describe the process of finding the equivalence between two sets:
\begin{itemize}
\item whenever encountering two wrapped expression sets we add the universal quantification token to each of them in two distinct goals:
\begin{alltt}
\fontsize{9}{10}\selectfont srl \{ESS\} = \{ESS'\} => \{[/] ESS\} = \{ESS'\} /\verb#\# \{ESS\} = \{[/] ESS'\} .
\end{alltt}
\item iterate through the expressions on the left-hand-side (similarly for the other direction):
\begin{alltt}
\fontsize{9}{10}\selectfont srl \{[/] (ES , ESS)\} = \{ESS'\} =>
     \{[/] ES\} = \{ESS'\} /\verb#\# \{[/] ESS\} = \{ESS'\} .
 srl \{ESS\} = \{[/] (ES' , ESS')\} =>
     \{ESS\} = \{[/] ES'\} /\verb#\# \{ESS\} = \{[/] ESS'\} .
\end{alltt}
\item when left with one expression on the left-hand-side, start iterating through the expressions on the right-hand-side until finding an equivalence (similarly for the other direction):
\begin{alltt}
\fontsize{9}{10}\selectfont srl \{[/] ES\} = \{ES' , ESS'\} => ES = ES' \verb#\#/ \{[/] ES\} = \{ESS'\} .
 srl \{ES , ESS\} = \{[/] ES'\}  => ES = ES' \verb#\#/ \{ESS\} = \{[/] ES'\} .
\end{alltt}
\item if no equivalence has been found, transform the current goal into a visible failure:
\begin{alltt}
\fontsize{9}{10}\selectfont srl \{ESS\} = emptyS => true = false .
 srl emptyS = \{ESS\} => true = false .
\end{alltt}

\end{itemize}

Finally, the type checker for structured expressions has a straightforward implementation. Its code does not appear in the generated specification as it is only used when the tool receives the expressions as input. This prevents obtaining the specification and starting the prover in case invalid expressions are provided.

\section{Discussion}
\label{SCP-sec:concl}
In this chapter we provided a decision procedure for the bisimilarity of generalised regular expressions. In order to enable the implementation of the decision procedure, we have exploited an encoding of coalgebra into algebra, and we formalised the equivalence between the coalgebraic concepts associated with non-deterministic coalgebras \cite{brs_lmcs} and their algebraic correspondents. This led to the definition of algebraic specifications (\algspec) that model both the language and the coalgebraic structure of expressions.
Moreover, we defined an equational deduction relation ($\vdashInd$), used on the algebraic side for reasoning on the bisimilarity of expressions.

The most important result of the parallel between the coalgebraic and algebraic approaches is given in Corollary~\ref{cor:ii}, which formalises the definition of the bisimulation relations in algebraic terms. Actually, this result is the key for proving the soundness of the decision procedure implemented in the automated prover {\CIRC} \cite{lucanu-etal-2009-calco}. As a coinductive prover, {\CIRC} builds a relation $ F$ closed under the application of $\delta_\itG$ with respect to $\vdashInd$ ($\algspec \cup \efr{ F} \vdashInd \efr{\delta_\itG({ F})}$), hence automatically computing a bisimulation the initial proof obligations belong to.

The approach we present in this chapter enables {\CIRC} to perform
reasoning based on bisimulations (instead of experiments
\cite{rosu-lucanu-2009-calco}). This way, the prover is extended to
checking bisimilarity in a large class of systems that can be
modelled as non-deterministic coalgebras. Note that the constructions above are all automated -- the (non-trivial) {\CIRC} algebraic specification describing $\algspec$, together with the interpolants implementing $\vdashInd$ are generated with the Maude tool presented in Section~\ref{sec:caseStudy}.

\newpage\null\thispagestyle{empty}\newpage
\chapter{Decorated trace and testing semantics coalgebraically}
\label{ch:dec-trace-testing}

The study of behavioural equivalence of systems has been a research
topic in concurrency for many years now. For different kinds of systems, several types of behavioural equivalences and preorders have been proposed
throughout the years, each suitable for use in different
contexts of application.

In Chapter~\ref{ch:dec-bisim} we showed how (co)algebras can be used in order to model and reason on bisimilarity of expressions describing non-deterministic systems.

The focus of this chapter is on a suite of other semantics of interest for labelled transition systems (LTS's), generative probabilistic systems (GPS's) and labelled transition systems with divergence. More explicitly, we consider \emph{decorated trace semantics} including ready, failure, (complete) trace, possible-futures, ready trace and failure trace for LTS's, as described in~\cite{Glabbeek01} and ready, (maximal) failure and (maximal) trace for GPS's, as introduced in~\cite{GPS-Jou-Smolka}. For the case of divergent LTS's, the emphasis is on \emph{must} and \emph{may testing semantics}~\cite{CleavelandH89}.

In short, our approach consists in providing a coalgebraic modelling of the aforementioned systems and their semantics. The latter are derived by employing the generalised powerset construction~\cite{gen-pow} and proved equivalent with their counterparts as defined in~\cite{CleavelandH89,Glabbeek01,GPS-Jou-Smolka}. This further allows reasoning on the corresponding notions of behavioural equivalence/preorder in terms of (Moore-) bisimulations.

We further provide the intuition behind decorated trace and testing semantics.

At the left-hand side of Figure~\ref{fig:lattice} we illustrate the hierarchy (based on the coarseness level) among bisimilarity, ready, failure, (complete) trace, possible-futures, ready trace and failure trace semantics for LTS's, as introduced in ~\cite{Glabbeek01}.
On the right-hand side a similar hierarchy is depicted for bisimilarity, ready, (maximal) failure and (maximal) trace semantics for GPS's, as in~\cite{GPS-Jou-Smolka}.
For example, for both types of systems, bisimilarity (the standard behavioural equivalence on $\F$-coalgebras) is the finest of the semantics, whereas trace semantics is the coarsest one. Moreover, note that for the case of GPS's, maximality does not yield more 
distinguishing power and ready and failure semantics are equivalent.
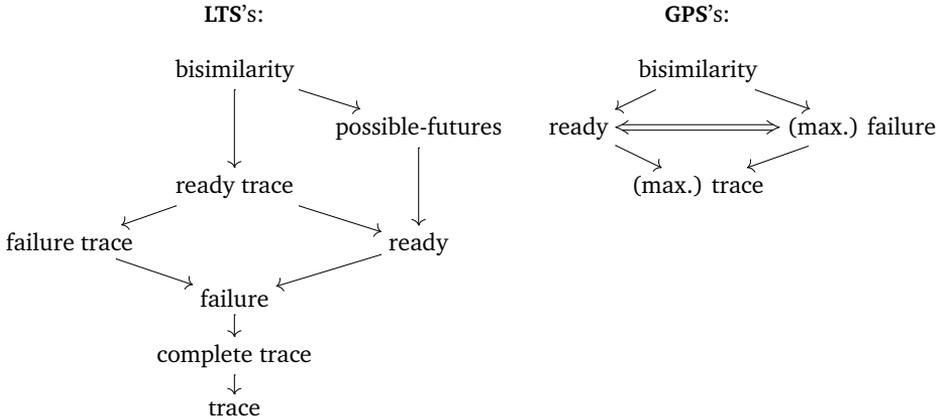
\begin{figure}%[H]
\centering
\[
\hspace{-9pt}
\small
\xymatrix@C=.1cm@R=.25cm{
& \textnormal{{$\mathbf {LTS}$}'s:} & & & & \textnormal{{$\mathbf {GPS}$}'s:} &\\
& \textnormal{bisimilarity}\ar[dd]\ar[dr] & & & & \textnormal{bisimilarity}\ar[dr]\ar[dl] &\\
& & \textnormal{possible-futures}\ar[dd] && \textnormal{ready}\ar[dr] \ar@{<=>}[rr]&& \textnormal{(max.) failure}\ar[dl]\\
& \textnormal{ready trace}\ar[dl]\ar[dr] & && & \textnormal{(max.) trace} &\\
\textnormal{failure trace}\ar[dr] & & \textnormal{ready}\ar[dl] && &  &\\
& \textnormal{failure}\ar[d] & && & &\\
& \textnormal{complete trace}\ar[d] & &&\\
& \textnormal{trace} & &&
}
\]
\caption{Lattices of semantic equivalences for LTS's and GPS's.}
\label{fig:lattice}
\end{figure}

%\begin{figure}[H]
%\centering
%\[
%\hspace{-18pt}
%\xymatrix@C=.2cm@R=.25cm{
%& \textnormal{{$\mathbf {LTS}$}'s:} & & & & &\\
%& \textnormal{bisimilarity}\ar[dd]\ar[dr] & & & & &\\
%& & \textnormal{possible-futures}\ar[dd] &&&&\\
%& \textnormal{ready trace}\ar[dl]\ar[dr] & && & &\\
%\textnormal{failure trace}\ar[dr] & & \textnormal{ready}\ar[dl] && &  &\\
%& \textnormal{failure}\ar[d] & && & &\\
%& \textnormal{complete trace}\ar[d] & &&\\
%& \textnormal{trace} & &&\\\\
%& \textnormal{{$\mathbf {GPS}$}'s:} &\\
%& \textnormal{bisimilarity}\ar[dr]\ar[dl] &\\
%\textnormal{ready}\ar[dr] \ar@{<->}[rr]|-{\ = \ }&& \textnormal{(max.) failure}\ar[dl]\\
%& \textnormal{(max.) trace} &
%}
%\]
%\caption{Lattices of semantic equivalences for LTS's and GPS's.}
%\label{fig:lattice}
%\end{figure}

In order to get some intuition on the type of distinctions the
equivalences above encompass, consider the following
LTS's: %over the alphabet $A = \{a, b, c\}$:
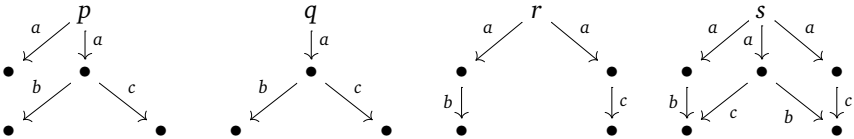
\vspace*{-.15cm}
\[
\vspace*{-.15cm}
\xymatrix@C=.6cm@R=.4cm{& p \ar[d]^{a} \ar[dl]_{a} & && q \ar[d]^{a} & && r \ar[dr]^{a}\ar[dl]_{a} & && s \ar[d]_a \ar[dr]^{a}\ar[dl]_{a} \\
\bullet &\bullet\ar[dr]^{c}\ar[dl]_{b}& & & \bullet\ar[dr]^{c}\ar[dl]_{b}& &\bullet \ar[d]_{b}&& \bullet\ar[d]^{c} &   \bullet\ar[d]_b&\bullet \ar[rd]_b\ar[ld]^c& \bullet\ar[d]^c\\
 \bullet&& \bullet & \bullet&&\bullet& \bullet&  & \bullet&\bullet&&\bullet }
\]
None of the top states of the systems above are bisimilar. The state $p$ is the only one among the four in which an action $a$ can lead to a deadlock state, whereas $q, r$ and $s$ have a different branching structures.

The traces of the states $p, q, r$ and $s$ are $\{a, ab, ac\}$, and therefore they are
all trace equivalent.
Of the four states above, $q$ and $r$ and $s$ are complete trace equivalent as they can execute the same traces that lead to states where no further action are possible, whereas $p$ is the only state that can trigger $a$ and terminate.

Ready (respectively, failure) semantics identifies states according to the set of
actions
they can (respectively, fail to) trigger immediately after a certain trace has been executed. None of the states above are ready equivalent; for example, after the execution of action $a$, process $p$ can reach a deadlock state whereas $q$ has always to choose between actions $b$ and $c$. Orthogonally, only $r$ and $s$ are failure equivalent.

Possible-futures semantics identifies states that can perform the same traces $w$ and, moreover, the states reached by executing such $w$'s are trace equivalent. None of the states above are possible-futures equivalent. For example, after triggering action $a$, $p$ can reach a deadlock state (with no further behaviour) whereas $q$ can execute the set of traces $\{b, c\}$.

Ready (respectively failure) trace semantics identifies states that can trigger the same traces $w$ and the (pairwise-taken) intermediate states determined by such $w$'s are ready (respectively refuse) to trigger the same sets of actions.
None of the systems above is ready trace equivalent. For example, after performing action $a$, process $q$ reaches a state that is ready to trigger both $b$ and $c$, whereas $r$ cannot. The analysis on failure trace equivalence follows a similar reasoning, but different results.

The corresponding semantic equivalences in Figure~\ref{fig:lattice} distinguish between $p, q, r$ and $s$ as summarised in the table below:
\[
\begin{array}{r | c | c | c | c | c | c |}
& p,q & p, r & p,s & q,r & q,s & r,s\\
\hline

\textnormal{bisimilarity} & \times & \times & \times & \times & \times & \times\\ 
\hline

\textnormal{trace} & \checkmark & \checkmark & \checkmark & \checkmark & \checkmark & \checkmark\\
\hline

\textnormal{complete trace} & \times & \times & \times & \checkmark & \checkmark & \checkmark\\
\hline

\textnormal{ready} &  \times & \times & \times & \times & \times & \times\\ 
\hline

\textnormal{failure} & \times & \times & \times & \times & \times & \checkmark\\ 
\hline

\textnormal{possible-futures} & \times & \times & \times & \times & \times & \times\\ 
\hline

\textnormal{ready trace} & \times & \times & \times & \times & \times & \times\\ 
\hline

\textnormal{failure trace} & \times & \times & \times & \times & \times & \checkmark\\ 
\hline

\end{array}
\]
where $\checkmark$ to stands for a ``yes'' answer with respect to the behavioural equivalence of two of the states $p,q,r$ and $s$, whereas $\times$ represents a ``no'' answer.

Intuitively, GPS's resemble LTS's, with the difference that each transition is labelled by both an action and the probability of that action being executed.
For more insight on decorated trace semantics for GPS's, consider the following systems:
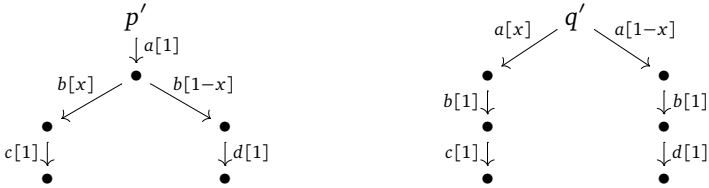
\[
\hspace{-3.5pt}
\xymatrix@C=.72cm@R=.3cm{
& p'\ar[d]^{a[1]} && \qquad\qquad& &  q'\ar[dl]_{a[x]}\ar[dr]^{a[1-x]}% & & & r'\ar[d]^{a[1]} & & & s'\ar[dl]_{a[x]}\ar[dr]^{a[1-x]}
\\
& \bullet\ar[dl]_{b[x]}\ar[dr]^{b[1-x]} && & \bullet\ar[d]_{{b[1]}} & & \bullet\ar[d]^{b[1]} &% & \bullet\ar[dl]_{b[x]}\ar[dr]^{c[1-x]} %& & \bullet\ar[d]_{b[1]} & & \bullet\ar[d]^{c[1]}
\\
\bullet\ar[d]_{c[1]} & & \bullet\ar[d]^{d[1]} && \bullet\ar[d]_{c[1]} & & \bullet\ar[d]^{d[1]} %& \bullet & & \bullet & \bullet & & \bullet
\\
\bullet & & \bullet && \bullet & & \bullet
}
\]

In the setting of GPS's, decorated trace semantics take into consideration paths $w$ which can be executed by a probabilistic process $p$. Reasoning on the corresponding equivalences is based on the sum of probabilities of occurrence of such $w$'s that, for example, lead $p$ to a set of processes, for the case of trace semantics, or to a set of processes that (fail to) trigger the same sets of actions as a first step, for ready (respectively, failure) semantics.

In~\cite{GPS-Jou-Smolka} a notion of \emph{maximality} was introduced for the case of trace and failure semantics. Intuitively, the former takes into consideration the probability of a process $p$ to execute a certain trace $w$ and terminate, whereas the latter takes into consideration the largest set of actions $p$ fails to trigger as a first step after the execution of $w$.
However,
it has been proven in~\cite{GPS-Jou-Smolka} that maximality does not increase the distinguishing power of decorated trace semantics and, moreover, ready and failure equivalence of GPS's coincide.

With respect to (maximal) trace semantics, amongst the systems above, $p'$ and $q'$ are equivalent: they have the same probability of executing traces $w \in \{\varepsilon, a, ab, abc, abd\}$. Moreover, each such $w$ leads $p'$ and $q'$ to sets of processes $S_{1}, S_{2}$ ready to fire the same actions. Consequently, $S_{1}$ and $S_{2}$ fail to trigger the same sets of actions as a first step. Hence, $p'$ and $q'$  are both ready and maximal failure equivalent at the same time. None of the processes above are bisimilar: the corresponding states reached via transitions labelled $a$ (with total probability $1$) display different behaviour as they either have different branching structure, or can trigger different actions.

Orthogonally, as previously stated, in this chapter we also focus on providing a coalgebraic modelling of must and may testing semantics for divergent LTS's.

Intuitively, in the setting of testing semantics, fixed a set of tests, two systems are deemed to be equivalent if they
pass exactly the same tests. With concurrent non-deterministic
processes, a system may pass a test in some, but not all, its
executions. This leads to the definitions of \emph{may testing} (a
system may pass a test in some execution) and \emph{must testing} (a
system must pass a test in all its executions).

However, alternative trace-based characterisations of must and may testing were provided in~\cite{CleavelandH89,DBLP:journals/tcs/NicolaH84,Hennessy:1988:ATP:50497}. Intuitively, must testing preorder abstracts from infinite internal computations. It relates two processes p and q only if, for each trace w, whenever p
does not engage in divergent behaviour in its attempt to execute w, then 
so does q. Moreover, $q$ has to be ``less non-deterministic'' than $p$ -- a property established based on the inclusion of the acceptance (ready) sets associated with $q$ and $p$, respectively.
Two processes are must equivalent whenever the must preorder relates them in both directions.
May testing preorder (respectively, equivalence) coincides with the usual language inclusion (respectively, equality).

Consider for an example the following two systems, where $\tau$ is used to represent an internal computation step:

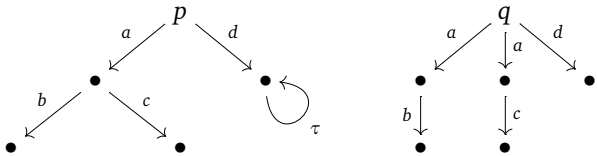
\[
\xymatrix@C=.72cm@R=.5cm{
& & p\ar[dl]_{a}\ar[dr]^{d} & & & & q\ar[dl]_{a}\ar[d]^{a}\ar[dr]^{d} &\\
& \bullet\ar[dl]_{b}\ar[dr]^{c} & & \bullet\ar@(d,r)_\tau & & \bullet\ar[d]_{b} & \bullet\ar[d]^{c} & \bullet\\
\bullet & & \bullet & & & \bullet & \bullet & &\\
}
\]

Processes $p$ and $q$ cannot be related in terms of the must testing semantics. On the one hand, $q$ does not diverge with respect to action $d$, whereas $p$ diverges. On the other hand, $p$ is less non-deterministic than $q$, as the ready set $\{\{b, c\}\}$ of $p$ after performing action $a$ is not included in the ready set $\{\{b\}, \{c\}\}$ of $q$. However, $p$ and $q$ are may testing equivalent as they both execute the same sets of (visible) traces $\{\eps, a, d, ab, ac\}$.

In this chapter we show how decorated trace, must and may testing semantics can be recovered in a coalgebraic setting by employing the generalised powerset construction in~\cite{gen-pow}. The derived coalgebraic characterisations leads to canonical representatives in terms of final Moore automata which further enabled reasoning by constructing bisimulations witnessing the desired notion of behavioural equivalence/preorder. Moreover, as we also saw in the previous chapter, this result is interesting from the point of view of tool development as well: 
construction of bisimulations is known to be particularly suitable for automation.

It is also interesting to observe that the spectrum of decorated trace semantics in Figure~\ref{fig:lattice} can be recovered from our coalgebraic modelling. The procedure is briefly summarised in Section~\ref{sec:recover-spectrum}, for the case of failure and complete trace semantics for LTS's, and ready and trace semantics for GPS's, respectively.
 
\textit{Organisation of the chapter.}
In Section~\ref{sec:dec-tr-lts} and Section~\ref{sec:dec-tr-gps}, we show how the powerset construction can be applied for determinising LTS's and GPS's, respectively, in terms of Moore automata $(X, f\colon X \to B \times X^A)$, in order to coalgebraically characterise the corresponding decorated trace semantics. Here we also prove that the obtained coalgebraic models are equivalent to the original definitions, and 
illustrate how one can reason about decorated trace equivalence by constructing (Moore) bisimulations.
A compact overview on the uniform coalgebraic framework is given in Section~\ref{sec:nuthsell}.
Section~\ref{sec:canon} discusses that the canonical representatives of LTS's and GPS's we obtain coalgebraically coincide with the corresponding minimal automata one would obtain by identifying all states equivalent with respect to a particular decorated trace semantics.
In Section~\ref{sec:recover-spectrum} we show that the spectrum of decorated trace semantics can be obtained from the coalgebraic modelling.
A coalgebraic modelling of may and must testing semantics, respectively, is provided in Section~\ref{sec:testing-semantics} by exploiting extensions of trace and failure semantics, respectively, to the context of LTS's with internal behaviour.
Finally, Section~\ref{sec:concl-traces} contains a summary of the results in this chapter.

\section{Decorated trace semantics of LTS's}
\label{sec:dec-tr-lts}

In this section, our aim is to provide a coalgebraic view on decorated trace equivalences of labelled transition systems (LTS's).
We  use the generalised powerset construction and
show how one can determinise arbitrary LTS's obtaining particular instances of Moore automata
(with different output sets) in order to model ready, failure, (complete) trace, possible-futures, ready trace and failure trace equivalences.
This paves the way to building a general framework for reasoning on decorated trace equivalences in a uniform fashion, in
terms of bisimulations (up-to context).

\def\tr#1{\stackrel{#1}{\to}}

Note that our results are derived in the context of image finite LTS's, in accordance with the setting proposed in~\cite{Glabbeek01}.
An LTS is a pair $(X, \delta)$ where $X$ is a set
of states and $\delta\colon X \to (\pow X)^A$ is a function assigning to
each state $x\in X$ and to each label $a\in A$ a finite set of possible
successors states. We write $x \xrightarrow{a} y$ whenever $y \in \delta(x)(a)$. We extend the notion of transition to words $w = a_1\ldots a_n \in A^*$ as follows: $x \xrightarrow{w} y$ if and only if $x \xrightarrow{a_1} \ldots \xrightarrow{a_n} y$. For $w = \eps$, we have $x \xrightarrow{\eps} y$ if and only if $y = x$.

The coalgebraic characterisation of ready, failure and (complete) trace was originally obtained in~\cite{gen-pow}. We recall it here, with a slight adaptation which will be useful for the generalisations we will explore. Given an arbitrary LTS
\[
(X, \delta \colon  X \rightarrow (\Powf X)^A),
\]
one constructs a \emph{decorated} LTS, which is a coalgebra of the functor $\F_{\I}(X) = \B_\I \times (\Powf X)^A$. More precisely, we construct 
\[
(X, <\overline{o}_\I, \delta> \colon  X \rightarrow \B_\I \times (\Powf X)^A),
\]
where the output operation
\[
\overline{o}_\I \colon X \rightarrow \B_\I
\]
provides the observations  of interest (the decorations) corresponding to the original LTS and depending on the equivalence ($\I$) we want to study.
Note that both the output operation and its codomain are parameterised by $\I$.
%and, moreover, $\B_{\I}$ carries a semilattice structure.

Then, the decorated LTS is determinised as depicted in Figure~\ref{fig:GenSettMoore}, according to the powerset construction summarised in diagram~(\ref{F-final}) in Section~\ref{prelim:gen-pow}.
Recall that the {generalised powerset construction} is applied in the framework of coalgebras $f\colon X \to \F T(X)$ for a functor $\F$ and a
monad $T$, with $\F T(X)$ a $T$-algebra. Intuitively, monads are used to hide computational effects such as non-determinism, whereas the requirement that $\F T(X)$ is an algebra for $T$ guarantees the unique extension of $f$ to a $T$-algebra homomorphism $f^{\sharp}$ representing a new coalgebra with state space hiding the computational effects.
Consequently, this extension enables reasoning on $\F$-equivalence in the coalgebra $f^{\sharp}$, rather than reasoning on the (finer) $\F T$-equivalence in the coalgebra $f$. 
%
%Consequently, this extension enables reasoning on $\F$-equivalence in the coalgebra $f^{\sharp}$, rather than reasoning on the (finer) $\F T$-equivalence in the coalgebra $f$. Equivalently, this allows shifting from verifying bisimilarity of decorated LTS's to verifying language equivalence of Moore automata accepting the same languages as the LTS's.

For the case of decorated LTS's, we instantiate $T$ with the powerset monad $(\pow, \eta, \mu)$ such that $\eta(x) = \{x\}$ and $\mu(U) = \bigcup_{S_{i}\in U}S_{i}$, and $\F$ with $\F_{\I} = B_{\I} \times (\pow(-))^{A}$. Moreover, $\F T(X)$ carries a $T$-algebra structure, that is a semilattice, as $\Pow(X)$ and $B_\I$ are a semilattices (as we shall see later, for each of the semantics $\I$), and product and exponentiation preserve the algebra structure. We will see that this extension enables shifting from reasoning on bisimilarity of decorated LTS's to reasoning on the (coarser) language (trace) equivalence.
Note that the semilattice structures ensure the existence of least upper bounds, which further enable the definition of $f^{\sharp} = \langle o, t\rangle$ and $\llbracket - \rrbracket$ as semilattice morphisms.
In Figure~\ref{fig:GenSettMoore} we use $\bigsqcup$ to denote both the operation of $\B_{\I}$ and the union of subsets in $\Powf X$.
\begin{figure}[ht]
\centering\vspace{-.7cm}
\[
\xymatrix@C=1cm@R=.5cm{
X \ar[r]^{\{-\}}\ar[dd]_{<\overline{o}_\I, \delta>} & \Powf X \ar@{-->}[rr]^{\llbracket - \rrbracket}\ar[ddl]^{<o,t>} && (\B_\I)^{A^*} \ar[dd]^{<\epsilon, (-)_a>}\\
 & &\\
\F_\I X = \B_\I \times (\Powf X)^A \ar@{-->}[rrr]_{id_{\B_\I} \times {\llbracket - \rrbracket}^A} & && \B_\I \times ((\B_\I)^{A^*})^A
}
\]
\[
{\small\begin{array}{l}
o(Y) = \bigsqcup_{y \in Y} \overline{o}_\I (y)\\
t(Y)(a) = \bigsqcup_{y \in Y} \delta(y)(a)
\end{array} \qquad\qquad \begin{array}{l}
\llbracket{Y}\rrbracket(\varepsilon) = o(Y)\\
\llbracket{Y}\rrbracket(a w) = \llbracket{\bigsqcup_{y\in Y} \delta(y)(a)}\rrbracket(w)\\
\end{array}  }
\]
\caption{The powerset construction for decorated LTS's.}
\label{fig:GenSettMoore}
\end{figure}

The coalgebraic modelling of possible-futures semantics could easily be recovered by following a similar approach. However, for the case of ready and failure trace semantics the transition structure of the LTS also needs to be slightly modified before the determinisation. This consists in changing  the alphabet $A$ to include additional information represented by sets of actions ready to be triggered as a first step. Consequently, to each LTS $(X, \delta \colon X\rightarrow(\Powf X)^{A})$ a unique coalgebra $(X, <\overline{o}_\I, \bar{\delta} \colon X\rightarrow(\Powf X)^{\bar{A}}>)$ is associated, defined in a natural fashion, as we will present later on. The construction in Figure~\ref{fig:GenSettMoore} is then applied on $(X, <\overline{o}_\I,\bar{\delta}>)$.

The explicit instantiations of $\overline{o}_\I$ and $\B_\I$ are provided later in this section, where we will also show that the coalgebraic modelling in fact coincides with the original definitions of the corresponding equivalences. This was not formally shown in~\cite{gen-pow}, for any of the aforementioned semantics.

The coalgebraic modelling of decorated trace semantics enables the definition of the corresponding equivalences as Moore bisimulations~\cite{Rutten00} ({\emph i.e.}, bisimulations for a functor $\M = \B_\I \times X^A$).
This way, checking behavioural equivalence of $x_1$ and $x_2$ reduces to checking the equality of their unique representatives in the final coalgebra: $\llbracket \{x_1\} \rrbracket$ and $\llbracket \{x_2\} \rrbracket$ .

In the subsequent sections we a) prove the details on the coalgebraic modelling of ready, failure, (complete) trace, possible-futures, ready trace and failure trace semantics, b) show that the corresponding representations coincide with their original definitions in\linebreak \cite{Glabbeek01} and c) demonstrate, by means of examples, how the associated coalgebraic framework can be used in order to reason on (some of) the aforementioned equivalences in terms of Moore bisimulations.

%The same steps are followed for each of the decorated trace semantics of interest. We proceed by first instantiating the ingredients of Fig.~\ref{fig:GenSettMoore} (that summarises the generalised powerset construction in~\cite{gen-pow}).
%
%For $\mathcal{I}$ ranging over $\Tr,\Ctr,
%\Fp, \Rp, \Pf, \Rtr$ and $\Ftr$, showing that the corresponding coalgebraic modelling and the set-theoretic definitions in~\cite{Glabbeek01} are equivalent reduces to proving that, given an arbitrary state $x$ of an LTS, $\I(x)$ is in one-to-one correspondence with the behaviour $\llbracket \{x\} \rrbracket$ in the final Moore coalgebra.

%For a more concrete insight, we will provide a series of examples of (possibly not) equivalent systems, and show how the coalgebraic machinery is used for reasoning on decorated equivalence.
%
\subsection{Ready and failure semantics}
\label{sec:ready-fail}
In this section we show how the ingredients of Figure~\ref{fig:GenSettMoore} can be instantiated in order to provide a coalgebraic modelling of ready and failure semantics. We also prove that the resulting coalgebraic characterisations of these semantics are equivalent to their original definitions in~\cite{Glabbeek01}. Moreover, we provide an optimisation that can be used when reasoning on failure equivalence, based on the isomorphism of downsets and antichains.

Consider an LTS $(X, \delta\colon  X \rightarrow (\Powf X) ^ A)$ and define, for a function $\varphi\colon A \rightarrow \Powf X$, the set of \emph{actions enabled by $\varphi$}:
\begin{equation}
\label{eq:def-ini}
\begin{array}{l}
I(\varphi) = \{a \in  A \mid \varphi(a) \not= \emptyset\},
\end{array}
\end{equation}
and the set of \emph{actions $\varphi$ fails to enable}:
\[
Fail(\varphi) = \{Z \subseteq A \mid Z \cap I(\varphi) = \emptyset\}.
\]
For the particular case $\varphi = \delta(x)$, $I(\delta(x))$ denotes the set of all (initial) actions ready to be fired by $x \in X$, and $Fail(\delta(x))$ represents the set of subsets of all (initial) actions that cannot be triggered by such $x$.

A \emph{ready pair} of $x$ is a pair $(w, Z) \in A^* \times \Powf A$ such that $x \xrightarrow{w} y$ and $Z = I(\delta(y))$.
A \emph{failure pair} of $x$ is a pair $(w, Z) \in A^* \times \Powf A$ such that $x \xrightarrow{w} y$ and $Z \in Fail(\delta(y))$.
We denote by $\Rs(x)$ and $\Fs(x)$, respectively, the sets of \emph{all ready pairs} and \emph{failure pairs}, respectively, associated with $x$.

Intuitively, ready semantics identifies states in $X$ based on the actions $a \in A$ they can immediately trigger after performing a certain action sequence $w \in A^*$, {\emph i.e.}, based on their ready pairs. It was originally defined as follows:
\begin{definition}[Ready equivalence~\cite{DBLP:journals/acta/OlderogH86,Glabbeek01}]
\label{def:R-equiv}
Let $(X, \delta \colon X \rightarrow (\Powf X)^A)$ be an LTS and $x, y \in X$ two states. States $x$ and $y$ are \emph{ready equivalent} ($\Rs$-equivalent) if and only if they have the same set of ready pairs, that is $\Rs(x) = \Rs(y)$, where
%We call $\R(x)$ the set of \emph{all ready pairs} of $x$.
\[
\Rs(x) = \{(w, Z) \in A^* \times \Powf A \mid \exists x' \in X .\, x \xrightarrow{w} x' \land Z = I(\delta(x'))\}.
\]
\end{definition}

Failure semantics identifies behaviours of states in $X$ according to their failure pairs.
\begin{definition}[Failure equivalence~\cite{Glabbeek01}]
\label{def:F-equiv}
Let $(X, \delta \colon X \rightarrow (\Powf X)^A)$ be an LTS and $x, y \in X$ two states. States $x$ and $y$ are \emph{failure equivalent} ($\Fs$-equivalent) if and only if $\Fs(x) = \Fs(y)$, where
\[
\Fs(x) = \{(w, Z) \in A^* \times \Powf A \mid \exists x' \in X .\, x \xrightarrow{w} x' \land Z \in Fail(\delta(x'))\}.
\]
\end{definition}

The coalgebraic modelling of ready, respectively, failure semantics is obtained in a uniform fashion, by instantiating the ingredients of Figure~\ref{fig:GenSettMoore} as follows. For $\I \in \{\Rs, \Fs\}$,
$\overline o_{\I} \colon X \rightarrow \Powf(\Powf A)$ is defined as:
\[
\overline{o}_\Rs(x)  =   \{I(\delta(x))\}\qquad\qquad
\overline{o}_\Fs(x)  =  Fail(\delta(x)).
\]
Intuitively, in the setting of ready semantics, the observations provided by the output operation refer to the sets of actions ready to be executed by the states of the LTS. Similarly, for failure semantics, the output operation refers to the sets of actions the states of the LTS cannot immediately fire.

\begin{remark}
\label{rem:pow-pow-ready}
Observe that the codomain of $\bar{o}_{\Rs}$ is $\Powf(\Powf A)$, and not $\Powf A$, as one might expect. However, this is consistent with the intended semantics. For $\B_{\mathcal I} = \B_\Rs = \B_{\Fs} = \Powf(\Powf A)$, the final Moore coalgebra has carrier $(\Powf(\Powf A))^{A^*}$ which is isomorphic to $\mathscr P(A^* \times \Powf(A))$ the type of $\Rs(x)$ and $\Fs(x)$. The unique homomorphism into the final coalgebra will associate to each state $\{x\}$ a function that for each $w \in A^*$ returns a set containing all sets $R_{x'}$ of ready (resp. failed) actions triggered by all $x'$ such that $x \xrightarrow{w} x'$, for $x, x' \in X$. 
\end{remark}
Next, we will prove the equivalence between the coalgebraic modelling of ready and failure semantics and their original definitions, presented above. More explicitly, given an arbitrary LTS $(X, \delta\colon X \rightarrow (\Powf X)^A)$ and a state $x\in X$, we want to show that $\llbracket \{x\} \rrbracket$ is equal to $\I(x)$, for $\I \in \{\Rs, \Fs\}$, depending on the semantics of interest. However, note that the definition of $\llbracket - \rrbracket$ is independent of $\I$; the difference is (implicitly) made by the output function $\overline{o}_{\I}$.

The behaviour of a state $x\in X$ is a function $\llbracket \{x\} \rrbracket\colon A^* \rightarrow \Powf(\Powf A)$, whereas $\I(x)$ is defined as a set of pairs in $A^* \times \Powf A$. We represent the set $\I(x) \in \Pow(A^* \times \Powf A)$ by a function $\varphi^{\I}_{x}\colon \Powf(\Powf A)^{A^*}$, where, for $w \in A^*$,
\[
\begin{array}{l}
\varphi^\Rs_{x}(w)=\{I(\delta(y)) \mid x \xrightarrow{w} y\}\\
\varphi^\Fs_x(w) = \{Z \subseteq A \mid x \xrightarrow{w} y \land Z \in Fail(\delta(y))\}.
\end{array}
\]
Showing the equivalence between the coalgebraic and the original definitions of ready, respectively, failure semantics reduces to proving that
\begin{equation}
\label{eq:interp-iso}
(\forall x \in X)\,.\,\llbracket \{x\} \rrbracket = \varphi^\I_{x}.
\end{equation}
\begin{theorem}
\label{thm:eqiv-ready}
Let $(X, \delta\colon X \rightarrow (\Powf X)^A)$ be an LTS. Then for all $x\in X$, $w\in A^*$, and $\I \in \{\Rs, \Fs\}$, $\llbracket \{x\} \rrbracket(w) = \varphi^\I_{x}(w)$.
\end{theorem}
\begin{proof}
For $\I$ ranging over $\{\Rs, \Fs\}$, the proof is
by induction on words $w \in A^*$.
We provide the details for the case of ready semantics. A similar reasoning can be applied for failure semantics.
\begin{itemize}\itemsep2pt
\item {\emph Base case.} $w = \eps$. We have:
$$
\begin{array}{r c l}
\llbracket \{x\} \rrbracket(\eps) & = & o(\{x\})=  \overline{o}_\I (x) = \{I(\delta(x))\}\\%[1.2ex]
\varphi^\Rs_{x}(\eps) & = & \{I(\delta(y)) \mid x \xrightarrow{\eps} y\} =  \{I(\delta(x))\}
\end{array}
$$
\item {\emph Induction step.} 
Consider $w \in A^*$ and assume, for all $x\in X$, $\llbracket \{x\} \rrbracket(w) = \varphi^\Rs_{x}(w)$. We want to prove that $\llbracket \{x\} \rrbracket(aw) = \varphi^\Rs_{x}(aw)$, where $a\in A$.
$$
\begin{array}{r c l}
\llbracket \{x\} \rrbracket(aw) & = &  \llbracket \delta(x)(a)\rrbracket(w) = \llbracket t(\{x\})(a) \rrbracket(w)\\
& = & \bigcup\limits_{x \xrightarrow{a} z} \llbracket \{z\} \rrbracket(w) \stackrel{\mathrm{IH}}=  \bigcup\limits_{x \xrightarrow{a} z} \varphi^\Rs_{z}(w) \\[2ex]
\varphi^\Rs_x(aw) & = & \{I(\delta(y)) \mid x \xrightarrow{aw} y\}\\%[0.5ex]
& = & \{I(\delta(y)) \mid  x \xrightarrow{a} z\land  z \xrightarrow{w} y \}\\%[0.5ex]
& = & \bigcup\limits_{x \xrightarrow{a} z} \{I(\delta(y)) \mid z \xrightarrow{w} y\}\\
& = & \bigcup\limits_{x \xrightarrow{a} z} \varphi^\Rs_{z}(w)
\end{array}
$$
\end{itemize}\vspace{-.7cm}
\end{proof}
%By Theorem~\ref{thm:eqiv-ready} it follows that~(\ref{eq:interp-iso}) holds, therefore the two representations of ready semantics are equivalent.
\begin{example}
\label{ex-ready-equiv}
In what follows we illustrate the equivalence between the coalgebraic and the original definitions of ready semantics by means of an example.
Consider the following LTS.
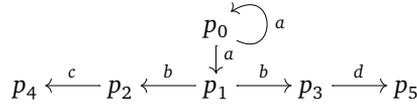
\[\xymatrix@C=.7cm@R=.35cm{
&& {p_0}\ar[d]^{a}\ar@(dr,ur)_{a}&\\
{p_4} &{p_2}\ar[l]_{c}& {p_1}\ar[r]^{b}\ar[l]_{b} &{p_3}\ar[r]^{d} & {p_5}
}\]
We write $a^n$ to represent the action sequence $aa\ldots a$ of length $n \geq 1$, with $n \in {\mathbb N}$. The set $\Rs(p_0)$ of all ready pairs associated with $p_0$ is:
\[
\{(\eps, \{a\}),
(a^n, \{a\}),
(a^n, \{b\}),
(a^n b, \{c\}),
(a^n b, \{d\}), (a^n bc, \emptyset),
(a^n bd, \emptyset)
\mid n \geq 1
\}.
\]
We can construct a Moore automaton, for $S = \{p_0, p_1, \ldots, p_5\}$,
\[
(\pow S, <o,t> \colon  \Powf S \rightarrow \Powf(\Powf A) \times (\Powf S)^A)
\]
 by applying the generalised powerset construction on the LTS above. The automaton will have $2^6=64$ states. We depict the accessible part from state $\{p_0\}$, where the output sets are indicated by double arrows:
\begin{figure}[ht]
\centering
$\xymatrix@R=.4cm@C=.3cm{
 &&&\{p_0\}\ar[d]^{a} \ar@{=>}[r] & {\{\{a\}\}\phantom{\{a,b\}}}\\
&& & \ar@{=>}[r]  \{p_0, p_1\} \ar[d]^{b}
\ar@(d,l)^{a}& \{\{a\}, \{b\}\}\\
\{\emptyset\}&{\{p_4\}}\ar@{=>}[l]&&{\{p_2, p_3\}}\ar@{=>}[r] \ar@/_{1pc}/[rr]|{d}\ar[ll]_-{c}&  \{\{c\}, \{d\}\}&  {\{p_5\}}\ar@{=>}[r]& \{\emptyset\}\hspace{1cm}
}$
\caption{Ready determinisation when starting from $\{p_0\}$.}
\label{fig:ready-proc1-det}
\end{figure}
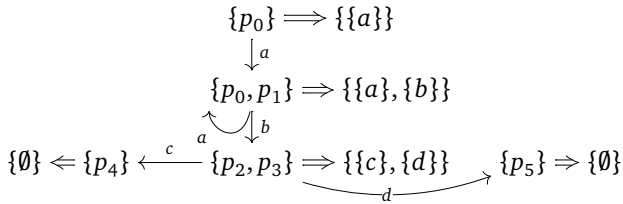
The output sets of a state $Y$ of the Moore automaton in Figure~\ref{fig:ready-proc1-det} is the set of actions associated with a certain state $y\in Y$ which can immediately be performed. For example, process $p_0$ in the original LTS above is ready to perform action $a$, whereas $p_1$ can immediately perform $b$. Therefore it holds that $o(\{p_0\}) = \{\{a\}\}$ and $o(\{p_0, p_1\}) = \{\{a\}, \{b\}\}$.

By simply looking at the automaton in Figure~\ref{fig:ready-proc1-det}, one can easily see that the set of action sequences $w \in A^*$ the state $\{p_0\}$ can execute, together with the corresponding possible next actions equals $\Rs(p_0)$. Therefore, the automaton generated according to the generalised powerset construction captures the set of all ready pairs of the initial LTS.
\end{example}

\begin{example} The last example considered in this section shows how the coalgebraic framework can be applied in order to reason on failure equivalence of LTS's. (Checking ready equivalence follows a similar approach.)
Consider the following two systems.
$$
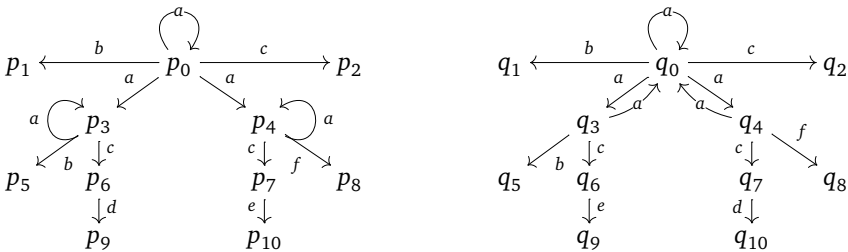
\xymatrix@R=.3cm@C=0.5cm{
p_1 & & p_0\ar@(ul,ur)|{a}\ar[ll]_{b}\ar[rr]^{c}\ar[dl]_{a}\ar[dr]^{a} & & p_2
& \hspace{10pt} &
q_1 & & q_0\ar@(ul,ur)|{a}\ar[ll]_{b}\ar[rr]^{c}\ar[dl]_{a}\ar[dr]^{a} & & q_2
\\
& p_3\ar@(dl,ul)^{a}\ar[dl]^{b}\ar[d]^{c} & & p_4\ar@(dr,ur)_{a}\ar[d]_{c}\ar[dr]_{f} &
& \hspace{10pt} &
& q_3\ar[dl]^{b}\ar[d]^{c}\ar@/_/[ur]|{a} & &
q_4\ar[d]_{c}\ar[dr]^{f}\ar@/^/[ul]|{a} &\\
p_5 & p_6\ar[d]^{d} & & p_7\ar[d]_{e} & p_8
& \hspace{10pt} &
q_5 & q_6\ar[d]^{e} & & q_7\ar[d]_{d} & q_8\\
& p_9 & & p_{10} &
& \hspace{10pt} &
& q_9 & & q_{10} &
}
$$
Let $Z = \{a_1, a_2, \ldots, a_n\}$ be the set of actions a process fails executing as a first step. For the simplicity of notation, we write $[{a_1a_2\ldots a_n}]$ to denote the set of all non-empty subsets $Z' \subseteq Z$.
For example, if $Z = \{a_1, a_2\}$, then $[{a_1a_2}]$ stands for $\{\{a_1\}, \{a_2\}, \{a_1, a_2\}\}$.

Note that $p_0$ and $q_0$ are $\Fs$-equivalent, according to Definition~\ref{def:F-equiv}, as they have the same sets of failure pairs $\Fs(p_0)$ and $\Fs(q_0)$, respectively, equal to:
$$
\begin{array}{l}
\{
(\eps, [def]),
(b, [abcdef]),
(c, [abcdef])\}\cup
\{
(a^n, [def]),
(a^n, [bde]),\\
\phantom{\{}(a^nb, [abcdef]),
(a^nc, [abcdef]),
(a^nc, [abcef]),
(a^nc, [abcdf]),\\
\phantom{\{}(a^nf, [abcdef]),
(a^ncd, [abcdef]),
(a^nce, [abcdef])\mid n \in \mathbb{N}, n\geq 1
\}.
\end{array}
$$
The same conclusion can be reached by checking behavioural equivalence of the two Moore automata generated according to the powerset construction, starting with $\{p_0\}$ and $\{q_0\}$. The fragments of the two automata starting from the states $\{p_0\}$ and $\{q_0\}$ are depicted in Figure~\ref{fig:failure-proc-det} at page~\pageref{fig:failure-proc-det}.
%\begin{figure}[ht]
%\centering
%\vspace{-.3cm}
%$\xymatrix@C=.5cm@R=.35cm{
%& \{p_0\}\ar@{=>}[r]\ar[dl]_{b}\ar[d]_{a}\ar[dr]^{c} &  [def]
%&  &
%& \{q_0\}\ar@{=>}[r]\ar[dl]_{b}\ar[d]_{a}\ar[dr]^{c} &[def]\\
%\{p_1\} \ar@{=>}[d]& \{p_0, p_3, p_4\}\ar@{=>}[d]\ar@(u,r)|{a}\ar@/_/[ddl]|{b}\ar@/^/[ddr]|{f}\ar@/^{2pc}/[dd]|{c} & {\{p_2\}}\ar@{=>}[d] &&
%{\{q_1\}}\ar@{=>}[d]  & \{q_0, q_3, q_4\}\ar@{=>}[d]\ar@(u,r)|{a}\ar@/_/[ddl]|{b}\ar@/^/[ddr]|{f}\ar@/^{2pc}/[dd]|{c} & \{q_2\}\ar@{=>}[d]\\
%{[abcdef]} & [def]\!\cup\![bde]\ & [abcdef] &&[abcdef]&[def]\!\cup\![bde]\ & [abcdef]& \\
%\{p_1, p_5\}\ar@{=>}[d] & \{p_2, p_6, p_7\}\ar@{=>}[d]\ar@/_/[ddl]|{d}\ar@/^/[ddr]|{e} & {\{p_8\}}\ar@{=>}[d]
%&  &
%\{q_1, q_5\} \ar@{=>}[d]&\{q_2, q_6, q_7\}\ar@{=>}[d]\ar@/^/[ddr]|{e}\ar@/_/[ddl]|{d} & \{q_8\}\ar@{=>}[d]
%\\
%[abcdef] &   {[abcdef]} \cup\ar@{}[d]|{\begin{array}{lr} \\{[abcef]} & \cup\\ {[abcdf]} &\end{array}} & [abcdef] &&[abcdef]&{[abcdef]} \cup\ar@{}[d]|{\begin{array}{lr} \\{[abcef]} & \cup\\ {[abcdf]} &\end{array}}& [abcdef]& \\
%\{p_9\}\ar@{=>}[d]&    &\{p_{10}\}\ar@{=>}[d]& &
%\{q_9\}\ar@{=>}[d]&&\{q_{10}\} \ar@{=>}[d]&\\
% [abcdef]&& [abcdef] &&
% [abcdef]&& [abcdef]
%}
%$
%\caption{Failure determinisation when starting from $\{p_0\}$ and $\{q_0\}$.}
%\label{fig:failure-proc-det}
%\end{figure}
%
The states $\{p_0\}$ and $\{q_0\}$ are Moore bisimilar, 
since their corresponding automata are isomorphic.

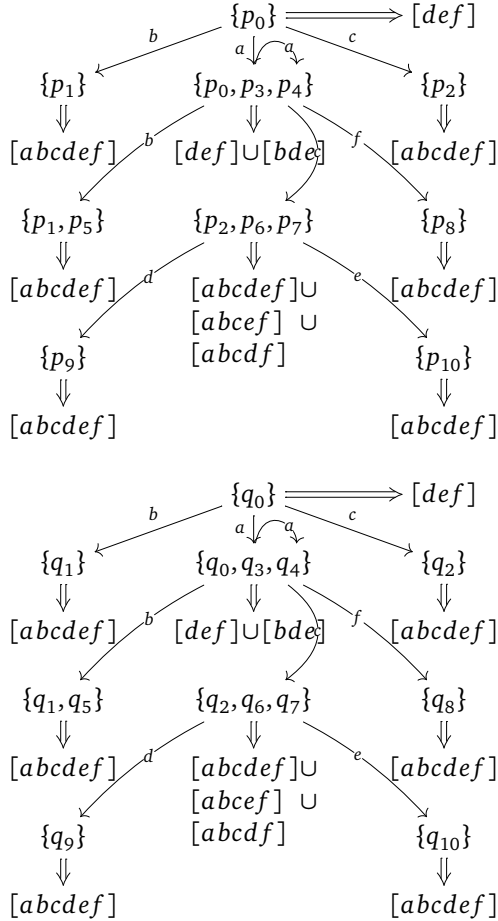
\begin{figure}[ht]
\centering
%\vspace{-.3cm}
$\xymatrix@C=.5cm@R=.35cm{
& \{p_0\}\ar@{=>}[r]\ar[dl]_{b}\ar[d]_{a}\ar[dr]^{c} &  [def]
&  &\\
\{p_1\} \ar@{=>}[d]& \{p_0, p_3, p_4\}\ar@{=>}[d]\ar@(u,r)|{a}\ar@/_/[ddl]|{b}\ar@/^/[ddr]|{f}\ar@/^{2pc}/[dd]|{c} & {\{p_2\}}\ar@{=>}[d] &&\\
{[abcdef]} & [def]\!\cup\![bde]\ & [abcdef] &&\\
\{p_1, p_5\}\ar@{=>}[d] & \{p_2, p_6, p_7\}\ar@{=>}[d]\ar@/_/[ddl]|{d}\ar@/^/[ddr]|{e} & {\{p_8\}}\ar@{=>}[d]
&  &\\
[abcdef] &   {[abcdef]} \cup\ar@{}[d]|{\begin{array}{lr} \\{[abcef]} & \!\!\!\!\cup\\ {[abcdf]} &\end{array}} & [abcdef] &&\\
\{p_9\}\ar@{=>}[d]&    &\{p_{10}\}\ar@{=>}[d]& &\\
 [abcdef]&& [abcdef] &&\\
& \{q_0\}\ar@{=>}[r]\ar[dl]_{b}\ar[d]_{a}\ar[dr]^{c} &[def]\\
{\{q_1\}}\ar@{=>}[d]  & \{q_0, q_3, q_4\}\ar@{=>}[d]\ar@(u,r)|{a}\ar@/_/[ddl]|{b}\ar@/^/[ddr]|{f}\ar@/^{2pc}/[dd]|{c} & \{q_2\}\ar@{=>}[d]\\
[abcdef]&[def]\!\cup\![bde]\ & [abcdef]& \\
\{q_1, q_5\} \ar@{=>}[d]&\{q_2, q_6, q_7\}\ar@{=>}[d]\ar@/^/[ddr]|{e}\ar@/_/[ddl]|{d} & \{q_8\}\ar@{=>}[d]\\
[abcdef]&{[abcdef]} \cup\ar@{}[d]|{\begin{array}{lr} \\{[abcef]} & \!\!\!\!\cup\\ {[abcdf]} &\end{array}}& [abcdef]& \\
\{q_9\}\ar@{=>}[d]&&\{q_{10}\} \ar@{=>}[d]&\\
 [abcdef]&& [abcdef]
}
$
\caption{Failure determinisation when starting from $\{p_0\}$ and $\{q_0\}$.}
\label{fig:failure-proc-det}
\end{figure}

\end{example}

\paragraph{optimisation for failure semantics.}
In this section we showed how failure semantics can be modelled in a coalgebraic setting, by employing the generalised powerset construction. More explicitly, given a state $p$ of an LTS $(S, \delta\,:\,S \rightarrow (\pow S)^{A})$, we showed how to build a (final) Moore coalgebra $(\pow S, \langle o,t\rangle\,:\,\pow S \rightarrow (\pow(\pow))^{A^{*}})$ ``capturing'' the corresponding set of failure pairs $\Fs(p)$, hence enabling reasoning on failure equivalence in terms of Moore bisimulations.

An optimised, equivalent modelling of failure semantics can be provided by exploiting the standard isomorphism between downsets and antichains. As we shall see, this enables reasoning on the corresponding equivalence more effectively, based on bisimulations of Moore automata with ``smaller'' output sets consisting of ready actions.

A \emph{downset} in $\pow(A)$ is a set $D\subseteq \pow(A)$ such that if $Z \in D$ and $Z' \subseteq Z$ then $Z'\in D$.
We use $\overline{\mathcal{D}}(\pow(A))$ denote the set of downsets of $\pow(A)$. 
Note that we can define a semilattice $(\overline{\mathcal{D}}(\pow(A)), \sqcup, 0)$ by taking $\sqcup$ as being the union and $0$ as the empty set.

An \emph{antichain} on $\pow(A)$ is a set $I\subseteq \pow(A)$ such that if $Z \in I$ then there exists no $Z'\in I$ such that $Z' \subset Z$. 
We use $\overline{\mathcal{A}}(\pow(A))$ denote the set of antichains of $\pow(A)$. 
Note that the union of antichains is not necessarily an antichain.
However, we can define a semilattice on $\overline{\mathcal{A}}(\pow(A))$ by taking the $\sqcup$ defined as $I_1 \sqcup I_2 = min (I_1 \cup I_2)$ where
\begin{equation}
\label{eq:min}
min(I) = \set{Z \in I \mid (\not \exists Z' \in I)\,.\, Z' \subset Z}.
\end{equation}
Now consider the homomorphisms $i \colon \overline{\mathcal{D}}(\pow(A)) \rightarrow \overline{\mathcal{A}}(\pow(A))$ defined as 
\begin{equation}
\label{eq:iso-i}
i(F) = min(\cup_{F_{i} \in F} \set{A - F_{i}})
\end{equation}
and $j \colon \overline{\mathcal{A}}(\pow(A)) \rightarrow \overline{\mathcal{D}}(\pow(A))$
defined as 
\begin{equation}
\label{eq:iso-j}
j(I) = \downarrow(\cup_{I_{i} \in I} \{A - I_{i}\}),
\end{equation}
where $\downarrow S$ denotes the downward closure of a set $S$. It is easy to see that one homomorphism is the inverse of the other and thus the 
semilattices $\overline{\mathcal{D}}(\pow(A))$ and $\overline{\mathcal{A}}(\pow(A))$ are isomorphic.

At this point, it is worth to observe that for all $X\in \pow(S)$, the Moore output function $o(X)$ is a downset 
(since $\overline{o}_{\Fs}(x)$ is a downset for all $x$, and since the union of downset is a downset). 
Therefore we can safely restrict the codomain of $o \colon \pow(S) \to \pow(\pow(A))$, 
to $o \colon \pow(S) \to \overline{\mathcal{D}}(\pow(A))$. By exploiting the isomorphism discussed above, 
we can instead define the function 
$o_{1}
\colon \pow(S) \to \overline{\mathcal{A}}(\pow(A))$ as follows: for all $X\in \pow(S)$
\begin{align*}
  o_{1} (X) &=
  \begin{cases}
    \set{I(\delta(x))} & \text{ if } X= \set x \text{ with } x\in S\\
    0 & \text{ if } X= 0\\
    min(o_{1}(X_1) \sqcup o_{1}(X_2)) & \text{ if } X= X_1 \sqcup X_2\\
  \end{cases}
\end{align*}

\begin{proposition}
\label{prop:iso-fail-ini}
 For all $X,Y\in \pow(S)$, $o(X)=o(Y)$ iff $o_{1}(X)=o_{1}(Y)$.
\end{proposition}
\begin{proof}
The proof follows from the fact that $o_{1} = i \circ o$ and that $i \colon \overline{\mathcal{D}}(\pow(A)) \to \overline{\mathcal{A}}(\pow(A))$ and 
$j\colon \overline{\mathcal{A}}(\pow(A)) \to \overline{\mathcal{D}}(\pow(A))$ are isomorphic.
\end{proof}

This optimisation can be applied also for the case of failure trace semantics in Section~\ref{sec:ready-fail-tr}. Moreover, as presented in Section~\ref{sec:fail-must}, the isomorphism of downsets and antichains is used for the coalgebraic modelling of must testing semantics.

\subsection{(Complete) trace semantics}
\label{sec:complete-tr}

In this section we model coalgebraically trace and complete trace semantics.
Similar to the previous section, we also show that the corresponding coalgebraic representations
of these semantics are equivalent to their original definitions.

Consider an LTS $(X, \delta\colon  X \rightarrow (\Powf X)^A)$. Trace semantics identifies states in $X$ according to the set of words $w \in A^*$ they can execute, whereas
complete trace semantics identifies states $x \in X$ based on their
set of complete traces. A trace $w \in A^*$ of $x$ is complete if and only if $x$
can perform $w$ and reach a deadlock state $y$ or, equivalently,
\[(\exists y \in X)\,.\,x \xrightarrow{w} y \land I(\delta(y))= \emptyset.\]
The difference between trace and complete semantics is
that the latter enables an external observer to detect stagnation, or deadlock states of a system.

Formally, trace and complete trace equivalences are defined as follows.
\begin{definition}[Trace equivalence~\cite{Hoare:1978:CSP:359576.359585,Glabbeek01}]
\label{def:T-equiv}
Let $(X, \delta \colon X \rightarrow (\Powf X)^A)$ be an LTS and $x, y \in X$ two states. States $x$ and $y$ are \emph{trace equivalent} ($\Ts$-equivalent) if and only if $\Ts(x) = \Ts(y)$, where
\begin{equation}
\label{eq:trace}
\Ts(x) = \{w \in A^* \mid \exists x' \in X .\, x \xrightarrow{w} x'\}.
\end{equation}
\end{definition}

\begin{definition}[Complete trace equivalence~\cite{sos}]
\label{def:CT-equiv}
Consider an LTS
$(X, \delta \colon X \rightarrow (\Powf X)^A)$ and $x, y \in X$ two states.
States $x$ and $y$ are \emph{complete trace equivalent} ($\CTs$-equivalent) if and only if $\CTs(x) = \CTs(y)$, where
\[
\CTs(x) = \{w \in A^* \mid \exists x' \in X .\, x \xrightarrow{w} x' \land I(\delta(x')) = \emptyset\}\text{.}
\]
\end{definition}

In what follows we instantiate the constituents of Figure~\ref{fig:GenSettMoore} in order to provide the associated coalgebraic modellings.

For $\I \in \{\Ts, \CTs\}$, the output function $\overline{o}_{\I}\colon X \rightarrow 2$ is:
\[
\overline{o}_\Ts(x) = 1\qquad \quad
\overline{o}_{\CTs}(x) =
\left \{
\begin{array}{ r l}
1 & \textnormal{ if $I(\delta(x)) = \emptyset$}\\
0 & \textnormal{ otherwise}
\end{array}
\right.
\]
Note that, for trace semantics, one does not distinguish between traces and complete traces. Intuitively, all states are accepting, so they have the same observable behaviour ({\emph i.e.}, $\overline{o}_{\Ts}(\varphi) = 1$), no matter the transitions they perform. On the other hand, complete trace semantics distinguishes between deadlock states and states that can still execute actions $a \in A$.

Consider, for example, the following LTS:
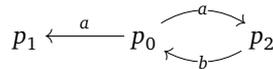
\[
\xymatrix@C=1cm{
p_1 & p_0\ar[l]_{a}\ar@/^0.8pc/[r]|{a} & p_2\ar@/^0.8pc/[l]|{b}
}\]
Observe that, for each $n \geq 0$, $(ab)^{n}a$ is a complete trace of $p_0$, as
\begin{equation}
\label{eq:trace-comp}
p_0 \xrightarrow{a} p_2 \xrightarrow{b} p_0 \xrightarrow{a} p_2 \xrightarrow{b} \ldots \xrightarrow{b} p_0 \xrightarrow{a} p_1
\end{equation}
where $p_1$ cannot perform any further action. 

The above behaviour, described in terms of transitions between states of the Moore automaton
derived according to the generalised powerset construction, can be depicted as follows:
\[
\{p_0\} \xrightarrow{a} \{p_1, p_2\} \xrightarrow{b} \{p_0\} \xrightarrow{a} \{p_1, p_2\} \xrightarrow{b} \ldots \xrightarrow{b} \{p_0\} \xrightarrow{a} \{p_1, p_2\}
\]
where $p_1$ is a deadlock state and $p_2$ is not.

Intuitively, for $n \geq 0$, we can state that $(ab)^{n}a$ is
%\begin{itemize}\itemsep3pt
%\item
a complete trace of $\{p_0\}$, as the deadlock state $p_2 \in \{p_1, p_2\}$ can be reached from $\{p_0\}$ by performing $(ab)^{n}a$ (see~(\ref{eq:trace-comp})).%

Therefore, given $Y_1, Y_2 \subseteq X$ and $w \in A^*$ such that $Y_1 \xrightarrow{w} Y_2$,
we observe that $w$ is a complete trace of $Y_1$ whenever there exists a deadlock state $y \in Y_2$.
Otherwise, $w$ is not a complete trace of $Y_1$.

In the coalgebraic modelling, the above observations with respect to the\linebreak (non)stagnating states appear in the definition of
the function $o\colon \Powf(X)  \rightarrow 2$.
Note that, for example, $o(\{p_{1}, p_{2}\}) = 1$ and $o(\{p_{0}\}) = 0$ for the case of complete trace equivalence, as $p_{1}$ is a deadlock state and $p_{0}$ is not. For trace semantics we have $o(\{p_{1}, p_{2}\}) = o(\{p_{0}\}) = 1$.

Here, $\B_{\I} = 2$ and the final Moore coalgebra in Figure~\ref{fig:GenSettMoore} is the set of languages $2^{A^*}$ over $A$ (and the transition structure $<\epsilon, (-)_a>$ is simply given by Brzozowski derivatives). Therefore, we can state that the map into the final coalgebra associates to each  state $Y \in \Powf X$ the set of all traces corresponding to states $y \in Y$, namely, the language:
\[
L = \bigcup_{y \in Y}\{w \in A^* \mid (\exists y' \in X)\,.\, y \xrightarrow{w} y'\} .
\]

The set $\Pow(A^*)$ is isomorphic to the set of functions $2^{A^*}$ which enables us to represent the set $\I(x)$ in terms its characteristic function $\varphi^\I_x\colon A^* \rightarrow 2$ defined, for $\I \in \{\Ts, \CTs\}$, $w \in A^*$, as follows:
\[
\varphi^\Ts_x(w)  =  1 \textnormal{ if } \exists y \in X\,.\, x \xrightarrow{w} y\quad
\varphi^{\CTs}_x(w) =
\left \{
\begin{array}{ r l}
1 & \textnormal{ if $\exists y \in X\,.\, x \xrightarrow{w} y\,\land\,I(\delta(y)) =
\emptyset$} \\
0  & \textnormal{ otherwise}.
\end{array}
\right.
\]
Proving the equivalence between the coalgebraic and the classic definition of (complete) trace semantics reduces to showing that
\begin{equation}
\label{eq:trace-equality}
(\forall x \in X)\,.\,\llbracket \{x\} \rrbracket = \varphi^\I_{x}.
\end{equation}
\begin{theorem}
\label{thm:eqiv-trace}
Let $(X, \delta\colon X \rightarrow (\Powf X)^A)$ be an LTS. Then for all $x\in X$ and $w\in A^*$, $\llbracket \{x\} \rrbracket(w) = \varphi^\I_{x}(w)$.
\end{theorem}

\begin{proof}
The proof is by induction on words ${w} \in A^{*}$ (similar to the proof of Theorem~\ref{thm:eqiv-ready}).
\end{proof}

\begin{example}
Consider the following two LTS's:
$$\xymatrix@C=0.5cm@R=0cm{
w_1& w_0\ar[l]_{a}\ar@(dr,ur)_{a}
& \hspace{20pt} & w'_0\ar@(dr,ur)_{a}.
}
$$
Observe that $w_0$ and $w'_0$ are trace equivalent (according to Definition~\ref{def:T-equiv}), as they output the same sets of traces
\[
\Ts(w_0) = \Ts(w'_0) = \{\eps\}\cup\{ a^n\mid  n\in \mathbb{N}, n\geq 1\}
\]
but they are not complete trace equivalent (according to Definition~\ref{def:CT-equiv}), as $w'_0$ can never reach a deadlock state, whereas $w_0$ can reach the stagnating state $w_1$.

The complete trace determinisation contains the sub-automata starting from states $\{w_0\}$ and $\{w'_0\}$ depicted in Figure~\ref{fig:compl-trace-proc2-det}:
\begin{figure}[ht]
\centering
$\xymatrix@C=0.25cm@R=0cm{
0&{\{w_0\}}\ar[rr]^-{a}\ar@{=>}[l]&
&{\{w_0,w_1\}}\ar@(dl,dr)|{a}\ar@{=>}[r] &1&&0& \{w'_0\}\ar@(dr,ur)_{a}\ar@{=>}[l]
}
$
\caption{Complete trace determinisation when starting from $\{w_0\}, \{w'_0\}$.}
\label{fig:compl-trace-proc2-det}
\end{figure}
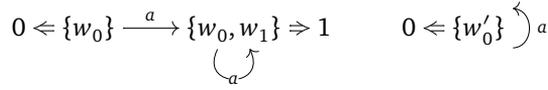
States $\{w_0\}$ and $\{w'_0\}$ are not behaviourally equivalent,
since $\{w_0, w_1\}$ outputs $1$, whereas $\{w'_0\}$ never reaches a state with this output. Hence, as expected, we will never be able to build a bisimulation containing states $\{w_0\}$ and $\{w'_0\}$.

On the other hand, in the setting of trace semantics, the determinised (Moore) automata associated with $w_{0}$ and $w_{0}'$, respectively, are similar to those depicted in Figure~\ref{fig:compl-trace-proc2-det}, with the difference that now all their states output value $1$. This makes the aforementioned automata bisimilar, hence providing a ``yes'' answer with respect to $\Ts$-equivalence of $w_{0}$ and $w_{0}'$, as anticipated.

\end{example}

\subsection{Possible-futures semantics}
\label{sec:poss-futures}
In what follows we provide a coalgebraic modelling of possible-futures semantics and show that it coincides with the original definition in~\cite{Glabbeek01}. We also give an example on how the generalised powerset construction and Moore bisimulations can be used in order to reason on possible-futures equivalence.

Let $(X, \delta \colon X \rightarrow (\Powf X)^{A})$ be an LTS. A \emph{possible future} of $x \in X$ is a pair $\langle w, T\rangle \in A^{*} \times \Pow(A^{*})$ such that $x \xrightarrow{w} y$ and $T = \Ts(y)$ (where $\Ts(y)$ is the set of traces of $y$, as in Section~\ref{sec:complete-tr}).

Possible-futures semantics identifies states that can trigger the same sets of traces $w \in A^{*}$ and moreover, by executing such $w$, they reach trace-equivalent states.

\begin{definition}[Possible-futures equivalence~\cite{DBLP:conf/focs/RoundsB81,Glabbeek01}]
\label{def:Pf-equiv}
Consider an LTS\linebreak $(X, \delta \colon X \rightarrow (\Powf X)^A)$ and $x, y \in X$ two states.
States $x$ and $y$ are \emph{possible-futures equivalent} ($\Pf$-equivalent) if and only if $\Pf(x) = \Pf(y)$, where
\[
\Pf(x) = \{\langle w, T \rangle \in A^{*} \times \Pow(A^{*}) \mid \exists x' \in X .\, x \xrightarrow{w} x' \land T = \Tr(x')\}\text{.}
\]
\end{definition}

The ingredients of Figure~\ref{fig:GenSettMoore} are instantiated as follows.

The output function $\bar{o}_{\Pf}\colon X \rightarrow \Pow(\Pow A^{*})$, which refers to the set of traces enabled by states $x \in X$ of the LTS, is defined as
\[
\bar{o}_{\Pf}(x) =  \{\Tr(x)\}.
\]
Here, $\B_{\I} = \B_{\Pf} = \Pow(\Pow A^{*})$ 
and the behaviour of a state $x \in X$ in the final coalgebra is given in terms of a function 
$\llbracket \{x\} \rrbracket \colon A^{*} \rightarrow \Pow(\Pow A^{*})^{A^*}$, 
which, intuitively, for each $w \in A^{*}$ returns the set of sets $T_y$ of traces corresponding to states $y \in X$ such that $x \xrightarrow{w} y$.

Next we want to show that for each $x \in X$, $\llbracket \{x\} \rrbracket$ and $\Pf(x)$ coincide.

First we choose to equivalently represent $\Pf(x) \in \Pow(A^{*} \times \Pow(A^{*}))$ -- the set of all possible futures of a state $x \in X$ -- in terms of $\varphi^{\Pf}_{x}\in (\Pow(\Pow A^{*}))^{A^{*}}$, where
\[
\varphi^{\Pf}_{x}(w)=\{\Tr(y) \mid x \xrightarrow{w} y\},
\]
%as 
%\begin{equation}
%\label{eq:iso-pf}
%(\Pow(\Pow A^{*}))^{A^{*}} \cong \Pow(A^{*} \times \Pow(A^{*})).
%\end{equation}
%\begin{proof}
%In order to prove~(\ref{eq:iso-pf}), we show that there exist
%\[
%\begin{array}{cc}
%f\colon (\Pow(\Pow A^{*}))^{A^{*}} \rightarrow \Pow(A^{*} \times \Pow(A^{*})), & f^{-1}\colon \Pow(A^{*} \times \Pow(A^{*})) \rightarrow (\Pow(\Pow A^{*}))^{A^{*}}
%\end{array}
%\]
%such that
%\[
%\begin{array}{ccc}
%f^{-1} \circ f = id_{(\Pow(\Pow A^{*}))^{A^{*}}} & \textnormal{ and } & f \circ f^{-1} = id_{\Pow(A^{*} \times \Pow(A^{*}))}.
%\end{array}
%\]
%We define
%\[
%\begin{array}{rcl}
%f(\varphi) & = & \{(w, T) \in A^{*} \times (\Pow A^{*}) \mid T \in \varphi(w)\}\\
%f^{-1}(V)(w) & = & \{T \in \Pow A^{*} \mid (w, T) \in V\}.
%\end{array}
%\]
%It follows that:
%\begin{align*}
%& f^{-1}(f(\varphi))(w')\\
%=~& f^{-1}(V = \{(w,T) \mid T \in \varphi(w)\})(w') \tag{by def. of $f$}\\
%=~& \varphi(w') \tag{by def. of $f^{-1}$}
%\end{align*}
%and
%\begin{align*}
%& f(f^{-1}(V))\\
%=~& \{(w,T) \mid T \in f^{-1}(V)(w)\} \tag{by def. of $f$}\\
%=~& V \tag{by def. of $f^{-1}$}
%\end{align*}
%Therefore~(\ref{eq:iso-pf}) holds.
%\end{proof}

Showing the equivalence between the coalgebraic and the original definition of possible-futures semantics reduces to proving that
\begin{equation}
\label{eq:poss-fut-iso}
(\forall x \in X)\,.\,\llbracket \{x\} \rrbracket = \varphi^\Pf_{x}.
\end{equation}
\begin{theorem}
\label{thm:eqiv-poss-fut}
Let $(X, \delta\colon X \rightarrow (\Powf X)^A)$ be an LTS. Then for all $x\in X$ and $w\in A^*$, $\llbracket \{x\} \rrbracket(w) = \varphi^\Pf_{x}(w)$.
\end{theorem}

\begin{proof}
The proof is by induction on ${w} \in A^{*}$ (similar to the proof of Theorem~\ref{thm:eqiv-ready}).
\end{proof}

\begin{example}
\label{ex:poss-fut}
Consider the following LTS's.
%\vspace{-40pt}
\[\small
\xymatrix@C=.45cm@R=0.3cm{
& & p_{0}\ar[dl]_{a}\ar[dr]^{a} & & & &\\
& p_{1}\ar[dl]_{b}\ar[d]_{a}\ar[dr]^{a} & & p_{2}\ar[d]_{a}\ar[dr]^{a} & &\\
p_{3} & p_{4}\ar[dl]_{b}\ar[d]_{c} & p_{5}\ar[d]_{c} & p_{6}\ar[d]_{c} & p_{7}\ar[d]_{c}\ar[dr]^{b} &\\
p_{8} & p_{9}\ar[d]_{d} & p_{10}\ar[d]_{e} & p_{11}\ar[d]_{d} & p_{12}\ar[d]_{e} & p_{13}\\
 & p_{14} & p_{15} & p_{16} & p_{17} &}
\xymatrix@C=.45cm@R=0.3cm{
& & & q_{0}\ar[dl]_{a}\ar[dr]^{a} & &\\
& & q_{1}\ar[dl]_{a}\ar[d]^{a} & & q_{2}\ar[dl]_{a}\ar[d]_{a}\ar[dr]^{b} &\\
& q_{3}\ar[dl]_{b}\ar[d]^{c} & q_{4}\ar[d]_{c} & q_{5}\ar[d]_{c} & q_{6}\ar[d]_{c}\ar[dr]^{b} & q_{7}\\
q_{8} & q_{9}\ar[d]_{d} & q_{10}\ar[d]_{e} & q_{11}\ar[d]_{d} & q_{12}\ar[d]_{e} & q_{13}\\
& q_{14} & q_{15} & q_{16} & q_{17} &
}\]

%\vspace{-40pt}
Note that $p_{0}$ and $q_{0}$ are possible-futures equivalent, as the traces both can follow are sequences $w \in \{a, ab, aa, aab, aac, aacd, aace\}$ and moreover, by triggering the same $w$ they reach states with equal sets of traces. The equivalence between $p_{0}$ and $q_{0}$ can be formally captured in terms of a bisimulation relation $R$ on the associated Moore automata (generated according to the generalised powerset construction) depicted in Figure~\ref{fig:poss-fut-det}, where
\[
\begin{array}{@{}r@{}l}
R = \{&
(\{p_{0}\}, \{q_{0}\}),
(\{p_{1}, p_{2}\}, \{q_{1}, q_{2}\}),
(\{p_{3}\}, \{q_{7}\}), (\{p_{8}, p_{13}\}, \{q_{8}, q_{13}\}),\\
&
(\{p_{5}, p_{5}, p_{6}, p_{7}\}, \{q_{3}, q_{4}, q_{5}, q_{6}\}),
 (\{p_{9}, p_{10}, p_{11}, p_{12}\}, \{q_{9}, q_{10}, q_{11}, q_{12}\}),\\
&
(\{p_{14}, p_{16}\}, \{q_{14}, q_{16}\}),
(\{p_{15}, p_{17}\}, \{q_{15}, q_{17}\})\,\,
\}.
\end{array}
\]

It is easy to check that $R$ is a bisimulation, since both automata in Figure~\ref{fig:poss-fut-det} are isomorphic. (Note that equality of the outputs -- which are sets of traces -- can be established using the framework introduced in Section~\ref{sec:complete-tr}.)

\begin{figure}[H]
%\hspace{-20pt}
\[
\begin{array}{c}
{
\!\!\!\xymatrix@C=.1cm@R=0.35cm{
& & \{p_{0}\}\ar[d]_{a}\ar@{=>}[r] & \{\Tr(p_{0})\}\phantom{\Tr(p_{0})\}}\\
 \{\emptyset\} & \{p_{8}, p_{13}\}\ar@{=>}[l]& \!\!\!\!\!\!\!\!\{p_{1}, p_{2}\}\ar[d]_{b}\ar[dl]_{a}\ar@{=>}[r] & \{\Tr(p_{1}), \Tr(p_{2})\}\\
o_1 & \{p_{4}, p_{5}, p_{6}, p_{7}\}\ar[u]_{b}\ar[d]^{c}\ar@{=>}[l] & \{p_{3}\}\ar@{=>}[r] & \{\emptyset\}\phantom{\Tr(p_{0},\Tr(p_{0})}\\
o_2 & \{p_{9}, p_{10}, p_{11}, p_{12}\}\ar[d]_{d}\ar[dr]^{e}\ar@{=>}[l] &  & \\
  \{\emptyset\} & \{p_{14}, p_{16}\}\hspace{.7cm}\ar@{=>}[l] &\!\!\!\!\!\!\!\!\!\!\!\!\!\!\!\! \{p_{15}, p_{17}\}\ar@{=>}[r] &\{\emptyset\}\phantom{\Tr(p_{0}),\Tr(p_{0})} 
  }
}\\\\
{\xymatrix@C=.1cm@R=0.35cm{
& & \{q_{0}\}\ar[d]_{a}\ar@{=>}[r] & \{\Tr(q_{0})\}\phantom{\Tr(p_{0})\}}\\
 \{\emptyset\} & \{q_{8}, q_{13}\}\ar@{=>}[l] & \!\!\!\!\!\!\!\!\{q_{1}, q_{2}\}\ar[d]_{b}\ar[dl]_{a}\ar@{=>}[r] & \{\Tr(q_{1}), \Tr(q_{2})\}\\
o'_1 & \{q_{3}, q_{4}, q_{5}, q_{6}\}\ar[u]_{b}\ar[d]^{c}\ar@{=>}[l] & \{q_{7}\}\ar@{=>}[r] & \{\emptyset\}\phantom{\Tr(p_{0},\Tr(p_{0})}\\
o'_2& \{q_{9}, q_{10}, q_{11}, q_{12}\}\ar[d]_{d}\ar[dr]^{e}\ar@{=>}[l] &  & \\
  \{\emptyset\} & \{q_{14}, q_{16}\}\hspace{.7cm}\ar@{=>}[l] &\!\!\!\!\!\!\!\! \!\!\!\!\!\!\!\! \{q_{15}, q_{17}\}\ar@{=>}[r] &\{\emptyset\}\phantom{\Tr(p_{0}),\Tr(p_{0})} 
}
}
\end{array}
\]
\caption{Possible-futures determinisation when starting from $\{p_0\}, \{q_0\}$.
$o_1 = \{\Tr(p_{4}), \Tr(p_{5}), \Tr(p_{6}), \Tr(p_{7})\}, o_2 = \{\Tr(p_{9}), \Tr(p_{10}), \Tr(p_{11}), \Tr(p_{12})\}$, $o'_1 = \{\Tr(q_{3}), \Tr(q_{4}), \Tr(q_{5}), \Tr(q_{6})\}$, $o'_2 = \{\Tr(q_{9}), \Tr(q_{10}), \Tr(q_{11}), \Tr(q_{12})\}$.}
\label{fig:poss-fut-det}
\end{figure}
\end{example}

\subsection{Ready and failure trace semantics}
\label{sec:ready-fail-tr}

In this section we provide a coalgebraic modelling of ready and failure trace semantics by employing the generalised powerset construction. Similarly to the other semantics tackled so far, we show a) that the coalgebraic representation coincides with the original definition in~\cite{Glabbeek01} and b) how to apply the coalgebraic machinery in order to reason on the corresponding equivalences.

%We proceed by recalling some basic concepts.

Intuitively, ready trace semantics identifies two states if and only if they can follow the same traces $w$, and moreover, the corresponding (pairwise-taken) states determined by such $w$'s have equivalent one-step behaviours. Failure trace semantics identifies states that can trigger the same traces $w$, and moreover, the (pairwise-taken) intermediate states occurring during the execution of a such $w$ fail triggering the same (sets of) actions.
Formally, the associated definitions are as follows:

\begin{definition}[Ready trace equivalence~\cite{DBLP:conf/icalp/Pnueli85,Glabbeek01}]
\label{def:Rtr-equiv}
Consider an LTS\linebreak $(X, \delta \colon X \rightarrow (\Powf X)^A)$ and $x, y \in X$ two states.
States $x$ and $y$ are \emph{ready trace equivalent} ($\Rtr$-equivalent) if and only if $\Rtr(x) = \Rtr(y)$, where
\[
\begin{array}{r l}
\Rtr(x) = \{& I_{0}a_{1}I_{1}a_{2}\ldots a_{n}I_{n} \in \Powf(A) \times (A \times \Powf(A))^{*} \mid\\
&
(\exists x_{1},\ldots, x_{n} \in X) \,.\,
x = x_{0} \xrightarrow{a_{1}} x_{1} \xrightarrow{a_{2}} \ldots \xrightarrow{a_{n}}x_{n}\, \land\\
&
(\forall i = 0,\,\ldots,\,n)\,.\, I_{i} = I(\delta(x_{i}))
\,\,\}\text{.}
\end{array}
\]
We call an element of $\Rtr(x)$ a \emph{ready trace} of $x$.
\end{definition}

\begin{definition}[Failure trace equivalence~\cite{DBLP:journals/tcs/Phillips87}]
\label{def:Ftr-equiv}
Let $(X, \delta \colon X \rightarrow (\Powf X)^A)$ be an LTS and $x, y \in X$ two states.
States $x$ and $y$ are \emph{failure trace equivalent} ($\Ftr$-equivalent) if and only if $\Ftr(x) = \Ftr(y)$, where
\[
\begin{array}{r l}
\Ftr(x) = \{& F_{0}a_{1}F_{1}a_{2}\ldots a_{n}F_{n} \in \Powf(A) \times (A \times \Powf(A))^{*} \mid\\
&
(\exists x_{1},\ldots, x_{n} \in X) \,.\,
x = x_{0} \xrightarrow{a_{1}} x_{1} \xrightarrow{a_{2}} \ldots \xrightarrow{a_{n}}x_{n}\, \land 
 F_{i} \in Fail(\delta(x_{i}))
\,\,\}\text{.}
\end{array}
\]
We call an element of $\Ftr(x)$ a \emph{failure trace} of $x$.
\end{definition}

In order to model these two equivalences coalgebraically we will have to apply the generalised powerset construction, from Figure~\ref{fig:GenSettMoore}, not only by adding the output function but also by changing the transitions of the LTS.

In particular, we have to add to transitions of shape $x \xrightarrow{a} y$ information regarding the sets of actions ready to be triggered by $x$. In the  new LTS we consider transitions of shape
$x \xrightarrow{<a, I(\delta(x))>} y$
therefore enabling the construction of Moore automata ``collecting'' states that have been reached not only via one-step transitions with the same label, but also from processes sharing the same initial behaviour. (Note that $F \in Fail(\delta(x))$ whenever $F \subseteq A - I(\delta(x))$.) 

We apply the generalised powerset construction to the decorated LTS:
\[
\xymatrix{X \ar[r]^-{<\overline{o}_\I, \overline\delta >}&  \Powf(\Powf(A))\times \Powf(X)^{A \times \Powf(A)}}
\]
where $\overline\delta$ is defined by first computing the set $I$ and then appending it to every successor of a state by using the strength of powerset:
\[
\small
\begin{array}{l}
\overline\delta = \xymatrix{X\ar[r]^-\delta & \ar[r]^-{<I,id>}\Powf(X)^A & \Powf(A)\times \Powf(X)^A \ar[r]^-{{\emph st}}& \Powf(\Powf(A)\times X)^A \to  \Powf(X)^{A \times \Powf(A)}}\\
\overline{\delta}(x)(\langle a, Z\rangle) =
\begin{cases}
    \delta(x)(a) & \text{ if } Z = I(\delta(x))\\
    \emptyset & \text{ otherwise. }
  \end{cases}
\end{array}
\] 
For $\I \in \{\Rtr,\, \Ftr\}$, the output function $\bar{o}_{\I}$ provides information with respect to the actions ready, respectively, failed to be triggered by a state $x \in X$ as a first step:
\[
\overline{o}_\Rtr(x)  =  \{I(\delta(x))\}\qquad\qquad
\overline{o}_\Ftr(x)  =  Fail(\delta(x)).
\]

%Each LTS $(X, \delta \colon X \rightarrow (\Powf X)^{A})$ is associated a unique coalgebra $(X, \bar{\delta} \colon X \rightarrow (\Powf X)^{\bar{A}})$, where
%\[
%\begin{array}{r c l}
%\bar{A} & = & A \times \Powf(A)\\
%\bar{\delta}(x)(\langle a, I_{x}\rangle) & = &
%\left \{
%\begin{array}{ r l}
%\delta(x)(a) & \textnormal{ if $I_{x} = I(\delta(x))$} \\
%\emptyset  & \textnormal{ otherwise}.
%\end{array}
%\right.
%\end{array}
%\]
%(See Example~\ref{ex:ready-tr-not} for a more concrete insight on the preprocessing procedure and its effect.)

%Using the new construction $(X, \bar{\delta} \colon X \rightarrow (\Powf X)^{\bar{A}})$ as a starting point, defining the constituents of Fig.~\ref{fig:GenSettMoore} follows the recipe described in Section~\ref{sec:ready-fail}.
%
%For $\I \in \{\Rtr,\, \Ftr\}$, the output function $\bar{o}_{\I}$ provides information with respect to the actions ready, respectively, failed to be triggered by a state $x \in X$ as a first step:
%\[
%\begin{array}{l}
%\overline{o}_\I \colon (\Powf X)^{\bar{A}} \rightarrow \Powf(\Powf A)\\%[0.5ex]
%\overline{o}_\Rtr(\bar{\delta}(x))  =  \{I(\delta(x))\}\\
%\overline{o}_\Ftr(\bar{\delta}(x))  =  Fail(\delta(x)).
%\end{array}
%\]

We need to show that for $x_{0}\in X$, there is a one-to-one correspondence between $\llbracket \{x_{0}\}\rrbracket$ and $\I(x_{0})$.
Intuitively, for ready trace semantics, for example, each behaviour
\[
\begin{array}{ll}
\llbracket \{x_{0}\} \rrbracket(\bar{w}) = \{Z_{n}^{j}\mid x_a \xrightarrow{w} x_j\}, &\textnormal{ with } \bar{w} = \langle a_{1}, Z_{0}\rangle\,\ldots\,\langle a_{n}, Z_{n-1}\rangle \in (A \times \Powf({A}))^{*}\\
&\text{ and } w= a_1 \ldots a_n \in A^*
\end{array}
\]
corresponds to a set of sequences of shape
\[Z_{0}a_{1}Z_{1}a_{2}\ldots Z_{n-1}a_{n}Z^{j}_{n} \in \I(x_{0}).\]

Given $x \in X$, for $\I \in \{\Rtr,\,\Ftr\}$, we again represent $\I(x) \in \Pow(\Powf(A) \times (A \times \Powf(A))^{*})$ by a function $\varphi^\I_{x}$:
\[
\begin{array}{l}
\varphi^\Rtr_{x}(\bar{w})=\{Z \subseteq A \mid x \xrightarrow{\bar{w}} y \land Z = I(\delta(y))\}\\[1.2ex]
\varphi^\Ftr_{x}(\bar{w})=\{Z \subseteq A \mid  x \xrightarrow{\bar{w}} y \land Z \in Fail(\delta(y))\}
\end{array}
\]
Showing the equivalence between the coalgebraic and the original definition of ready and failure trace semantics consists in proving that
\begin{equation}
\label{eq:ready-tr-iso}
(\forall x \in X)\,.\,\llbracket \{x\} \rrbracket = \varphi^\I_{x}.
\end{equation}
\begin{theorem}
\label{thm:eqiv-ready-tr}
Let $(X, \delta\colon X \rightarrow (\Powf X)^A)$ be an LTS. Then for all $x \in X$ and $\bar{w}\in (A \times \Powf({A}))^{*}$, $\llbracket \{x\} \rrbracket(\bar{w}) = \varphi^\I_{x}(\bar{w})$.
\end{theorem}

\begin{proof}
The proof follows by induction on words ${w} \in (A \times \Powf(A))^{*}$ (similar to the proof of Theorem~\ref{thm:eqiv-ready}).
\end{proof}

\begin{example}
\label{ex:ready-tr-not}
Consider the following two systems:
\[\xymatrix@C=.6cm@R=0.35cm{
& & p_{0}\ar[dl]_{a}\ar[dr]^{a} & & & & & q_{0}\ar[dl]_{a}\ar[dr]^{a} & &\\
& p_{1}\ar[dl]_{b}\ar[d]_{c} & & p_{2}\ar[d]_{c}\ar[dr]^{f} & & & q_{1}\ar[dl]_{b}\ar[d]_{c} & & q_{2}\ar[d]_{c}\ar[dr]^{f} & \\
p_{3} & p_{4}\ar[d]_{d} & & p_{5}\ar[d]_{e} & p_{6} & q_{3} & q_{4}\ar[d]_{e} & & q_{5}\ar[d]_{d} & q_{6} \\
& p_{7} & & p_{8} & & & q_{7} & & q_{8} & \\  
}\]
Note that they are not ready trace equivalent as, for example, $\{a\}a\{c,f\}c\{e\}$ is a ready trace of $p_{0}$ but not of $q_{0}$. Moreover, they are not failure trace equivalent as, for example, $\{b,c,d,e,f\}a\{a,d,e,f\}c\{a,b,c,e,f\}d\{a,b,c,d,e,f\}$ is a failure trace of $p_{0}$ but not of $q_{0}$.

It is easy to check that by taking exactly the generalised powerset construction (starting with $\{p_{0}\}, \{q_{0}\}$) without changing the transition function, as in Section~\ref{sec:ready-fail},  one gets two bisimilar Moore automata (for both the case of ready and failure trace equivalence).
This would indicate that the initial LTS's are behavioural equivalent (which is not the case for ready and failure trace!).

The change in the transition function generates the automata (with labels in $A \times \Powf(A)$) in Figure~\ref{fig:preproc-rt}. Then, for both semantics studied in this section, the determinisation derives the two Moore automata in Figure~\ref{fig:det-preproc-rt}.

For ready trace semantics it holds that:
\[
\begin{array}{l}
o_{0} = \overline{o}_{0} = \{\{a\}\}\quad
o_{12} = \overline{o}_{12} = \{\{b,c\}, \{c,f\}\}\quad
o_{4} = \overline{o}_{5} = \{\{d\}\}\quad
o_{5} = \overline{o}_{4} = \{\{e\}\}\\
o_{3} = o_{6} = o_{7} = o_{8} = \overline{o}_{3} = \overline{o}_{6} = \overline{o}_{7} = \overline{o}_{8} = \{\emptyset\}.
\end{array}
\]

Hence, the systems in Figure~\ref{fig:det-preproc-rt} are not bisimilar as, for example, both states $\{p_{4}\}$ and $\{q_{4}\}$ can be reached via transitions labelled the same, but they output different sets of ready actions -- namely $\{\{d\}\}$ and $\{\{e\}\}$, respectively. Therefore, we conclude that $p_{0}$ and $q_{0}$ are not ready trace equivalent.

Similarly, for failure trace we have:
\[
\begin{array}{l}
o_{0} = \overline{o}_{0} = [bcdef]\quad
o_{12} = \overline{o}_{12} = [adef] \cup [abde]\quad\\
o_{4} = \overline{o}_{5} = [abcef]\quad
o_{5} = \overline{o}_{4} = [abcdf]\\
o_{3} = o_{6} = o_{7} = o_{8} = \overline{o}_{3} = \overline{o}_{6} = \overline{o}_{7} = \overline{o}_{8} = [abcdef].
\end{array}
\]

As before, the automata in Figure~\ref{fig:det-preproc-rt} are not bisimilar as, for example, both $\{p_{4}\}$ and $\{q_{4}\}$ are reached via transitions labelled the same, but have different outputs. Therefore we conclude that $p_{0}$ and $q_{0}$ are not failure trace equivalent.

The purpose of changing the transition labels with sets of ready actions is to collect in a Moore state only states of the initial LTS's that have been reached from ``parents'' with the same one-step (initial) behaviour. Or dually, to distinguish between states that have ``parents'' ready, respectively, failing to trigger different sets of actions. This way one avoids the unfortunate situation of encapsulating, for example, the states $p_{4}, p_{5}$, respectively $q_{4}, q_{5}$, fact which eventually would lead to providing a positive answer with respect to both ready and failure trace equivalence of $p_{0}$ and $q_{0}$.

In other words, the change in the transition function is needed in order to guarantee that whenever two states of an LTS are ready/failure trace equivalent, the (pairwise-taken) states determined by the executions of a given trace have the same initial behaviour.

%\vspace{-45ex}
\begin{figure}%[H]
\[\xymatrix@C=.8cm@R=0.5cm{
& & p_{0}\ar[dl]_{\langle a, \{a\} \rangle}\ar[dr]^{\langle a, \{a\} \rangle} & & & & & q_{0}\ar[dl]_{\langle a, \{a\} \rangle}\ar[dr]^{\langle a, \{a\} \rangle} & &\\
& p_{1}\ar[dl]_{\langle b, \{b, c\}\rangle}\ar[d]^{\langle c, \{b, c\}\rangle} & & p_{2}\ar[d]_{\langle c, \{c,f\}\rangle}\ar[dr]^{\langle f, \{c,f\}\rangle} & & & q_{1}\ar[dl]_{\langle b, \{b, c\}\rangle}\ar[d]^{\langle c, \{b, c\}\rangle} & & q_{2}\ar[d]_{\langle c, \{c,f\}\rangle}\ar[dr]^{\langle f, \{c,f\}\rangle} & \\
p_{3} & p_{4}\ar[d]_{\langle d, \{d\}\rangle} & & p_{5}\ar[d]^{\langle e, \{e\}\rangle} & p_{6} & q_{3} & q_{4}\ar[d]_{\langle e, \{e\}\rangle} & & q_{5}\ar[d]^{\langle d, \{d\}\rangle} & q_{6} \\
& p_{7} & & p_{8} & & & q_{7} & & q_{8} & \\  
}\]
\caption{Altered transition function before determinisation.}
\label{fig:preproc-rt}
\end{figure}
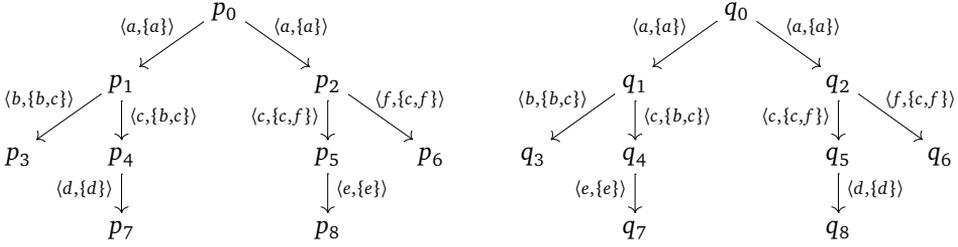

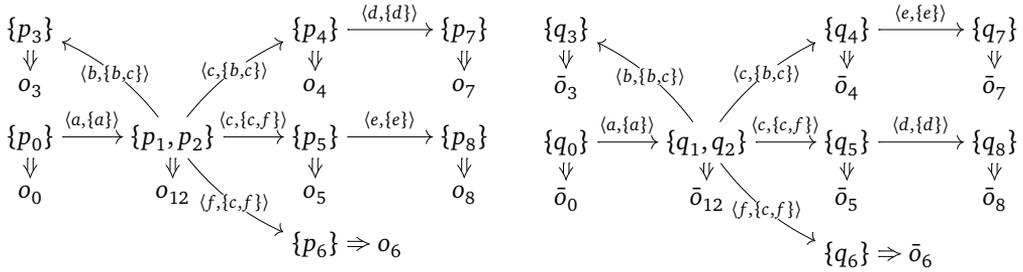
\begin{figure}%[H]
\[\xymatrix@C=.28cm@R=0.2cm{
 \{p_{3}\}\ar@{=>}[d]  &&&& \{p_{4}\}\ar[rr]^{\langle d, \{d\} \rangle}\ar@{=>}[d] && \{p_{7}\}\ar@{=>}[d] & \\
o_{3} & &&& o_{4} && o_{7} \\
\{p_{0}\}\ar[rr]^-{\langle a, \{a\} \rangle}\ar@{=>}[d] && \{p_{1}, p_{2}\}\ar@{=>}[d]
\ar@/_/[uull]|-{\langle b, \{b,c\}\rangle}
\ar@/^/[uurr]|{\langle c, \{b,c\}\rangle}
\ar[rr]^-{\langle c, \{c,f\}\rangle}
\ar@/_/[ddrr]|{\langle f, \{c,f\}\rangle}
&& \{p_{5}\}\ar@{=>}[d]\ar[rr]^{\langle e, \{e\} \rangle} && \{p_{8}\}\ar@{=>}[d]\\
 o_{0}&&o_{12} & & o_{5} && o_{8} &\\
& & &&\{p_{6}\}\ar@{=>}[r] & o_{6} &
}
\xymatrix@C=.28cm@R=0.2cm{
 \{q_{3}\}\ar@{=>}[d]  &&&& \{q_{4}\}\ar[rr]^{\langle e, \{e\} \rangle}\ar@{=>}[d] && \{q_{7}\}\ar@{=>}[d] & \\
\bar o_{3} & &&& \bar o_{4} && \bar o_{7} \\
\{q_{0}\}\ar[rr]^-{\langle a, \{a\} \rangle}\ar@{=>}[d] && \{q_{1}, q_{2}\}\ar@{=>}[d]
\ar@/_/[uull]|-{\langle b, \{b,c\}\rangle}
\ar@/^/[uurr]|{\langle c, \{b,c\}\rangle}
\ar[rr]^-{\langle c, \{c,f\}\rangle}
\ar@/_/[ddrr]|{\langle f, \{c,f\}\rangle}
&& \{q_{5}\}\ar@{=>}[d]\ar[rr]^{\langle d, \{d\} \rangle} && \{q_{8}\}\ar@{=>}[d]\\
 \bar o_{0}&&\bar o_{12} & & \bar o_{5} && \bar o_{8} &\\
& & &&\{q_{6}\}\ar@{=>}[r] & \bar o_{6} &
}
\]
\caption{Determinisation starting from $\{p_{0}\}, \{q_{0}\}$.}
\label{fig:det-preproc-rt}
\end{figure}
\end{example}

\section{Decorated trace semantics of GPS's}
\label{sec:dec-tr-gps}

In this section we show how the generalised powerset construction for coalgebras $f\colon  X \rightarrow \F T(X)$ for a functor $\F$ and a monad $T$ in~(\ref{F-final}), Section~\ref{prelim:gen-pow}, can be instantiated in order to provide coalgebraic modellings of decorated trace semantics for generative probabilistic systems (GPS's). More explicitly, we show how the determinisation procedure can be applied in order to derive coalgebraic representations of ready, (maximal) failure and (maximal) trace semantics, equivalent to their standard definitions in~\cite{GPS-Jou-Smolka}.

A GPS is similar to an LTS, but each transition is labelled by both an action and a probability $p$. More precisely, the transition dynamics is given by a \emph{probabilistic transition function} $\mu \colon  X \times A \times X \rightarrow [0,1]$ satisfying
\begin{equation}
\label{eq:prob-trans-fct}
(\forall x \in X)\,.\,\sum_{\substack{a \in A\\ y \in X}}\mu(x,a,y) \leq 1,
\end{equation}
where $X$ is the state space and $A$ is the alphabet of actions. For simplicity, we write $\mu_{a}(x,y)$ in lieu of $\mu(x,a,y)$ and 
we will use the notation $x \xrightarrow{a[v]} y$ for $\mu_a(x,y) =v$.  We extend $\mu$ to words $w \in A^{*}$:
\[
 \mu_{\varepsilon}(x,y) =
\begin{cases}
1 & \textnormal{if } x = y\\
0 & \textnormal{if } x \not= y
\end{cases}\qquad \mu_{aw}(x,y) = \sum_{\substack{x' \in X}} \mu_{a}(x,x') \times \mu_{w}(x',y)
\]
Intuitively, $\mu_{w}(x,y)$ represents the sum of the probabilities associated with all traces $w$ from $x$ to $y$. Moreover, we write
\[\mu_{\ZA}(x,\ZP) = 1 - \sum_{\substack{a \in A \\ y \in X}} \mu(x,a,y)\]
for the probability of $x$ to \emph{terminate}, where $\ZA$ is a special symbol not in $A$, called the \emph{zero action}, and $\ZP$ is the (deadlock-like) \emph{zero process} whose only transition is $\mu_{\ZA}(\ZP,\ZP) = 1$.

Similarly to the case of LTS's, the set of initial actions that can be triggered (with a probability greater than $0$) from $x \in X$ is given by
\[
I(x) = \{a \in A \mid (\exists y \in X)\,.\,\mu_{a}(x,y)\, \textnormal{\textgreater}\, 0 \},
\]
whereas failure sets $Z \in \pow A$ satisfy the condition $Z \cap I(x) = \emptyset$. We write $Fail(x)$ to represent the set of all failure sets of $x$.

The decorated trace semantics for GPS's considered in this paper can be intuitively described as follows.
Given two states $x, y \in X$, we say that $x$ and $y$ are equivalent whenever traces $w \in A^{*}$
\begin{itemize}
\item lead, with the same probability, $x$ and $y$ to processes that trigger (respectively, fail to execute) as a first step the same sets of actions, for the case of ready (respectively, failure) semantics. Note that maximal failure semantics takes into consideration only the largest sets of failure actions ({\emph i.e.}, $A - I(x),\,A - I(y)$).

\item can be executed with the same probability from both $x$ and $y$, for the case of trace semantics and, moreover, lead $x$ and $y$ to processes that have the same probability to terminate, for the case of maximal trace semantics.
\end{itemize}

To model GPS's, we consider $\PD(X)$ -- the (finitely supported sub)probability distribution functor defined on \textbf{Set}. $\PD$ maps a set $X$ to
\[
\PD(X) = \{\varphi\colon X \rightarrow [0,1] \mid supp(\varphi) \textnormal{ is finite and } \sum_{\substack{x\in X}} \varphi(x) \leq 1\},
\]
where $supp(\varphi) = \{x \in X \mid \varphi(x)\, \textnormal{\textgreater}\, 0\}$ is the \emph{support} of $\varphi$. Given a function $g\colon X \rightarrow Y$, ${\PD(g)} \colon  \PD(X) \rightarrow \PD(Y)$ is defined as
\[
\PD(g)(\varphi) = \lambda y\,.\,\sum_{\substack{g(x) = y}} \varphi(x).
\]

%To model GPS's we consider the (sub)probability distribution monad $(\PD(X), \eta, \mu)$ where
%\[
%\begin{array}{l@{\hspace{1cm}}l}
%\eta\colon X \rightarrow \PD(X) & \mu\colon \PD(\PD(X)) \rightarrow \PD(X)\\
%\eta(x) = \lambda y\,.\,\left \{ \begin{array}{c l}1 & \textnormal{if } x = y\\ 0 & \textnormal{ otherwise} \end{array}\right. & \mu(\psi) = \lambda x\,.\,\sum\limits_{{\varphi \in supp(\psi)}} \varphi(x) \times \psi(\varphi)
%\end{array}
%\]
%The algebras for this monad are the so-called positive convex structures~\cite{doberkat}.

A GPS is a coalgebra
\[
(X, \delta\colon X \rightarrow (\PD(X))^{A})
\] such that $\delta(x)(a)(y) = \mu_a(x,y)$\footnote{Note that the coalgebraic type directly corresponds to reactive systems~\cite{DBLP:journals/tcs/BartelsSV04}. The embedding of generative into reactive is injective and poses no problems semantic-wise.
%Moreover, for the use of the generalised powerset construction the reactive type is more suitable.
In the sequel, when we write ``Let $(X,\delta\colon X \to (\PD(X))^A)$ be a GPS'' we implicitly mean a coalgebra of this type originating from a GPS defined by a probabilistic function $\mu \colon X\times A\times X \to [0,1]$ as in~(\ref{eq:prob-trans-fct}).}.

To each GPS we associate a \emph{decorated} GPS's \[(X, <\overline{o}_{\I}, \delta> \colon X \rightarrow B_{\I} \times (\PD(X))^{A})\] ``parameterised'' by $\I$, depending on the semantics under consideration. 

Decorated GPS's can be determinised according to 
the generalised powerset construction as illustrated in Figure~\ref{fig:GenSettMooreGPS}, where $\F$ is $B_{\I} \times (-)^{A}$ and $T$ is instantiated with the probability distribution monad $(\PD, \eta, \mu)$:
\[
\begin{array}{l@{\hspace{1cm}}l}
\eta\colon X \rightarrow \PD(X) & \mu\colon \PD(\PD(X)) \rightarrow \PD(X)\\
\eta(x) = \lambda y\,.\,\left \{ \begin{array}{c l}1 & \textnormal{if } x = y\\ 0 & \textnormal{ otherwise} \end{array}\right. & \mu(\psi) = \lambda x\,.\,\sum\limits_{{\varphi \in supp(\psi)}} \varphi(x) \times \psi(\varphi)
\end{array}
\]
Algebras for this monad are the so-called positive convex structures~\cite{doberkat}.

Moreover, for each of the semantics of interest the observations set $B_{\I}$ has to carry a $\PD$-algebra structure, or, equivalently, there has to exist a morphism $h_{\I}$ such that $(B_{\I}, h_{\I}\colon \PD(B_{\I}) \rightarrow B_{\I})$ is a $\PD$-algebra (as introduced in Definition~\ref{def:alg-of-monad}, in Section~\ref{prelim:gen-pow}).
\begin{figure}[ht]\vspace{-.7cm}
\[
\xymatrix@C=1cm@R=.5cm{
X \ar[r]^{\eta}\ar[dd]_{<\overline{o}_\I,\delta>} & \PD(X) \ar@{-->}[rr]^{\llbracket - \rrbracket}\ar[ddl]^{<o,t>} && (\B_\I)^{A^*} \ar[dd]^{<\epsilon, (-)_a>}\\
&&\\
\B_\I \times (\PD(X))^A \ar@{-->}[rrr]_{id_{\B_\I} \times {\llbracket - \rrbracket}^A} & && \B_\I \times ((\B_\I)^{A^*})^A
}
\]
\[
{\small\begin{array}{l}
o = h_{\I} \circ \PD(\overline{o}_{\I})\\
t(\varphi)(a)(y) = \sum\limits_{\substack{x \in supp(\varphi)}} \delta(x)(a)(y) \times \varphi(x)
\end{array}\qquad\qquad
\begin{array}{l}
\bb{\varphi}(\varepsilon) = o(\varphi)\\
\bb{\varphi}(a w) = \bb{t(\varphi)(a)}(w)\\
\end{array}  }
\]
\caption{The powerset construction for decorated GPS's.}
\label{fig:GenSettMooreGPS}
\end{figure}

The ingredients $\overline{o}_{\I}, B_{\I}$ and $h_{\I}$ of Figure~\ref{fig:GenSettMooreGPS} are explicitly defined in the subsequent sections for each of the coalgebraic decorated trace semantics. The latter are also proven to be equivalent with their corresponding definitions in~\cite{GPS-Jou-Smolka}.

\subsection{Ready and (maximal) failure semantics}
\label{sec:r-f-mf}

In this section we provide the detailed coalgebraic modelling of ready and (maximal) failure semantics and show the equivalence with their counterparts defined in~\cite{GPS-Jou-Smolka}, as follows:

\begin{definition}[Ready equivalence~\cite{GPS-Jou-Smolka}]
\label{def:ready-gps}
The \emph{ready function} \[\PR \colon  X \rightarrow((A^{*} \times \pow A) \rightarrow [0,1])\] is given by
\[
\PR(x)((w,I)) = \sum_{\substack{I=I(y)}} \mu_{w}(x,y).
\]
We say that $x,x' \in X$ are \emph{ready equivalent} whenever $\PR(x) = \PR(x')$.
\end{definition}

\begin{definition}[Failure equivalence~\cite{GPS-Jou-Smolka}]
\label{def:failure-gps}
The \emph{failure function} \[\PF \colon  X \rightarrow((A^{*} \times \pow A) \rightarrow [0,1])\] is given by
\[
\PF(x)((w,Z)) = \sum_{\substack{Z\cap I(y) = \emptyset}} \mu_{w}(x,y).
\]
We say that $x,x' \in X$ are \emph{failure equivalent} whenever $\PF(x) = \PF(x')$.
\end{definition}

\begin{definition}[Maximal failure equivalence~\cite{GPS-Jou-Smolka}]
\label{def:max-failure-gps}
The \emph{maximal failure function} $\PMF \colon  X \rightarrow((A^{*} \times \pow A) \rightarrow [0,1])$ is given by
\[
\PMF(x)((w,Z)) = \sum_{\substack{Z=A-I(y)}} \mu_{w}(x,y).
\]
We say that $x,x' \in X$ are \emph{maximal failure equivalent} whenever $\PMF(x) = \PMF(x')$.
\end{definition}

Intuition:
\emph{ready} and \emph{(maximal) failure semantics}, respectively, identify states which have the same probability of reaching processes sharing the same sets of ready actions $I$, or (maximal) sets of failure actions $Z$, respectively, by executing the same traces $w \in A^{*}$. Consequently, appropriate modellings in the coalgebraic setting should capture sets of traces $w$, together with some notion of observations based on execution probabilities of such $w$'s and sets of ready/(maximal) failure actions.

As a first step we define $B_{\I}$, the observation set in Figure~\ref{fig:GenSettMooreGPS}, as $[0,1]^{\pow(A)}$, for ready, failure and maximal failure semantics (for which, for consistency of notation, $\I$ will be instantiated with $\PR$, $\PF$ and $\PMF$, respectively).

The associated ``decorating'' functions $\overline{o}_{\I}\colon X \rightarrow [0,1]^{\pow(A)}$ 
are defined for $x \in X$ as:
$$
\overline{o}_{\PR}(x)(I) = \begin{cases} 1 & \text{ if } I = I(x)\\
0 & \text{ otherwise.}
\end{cases}\qquad 
\overline{o}_{\PF}(x)(Z)  = \begin{cases} 1 & \text{ if } Z \cap I(x) = \emptyset\\
0 & \text{ otherwise.}
\end{cases}
$$
$$
\overline{o}_{\PMF}(x)(Z) =  \begin{cases} 1 & \text{ if } Z = A-I(x)\\ 0 & \text{ otherwise.}
\end{cases}
$$
For the generalised powerset construction for GPS's, $B_{\I} = [0,1]^{\pow(A)}$ is required to carry a $\PD$-algebra structure. This structure is given by the pointwise extension of the free algebra structure in $[0,1] = \PD(1)$:
\begin{align*}
& h_{\I} \colon \PD([0,1]^{\pow(A)}) \rightarrow [0,1]^{\pow(A)} \\
& h_{\I}(\varphi)(Z) = \sum_{f \in supp(\varphi)} \varphi(f) \times f(Z).
\end{align*}

It is easy to check that, for $\I \in \{\PR, \PF, \PMF\}$, the output function $o = h_{\I} \circ \PD(\overline{o}_{\I})$ is explicitly defined, for $\varphi \in \PD(X)$, as:
\[
o(\varphi)(S) = \sum_{x \in supp(\varphi)} \varphi(x) \times \overline{o}_{\I}(x)(S).
\]
This enables the modelling of the behaviour of GPS's in terms of (final) Moore machines with state space in $(B_{\I})^{A^{*}}$ and observations in $B_{\I}$. More explicitly, given a GPS $(X, \delta)$, the decorated trace behaviour of $x \in X$ is represented in the coalgebraic setting by $\bb{\eta(x)} \in (B_{\I})^{A^{*}} = ([0,1]^{\pow(A)})^{A^*} \cong [0,1]^{A^* \times {\pow(A)}}$, precisely the type of the functions in Definitions~\ref{def:ready-gps}--\ref{def:max-failure-gps}. This paves the way for reasoning on ready and (maximal) failure equivalence by coinduction, in terms of Moore bisimulations.

\begin{example}
\label{ex:ready-fail-gps}
Consider, for example, the following GPS's:
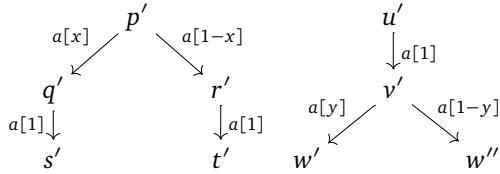
\[
\xymatrix@C=.6cm@R=.4cm{
& p'\ar[dl]_{a[x]}\ar[dr]^{a[1-x]} & & & u'\ar[d]^{a[1]}\\
q'\ar[d]_{a[1]} & & r'\ar[d]^{a[1]} & & v'\ar[dl]_{a[y]}\ar[dr]^{a[1-y]} &\\ 
s' & & t' & w' & & w''
}
\]
States $p'$ and $u'$ are ready equivalent, as their corresponding ready functions
in Definition~\ref{def:ready-gps} are equal:
\[
\begin{array}{rcl}
\PR(p')(\varepsilon, \emptyset) & = & 0\qquad\PR(p')(\varepsilon, \{a\})  =  1 \qquad \PR(p')(a, \emptyset) =  0 \\
\PR(p')(a, \{a\}) & = & \mu_{a}(p',q') + \mu_{a}(p',r')~~  = x + (1-x)  =  1\\
\PR(p')(aa, \emptyset) & = & \mu_{aa}(p',s') + \mu_{aa}(p',t')~~  =  x \times 1 + (1-x) \times 1  =  1\\
\PR(p')(aa, \{a\}) & = & 0 \qquad
\PR(u')(\varepsilon, \emptyset)  =  0\qquad
\PR(u')(\varepsilon, \{a\})  =  1\\
\PR(u')(a, \{a\}) & = & \mu_{a}(u',v')= 1 \qquad
\PR(u')(a, \emptyset)  =  0 \qquad \PR(u')(aa, \{a\})  =  0\\
\PR(u')(aa, \emptyset) & = & \mu_{aa}(u',w') + \mu_{aa}(u',w'') = 1 \times y + 1 \times(1-y) = 1
\end{array}
\]
Intuitively, $\PR(p')(\varepsilon, \emptyset) = 0$ states that from $p'$, by executing the empty trace $\varepsilon$, the probability to reach states that cannot further trigger any action is $0$. This is indeed the case, as $p'$ can always fire $a$ as a first step. Similarly, $\PR(u')(a, \{a\}) = 1$ states that the probability of performing $a$ from $u'$ and reaching states with the ready set $\{a\}$ is $1$. This because $u' \xrightarrow{a[1]} v'$ and $I(v') = \{a\}$.
Nevertheless, the aforementioned ready equivalence follows according to the hierarchy in the right-hand side of Figure~\ref{fig:lattice}, as $p'$ and $u'$  are probabilistic bisimilar as well.

The same answer with respect to the ready equivalence of $p'$ and $u'$ is obtained by applying the coalgebraic framework. As illustrated below, the corresponding Moore automata derived starting from $p'$ and $u'$, respectively, are bisimilar; they have the same branching structure and equal outputs:
\[
\xymatrix@C=.8cm@R=.2cm{
p':& \varphi_{1}\ar[r]^{a}\ar@{=>}[d] & \varphi_{2}\ar[r]^{a}\ar@{=>}[d] & \varphi_{3}\ar@{=>}[d] & & u': & \alpha_{1}\ar[r]^{a}\ar@{=>}[d] & \alpha_{2}\ar[r]^{a}\ar@{=>}[d] & \alpha_{3}\ar@{=>}[d]\\
& o_{\varphi_{1}} & o_{\varphi_{2}} & o_{\varphi_{3}} & & & o_{\alpha_{1}} & o_{\alpha_{2}} & o_{\alpha_{3}}
}
\]
The state spaces of the aforementioned Moore machines consist of the functions:
\[
\begin{array}{rcl}
\varphi_{1} & = & \eta(p') = \{p' \rightarrow 1,\, q' \rightarrow 0,\, r' \rightarrow 0,\, s' \rightarrow 0,\, t' \rightarrow 0\}\\
\varphi_{2} & = & \{p' \rightarrow 0,\, q' \rightarrow x,\, r' \rightarrow 1-x,\, s' \rightarrow 0,\, t' \rightarrow 0\}\\
\varphi_{3} & = & \{p' \rightarrow 0,\, q' \rightarrow 0,\, r' \rightarrow 0,\, s' \rightarrow 1,\, t' \rightarrow 1\}\\[1ex]
\alpha_{1} & = & \eta(u') = \{u' \rightarrow 1,\, v' \rightarrow 0,\, w' \rightarrow 0,\, w'' \rightarrow 0\}\\
\alpha_{2} & = & \{u' \rightarrow 0,\, v' \rightarrow 1,\, w' \rightarrow 0,\, w'' \rightarrow 0\}\\
\alpha_{3} & = & \{u' \rightarrow 0,\, v' \rightarrow 0,\, w' \rightarrow y,\, w'' \rightarrow 1-y\}.
\end{array}
\]
The associated observations are:
\[{\small
o_{\varphi_{1}}= o_{\alpha_{1}}  =  o_{\varphi_{2}}= o_{\alpha_{2}}  = (\emptyset \mapsto  0,  \{a\}\mapsto  1), 
o_{\varphi_{3}} = o_{\alpha_{3}}  =  (\emptyset\mapsto 1, \{a\}\mapsto 0.)}
\]
The functions $\varphi_{2}$, $\varphi_{3}$, $\alpha_{2}$ and $\alpha_{3}$ together with their outputs are easily determined based on the operations of the corresponding Moore coalgebra (as depicted in Figure~\ref{fig:GenSettMooreGPS}).

The connection between the behaviour, {\emph i.e.}, ready function of $p'$ (respectively, $u'$) and $\varphi_{i}$ (respectively, $\alpha_{i}$), for $i \in \{1,2,3\}$, is straightforward. Each of the functions $\varphi_{1}, \varphi_{2}$ and $\varphi_{3}$ captures the behaviour of the system starting from $p'$, after executing the traces $\varepsilon, a$ and $aa$, respectively. Note that, for example, the values of the ready function for trace $\varepsilon$ and ready sets $\emptyset$ and $\{a\}$, respectively, are in one to one correspondence with the assignments in $o_{\varphi_{1}}$. Similarly for the case of $u'$.

By following the same approach, the coalgebraic machinery provides an ``yes'' answer with respect to (maximal) failure equivalence of $p'$ and $u'$ as well. This is also in agreement with the results in~\cite{GPS-Jou-Smolka} stating that ready and (maximal) failure equivalence for GPS's have the same distinguishing power.
\end{example}

%For each of the decorated trace semantics coalgebraically modelled in this section, proving the equivalence with their correspondents in~\cite{GPS-Jou-Smolka} 
%consist of the following steps. First, the ready, failure and maximal failure functions in Definition~\ref{def:ready-gps}, Definition~\ref{def:failure-gps} and Definition~\ref{def:max-failure-gps}, respectively, are extended, from $X$ to $\PD(X)$:
%
%\begin{align*}
%(\forall \I \in \{\PR, \PF, \PMF\})\hspace{10pt} & \overline{\I}\colon \PD(X) \rightarrow A^{*} \rightarrow \pow(\pow A \times [0,1])\\
%& \overline{\PR}(\varphi)(w) = \{(I, \sum_{\substack{I=I(y)}} (\sum_{\substack{x \in supp(\varphi)}} \varphi(x) \times \mu_{w}(x,y)))\}\\
%& \overline{\PF}(\varphi)(w) = \{(Z, \sum_{\substack{Z\cap I(y) = \emptyset}} (\sum_{\substack{x \in supp(\varphi)}} \varphi(x) \times \mu_{w}(x,y)))\}\\
%& \overline{\PMF}(\varphi)(w) = \{(Z, \sum_{\substack{Z = A-I(y)}} (\sum_{\substack{x \in supp(\varphi)}} \varphi(x) \times \mu_{w}(x,y)))\}
%\end{align*}
%
%Note that, for all $\I$ in $\{\PR, \PF, \PMF\}$ the following holds:
%\[
%(\forall x \in X,\, w\in A^{*})\,.\, \overline{\I}(\eta(x))(w) = \{(S,\, \I(x)(w, S))\}.
%\]

The equivalence between the coalgebraic and the original definitions of the decorated trace semantics $\I \in \{\PR, \PF, \PMF\}$ in~\cite{GPS-Jou-Smolka} consists in showing that, given a GPS $(X, \delta)$, $x \in X$, $w \in A^{*}$ and $S\subseteq A$, it holds that $\bb{\eta(x)}(w)(S) = I(x)(w,S)$.
\begin{theorem}
\label{thm:equiv-r-f-mf}
Let $(X, \delta\colon X \rightarrow (\PD(X))^{A})$ be a GPS and $(\PD(X), \langle o, t\rangle)$ be its associated determinisation as in Figure~\ref{fig:GenSettMooreGPS}. Then, for all $x\in X$, $w \in A^{*}$ and $S\subseteq A$, it holds
\[
\bb{\eta(x)}(w)(S) = {\I}(x)(w,S).
\]
\end{theorem}
\begin{proof}
The proof is similar for all $\I$ in $\{\PR, \PF, \PMF\}$, by induction on $w \in A^{*}$. \begin{itemize}\itemsep6pt
\item \emph{Base case --} $w = \varepsilon$: $\bb{\eta(x)}(\varepsilon)(S) = \overline o_{\I}(x)(S) = \I(x)(\varepsilon,S)$.
\item \emph{Induction step.} Here, we will use the fact that the map into the final coalgebra is also an algebra map and the equality 
$$
\I(x)(aw,S) = \sum_{y\in Y}\mu_a(x, y)\times {\I}(x)(w)(S).
$$
Consider $aw \in A^{*}$ and assume $\bb{\eta(y)}(w)(S) = {\I}(y)(w,S)$, for all $y\in X$. We want to prove that $\bb{\eta(x)}(aw)(S) = {\I}(x)(aw)(S)$, for $a \in A$.
\begin{align*}
\bb{\eta(x)}(aw)(S) =~&\bb{\delta(x)(a)}(w)(S)\\
=~& \sum_{y\in Y}\delta(x)(a)(y)\times \bb{\eta(y)}(w)(S) \tag{$\bb{-}$ is an algebra map}\\
=~& \sum_{y\in Y}\delta(x)(a)(y)\times {\I}(x)(w)(S)\tag{IH}\\
=~& \sum_{y\in Y}\mu_a(x, y)\times {\I}(x)(w)(S) \tag{$\mu_{a}(x,x') = \delta(x)(a)(x')$}\\
=~& {\I}(x)(aw)(S)
\end{align*}
\end{itemize}\vspace{-.6cm}
\end{proof}

\subsection{(Maximal) trace semantics}

In this section we provide the coalgebraic modelling of (maximal) trace semantics for GPS's. The approach resembles the one in the previous section: we first recall the aforementioned semantics as introduced in~\cite{GPS-Jou-Smolka}, and then show how to instantiate the ingredients of Figure~\ref{fig:GenSettMooreGPS} in order to capture the corresponding behaviours in terms of (final) Moore coalgebras. As a last step, we prove the equivalence between the coalgebraic modellings and the original definitions in~\cite{GPS-Jou-Smolka}.

\begin{definition}[(Maximal) trace equivalence~\cite{GPS-Jou-Smolka}]
\label{def:trace-gps}
The \emph{trace function}\newline $\PT \colon  X \rightarrow(A^{*} \rightarrow [0,1])$ is given by
\[
\PT(x)(w) = \sum_{\substack{y \in X}} \mu_{w}(x,y).
\]
The \emph{maximal trace function} $\PMT \colon  X \rightarrow(A^{*} \rightarrow [0,1])$ is given by
\[
\PMT(x)(w) = \mu_{w0}(x,{\mathbf 0}).
\]
We say that $x,x' \in X$ are \emph{trace equivalent} whenever $\PT(x) = \PT(x')$. If $\PMT(x) = \PMT(x')$ holds as well, then we say that $x$ and $x'$ are \emph{maximal trace equivalent}.
\end{definition}

%\begin{definition}[Trace equivalence~\cite{GPS-Jou-Smolka}]
%\label{def:trace-gps}
%The \emph{trace function} $\PT \colon  X \rightarrow(A^{*} \rightarrow [0,1])$ is given by
%\[
%\PT(x)(w) = \sum_{\substack{y \in X}} \mu_{w}(x,y).
%\]
%We say that $x,x' \in X$ are \emph{trace equivalent} whenever $\PT(x) = \PT(x')$.
%\end{definition}
%
%\begin{definition}[Maximal trace equivalence~\cite{GPS-Jou-Smolka}]
%\label{def:maximal-trace-gps}
%The \emph{maximal trace function} $\PMT \colon  X \rightarrow(A^{*} \rightarrow [0,1])$ is given by
%\[
%\PMT(x)(w) = \sum_{\substack{y \in X}} \mu_{w0}(x,{\mathbf 0}).
%\]
%We say that $x,x' \in X$ are \emph{maximal trace equivalent} whenever $\PT(x) = \PT(x')$ and $\PMT(x) = \PMT(x')$.
%\end{definition}

From the definition above, it can be easily seen at an intuitive level that trace equivalence identifies processes that can execute with the same probability the same sets of traces $w \in A^{*}$. Moreover, maximal trace equivalence takes into consideration the probability of not triggering any action after the performance of such $w$'s.

Therefore, we choose the set of observations $B_{\I}$ (where $\I = \PT$ for trace and $\I = \PMT$ for maximal trace semantics) to denote probabilities (of processes to execute $w \in A^{*}$, or stagnate after triggering such $w$'s) ranging over $[0,1]$.

We define the ``decorating'' functions, for $\I \in \{\PT, \PMT\}$, $\overline{o}_{\I}\colon X \rightarrow [0,1]$ by 
$$
\overline{o}_{\PT}(x) = 1 \qquad \overline{o}_{\PMT}(x) = \mu_{0}(x, {\mathbf 0})
$$
The (Moore) output function $o$ is given by, for all $\varphi \in \PD(X)$,
\[
o(\varphi) = \sum_{\substack{x \in supp(\varphi)}} \varphi(x) \times \overline{o}_{\I}(x).
\]
We can now show the equivalence between the coalgebraic and the original definition of (maximal) trace semantics.
\begin{theorem}
\label{thm:equiv-trace}
Let $(X, \delta\colon X \rightarrow (\PD(X))^{A})$ be a GPS and $(\PD(X), \langle o, t\rangle)$ be its associated determinisation as in Figure~\ref{fig:GenSettMooreGPS}. Then, for all $x \in X$ and $w \in A^{*}$:
\[
\bb{\eta(x)}(w) = \I(x)(w).
\]
\end{theorem}
\begin{proof}
By induction on $w \in A^{*}$, similar to Theorem~\ref{thm:equiv-r-f-mf}.
\end{proof}
Consider, for instance, the systems $p'$ and $u'$ in Example~\ref{ex:ready-fail-gps}. They are trace equivalent as they both can execute traces $\varepsilon, a$ and $aa$ with total probability $1$. Consequently, they are maximal trace equivalent as well: for sequences $\varepsilon$ and $a$, their associated maximal trace functions compute value $0$, whereas for $aa$ the latter return value $1$.

The same answer with respect to (maximal) trace equivalence of $p'$ and $u'$ is obtained by reasoning on bisimilarity of their associated determinisations derived according to the powerset construction. It is easy to check that in the current setting, the Moore automata corresponding to $\varphi_{1}$ and $\alpha_{1}$ in Example~\ref{ex:ready-fail-gps} output
\begin{itemize}
\item for the case of trace semantics:
\[
(\forall i\in \{1,2,3\})\,.\,o_{\varphi_{i}} = o_{\alpha_{i}} = 1;
\]
\item for the case of maximal trace semantics:
\[
\begin{array}{c}
(\forall i\in \{1,2\})\,.\,o_{\varphi_{i}} = o_{\alpha_{i}} = 0 \text{ and }
o_{\varphi_{3}} = o_{\alpha_{3}} = 1.
\end{array}
\]
\end{itemize}
Therefore $\varphi_{1}$ and $\alpha_{1}$ are bisimilar. Hence, $p'$ and $u'$ are (maximal) trace equivalent.

\section{Decorated trace semantics in a nutshell}
\label{sec:nuthsell}
Next we provide a more compact overview on the coalgebraic machinery introduced in Section~\ref{sec:dec-tr-lts} and Section~\ref{sec:dec-tr-gps}. This also in order to emphasise on the generality and uniformity of our coalgebraic framework.

Recall that for each of the decorated trace semantics we first instantiate the constituents of Figure~\ref{fig:GenSettMoore} (summarising the generalised powerset construction). 
Moreover, for the case of LTS's,
the original definitions of the semantics under consideration are provided with equivalent representations in terms of functions $\varphi^{\I}_{Y}$, paving the way to their interpretation in terms of final Moore coalgebras.

All these are summarised in Figure~\ref{fig:nutshell}, for an arbitrary LTS $(X, \delta \colon X \rightarrow (\Powf X)^{A})$ and an arbitrary GPS $(X, \delta \colon X \rightarrow (\PD X)^{A})$.

Once the ingredients of Figure~\ref{fig:GenSettMoore} and, for LTS's, functions $\varphi^{\I}_{Y}$ are defined, we formalise the equivalence between the coalgebraic modelling of $\I$-semantics and its 
original definition.

For the case of LTS's, for $\I$ ranging over $\Tr,\Ctr,
\Fp, \Rp, \Pf, \Rtr$ and $\Ftr$, we show that the following result holds:

\begin{theorem}
\label{thm:eqiv-gen}
Let $(X, \delta\colon X \rightarrow (\Powf X)^A)$ be an LTS. 
For all $x\in X$, $\llbracket \{x\} \rrbracket = \varphi^\I_{x} \cong \I(x)$.
\end{theorem}
Orthogonally, for the case of GPS's, for $\I$ ranging over $\PR, \PF, \PMF, \PT$ and $\PMT$, we prove the following:
\begin{theorem}
\label{thm:eqiv-gen-gps}
Let $(X, \delta\colon X \rightarrow (\PD X)^A)$ be a GPS.
For all $x \in X$,
$\bb{\eta(x)} = \I(x)$.
\end{theorem}

For each of the semantics under consideration, the proofs of Theorem~\ref{thm:eqiv-gen} and Theorem~\ref{thm:eqiv-gen-gps}, follow by induction on words over the corresponding action alphabet. For more details see the proof of Theorem~\ref{thm:eqiv-ready} in Section~\ref{sec:ready-fail} (for LTS's) and Theorem~\ref{thm:equiv-r-f-mf} in Section~\ref{sec:r-f-mf} (for GPS's), respectively.

\begin{remark}
It is worth observing that by instantiating $T$ with the identity functor, $\F$ with $\pow(-)^{A}$ and, respectively, $\PD(-)^{A}$ in~(\ref{F-final}), in Section~\ref{prelim:gen-pow}, one gets the coalgebraic modelling of the standard notion of bisimilarity for LTS's and, respectively, GPS's.
\end{remark}

Concrete examples on how to use the coalgebraic frameworks are provided for each of the decorated trace semantics. We show how to
%(apply the preprocessing procedure and)
derive determinisations of LTS's and GPS's in terms of Moore automata, which eventually are used to reason on the corresponding equivalences in terms of Moore bisimulations.

\begin{figure}[ht]
\[
%\hspace{-19pt}
\setlength{\arraycolsep}{3pt}
\begin{array}{||c|c|c||c|c|c}
\hline
\I &  \B_{\I} & \bar{o}_{\I}\colon X \to \B_\I\\
\hline\hline
\Rp &  \Powf(\Powf A) &
\overline{o}_\R(x)  =  \{I(\delta(x))\}\\
\hline
\Fp &  \Powf(\Powf A) &
\overline{o}_\Fs(x) = Fail(\delta(x))
\\
\hline
\Ts &  2 &
\overline{o}_\Ts(x) = 1\\
\hline
\CTs & 2 &
\bar{o}_{\CTs}(x) = \left \{
\begin{array}{ r l}
1 & \textnormal{ if $I(\delta(x)) = \emptyset$}\\
0 & \textnormal{ otherwise}
\end{array}
\right.\\
\hline
\Pf & \Pow(\Pow A^{*}) &
\bar{o}_{\Pf}(x) = \{\Ts(x)\}
\\
\hline
\Rtr &
\Powf(\Powf A) &
\overline{o}_\Rtr(x)  =  \{I(\delta(x))\}\\
\hline
\Ftr &
\Powf(\Powf A) &
\overline{o}_\Ftr(x)  =  Fail(\delta(x))\\
\hline
\PR & 
[0,1]^{\pow(A)} &
\overline{o}_{\PR}(x)(I) = \begin{cases} 1 & \text{ if } I = I(x)\\
0 & \text{ otherwise}\end{cases}
\\
\hline
\PF & [0,1]^{\pow(A)} &
\overline{o}_{\PF}(x)(Z)  = \begin{cases} 1 & \text{ if } Z \cap I(x) = \emptyset\\
0 & \text{ otherwise}
\end{cases}
\\
\hline
\PMF & 
[0,1]^{\pow(A)} &
\overline{o}_{\PMF}(x)(Z) =  \begin{cases} 1 & \text{ if } Z = A-I(x)\\ 0 & \text{ otherwise}
\end{cases}
\\
\hline
\PT & 
[0,1] &
\overline{o}_{\PT}(x) = 1\\
\hline
\PMT & 
[0,1] &
\overline{o}_{\PMT}(x) = \mu_{0}(x, {\mathbf 0})\\
\hline
\hline
\end{array}
\]
\caption{The coalgebraic framework in a nutshell.}
\label{fig:nutshell}
\end{figure}

%\begin{figure}[ht]
%\[\small
%\hspace{-19pt}
%\setlength{\arraycolsep}{3pt}
%\begin{array}{|c|c|c||c|c|c|}
%\hline
%\I &  \B_{\I} & \bar{o}_{\I}\colon X \to \B_\I & \I &  \B_{\I} & \bar{o}_{\I}\colon X \to \B_\I  \\
%\hline\hline
%\Rp &  \Powf(\Powf A) &
%\overline{o}_\R(x)  =  \{I(\delta(x))\}
%&
%\Ftr &
%\Powf(\Powf A) &
%\overline{o}_\Ftr(x)  =  Fail(\delta(x))\\
%\hline
%\Fp &  \Powf(\Powf A) &
%\overline{o}_\Fs(x) = Fail(\delta(x))
%&\PR & 
%[0,1]^{\pow(A)} &
%\overline{o}_{\PR}(x)(I) = \begin{cases} 1 & \text{ if } I = I(x)\\
%0 & \text{ otherwise}\end{cases}
%\\
%\hline
%\Ts &  2 &
%\overline{o}_\Ts(x) = 1
%&\PF & [0,1]^{\pow(A)} &
%\overline{o}_{\PF}(x)(Z)  = \begin{cases} 1 & \text{ if } Z \cap I(x) = \emptyset\\
%0 & \text{ otherwise}
%\end{cases}
%\\
%\hline
%\CTs & 2 &
%\bar{o}_{\CTs}(x) = \left \{
%\begin{array}{ r l}
%1 & \textnormal{ if $I(\delta(x)) = \emptyset$}\\
%0 & \textnormal{ otherwise}
%\end{array}
%\right.
%& \PMF & 
%[0,1]^{\pow(A)} &
%\overline{o}_{\PMF}(x)(Z) =  \begin{cases} 1 & \text{ if } Z = A-I(x)\\ 0 & \text{ otherwise}
%\end{cases}
%\\
%\hline
%\Pf & \Pow(\Pow A^{*}) &
%\bar{o}_{\Pf}(x) = \{\Ts(x)\}
% & \PT & 
%[0,1] &
%\overline{o}_{\PT}(x) = 1
%\\
%\hline
%\Rtr &
%\Powf(\Powf A) &
%\overline{o}_\Rtr(x)  =  \{I(\delta(x))\}
%&\PMT & 
%[0,1] &
%\overline{o}_{\PMT}(x) = \mu_{0}(x, {\mathbf 0})\\
%\hline
%\hline
%\end{array}
%\]
%\caption{The coalgebraic framework in a nutshell.}
%\label{fig:nutshell}
%\end{figure}

\section{Canonical representatives}
\label{sec:canon}
Given a {\em decorated} system $(X, <\overline{o}_\I, \delta>)$, we showed in 
the previous sections how to construct a determinisation
$(T(X), <o,t> )$, with $T = \Powf$ for the case of LTS's, and $T = \PD$ for GPS's, respectively. 
The map $\llbracket - \rrbracket \colon T X \to \B_\I^{A^*}$ provides us with a  \emph{canonical representative}
of the behaviour of each state in $T X$. The image $(C,\delta')$ of $(T X, <o,t> )$, via the map $\llbracket  - \rrbracket$, can be viewed as the minimisation with respect to 
the equivalence $\I$. 

Recall that the states of the final Moore coalgebra $(\B_\I^{A^*}, <\epsilon, (-)_a>)$ are functions $\varphi \colon A^* \to \B_\I$ 
and that their decorations and transitions are given by the functions 
$\epsilon \colon \B_\I^{A^*} \to \B_\I$ and $(-)_a\colon  \B_\I^{A^*} \to  (\B_\I^{A^*})^A$, defined in Example~\ref{eg:Moore} in Section~\ref{prelim:coalg}.
The states of the canonical representative $(C,\delta')$ are also functions $\varphi \colon A^* \to \B_\I$, {\emph i.e.}, $C \subseteq \B_\I^{A^*}$. 
Moreover, the function ${\delta}' \colon C \to \B_\I \times C^A$ is simply the restriction of $<\epsilon, (-)_a>$ to $C$, 
that means ${\delta'}(\varphi)=<\varphi(\epsilon), (\varphi)_a>$ for all $\varphi \in C$.

Finally, it is interesting to observe that for LTS $\B_\I^{A^*}$ carries a semilattice structure (inherited from $\B_\I$) and that 
$\llbracket - \rrbracket \colon \Powf X \to \B_\I^{A^*}$ is a semilattice homomorphism. From this observation, it is immediate to conclude 
that also $C$ is a semilattice, but it is not necessarily freely generated, {\emph i.e.}, it is not necessarily a powerset. Similarly, for GPS $\B_\I^{A^*}$ carries a positive convex algebra structure (these are the $\PD$-algebras) and $\llbracket - \rrbracket \colon \PD X \to \B_\I^{A^*}$ is a positive convex algebra homomorphism. Again, from this observation, we know that also $C$ is a positive convex algebra (not necessarily freely generated).

\section{Recovering the spectrum}
\label{sec:recover-spectrum} 

We will briefly explain how to recover the spectrums from Figure~\ref{fig:lattice} from the coalgebraic modelling. First, we recall the following folklore result from coalgebra theory which is the key behind building the spectrum. 
Let $\mathbf{Coalg}_{f}(\F)$ denote the category of all $\F$-coalgebras with a free carrier (arising from a powerset construction) and $\F$-homomorphisms. That is, the objects are of the form $T(X) \to \F T(X)$. Given two functors $\F$ and $\mathscr G$, if one can construct a functor $\sigma \colon \mathbf{Coalg}_{f}(\F) \to \mathbf{Coalg}_{f}(\mathscr G)$ then $\sim_\F \subseteq \sim_{\mathscr G}$.

In the current setting, we apply this to the category $\mathbf{Coalg}_{f}(\F)$ of all $\F$-coalgebras with a free carrier (arising from a powerset construction) and $\F$-homomorphisms. That is, the objects are of the form $T(X) \to \F T(X)$.

For all the relations in the spectrum we can indeed define such $\sigma$. We illustrate here the case for failure and complete trace. 
$$
\xymatrix{(\Powf(X) \ar[r]^-{<o_{\F},t>}& \Powf(\Powf(A^*)) \times \Powf(X)^A}) \stackrel\sigma \mapsto 
\xymatrix{(\Powf(X) \ar[r]^-{<o_{\CTs},t>}& 2 \times \Powf(X)^A)} 
$$
In order to prove that $\sigma$ is a functor we need to show that it preserves homomorphisms.
\begin{lemma}
\label{lm:rec-fail=ct}
Consider $f\colon \Powf(X) \to  \Powf(Y)$ such that $o_\F = o_\F\circ f$. Then $o_\CTs = o_\CTs \circ f$. 
\end{lemma}
\begin{proof}
$$
\begin{array}{lcl}
&& o_\F(S) = o_\F \circ f (S) \\
&\iff&  \{ Z \subseteq A \mid Z \cap I(\delta(s)) = \emptyset, s\in S\} = \{ Z  \subseteq A  \mid Z \cap I(\delta(s')) = \emptyset, s'\in f(S)\} \\
&\iff& \forall_{s\in S} \exists_{s'\in f(S)}\   Z \cap I(\delta(s)) = \emptyset \iff  Z \cap I(\delta(s')) = \emptyset \text{ and vice-versa.}\\
&\Rightarrow& \forall_{s\in S} \exists_{s'\in f(S)}\   I(\delta(s)) = \emptyset \iff  I(\delta(s')) = \emptyset \text{ and vice-versa.}\\
&\iff&  \bigvee_{s\in S} (I(\delta(s)) = \emptyset) =  \bigvee_{s'\in f(S)} (I(\delta(s')) = \emptyset) \\
&\iff& o_\CTs (S) = o_\CTs \circ f (S)
\end{array} 
$$
\end{proof}

Note that this is different from the technique used to recover a hierarchy of probabilistic systems in~\cite{DBLP:journals/tcs/BartelsSV04} where injective natural transformations were defined between functor types and then it was shown that bisimilarity was reflected by these transformations. Here, the situation is different and, for several different equivalences, we have the same functor ({\emph e.g.}, for $\CTs$ and $\Tr$).

In the case of the probabilistic spectrum similar proofs can be given. We illustrate it for the case of probabilistic ready and trace semantics. 
$$
\xymatrix{(\PD(X) \ar[r]^-{<o_{\PR},t>}& [0,1]^{\Powf(A^*)} \times \PD(X)^A}) \quad \stackrel\sigma \mapsto \quad
\xymatrix{(\PD(X) \ar[r]^-{<o_{\PT},t>}& [0,1] \times \PD(X)^A)} 
$$
Again, in order to prove that $\sigma$ is a functor we need to show that it preserves homomorphisms.
\begin{lemma}
\label{lm:rec-ready-trace}
Consider $f\colon \PD(X) \to  \PD(Y)$ such that $o_{\PR} = o_{\PR}\circ f$. Then $o_{\PT} = o_{\PT} \circ f$.
\end{lemma}
\begin{proof}
$$
\begin{array}{lcl}
 o_{\PR}(\varphi) = o_{\PR} \circ f (\varphi) 
&\iff&  \sum\limits_{\scriptsize\begin{array}{c}x\in X\\I = I(x)\end{array}} \!\!\!\!\!\varphi(x)\ \  = \sum\limits_{\scriptsize\begin{array}{c}y\in Y\\I = I(y)\end{array}}\!\!\!\! f(\varphi)(y), \quad \text{for all $I\subseteq A$.}\\\\
&\Rightarrow&  \sum\limits_{I\subseteq A}\sum\limits_{\scriptsize\begin{array}{c}x\in X\\I = I(x)\end{array}} \!\!\!\!\!\varphi(x)\ \  = \sum\limits_{I\subseteq A}\sum\limits_{\scriptsize\begin{array}{c}y\in Y\\I = I(y)\end{array}}\!\!\!\! f(\varphi)(y)\\\\
&\iff&  \sum\limits_{x\in X} \varphi(x)  = \sum\limits_{y\in Y} f(\varphi)(y) \\[1.9ex]
&\iff& o_{\PT} (\varphi)= o_{\PT} \circ f (\varphi)
\end{array} 
$$
\end{proof}

\section{Testing semantics}
\label{sec:testing-semantics}

In this section we show how must and may testing~\cite{CleavelandH89,DBLP:journals/tcs/NicolaH84,Hennessy:1988:ATP:50497} can be modelled coalgebraically by exploiting the generalised powerset construction in the context of LTS's with internal behaviour.
As we shall see, the modelling of may testing is derived based the coalgebraic characterisation of trace semantics in Section~\ref{sec:complete-tr}, in a straightforward fashion. The coalgebraic characterisation of must testing follows as an ``extension to divergence'' of failure semantics in Section~\ref{sec:ready-fail}.

In our approach we consider LTS's on an alphabet $A+ \set{\tau}$,
where $\tau$ is a special label representing \emph{internal
actions}.
We write
$\xRightarrow{\varepsilon}$ to represent $\xrightarrow{\tau}^{*}$ the
reflexive and transitive closure of $\xrightarrow{\tau}$ and, for
$a\in A$, by $\xRightarrow{a}$ we denote $\xrightarrow{\tau}^{*}
\xrightarrow{a} \xrightarrow{\tau}^{*}$. For $w\in A^*$,
$\xRightarrow{w}$ is defined inductively, in the obvious way.

\subsection{From traces to may testing}
\label{sec:trace-may}
In this section we show how may testing semantics can be modelled in the coalgebraic setting.

Intuitively, may testing relates processes in terms of the observable traces (consisting of actions different from $\tau$) they can execute, by ignoring (any number of) occurrences of the internal action $\tau$.

Let $L(p)$ represent the set of observable traces associated with a state $p$ of an LTS with actions in $A \cup \set{\tau}$, referred to as the \emph{language} of $p$:
\begin{equation}
\label{eq:lang-may}
L(p) = \set{w \in (A-\set{\tau})^{*} \mid (\exists p')\,.\,p \xRightarrow{w} p'\ }.
\end{equation}

In~\cite{CleavelandH89}, an alternative characterisation may testing semantics is defined as follows.

\begin{definition}[May semantics~\cite{CleavelandH89}]
\label{def:must-preord} Let $x$ and $y$ be two states of an LTS. We write
$x \may y$ iff $L(x) \subseteq L(y)$. We say that $x$ and $y$ are \emph{may-equivalent} ($x \sim_{may} y$) iff $x \may y$ and $y \may x$.
\end{definition}

The connection with trace semantics in Section~\ref{sec:complete-tr} is rather obvious: both may and trace distinguish processes depending on their languages. 
Hence, we further provide an extension of the coalgebraic modelling of trace semantics to the context of LTS's with internal behaviour, and show it corresponds precisely the may testing as given in Definition~\ref{def:must-preord}.

To begin with, we model LTS's with \emph{internal behaviour} as coalgebras $(S,t\colon S \rightarrow (\pow S)^{A})$, such that, for $x \in S$ and $a \in A$:
\begin{equation}
\label{eq:may-t}
t(x)(a) = \set{y \mid x \xRightarrow{a} y}.
\end{equation}

Then, we decorate LTS's by means of a function $o\ct S \rightarrow 2$ such that, for all $x\in S$
\[
o(x) = 1
\]
and apply the generalised powerset construction as depicted in Figure~\ref{fig:GenSettMoore} in Section~\ref{sec:dec-tr-lts}. Similarly to the case of trace semantics, the final Moore coalgebra is $2^{A^{*}}$ -- the set of languages over $A$. Therefore, by the definition of the transition function $t$ in~(\ref{eq:may-t}), it immediately follows that the behaviour map $\bb{-}$ captures precisely the languages of states in $S$. Namely, for all $x \in S$:
\[
\bb{\set{x}} \cong L(x).
\]

Note that $2^{A^{*}}$ carries a join semilattice structure, where identity is the empty language and join is the union of languages. Consider $\sqsubseteq$ the associated preorder. At this point, the coalgebraic modelling of may testing semantics is straightforward:

\begin{theorem}
\label{thm:may-tr}
Let $x$ and $y$ be two states of an LTS. Then
\[
x \may y \text{ iff } \bb{\set{x}} \sqsubseteq \bb{\set{y}} \text{ and } x \sim_{may} y \text{ iff } \bb{\set{x}} = \bb{\set{y}}.
\]
\end{theorem}

\subsection{From failures to must testing}
\label{sec:fail-must}

In what follows we provide a coalgebraic handling of must testing semantics~\cite{DBLP:journals/tcs/NicolaH84,Hennessy:1988:ATP:50497}, and show the connection between our approach and the framework used for the corresponding (alternative) modelling in~\cite{CleavelandH89}.

Intuitively, must testing relates processes based on the traces that do not lead to divergent states ({\it i.e.}, states that can engage into infinite internal computations), and a notion of non-determinism captured in terms of antichains of corresponding ready actions. By exploiting the isomorphism of antichains and downsets introduced in Section~\ref{sec:ready-fail}, it was easy to
observe that must testing coincides failure semantics for LTS's without internal behaviour (as formalised in Proposition~\ref{prop:must-fail-LTS} later on in this section).
With this intuition in mind, we provide an extension of failure semantics to the context of divergent LTS's and show it coincides with must testing semantics. The aforementioned coincidence is proven by employing a ``lifting'' of the isomorphism of downsets and antichains encompassing information on both the degree of non-determinism and divergence of processes.

%With this intuition in mind, we provide an extension of failure semantics to the context of divergent LTS's which further enables reasoning on the desired notion of behavioural compliance in terms of bisimulations (on final Moore coalgebras) ``abstracting'' from divergent behaviours. Eventually, we prove (in Theorem~\ref{thm:fail-mst-conn}) that the coalgebraic semantics we obtain coincides must testing semantics as introduced in~\cite{CleavelandH89}.

We first recall some notations in~\cite{CleavelandH89}.
The
\emph{acceptance set of $x$ after $w$} is $A(x,w) = \set{\set{a \in A
    \mid x' \xrightarrow{a}} \mid x \xRightarrow{w} x' \land x' \not
  \xrightarrow{\tau}}$. Intuitively, it represents the set of actions
that can be fired after ``maximal'' executions of $w$ from $x$, those
that cannot be extended by some $\tau$-labelled transitions.

The possibility of an LTS to execute $\tau$-actions forever is referred to as
\emph{divergence}. We write $x \not \downarrow$ whenever $x$ diverges.
Dually, the convergence relation $x\downarrow w$ for a state $x$ and a
word $w\in A^*$ is inductively defined as follows: $x \downarrow
\varepsilon$ iff $x$ does not diverge and $x \downarrow aw'$ iff (a)
$x\downarrow \varepsilon$ and (b) if $x \xRightarrow{a} x'$, then
$x'\downarrow w'$.

Given two sets $B,C\in \pow(\pow(A))$, we write $B \subset \subset C$ iff for all $B_i \in B$, there exists $C_{i} \in C$ such that 
$C_{i} \subseteq B_{i}$.

With  these ingredients, it is possible to introduce must preorder and equivalence.
%
%Next we introduce \emph{must preorder} ($\mst$) as in~\cite{CleavelandH89}. Intuitively, $\mst$ abstracts from divergent behaviour and relates processes based on their degree on non-determinism.
%If $p \mst p'$ then the (possible empty) set of processes reached from $p'$ in its attempt to execute a word $w \in A^{*}$ is ``less non-deterministic'' than the corresponding set reached from $p$. This is captured in terms of the acceptance sets $A(p',w)$ and $A(p,w)$ and the inclusion relation on antichains (see (\ref{def:sub-sub-antichain}) in Section~\ref{sec:classical-opt}) as $A(p', w) \inclAnti A(p,w)$.
%
\begin{definition}[Must semantics~\cite{CleavelandH89}]
\label{def:must-preord} Let $x$ and $y$ be two states of an LTS. We write
$x \mst y$ iff for all words $ w \in A^{*}$, if $x\downarrow w$ then $y \downarrow w$ and $A(y, w) \inclAnti A(x,w)$. We say that $x$ and $y$ are \emph{must-equivalent} ($x \sim_{mst} y$) iff $x \mst y$ and $y \mst x$.
\end{definition}
As an example, consider the LTS's depicted below. States $x_4,x_5$ and $y_1$ are divergent. All the other states diverge for words containing the letter $b$ and converge for words on $a^*$. For these words and states $x, x_{1}, x_{2}, x_{3}$ and $y$, the corresponding acceptance sets equal $\set{\set{a,b}}$. In particular, note that $A(x_2,\varepsilon)$ is $\set{\set{a,b}}$ and not $\set{\set{b},\set{a,b}}$. It is therefore easy to conclude that $x,x_1,x_2, x_3$ and $y$ are all must equivalent.
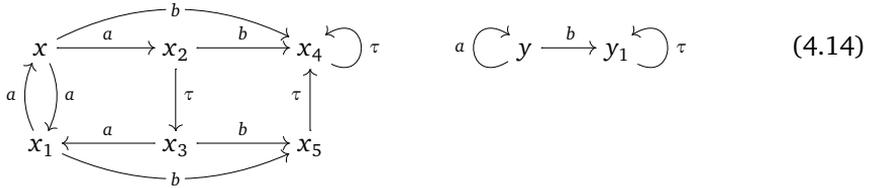
\begin{equation}\label{eq:LTSMUST2}
\xymatrix@C=1.2cm@R=.8cm{
x \ar@/^1.2pc/[rr]|{\;b\;} \ar[r]^{a} \ar@/^/[d]^{a} & x_2 \ar[d]^{\tau} \ar[r]^b & x_4  \ar@(dr,ur)_{\tau}\\
x_1 \ar@/^/[u]^a \ar@/_1.1pc/[rr]|{\;b\;}  & x_3 \ar[l]_a \ar[r]^b & x_5 \ar[u]^{\tau} \\
}\;\;\;\hspace{15pt}
\xymatrix@C=.7cm@R=.4cm{
y\ar@(dl,ul)^{a}\ar[r]^{b} & y_{1}\ar@(dr,ur)_{\tau}
}
\end{equation}

\paragraph{Coalgebraic characterisation of must semantics.}
In what follows we show how must testing semantics can be captured in terms of coalgebras.

In order to proceed, we have to properly tackle \emph{internal behaviour} and \emph{divergence}.
We model LTS's on $A+\set{\tau}$ in terms of coalgebras $(S, t\colon S \rightarrow (1 + \pow S)^{A})$, where $1 = \set{\top}$ is the singleton set, and for $x \in S$,
\begin{align}
\label{eq:div-LTS-trans}
  t(x)(a) =    \top, \textnormal{ if } x \not \downarrow a \qquad  t(x)(a) =
    \set{y \mid x \xRightarrow{a} y},  \text{ otherwise}.
\end{align}
Note that we use $x \not \downarrow a$ as a shorthand for $x \not \downarrow a\varepsilon$.
Intuitively, a state $x \in S$ that displays divergent behaviour with respect to an action $a \in A$ is mapped to $\top$. Otherwise $t$ computes the set of states that can be reached from $x$ through $a$ (by possibly performing a finite number of $\tau$-transitions).

Similarly to failure equivalence in Section~\ref{sec:ready-fail}, we decorate the states of the LTS by means of a function $o\colon S \to 1+\pow(\pow(A))$
defined as follows:
\begin{align}
\label{eq:div-LTS-dec}
  o(x) &=
  \begin{cases}
    \top& \textnormal{ if } x \not \downarrow\\
    \bigcup_{x\xrightarrow{\tau}x'} o(x')& \textnormal{ if } x
    \xrightarrow{\tau}\\
    \Fail(t(x)) & \text{ otherwise}.
  \end{cases} 
\end{align}

Note that $(S,<o,t>)$ is an $FT$-coalgebra 
for the functor $F(S) = (1 + \pow(\pow A)) \times S^{A}$ and the monad $T(S)=1 + \pow S$.
% with unit $\eta$ and multiplication $\mu$ defined as
%\begin{align*}
%	\eta(x) &= \set{x} &
%  \mu(U) &=
%  \begin{cases}
%  	\top & \textnormal{if } U = \top \text{ or } \top \in U\\
%    \bigcup_{U' \in U} U' & \text{otherwise.}\\
%  \end{cases}
%\end{align*}
$T$-algebras are semilattice with bottom and an extra element $\top$
acting as \emph{top} ({\emph i.e.}, such that $x \sqcup \top = \top$ for
all $x$). For any set $U$, $1+ \pow (U)$ carries a semilattice with
bottom and top: bottom is the empty set; top is the element $\top\in
1$; $X \sqcup Y$ is defined as the union for arbitrary subsets
$X,Y\in\pow(U)$ and as $\top$ otherwise. Consequently, $1 + \pow(\pow
A)$ and $FT(S)$ carry a $T$-algebra structure as well.
This enables the application of the generalised powerset construction (Section~\ref{prelim:gen-pow}) associating to each $FT$-coalgebra $(S, <o, t>)$  
the $F$-coalgebra $(1+ \pow S,  <o^{\sharp}, t^{\sharp}>)$ defined for all $X \in 1+\pow S$ as expected:
%\begin{equation}\label{def:otmst}
\begin{align*}
  o^{\sharp} (X) &=
  \begin{cases}
    %o(x) & \text{ if } X = \set{x}\text{ with } x\in S\\
    \top & \text{ if } X = \top\\
    \bigsqcup_{x\in X} o(x) & \text{ if } X \in \pow(S)\\
%    0 & \text{ if } X= 0\\
%    o^{\sharp}(X_1) \sqcup o^{\sharp}(X_2) & \text{ if } X= X_1 \sqcup X_2\\
  \end{cases}&
  %\hspace{8pt}
  t^{\sharp} (X)(a) &= 
  \begin{cases}
    %t(x)(a) & \text{ if } X= \set x \text{ with } x\in S\\
    \top & \text{ if } X = \top\\
    \bigsqcup_{x\in X} t(x)(a) & \text{ if } X \in \pow(S)\\
%    0 & \text{ if } X= 0\\
%    t^{\sharp}(X_1)(a) \sqcup t^{\sharp}(X_2)(a) & \text{ if } X= X_1 \sqcup X_2.\\
  \end{cases}
\end{align*}
%\end{equation}
Note that in the above definitions, $\sqcup$ is not simply the union of subsets (as it was the case for failure), but it is the join operation in $1+\pow (\pow(A))$ and $1+\pow( \pow(S))$.  
Moreover, $(1+ \pow S,  <o^{\sharp}, t^{\sharp})>$ is a Moore machine with output in $1 + \pow(\pow A)$ and, therefore induces a function $\bb{-}\colon (1+\pow(S))\to (1 + \pow(\pow A))^{A^*}$.
The semilattice structure of $1+\pow(\pow(A))$ can be easily lifted to $(1 + \pow(\pow A))^{A^*}$: 
bottom, top and $\sqcup$ are defined pointwise on $A^*$.
We denote by $\tauFailFin$ 
the preorder on $(1 + \pow(\pow A))^{A^*}$ induced by this semilattice.

A result (based on the isomorphism between downsets and antichains) similar to the one for failures, in Section~\ref{sec:ready-fail}, can also be derived in a modular fashion, for the case of LTS's decorated with outputs in $1 + \pow(\pow A)$.

As shown in Section~\ref{sec:ready-fail}, both the set of downsets $\overline{\mathcal{D}}(\pow(A))$, and the set of antichains $\overline{\mathcal{A}}(\pow(A))$ carry join-semillatice structures. It is easy to see that the corresponding extensions to $1 + (-)$ are join-semilattices with bottom as $0$, top as $\top$ (which, intuitively, plays the role of the greatest element) and $\sqcup$ extended as $\top \sqcup C = \top$
for $C \in 1 + \overline{\mathcal{D}}(\pow(A))$, or $C \in 1+\overline{\mathcal{A}}(\pow(A))$, respectively.

The isomorphism $1+i \colon 1 + \overline{\mathcal{D}}(\pow(A)) \rightarrow 1 + \overline{\mathcal{A}}(\pow(A))$ 
follows immediately from the isomorphism $i \,:\, \overline{\mathcal{D}}(\pow(A)) \rightarrow \overline{\mathcal{A}}(\pow(A))$ in~(\ref{eq:iso-i}) in Section~\ref{sec:ready-fail},
by defining
$$
  (1+i)(\top) = \top\enspace \qquad\qquad
  (1+i)(F) = i(F),\quad F\neq \top.
$$
%It can be easily shown that $1+i$ is an homomorphism (of semilattices with top and bottom) as well.

In the sequel, we will exploit $1+i$ to define a ``more efficient'' characterisation of the function $o^{\sharp}\colon 1 + \pow(S) \rightarrow 1 + \pow(\pow(A))$, also useful to prove the soundness of the coalgebraic modelling of must testing semantics (formalised in Theorem~\ref{thm:fail-mst-conn}).

As a first step, observe that the function $o \colon S \to 1+\pow(\pow(A))$ can be restricted to 
$o \colon S \to 1+ \overline{\mathcal{D}}(\pow(A))$ (since if $x\downarrow$ then $o(x)$ is a downset and the union of downsets is a downset, otherwise $o(x)=\top$). 
In analogy with Section~\ref{sec:ready-fail}, we define $o_{2}\colon S \to 1 + \overline{\mathcal{A}}(\pow(A))$
\begin{align*}
  o_{2}(x) &=
  \begin{cases}
    \top & \textnormal{if } x \not \downarrow\\
    min(\cup_{x\xrightarrow{\tau}x'} \overline{o}(x')) & \textnormal{if } x \xrightarrow{\tau}\\[0.5ex]
    \set{I(t(x))} & \text{otherwise}.\\
  \end{cases}
\end{align*}
and $o_{2}^{\sharp} \,:\, 1 + \pow(S) \rightarrow 1 + \overline{\mathcal{A}}(\pow(A))$ as
\begin{align*}
o_{2}^{\sharp} (X) &=
  \begin{cases}
    o_{2}(x) & \text{ if } X = \set{x}\text{ with } x\in S\\
    \top & \text{ if } X = \top\\
    0 & \text{ if } X= 0\\
    min(o_{2}(X_1) \sqcup o_{2}(X_2)) & \text{ if } X= X_1 \sqcup X_2\\
  \end{cases}
\end{align*}

Proposition~\ref{prop:iso-must-ini} states that it is equivalent computing $o^{\sharp}$ or computing $o_{2}^{\sharp}$. To this aim, we need the following lemmas.

\begin{lemma}\label{lemma:ioverlineo}
$(1+i) \circ o = o_{2}$
\end{lemma}
\begin{proof}
If $x\not \downarrow$, then  $o_{2}(x)=\top = (1+i) \circ o(x)$.
If $x\downarrow$ and $x \not \tr{\tau}$, then $o_{2}(x)= \set{I(t(x))} = (1+i) (\Fail(t(x)))= (1+i) \circ o(x)$.
If $x\downarrow$ and $x \tr{\tau}$, then observe that $o(x) = \bigcup \{\Fail(t(x'))\mid x \xrightarrow{\tau}^{*} x'\text{ and } x' \not\xrightarrow{\tau}\}$ and that $$o_{2}(x)=    min\left(\bigcup \set{I(t(x'))\mid x \xrightarrow{\tau}^{*} x'\text{ and }x' \not\xrightarrow{\tau}}\right).$$ We obtain the conclusion by the previous case and by the fact that $i$ is a homomorphism of semilattices.
\end{proof}

\begin{lemma}\label{lemma:io}
$(1+i) \circ o^{\sharp} = o_{2}^{\sharp}$
\end{lemma}
\begin{proof}
Follows immediately by Lemma~\ref{lemma:ioverlineo} and the fact that $(1+i)$ is a homomorphism of semilattices.
\end{proof}

\begin{proposition}
\label{prop:iso-must-ini}
 For all $X,Y\in \pow(S)$, $o^{\sharp}(X)=o^{\sharp}(Y)$ iff $o_{2}^{\sharp}(X)=o_{2}^{\sharp}(Y)$.
\end{proposition}
\begin{proof}
 Follows from Lemma \ref{lemma:io} and the fact that $(1+i)$ is an  isomorphism of semilattices.
\end{proof}

\begin{remark}
\label{rm:conn-fail-acc}
Note that the relation $\subset \subset$ used for defining $\mst$:
\begin{equation}
\label{def:sub-sub-antichain}
B \inclAnti C \texttt{ iff } (\forall B_{i} \in B).(\exists C_{i} \in C). C_{i} \subseteq B_{i}
\end{equation}
is the ordering induced by $\sqcup$ in $\overline{\mathcal{A}}(\pow(A))$:
\begin{align*}
& min(B \sqcup C) = C\\
~\texttt{ iff }& min(B \sqcup C) = min(C) \tag{as $C \in \overline{\mathcal{A}}(\pow(A)$}\\
~\texttt{ iff }& (\forall B_{i} \in B).(\exists C_{i} \in C). C_{i} \subseteq B_{i} \tag{by definition of $min$}\\
~\texttt{ iff }& B \inclAnti C \tag{by definition of $\inclAnti$}
\end{align*}
\end{remark}

We formalise the coalgebraic modelling of must semantics in the following theorem.
\begin{theorem}
\label{thm:fail-mst-conn}
%Let $(S, \delta)$ be an LTS with divergence and $x, y \in S$.
Let $x$ and $y$ be two states of an LTS. Then

\[
x \mst y \textnormal{ iff } \bb{\set y} \tauFailFin \bb{\set x} \textnormal{ and } x\sim_{mst}y \textnormal{ iff } \bb{\set{x}}=\bb{\set{y}}.
\]
\end{theorem}
The morphism $o_{2}^{\sharp} \colon 1 + \pow(S) \rightarrow 1 + \overline{\mathcal{A}}(\pow(A))$ is useful to prove that the preorders $\tauFailFin$ and $\mst$ coincide.
Indeed, the Moore machine $(1+\pow S, <o_{2}^{\sharp},t^{\sharp}>)$ induces the morphism $\bb{-}_{2} \colon 1+ \pow S \to (1+ \overline{\mathcal{A}}(\pow(A)))^{A^*}$ defined for all $X \in 1+\pow(S)$ as 
\[
\begin{array}{lr}
\bb{X}_{2}(\varepsilon) = o_{2}^{\sharp}(X) &\hspace{10pt} \bb{X}_{2}(a w) = \bb{t^{\sharp}(X)(a)}_{2}(w).
\end{array}
\] 
The isomorphism $(1+i)\colon 1+\overline{\mathcal{D}}(\pow(A)) \to 1+\overline{\mathcal{A}}(\pow(A))$ can be extended to the isomorphism $(1+i)^{A^*}\colon (1+\overline{\mathcal{D}}(\pow(A)))^{A^*} \to (1+\overline{\mathcal{A}}(\pow(A)))^{A^*}$, defined for every function $\phi \in (1+\overline{\mathcal{D}}(\pow(A)))^{A^*}$ and word $w\in A^{*}$ as
\[(1+i)^{A^*}(\phi)(w)=(i+1)(\phi(w)).\]
Note that the function $\bb{-} \colon 1+\pow(S) \to (1+\pow\pow(A))^{A^*}$ can be restricted to $\bb{-} \colon 1+\pow(S) \to (1+\overline{\mathcal{D}}(\pow(A)))^{A^*}$.
\begin{proposition}
 $(1+i)^{A^*} \circ \bb{-} = \bb{-}_{2}$
\end{proposition}
\begin{proof}
This can be proved by ordinary induction on words, exploiting Lemma \ref{lemma:io} for the base case.
\end{proof}
The ordering $\tauFailFind$ induced by the semilattice structure of $(1+\overline{\mathcal{A}}(\pow(A)))^{A^*}$ is given as follows: for all $\phi,\psi \in (1+\overline{\mathcal{A}}(\pow(A)))^{A^*}$, $\phi \tauFailFind \psi$ iff for all $w\in A^*$
\begin{enumerate}
\item  if $\phi(w)=\top$ then $\psi(w)= \top$ and
\item  if $\phi(w) \neq \top$ then either $\psi(w)= \top$ or $\phi(w) \inclAnti \psi(w)$.
\end{enumerate}

Observe that $\bb{X}_{2}(w)=\top$ iff $X \not \downarrow w$. Furthermore, whenever $X=\set{x}$ and $x \downarrow w$, $\bb{X}_{2}(w)=A(x,w)$. As a consequence, the following proposition holds.
\begin{proposition}
 $\bb{\set{x}}_{2} \tauFailFind \bb{\set{y}}_{2}$ iff $y\mst x$
\end{proposition}
\begin{proof}
Suppose that $\bb{\set{x}}_{2} \tauFailFind \bb{\set{y}}_{2}$ and take one word $w\in A^*$. If $y \downarrow w$, then $\bb{\set{y}}_{2}(w)\neq \top$ and also $\bb{\set{x}}_{2}(w)\neq \top$, that is $x \downarrow w$. This means that $\bb{\set{x}}_{2}(w) \inclAnti \bb{\set{y}}_{2}(w)$, that is $A(x,w)\inclAnti A(y,w)$. summarising $y\mst x$.
 
Now suppose that $y\mst x$ and take one word $w\in A^*$. If $\bb{\set{x}}_{2}(w)=\top$, then $x\not \downarrow w$. This implies that also $y\not \downarrow w$ (and thus $\bb{\set{y}}_{2}(w)=\top$) because otherwise the hypothesis $y\mst x$ would be violated. If $\bb{\set{x}}_{2}(w) \neq \top$, then we have two possibilities: (a) $y  \downarrow w$ or (b) $y \not \downarrow w$. For (a), we have that $A(x,w) \inclAnti A(y,w)$, that is $\bb{\set{x}}_{2}(w) \inclAnti \bb{\set{y}}_{2}(w)$. For (b), we immediately have that $\bb{\set{y}}_{2}= \top$.  
\end{proof}

From the two above propositions, Theorem \ref{thm:fail-mst-conn} follows immediately.

Note that in absence of divergence, the ``decorating'' function in~(\ref{eq:div-LTS-dec}) and the transition function in~(\ref{eq:div-LTS-trans}) correspond precisely to $\overline{o}_{\Fs}$ and $\delta$ in Section~\ref{sec:ready-fail}, for the case of failure semantics.
Hence, by Theorem~\ref{thm:fail-mst-conn}, Definition~\ref{def:must-preord} and Remark~\ref{rm:conn-fail-acc} it follows immediately that must and failure semantics coincide in the context of LTS's without internal behaviour.
\begin{proposition}
\label{prop:must-fail-LTS}
Consider two states $x,y$ of an LTS without internal behaviour. Then
\[
\begin{array}{l}
x \mst y \text{ iff } \Fs(y) \subseteq \Fs(x)\\
x \sim_{mst} y  \text{ iff } \Fs(x) = \Fs(y).
\end{array}
\]
\label{prop:must-fail-LTS}
\end{proposition}

\begin{remark}
Note that according to the definition of $\tauFailFin$, $\bb{\set y} \tauFailFin \bb{\set x}$ iff $\bb{\set y}\sqcup \bb{\set x} = \bb{\set x}$, and since 
 $\bb{-}$ is a $T$-homomorphism (namely it preserves bottom, top and $\sqcup$), the latter equality holds iff $\bb{\set{y, x}} = \bb{\set x}$. Summarising,
\[x \mst y \text{ iff } \bb{\set{x, y}} = \bb{\set{x}}.\]
\end{remark}

Consider, once more, the LTS in~\eqref{eq:LTSMUST2}. The part of the Moore machine $(1+\pow(S),<o^{\sharp},t^{\sharp}>)$ which is reachable from $\set{x}$ and $\set{y}$ is depicted below (the output function $o^{\sharp}$ maps $\top$ to $\top$ and the other states to $\set{0}$). 
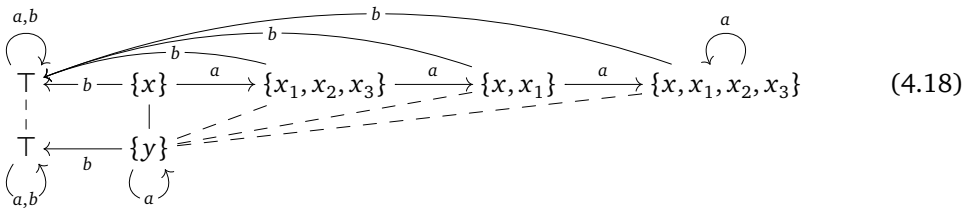
\begin{equation}\label{eq:DETMUST}
{\xymatrix@C=1cm@R=.3cm{
{\top} \ar@(ul,ur)^-{a,b} \ar@{--}[d]& {\set{x}} \ar@{--}[d] \ar[l]|{\;b\;} \ar[r]^-{a} & {\set{x_{1},x_{2},x_{3}}} \ar@{--}[dl] \ar@/_{1.0pc}/[ll]|{\;b\;}\ar[r]^-{a} & {\set{x,x_1}} \ar@{--}[dll] \ar@/_{1.6pc}/[lll]|{\;b\;} \ar[r]^-{a} & {\{x,x_{1},x_{2},x_{3}\}} \ar@{--}[dlll] \ar@/_{2.2pc}/[llll]|{\;b\;}\ar@(ur,ul)_{a} \\
{\top} \ar@(dl,dr)|-{a,b} & {\set{y}} \ar@(dl,dr)|{\;a\;}\ar[l]^{b}
}}
\end{equation}
The relation consisting of dashed lines is a bisimulation proving that $\bb{\set{x}}=\bb{\set{y}}$, i.e., that $x \sim_{mst} y$.

Our construction is closely related to the one
in~\cite{CleavelandH89}, that transforms LTS's into (deterministic)
acceptance graphs. We further provide more details on the connection between the coalgebraic machinery for reasoning on must preorder and the corresponding framework in~\cite{CleavelandH89}.

\paragraph{Moore machines and acceptance graphs.}
As previously introduced in this section, the behaviour of an LTS with divergence ${\mathcal L} = (S, t \colon S \rightarrow (1 + \pow S)^{A})$ 
can be captured in terms of a Moore machine
\[{\mathcal M} = (1+\pow S, \langle o^{\sharp}, t^{\sharp} \rangle \colon 1 + \pow S \rightarrow (1 + \pow(\pow A)) \times (1 + \pow S)^{A})\]
derived according to the powerset construction, and that reasoning on must preorder is equivalent to reasoning on the preorder $\sqsubseteq_{\mathcal M}$ on the final Moore coalgebra, as stated in Theorem~\ref{thm:fail-mst-conn}.

In~\cite{CleavelandH89} must preorder is established in terms of a notion of prebisimulation ($\sqsubseteq_{\langle \Pi, 0\rangle}$) on the so-called ``acceptance graphs'' generated from such $\mathcal L$'s, denoted by $ST({\mathcal L})$.
Intuitively, an acceptance graph $ST({\mathcal L})$ consists of a set of nodes $p$ of shape $\langle Q, b\rangle \in \pow S \times \set{\true, \false}$, where $Q$ is a set of states in $S$, and $b$ is associated the boolean value ${\true}$ whenever all states in $Q$ converge (written $Q \downarrow$) and ${\false}$ otherwise.

Orthogonally to the Moore machines with output in $1 + \pow(\pow A)$, for a node $p = \langle Q, b \rangle$ in $ST({\mathcal L})$, the information representing the divergence of (states in) $Q$ is given by $p.\emph{closed}\,(= b)$, and the corresponding (minimised) acceptance set consisting of visible actions that can be triggered as a first step from the states in $Q$  is represented by $p.\emph{acc}$ (defined later on in this section). Moreover, (deterministic) transitions in $ST(\mathcal L)$ are of shape $\langle Q_{1}, b_{1} \rangle \xrightarrow{a} \langle Q_{2}, b_{2} \rangle$, where $a \in A$ and $Q_{2}$ is the set of $a$-successors of states in $Q_{1}$, computed with respect to $\xRightarrow{a}$.

Based on the resemblance between the aforementioned Moore machines and acceptance graphs, we consider worth investigating to what extent these constructions and the corresponding ``alternative'' semantics used for reasoning on must preorder are connected.

In what follows we recall the formal definition of acceptance graphs as introduced in\linebreak \cite{CleavelandH89}, show they are isomorphic (up-to divergent behaviours) with the Moore machines used for the coalgebraic modelling of must semantics.

We proceed by first providing the basic ingredients needed for the definition of acceptance graphs.

Consider $Q \in \pow S$. The \emph{$\varepsilon$-closure} of a $Q$ is $Q^{\varepsilon} = \set{p \mid q \xRightarrow{\varepsilon} p \land q \in Q}$. 
The set of \emph{direct $a$-successors} of states $q \in Q$ is $D(Q, a) = \set{q' \mid q \xrightarrow{a} q' \land q \in Q}$, 
where $a \in A \cup \set{\tau}$.

\begin{definition}[Acceptance graphs~\cite{CleavelandH89}]
\label{def:acc-graphs}
Consider ${\mathcal L}$ an LTS with divergence, with state space $S$ and visible actions labelled in $A$. The corresponding \emph{acceptance graph} $ST({\mathcal L}) = (T, A \cup \set{\tau}, \rightarrow)$ is defined as follows.
\begin{enumerate}
\item $T = \set{\langle Q, b\rangle \in \pow S \times \set{\true, \false} \mid Q = Q^{\varepsilon} \land (b = \true \Rightarrow Q \downarrow)}$.

\item For $p = \langle Q, b\rangle \in T$ define $p.{\emph closed} = b$ and
\begin{align*}
p.{\emph acc} &=
  \begin{cases}
    0 & \text{ if } p.{\emph closed} = \false\\
    min(\set{\set{a \in A \mid q \xrightarrow{a}} \mid q \in Q \land q \not \xrightarrow{\tau}}) & \text{otherwise}.
  \end{cases}
\end{align*}
(We refer to~(\ref{eq:min}) in Section~\ref{sec:ready-fail} for the definition of $min$.)

\item A transition $\langle Q_{1}, b_{1} \rangle \xrightarrow{a}_{T} \langle Q_{2}, b_{2} \rangle$ is performed exactly when the following hold:
$a \not = \tau$, $D(Q_{1}, a)^{\varepsilon} = Q_{2}$, and $b_{1} = \true \land (Q_{2} \downarrow\,\, \Rightarrow b_{2} = \true)$.
\end{enumerate}
\end{definition}

It is worth observing that, according to Definition~\ref{def:acc-graphs}, acceptance graphs are deterministic, and moreover, there are no outgoing transitions from divergent states (see (c) above). Considering graphs satisfying the latter property comes as a natural consequence of the fact that the must preorder considers divergence catastrophic, as can be inferred from Definition~\ref{def:must-preord}.

Given an LTS with divergence $\mathcal L$, and $q$ a state of $\mathcal L$, the node in $ST({\mathcal L})$ corresponding to $q$ is $\langle \set{q}^{\varepsilon}, q \downarrow \varepsilon\rangle$. Orthogonally, the state corresponding to $q$ in the Moore machine derived according to the powerset construction is $\set{q}$.

For an example, consider the following LTS:
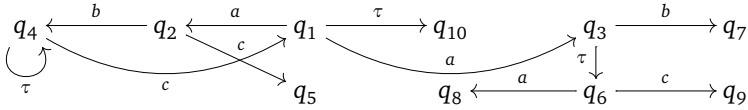
\[
\xymatrix@C=1.3cm@R=.4cm{
q_{4}\ar@(dl,dr)_{\tau}\ar@/_1.5pc/[rr]_{c}  & q_{2}\ar@{->}[l]_{b}\ar@{->}[dr]^{c}  & q_{1}\ar@{->}[r]^{\tau}\ar@{->}[l]_{a}\ar@/_1.5pc/[rr]^{a} & q_{10} & q_{3}\ar@{->}[d]_{\tau}\ar@{->}[r]^{b}& q_{7}\\
& &q_{5}  & q_{8} & q_{6}\ar@{->}[l]_{a}\ar@{->}[r]^{c} & q_{9}
}
\]
The associated Moore determinisation ${\mathcal M}$ when starting from $q_{1}$ and the corresponding acceptance graph $ST({\mathcal L})$, respectively, are illustrated as follows.
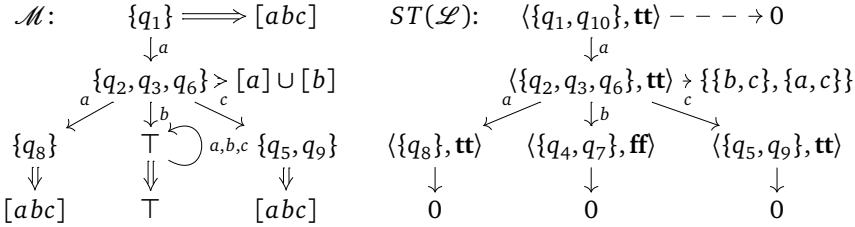
\[
\xymatrix@C=.1cm@R=.3cm{
{\mathcal M}:&  \set{q_{1}}\ar@{=>}[r]\ar@{->}[d]^{a} & [abc] & & ST({\mathcal L}):& \langle \set{q_{1}, q_{10}}, \true \rangle\ar@{-->}[r]\ar@{->}[d]^{a} & 0\\
 & \set{q_{2}, q_{3}, q_{6}}\ar@{=>}[r]\ar@{->}[dl]_{a}\ar@{->}[d]^{b}\ar@{->}[dr]^{c} & [a]\cup [b] & &  & \langle\set{q_{2}, q_{3}, q_{6}}, \true\rangle\ar@{-->}[r]\ar@{->}[dl]_{a}\ar@{->}[d]^{b}\ar@{->}[dr]^{c} & \set{\set{b,c}, \set{a,c}}\\
\set{q_{8}}\ar@{=>}[d] & \top\ar@{=>}[d]\ar@(dr,ur)_{a,b,c} & ~~~\set{q_{5}, q_{9}}\ar@{=>}[d] & & \langle \set{q_{8}}, \true\rangle\ar@{-->}[d] & \langle \set{q_{4}, q_{7}}, \false\rangle\ar@{-->}[d] & \langle \set{q_{5}, q_{9}}, \true\rangle\ar@{-->}[d]\\
[abc] & \top & [abc] & &  0& 0 & 0\\
}\]

Recall from Section~\ref{sec:ready-fail} that for the simplicity of notation we write, for example, $[abc]$ in order to denote the powerset of $\set{a,b,c}$.
In $ST({\mathcal L})$, the notation $\langle Q, b\rangle \dashrightarrow B$ represents a node $p = \langle Q, b\rangle$ such that $p.{\emph acc} = B$ and $p.{\emph closed} = b$.

Observe that: both $\mathcal M$ and $ST({\mathcal L})$ are deterministic, transitions starting from divergent states $\top$ in $\mathcal M$ always produce output $\top$, whereas in $ST({\mathcal L})$ divergent nodes $p = \langle Q^{\varepsilon}, \false \rangle$ are deadlock-like and, moreover, $p.{\emph acc} = 0$.

Given an LTS ${\mathcal L}$ with state space $S$,
the connection between non-divergent nodes $Q$ in the corresponding Moore machine ${\mathcal M} = (1 + \pow S, \langle o^{\sharp}, t^{\sharp}\rangle)$ and those in the associated acceptance graph $ST({\mathcal L})$ is obvious. Each such Moore state $Q$ corresponds to a node $p=\langle Q^{\varepsilon}, \true\rangle$ in the acceptance graph such that $p.{\emph acc} = i(o^{\sharp}(Q))$, where $i$ is the isomorphism between downsets and antichains defined in Section~\ref{sec:ready-fail}.

For example, state $\set{q_{1}}$ in $\mathcal M$ is in one to one correspondence with $p = \langle \set{q_{1}^{\varepsilon}} = \set{q_{1}, q_{10}}, \true\rangle$ in $ST({\mathcal L})$, and, moreover:
\begin{align*}
 i(o^{\sharp}(\set{q_{1}}))=~& i(F = \set{0, \set{a}, \set{b}, \set{c}, \set{a,b}, \set{a,c}, \set{b,c}, \set{a,b,c}})\\
=~& min(\cup_{F_{i} \in F} \set{A - F_{i}})= 0 = p.{\emph acc}.
\end{align*}

As already hinted, a divergent set of states $Q$ is represented by $\top$ in the Moore machine derived from an LTS, and it corresponds to a node $p = \langle Q^{\varepsilon}, \false\rangle$ in the associated acceptance graph, such that $p$ has no outgoing transitions and $p.{\emph acc} = 0$. For this case we refer to the states $Q = \set{q_{4}, q_{7}}$ in $\mathcal L$.

An important remark is that divergent nodes and their successors in the Moore machines can safely be ignored when reasoning on $\sqsubseteq_{\mathcal M}$. This follows as a consequence of:
\begin{equation}
\label{eq:mst-no-div}
\begin{array}{rl}
& \bb{X} \sqsubseteq_{\mathcal M} \bb{Y}\\
\text{iff}& (\forall w \in A^{*})\,.\, \bb{X}(w) \sqsubseteq \bb{Y}(w)\\
\text{iff}& (\forall w \in A^{*})\,.\, Y \downarrow~ \Rightarrow (X \downarrow \land~ \bb{X}(w) \sqsubseteq \bb{Y}(w))\\
& \text{(as if $X \not \downarrow w$ then $\llbracket X \rrbracket(w) = \top$, which follows by induction on $w \in A^{*}$)}
\end{array}
\end{equation}
for all $X, Y \in 1 + \pow S$, where $S$ is the state space of the LTS and $A$ is the corresponding action alphabet.

Hence, the corresponding subsequent transitions
$\xymatrix@C=.3cm@R=.4cm{\top & \top\ar@{=>}[l]\ar@(dr,ur)_{a}}$
can be ignored as well, for all $a \in A$.

As a last ingredient in showing the connection between the Moore machine and the acceptance graph associated with an LTS with divergence, we make the following observations. Transitions $\xymatrix@C=.4cm@R=.4cm{o_{1} & Q_{1}\ar@{=>}[l]\ar@{->}[r]^{a} & Q_{2}\ar@{=>}[r] & o_{2}}$ between non-divergent states $Q_{1}, Q_{2}$ correspond to transitions $\xymatrix@C=.4cm@R=.4cm{p_{1} = \langle Q_{1}^{\varepsilon}, \true\rangle\ar@{->}[r]^{a} & \langle Q_{2}^{\varepsilon}, \true\rangle = p_{2}}$ such that $p_{i}.{\emph acc} = i(o_{i})$, for $i \in \set{1,2}$.
Each transition  $\xymatrix@C=.4cm@R=.4cm{o_{1} & Q_{1}\ar@{=>}[l]\ar@{->}[r]^{a} & \top\ar@{=>}[r] & \top}$ with $Q_{1}$ a non-divergent state matches a transition ${p_{1} = \langle Q_{1}^{\varepsilon}, \true\rangle} \xrightarrow{a} p_{2}$ such that $p_{1}.{\emph acc} = i(o_{1})$, $p_{2}.{\emph closed} = \false$ and $p_{2}.{\emph acc} = 0$.

At this point we conclude that, given an LTS with divergence $\mathcal L$, the Moore machine derived according to the powerset construction and the corresponding acceptance graph $ST({\mathcal L})$ are isomorphic up-to divergent behaviours.

\section{Discussion}
\label{sec:concl-traces}

In this chapter, we have proved that the  coalgebraic characterisations of decorated trace semantics for labelled transition systems and generative probabilistic systems, respectively, are equivalent with the corresponding standard definitions in~\cite{Glabbeek01} and~\cite{GPS-Jou-Smolka}. More precisely, we have shown that for a state $x$, the coalgebraic canonical representative
$\llbracket \{x\}\rrbracket$,
given by determinisation and finality,
coincides with the classical semantics $\I(x)$, for $\I$ ranging over $\Tr,\Ctr,
\Rp, \Fp, \Pf, \Rtr$ and $\Ftr$,
representing the traces, complete traces, ready pairs, failure pairs, possible futures, ready traces and, respectively, failure traces of $x$ in a labelled transition system. Similar equivalences have been proven for $\I$ ranging over $\PR, \PF, \PMF, \PT$ and $\PMT$ representing the ready, failure, maximal failure, trace and maximal trace functions for the case of probabilistic systems.

We also showed that the spectrum of decorated trace semantics can be recovered from the coalgebraic modelling.

Moreover, we provided an extension of trace and failure semantics to the context of labelled transition systems with internal behaviour, which further enabled the coalgebraic modelling of may and must testing semantics in~\cite{CleavelandH89} via the generalised powerset construction. A similar idea of system determinisation was also applied in~\cite{CleavelandH89}, in a non-coalgebraic setting where, in the absence of internal actions and divergence, respectively, may testing coincides with trace and
{must testing} coincides with failure semantics, respectively. The connection with this work is also studied in this chapter.
%A coalgebraic characterisation of the spectrum was also attempted in~\cite{Monteiro08}, in a
%somewhat {\em ad hoc} fashion. Connections with these works are still to be explored.

In addition, we have illustrated 
how to reason about decorated trace and testing semantics using coinduction,
by constructing suitable Moore bisimulations. This is a sound and complete proof technique, and represents an important step towards automated reasoning,
as it opens the way for the use of, for instance, coinductive theorem provers such as CIRC~\cite{rosu-lucanu-2009-calco}.

%%%%%%%%%%%%%%%% end MFPS %%%%%%%%%%%%%%%%%%%%%%%%

\chapter{Algorithms for decorated trace and testing semantics}
\label{ch:algorithms}

In Chapter~\ref{ch:dec-trace-testing} we provided a coalgebraic handling of a suite of semantics for different types of systems. These consist of decorated trace semantics for labelled transition systems and generative probabilistic systems, and may/must testing semantics for labelled transition systems with internal behaviour. In this chapter we focus on deriving algorithms for reasoning on failure and must testing, but our considerations hold also for the other decorated trace semantics for LTS's in Chapter~\ref{ch:dec-trace-testing}, and for may testing semantics.

The problem of automatically checking these notions of behavioural equivalence is usually
reduced to the problem of checking bisimilarity, as implemented in
several
tools~\cite{cleaveland1993concurrency,DBLP:conf/cav/CleavelandS96,DBLP:journals/topnoc/CalzolaiNLT08,DBLP:conf/tacas/CranenGKSVWW13}
and proposed in~\cite{CleavelandH89} which introduces a procedure for
checking testing equivalences. The idea is the following.
First, non-deterministic systems, represented by labelled transition
systems (LTS's), are transformed into deterministic ``acceptance
graphs'' with a construction which is reminiscent of the
\emph{determinisation} of non-deterministic automata (NDA's). Then,
since bisimilarity in acceptance graphs coincides with testing
equivalence in the original LTS's, one checks bisimilarity via the
so-called \emph{partition refinement} 
algorithm~\cite{Kanellakis:1983:CEF:800221.806724,Paige:87:SIAM}. Such
algorithm, which is the best-known for minimising LTS's with respect to
bisimilarity, is analogous to Hopcroft's minimisation
algorithm~\cite{minimisation} for deterministic automata (DA's)
with respect to language equivalence. In both, a partition of the state space
is iteratively refined until the largest fixed-point is reached. In a
nutshell, the procedure for checking testing semantics adopted
in~\cite{CleavelandH89} is in essence the same as the classical
procedure for checking language equivalence of non-deterministic
automata: first determinise and then compute a largest fixed-point.

%\enlargethispage{\baselineskip}\enlargethispage{\baselineskip}

This observation led us to experiment with applying other interesting language equivalence algorithms, not available for bisimilarity, to solve the problem of checking must and failure semantics. In order to achieve this, we took a coalgebraic perspective of the problem at hand, which allowed us to study the constructions and the semantics in a uniform fashion. The abstract coalgebraic framework enabled a unified study of different kinds of state based systems: (a) both the determinisation of NDA's and 
 the construction of acceptance graphs in~\cite{CleavelandH89} are instances of the generalised powerset construction~\cite{CancilaHL03,Lenisa19992,fsttcs}, 
and (b) the iterations of both the Hopcroft and the partition refinement algorithms are in one-to-one correspondence with the so-called construction of the terminal sequence~\cite{ak:fix-point-set-functor,worrel-final}. While (b) is well-known in the community of coalgebras~\cite{ABHKMS12,FMT2005,k:logics-coalgebras,s:relating-coalgebraic-bisimulation-journal}, (a) is the key observation of this work, which enabled us to devise other algorithms for must and failure semantics (introduced in Section~\ref{sec:ready-fail} and Section~\ref{sec:fail-must}, respectively).

First, we consider \emph{Brzozowski's algorithm}~\cite{Brzozowski}
which transforms an NDA into the minimal deterministic automaton
accepting the same language: the input automaton is reversed (by
swapping final and initial states and reversing its transitions),
determinised, reversed and determinised once more. This somewhat
intriguing algorithm can be explained in terms of duality and
coalgebras~\cite{Bonchi:2012:BAA:2340820.2340823,BKP}. In particular,
the approach in~\cite{Bonchi:2012:BAA:2340820.2340823} allows us to
extend it to Moore machines, which paves the way to adapt Brzozowski's
algorithm for checking testing semantics.

Next, we consider several more efficient algorithms that have been recently introduced in a series of papers~\cite{tacas10,BONCHI:2012:HAL-00639716:4,DoyenR10,CAV06}. 
These algorithms rely on different kinds of \emph{(bi)simulations up-to}, which are proof techniques originally proposed for process calculi~\cite{milner89,MilnerPW92,San98MFCS}. 
From these algorithms, we choose the one in~\cite{BONCHI:2012:HAL-00639716:4} (\HKC), which can be easily proved correct using coalgebraic techniques. \HKC\ can be easily adapted to check must testing, once a coalgebraic characterisation of must equivalence is given.

Comparing the efficiency of these three families of algorithms (partition
refinement\linebreak \cite{CleavelandH89}, Brzozowski and bisimulations up-to)
is not a trivial task. Both the problems of checking language and
testing equivalence are PSPACE-complete as shown in~\cite{MS73}
and~\cite{Kanellakis:1983:CEF:800221.806724}, respectively. However,
in both cases, the theoretical complexity appears not to be
problematic in practice, so that an empirical evaluation is more
desirable. In~\cite{TabakovVardi05,Wat95,Wat01}, experiments have
shown that Brzozowski's algorithm performs better than Hopcroft's one  for ``high-density'' NDA's, while Hopcroft's algorithm is more efficient for generic
NDA's. Both algorithms appear to be rather inefficient compared to
those of the new
generation~\cite{tacas10,BONCHI:2012:HAL-00639716:4,DoyenR10,CAV06}. It
is out of the scope of this work to present an experimental
comparison of the adaptation of these algorithms for must equivalence;
we confine our results to showing that each approach can be more
efficient than the others on concrete examples.

%\paragraph{Contributions.}
%
%\clearpage
%\noindent
summarising, the main contributions of this chapter are:
\begin{itemize}\itemsep0pt
\item The adaptation of \HKC\ and Brzozowski's algorithm for failure
  and must semantics. For the latter, this includes an optimisation
  which avoids an expensive determinisation step. All the
  observations for failure can be used for various other decorated
  trace semantics, such as ready and ready trace.
  \item An interactive applet\footnote{\url{http://perso.ens-lyon.fr/damien.pous/brz}} allowing one to experiment with these
  algorithms.
  \item Experiments checking the equivalence of an ideal and a distributed multiway synchronisation protocol~\cite{ParrowS96}.
  \item At a more conceptual level, the present work also shows that the
coalgebraic analysis of systems yields not only a good mathematical
theory of their semantics but also a rich playground to devise
algorithms.
\end{itemize}

%\paragraph{Related Work.}
%This paper draws inspiration from three others: the coalgebraic characterisation of several decorated trace equivalences in~\cite{dec-tr-MFPS12}; the algorithm \HKC\ for checking language equivalence introduced in~\cite{BONCHI:2012:HAL-00639716:4}; the coalgebraic understanding of Brzozowski's algorithm in~\cite{Bonchi:2012:BAA:2340820.2340823}. 
%%
%Among other equivalences, we studied \emph{failure equivalence} in~\cite{dec-tr-MFPS12}; the characterisation of must given in this paper adds internal behaviour and divergence. Another coalgebraic outlook on must is presented in~\cite{DBLP:journals/tcs/BorealeG06} which introduces a fully abstract coalgebraic semantics for CSP. The main difference with our work consists in the fact that~\cite{DBLP:journals/tcs/BorealeG06} builds a coalgebra from the syntactic terms of CSP, while here we build a coalgebra starting from LTS's via the generalised powerset construction \cite{fsttcs}. Our approach puts in evidence the underlying semilattice structure which is needed for defining bisimulations up-to congruence and \HKC. As a further coalgebraic approach to testing, it is worth mentioning test-suites~\cite{DBLP:journals/entcs/Klin04}, which however do not tackle must testing. 

\textit{Organisation of the chapter.}
We first recall the word automata, the algorithms we will start with,
and their coalgebraic description (Sect.~\ref{CONC-sec:language}). We adapt
these algorithms to failure semantics (Sections~\ref{ssec:hkcFailure},~\ref{app:proofofcorrectness},~\ref{sec:brzozowski-failure},~\ref{sec:proof:Brzozowski}), and then
to must semantics (Sections~\ref{ssec:HKCmust},~\ref{app:proofofcorrectness},~\ref{sec:brzoz-must},~\ref{sec:proof:Brzozowski-must}) for finite machines: although failure
semantics can be seen as a special case of must semantics, the first
generalisation is important for the sake of clarity. We finally give
examples illustrating the relative behaviour of the various algorithms
(Sections~\ref{sec:exp},~\ref{sec:concrete-tests}), before concluding (Section~\ref{sec:concl-concur}).

\section{Language equivalence}
\label{CONC-sec:language}

The core of this chapter is about the problem of checking whether two
states in a finite transition system are behavioural equivalent, for a certain
notion of equivalence. More explicitly, we will reduce the problem of reasoning on failure and must testing semantics, respectively,
to the classical problem of checking language
equivalence.

We proceed by first providing a short overview on deterministic automata (DA's), Moore machines and non-deterministic automata (NDA's), and the problem of recovering language semantics of NDA's, in the coalgebraic setting.

We recall again that a \emph{deterministic automaton} over the input alphabet $A$ is a pair
$(S,<o,t>)$, where $S$ is a set of states and $<o,t> \colon S \to
2\times S^A$ is a function with two components: $o$, the output
function, determines whether a state $x$ is final ($o(x) = 1$) or not
($o(x) = 0$); and $t$, the transition function, returns for each
state and each input letter, the next state.

From any DA, there exists a function $\bb{-}\colon S \to
2^{A^*}$ mapping states to formal languages, defined as follows, for
all $x\in S$:
\begin{equation}\label{def:lang}
 \bb{x}(\varepsilon) = o(x) \hspace{1cm}
 \bb{x}(a\cdot w) = \bb{t(x)(a)}(w)
\end{equation}
The language $\bb{x}$ is called the language accepted by $x$, and it consists of all words $w \in A^{*}$ which, if executed from $x$, lead to a final (or accepting) state. Given an
automaton $(S,<o,t>)$, the states $x,y\in S$ are said to be
\emph{language equivalent} iff they accept they same language.

Throughout this chapter, we will use \emph{Moore machines} which are coalgebras for the functor $F(S)=B\times S^A$. These are very similar to DA's, but with outputs in any (fixed) set $B$. The unique $F$-homomorphism to the final coalgebra $\bb{-}\colon S \to B^{A^*}$  is defined exactly as for DA's by the equations in~\eqref{def:lang}. Note that the behaviours of Moore machines are functions $\varphi \colon A^* \to B$, rather than subsets of $A^*$. For each behaviour $\varphi \in B^{A^*}$, there exists a minimal Moore machine realising it.

A \emph{non-deterministic automaton} is similar to a DA but the
transition function returns a set of next-states 
instead of a single state. Thus, an NDA over the input alphabet $A$
is a pair $(S,<o,t>)$, where $S$ is a set of states and
$<o,t> \colon S \to 2\times (\pow(S))^A$. An example is depicted below
(final states are overlined, labelled edges represent transitions).
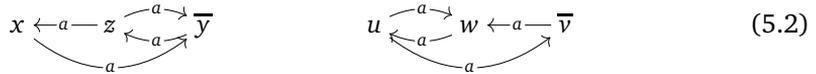
\begin{equation}\label{ex:automata}
  \nfa{\xymatrix @R=.4cm@C=2.3em {%
    \state{x} \ar@/_1.6em/[rr]|{\,a\,} & %
    \state{z} \ar[l]|{\,a\,}\ar@/^/[r]|{\,a\,} & %
    \fstate{y} \ar@/^/[l]|{\,a\,} &&
    \state{u} \ar@/^/[r]|{\,a\,}\ar@/_1.6em/[rr]|{\,a\,}& %
    \state{w} \ar@/^/[l]|{\,a\,} & %
    \fstate{v} \ar[l]|{\,a\,} }}
\end{equation}

Classically, in order to recover language semantics of NDA, one uses
the \emph{powerset construction} (see Section~\ref{prelim:gen-pow} for a reminder), transforming every NDA
$(S,<o,t>)$ into the DA $(\pow(S), <o^{\sharp}, t^{\sharp}>)$ where
$o^{\sharp} \colon \pow(S) \to 2$ and $t^{\sharp} \colon \pow(S) \to
\pow(S)^A$ are defined for all $X\in \pow(S)$ as
\begin{align*}
  o^{\sharp} (X) &= \bigsqcup_{x\in X}o(x) &
  t^{\sharp} (X)(a) &= \bigsqcup_{x\in X}t(x)(a)\enspace.
\end{align*}

Note that we use $\sqcup$ to denote the ``Boolean or'' in
$2$, the union of languages in $2^{A^*}$ and the union of sets in
$\pow(S)$.

For instance with the NDA from~\eqref{ex:automata},
$o^{\sharp}(\set{x,y})=0\sqcup1 =1$ (i.e., the state $\set{x,y}$ is
final) and $t^{\sharp}(\set{x,y})(a)=\set{y}\sqcup\set{z} =\set{y,z}$
(i.e., $\set{x,y}\tr{a}\set{y,z}$).

Since $(\pow(S),<o^{\sharp}, t^{\sharp}>)$ is a deterministic
automaton, we can now apply the language semantics above, yielding a
function $\bb{-}\colon \pow(S) \to 2^{A^*}$ mapping \emph{sets} of
states to languages. Given two states $x$ and $y$, we say that they
are language equivalent iff $\bb{\{x\}}=\bb{\{y\}}$. More generally,
for two sets of states $X,Y\subseteq S$, we say that $X$ and $Y$ are
language equivalent iff
$\bb{X}=\bb{Y}$. % In the above example, we have $\bb{\{x\}}=\bb{\{u\}}$.

In order to introduce the algorithms in full generality, it is important to recall here that the sets $2$, $\pow(S)$, $\pow(S)^A$, $2\times\pow(S)^A$ and $2^{A^*}$ carry 
a semilattice with bottom structure $(X,\sqcup, 0)$ and that the functions $<o^{\sharp},t^{\sharp}>\colon \pow(S) \to 2 \times \pow(S)^A$
and $\bb{-}\colon \pow(S) \to 2^{A^*}$ are homomorphisms of semilattices with bottom. In the rest of the chapter we will indiscriminately use $0$ to denote
the element $0\in 2$, the empty language in $2^{A^*}$ and the empty
set in $\pow(S)$. 

\subsection{Language equivalence via bisimulation up-to: \HKC\ }
\label{CONC-sec:via-bisim-up-to}
We recall the algorithm \HKC\ from~\cite{BONCHI:2012:HAL-00639716:4}. 
We first provide a notion of bisimulation on sets of states, underlying the notion of progression. Note that this is equivalent to the bisimulation introduced in Section~\ref{prelim:coalg}, but more appropriate for the proofs in this chapter.

\begin{definition}[Progression, Bisimulation]\label{def:bisimulation}
  Given two relations $R,R'\subseteq \pow(S) \times \pow(S)$, $R$
  \emph{progresses to} $R'$, denoted $R\prog R'$, if whenever
  $X\mathrel R Y$ then
  \vspace{-.2em}
  \begin{center} 1.\ $o^{\sharp}(X) = o^{\sharp}(Y)$ \quad and
    \quad 2.\ for all $a\in A$, $t^{\sharp}(X)(a) \mathrel{R'}
    t^{\sharp}(Y)(a)$. 
  \end{center}
  A \emph{bisimulation} is a relation $R$ such that $R\prog R$.
\end{definition}
This definition considers the states, the transitions and the outputs
of the \emph{determinised} NDA. For this reason, the bisimulation
proof technique is sound and complete for language equivalence.
% rather than for the standard notion of bisimilarity by Milner and Park~\cite{milner89}.
%
\newcommand\bisimulationlangequivalence{%
  For all $X,Y\in \pow(S)$, $\bb{X}=\bb{Y}$ iff there exists a bisimulation
  that relates $X$ and $Y$.}

Consequently, the coinduction proof principle is stated as follows.  

\begin{proposition}[Coinduction]
  \label{prop:bisimulation-langequivalence}
  \bisimulationlangequivalence
\end{proposition}
For an example, suppose that we want to prove the equivalence of $\set{x}$ and $\set{u}$ of the NDA  
in \eqref{ex:automata}. The part of the determinised NDA that is reachable from 
$\set{x}$ and $\set{u}$ is depicted below. The relation consisting of dashed and dotted lines is a bisimulation which proves that $\bb{\set{x}}=\bb{\set{u}}$.
\begin{equation}\label{executionNFA}
 \begin{tabular}{c}
$
  \dfa{\xymatrix @C=1em @R=1.7em {%
    \state{\set x}\ar[r]^a\ar@{--}[d]_1 & %
    \fstate{\set y}\ar[r]^a\ar@{--}[d]_2 & %
    \state{\set z}\ar[r]^(.45)a\ar@{--}[d]_3 & %
    \fstate{\set{x,y}}\ar[r]^a\ar@{.}[d] & %
    \fstate{\set{y,z}}\ar[r]^(.45)a\ar@{.}[ld] & %
    \fstate{\set{x,y,z}}\ar@(dr,dl)^a\ar@{.}[lld] \\
    \state{\set u}\ar[r]_(.45)a & %
    \fstate{\set{v,w}}\ar[r]_a & %
    \state{\set{u,w}}\ar[r]_(.45)a & %
    \fstate{\set{u,v,w}}\save !R(.9)\ar@(dr,ur)_a\restore }}
$
\end{tabular}
\end{equation}

The dashed lines (numbered by $1$, $2$, $3$) form a smaller relation
which is not a bisimulation, but a \emph{bisimulation up-to congruence}: the
equivalence of $\set{x,y}$ and $\set{u,v,w}$ can be
immediately deduced from the fact that $\set x$ is related to $\set u$
and $\set y$ to $\set{v,w}$.
In order to formally introduce bisimulations up-to congruence, we need
to define first the \emph{congruence closure} $c(R)$ of a relation
$R\subseteq\pow(S)\times\pow(S)$. This is done inductively, by the
following rules: 
\begin{equation}\label{eq:congclos}
%\hypertarget{rules:c}{
%\begin{mathpar}
{\small
\dfrac{X \mathrel R Y}{X \mathrel{c(R)} Y} \;
\dfrac{ }{X \mathrel{c(R)} X} \;
  \dfrac{X \mathrel{c(R)} Y}{Y \mathrel{c(R)} X} \;
  \dfrac{X \mathrel{c(R)} Y \; Y \mathrel{c(R)} Z}{X \mathrel{c(R)} Z} \;  
  \dfrac{X_1 \mathrel{c(R)} Y_1 \,\,\,\,\, X_2 \mathrel{c(R)}
    Y_2}{X_1 \sqcup X_2 \mathrel{c(R)} Y_1 \sqcup Y_2}
    }
%\end{mathpar}
%}
\end{equation}
Note that the term ``congruence'' here is intended with respect to the semilattice structure carried by the state space $\pow(S)$ of the determinised automaton. Intuitively, $c(R)$ is the smallest equivalence relation containing $R$ and which is
closed w.r.t $\sqcup$.
\begin{definition}[Bisimulation up-to congruence]\label{def:bisimulationuptocongruence}
  A relation $R\subseteq\pow(S)\times \pow(S)$ is a \emph{bisimulation up-to $c$}
  if $R\prog c(R)$, i.e., whenever $X\mathrel R Y$ then
  \vspace{-.2em}
  \begin{center} 1.\ $o^{\sharp}(X) = o^{\sharp}(Y)$ \quad and
    \quad 2.\ for all $a\in A$, $t^{\sharp}(X)(a) \mathrel{c(R)}
    t^{\sharp}(Y)(a)$. 
  \end{center}
\end{definition}
\begin{theorem}[\cite{BONCHI:2012:HAL-00639716:4}]\label{thm:upto:c}
  Any bisimulation up-to $c$ is contained in a bisimulation.
\end{theorem}
Figure~\ref{fig:hkc} shows the corresponding algorithm, parametric on
$o^\sharp$, $t^\sharp$, and $c$. Starting from an NDA $(S,<o,t>)$ and
considering the determinised automaton $(S,<o^\sharp,t^\sharp>)$, it
can be used to check language equivalence of two sets of states $X$
and $Y$. Starting from the pair $(X,Y)$, the algorithm builds a
relation $R$ that, in case of success, is a bisimulation up-to
congruence. In order to do that, it employs the set $todo$ which,
intuitively, at any step of the execution, contains the pairs
$(X',Y')$ that must be checked: if $(X',Y')$ already belongs to
$c(R\cup todo)$, then it does not need to be checked. Otherwise, the
algorithm checks if $X'$ and $Y'$ have the same outputs. If
$o^{\sharp}(X')\neq o^{\sharp}(Y')$ then $X$ and $Y$ are different,
otherwise the algorithm inserts $(X',Y')$ in $R$ and, for all $a\in
A$, the pairs $(t^{\sharp}(X')(a), t^{\sharp}(Y')(a))$ in $todo$. The
check $(X',Y')\in c(R\cup todo)$ at step \texttt{2.2} is done with the
rewriting algorithm of~\cite[Section 3.4]{BONCHI:2012:HAL-00639716:4}.
\vspace{-.1cm}
\begin{proposition}
  \label{prop:naive:correct}
  For all $X,Y\in \pow(S)$, $\bb{X}= \bb{Y}$ iff
  \emph{\hkc X Y}.
\end{proposition}%\vspace{-.4cm}

The iterations corresponding to the execution of {\hkc {\set{x}}
  {\set{u}}} on the NDA in~\eqref{ex:automata} are concisely described
by the numbered dashed lines in \eqref{executionNFA}.
%
%Executing \HKC\ on the NDA in~\eqref{ex:automata} to check that
%$\set{x}$ and $\set{u}$ are language equivalent, we go three times
%through step~\texttt{2.5}, to insert the three dashed pairs into the
%candidate relation. 
Observe that only a small portion of the
determinised automaton is explored; this fact usually makes \HKC\ more
efficient than the algorithms based on minimisation, that need to
build the whole reachable part of the determinised automaton.

\begin{figure}[t]
\begin{tabular}{l@{\quad}p{5cm}}
\underline{\hkc X Y}:\\
$\begin{array}{l}
\texttt{(1) $R$ is empty; $todo$ is $\set{(X,Y)}$;}\\[0.5ex]
\texttt{(2) while $todo$ is not empty, do} \\[0.5ex]
\hspace{10pt}\texttt{(2.1) extract $(X',Y')$ from $todo$;}\\[0.5ex]
\hspace{10pt}\texttt{(2.2) if $(X',Y')\in c(R\cup todo)$ then continue;}\\[0.5ex]
\hspace{10pt}\texttt{(2.3) if $o^{\sharp}(X') \neq o^{\sharp}(Y')$ then return false;}\\[0.5ex]
\hspace{10pt}\texttt{(2.4) for all $a\in A$, insert $(t^{\sharp}(X')(a),\,t^{\sharp}(Y')(a))$ in $todo$;}\\[0.5ex]
\hspace{10pt}\texttt{(2.5) insert $(X',Y')$ in $R$;}\\[0.5ex]
\texttt{(3) return true;}
\end{array}$\\\\
\underline{\texttt{\large{Brzozowski}}}:\\
$\begin{array}{l}
\texttt{(1) reverse and determinise;}\\[0.5ex]
\texttt{(2) take the reachable part;}\\[0.5ex]
\texttt{(3) reverse and determinise;}\\[0.5ex]
\texttt{(4) take the reachable part.}
\end{array}$
\end{tabular}
%\nocaptionrule
\caption{Generic \HKC\ algorithm, parametric on $o^\sharp$, $t^\sharp$ and $c$. Generic Brzozowski's algorithm, parametric on {{\texttt reverse} and {\texttt determinise}}. Instantiation to language/failure/must equivalence.}
\label{fig:hkc}
\end{figure}

\subsection{Language equivalence via Brzozowski's algorithm}
\label{brzo}
%For any language $L\in 2^{A^*}$, there exists a \emph{minimal deterministic automaton} containing a state $x$ such that $\bb{x}=L$. Such automaton is minimal in the number of states and it is unique up to isomorphism.
%
%Therefore, 
The problem of checking language equivalence of two sets of states $X$
and $Y$ of a non-deterministic finite automaton can be reduced to that
of building the minimal DA for $\bb{X}$ and $\bb{Y}$ and checking
whether they are the same (up to isomorphism). The most well-known procedure consists in first determinising the NDA and then minimising it with Hopcroft's algorithm~\cite{minimisation}. Another interesting solution is Brzozowski's algorithm~\cite{Brzozowski}.%, that we explain in this section. %Since the problem of checking language equivalence for NFAs is PSPACE\cite{xxx}, the worst case complexity of both approaches is irrelevant. In \cite{xxx}, it is performed an empirical evaluation of the two approaches, showing that the Brzozowski algorithms performs rather well in practice.

To explain the latter, it is convenient to consider a set of \emph{initial states} $I$.
Given an NDA $(S, <o,t>)$ and a set of states $I$, Brzozowski's algorithm computes the minimal automaton for the language $\bb{I}$ by performing the 4 steps in Figure~\ref{fig:hkc}. 

The operation ``reverse and determinise'' takes as input an NDA $(S, <o,t>)$ and returns a DA 
$(\pow(S),<\overline{o}_{R}, \overline{t}_{R}>)$ where the functions $\overline{o}_{R} \colon \pow(S) \to 2$ and $\overline{t}_{R} \colon \pow(S) \to \pow(S)^A$ are defined for all $X\in \pow(S)$ as
\begin{align*}
\overline{o}_{R}(X)  &= 1 \text{ iff } X\cap I \neq 0 &
\overline{t}_R(X)(a)  &= \set{x \in S \mid t(x)(a) \cap X \neq 0 }
\end{align*}
and the new initial state is the old set of final states:
$\overline{I}_{R} = \set{x \mid o(x)=1}$. The second step consists in
taking the part of $(\pow(S),<\overline{o}_{R}, \overline{t}_{R}>)$
which is reachable from $\overline{I}_{R}$. The third and the fourth
steps perform this procedure once more.

As an example, consider the NDA in \eqref{ex:automata} with the set of initial states $I=\set{x}$.
Brzozowski's algorithm builds the minimal DA accepting $\bb{\set{x}}$ as follows. After the first two steps, it returns the following DA where the initial state is $\set{y}$.
$$
  \dfa{\xymatrix @C=1em @R=1.7em {%
    \state{\set y}\ar[r]^(0.4)a & %
    \fstate{\set{x,z}}\ar[r]^a & %
    \state{\set{z,y}}\ar[r]^(.5)a & %
    \fstate{\set{x,y,z}}\hspace*{-.6cm}& {\ }\ar@(dr,ur)_a }}
$$
After steps \texttt{3} and \texttt{4}, it returns the DA below with initial state $\set{\set{x,z}\set{x,y,z}}$.
$$
{\small
\!\!\dfa{\xymatrix @C=1em @R=1.7em {%
    \state{\set{\set{x,z}\set{x,y,z}  }}\ar[r]^(0.45)a & %
    \fstate{\set{\set{y}\set{z,y}\set{x,y,z}}}\ar[dl]_(0.47)a\\ %
    \state{\set{\set{x,z}\set{z,y}\set{x,y,z}}}\ar[r]^(.51)a & %
    \fstate{\set{\set{y}\set{x,z}\set{z,y}\set{x,y,z}}}\hspace*{-.6cm}& {\ \ }\ar@(dr,ur)_a } }
}$$
Computing the minimal NDA in \eqref{ex:automata} with the set of initial
states $I=\set{u}$ results in an isomorphic automaton, showing the
equivalence of $x$ and $u$.

\section{Algorithms for failure and must testing semantics}
\label{sec:fail}
In this section we show how the algorithms \HKC\ and Brzozowski can be adapted for reasoning on failure and must testing semantics.
Next we briefly summarise the coalgebraic modelling of these semantics via the generalised powerset construction, as introduced in Chapter~\ref{ch:dec-trace-testing}.

An LTS over the %input 
alphabet $A$ is
a pair $(S,t)$ with $t\colon S \to \pow(S)^A$.
For a function $\varphi\in \pow (S)^A$, $I(\varphi)$ denotes the set
of all labels ``enabled'' by $\varphi$, given by $I(\varphi) =
\{a\in A \mid \varphi(a)\neq 0\}$, while $\Fail(\varphi)$ denotes
the set $\set{Z \subseteq A \mid Z \cap I(\varphi)=0}$.
%Let $(X, \delta)$ be a LTS and $x\in X$ be a state. 
A \emph{failure pair} of a state $x\in S$ is a
pair $(w, Z)\in A^*\times \pow (A)$ such that $x\tr{w}y$ and $Z\in
\Fail(t(y))$. The set of failure pairs of $x$ is denoted by
${\Fs}(x)$.
Given two states $x,y\in S$, $x$ is \emph{failure equivalent} to $y$ ($x\sim_{{\Fs}}y$) 
if and only if ${\Fs}(x)= {\Fs}(y)$.

In short, the coalgebraic modelling of failure semantics in Section~\ref{sec:ready-fail} is as follows.
First, the states of an LTS $(S,t)$ are decorated by means of the output function $o\colon S \to \pow(\pow(A))$ defined as 
\begin{equation}
\label{eq:output-fail-2}
o(x)=\Fail(t(x)).
\end{equation} 
Then, the decorated LTS $(S,<o,t>)$ is translated, using the
generalised powerset construction from Section~\ref{prelim:gen-pow}, into a
Moore machine $(\pow(S),<o^{\sharp},t^{\sharp}>)$ with
$o^{\sharp}\colon \pow(S) \to \pow(\pow(A))$ and $t^{\sharp} \colon
\pow(S) \to \pow(S)^A$ defined for all $X\in \pow(S)$ as
\begin{align} 
\label{def:Moore-fail}
  o^{\sharp} (X) &= \bigsqcup_{x\in X} o(x) &
  t^{\sharp} (X)(a) &= \bigsqcup_{x\in X}t(x)(a)
\end{align} 
where, in the left equation, $\sqcup$ denotes the union
of subsets in $\pow(\pow (A))$. Note that
$(\pow(S),<o^{\sharp},t^{\sharp}>)$ is a Moore machine with outputs in
$\pow(\pow(A))$. The map into the final Moore coalgebra $\bb{-}\colon \pow(S) \to
(\pow(\pow(A)))^{A^\star}$ associates to a set of states their
``behaviours''. The latter are in one-to-one correspondence with
failure pairs. More explicitly, for all $x\in S$, $Z\in \pow(\pow(A))$ and $w\in A^*$:
\begin{equation}
\label{lemma:corr}
Z\in \bb{\set{x} }(w) \text{ iff } (w,Z)\in {\Fs}(x).
\end{equation}

Hence,
\begin{equation}
\label{prop:fail-beh-equiv}
 x\sim_{\Fs} y \text{ iff } \bb{\set{x}} = \bb{ \set{y}}.
\end{equation}

The trace-based characterisation of must testing in~\cite{CleavelandH89} leads to a similar coalgebraic representation via the generalised powerset construction. In Section~\ref{sec:fail-must} we modelled LTS's with internal behaviour and divergence as coalgebras $(S, t\colon S \rightarrow (1 + \pow(S))^{A})$ such that, for all $x\in S$ and $a\in A$
\begin{align}
\label{def:div-LTS-trans}
  t(x)(a) =    \top, \textnormal{ if } x \not \downarrow a \qquad  t(x)(a) =
    \set{y \mid x \xRightarrow{a} y},  \text{ otherwise}.
\end{align}
Recall that $\xRightarrow{a}$ denotes the execution of an action $a\in A$, possibly preceded or followed by (any number of) internal steps $\tau$, $\downarrow$ is the convergence predicate, and $1 = \set{\top}$ is used to coalgebraically ``capture'' divergent behaviours.

Then, we decorate such LTS's by means of a function $o\ct S \rightarrow 1 + \pow(\pow A)$ such that, for all $x \in S$ and $a \in A$
\begin{align}
\label{def:div-LTS-dec}
  o(x) &=
  \begin{cases}
    \top& \textnormal{ if } x \not \downarrow\\
    \bigcup_{x\xrightarrow{\tau}x'} o(x')& \textnormal{ if } x
    \xrightarrow{\tau}\\
    \Fail(t(x)) & \text{ otherwise}.
  \end{cases} 
\end{align}

Finally, we apply the generalised powerset construction and derive a Moore machine $(1+\pow(S), \langle o^{\sharp}, t^{\sharp}\rangle)$ defined for all $x \in 1 + \pow(S)$ and $a \in A$ as
\begin{align}
\label{def:Moore-must}
  o^{\sharp} (X) &=
  \begin{cases}
    %o(x) & \text{ if } X = \set{x}\text{ with } x\in S\\
    \top & \text{ if } X = \top\\
    \bigsqcup_{x\in X} o(x) & \text{ if } X \in \pow(S)\\
%    0 & \text{ if } X= 0\\
%    o^{\sharp}(X_1) \sqcup o^{\sharp}(X_2) & \text{ if } X= X_1 \sqcup X_2\\
  \end{cases}&
  %\hspace{8pt}
  t^{\sharp} (X)(a) &= 
  \begin{cases}
    %t(x)(a) & \text{ if } X= \set x \text{ with } x\in S\\
    \top & \text{ if } X = \top\\
    \bigsqcup_{x\in X} t(x)(a) & \text{ if } X \in \pow(S)\\
%    0 & \text{ if } X= 0\\
%    t^{\sharp}(X_1)(a) \sqcup t^{\sharp}(X_2)(a) & \text{ if } X= X_1 \sqcup X_2.\\
  \end{cases}
\end{align}

The state space $(1+\pow(\pow(A))^{A^{*}}$ of the final Moore coalgebra carries the structure of a join semilattice with top, inducing a partial order $\sqsubseteq_{\mathcal M}$. This, together with the behaviour map $\bb{-}\ct 1 + \pow(S) \rightarrow (1+\pow(\pow(A))^{A^{*}}$ further enabled formalising must testing preorder and equivalence, respectively, as follows:
\begin{equation}
\label{eq:formalis-must}
\begin{array}{l}
x \mst y \text{ iff } \bb{\set{y}} \sqsubseteq_{\mathcal M} \bb{\set{x}}\\
x \sim_{mst} y  \text{ iff } \bb{\set{x}} = \bb{\set{y}}.
\end{array}
\end{equation}

\subsection{\HKC\ for failure semantics}\label{ssec:hkcFailure}
The algorithm \HKC\ in Figure~\ref{fig:hkc} can be used to check failure equivalence on an LTS $(S,t)$ by taking $o^{\sharp}$ and $t^{\sharp}$ as defined in~(\ref{def:Moore-fail}). Then, the congruence closure $c$ is defined as for
language equivalence in~(\ref{eq:congclos}). The analogue of
Proposition~\ref{prop:naive:correct} can be proved in exactly the same
way (check Section~\ref{app:proofofcorrectness}): in particular, soundness of bisimulation up-to-congruence is guaranteed from the fact that $(\pow(S),<o^{\sharp},t^{\sharp}>)$ is a bialgebra.

We provide an example of using bisimulation up-to congruence for reasoning on failure semantics.
Consider the following systems, where $n$ is an arbitrary natural number:
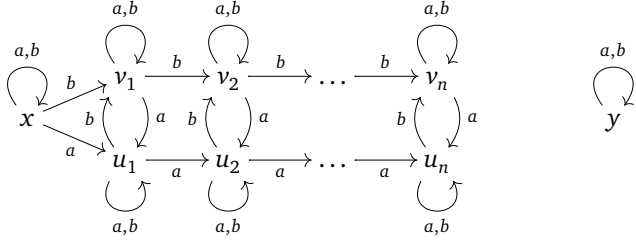
\[
\xymatrix@C=.8cm@R=.1cm{
& v_{1}\ar@(ul,ur)^{a,b}\ar[r]^{b}\ar@/^.7pc/[dd]^{a} & v_{2}\ar[r]^{b}\ar@(ul,ur)^{a,b}\ar@/^.7pc/[dd]^{a} & \ldots\ar[r]^{b} & v_{n}\ar@(ul,ur)^{a,b}\ar@/^.7pc/[dd]^{a} & &&\\
x\ar@(ul,ur)^{a,b}\ar[ur]^{b}\ar[dr]_{a} & &  & & & & y\ar@(ul,ur)^{a,b}\\
& u_{1}\ar@(dl,dr)_{a,b}\ar[r]_{a}\ar@/^.7pc/[uu]^{b} & u_{2}\ar@(dl,dr)_{a,b}\ar[r]_{a}\ar@/^.7pc/[uu]^{b} & \ldots\ar[r]_{a} & u_{n}\ar@(dl,dr)_{a,b}\ar@/^.7pc/[uu]^{b} & &&
}
\]
It is easy to see that $x$ and $y$ are bisimilar: intuitively, all the states of the automata depicted above can trigger actions $a$ and $b$ as a first step and, moreover, all their subsequent transitions lead to states with the same behaviour. Therefore $x$ and $y$ are also $\Fs$-equivalent, according to van Glabbeek's lattice of semantic equivalences~\cite{Glabbeek01} (partially) illustrated in Figure~\ref{fig:lattice} in Chapter~\ref{ch:dec-trace-testing}. 

The coalgebraic machinery provides a ``yes'' answer with respect to\linebreak $\Fs$-equivalence of the two LTS's as well. After determinisation, $\{x\}$ can reach all states of shape: $\{x\} \cup \overline{u}_{i}$, $\{x\} \cup \overline{v}_{i}$, $\{x\} \cup \overline{u}_{i} \cup \overline{v}_{i}$, for $i \in \{1,\ldots, n\}$ and $\{x\} \cup \overline{u}_{j} \cup \{v_{1}\}$, $\{x\} \cup \overline{v}_{j} \cup \{u_{1}\}$, respectively, for $j \in \{2,\ldots, n\}$. (We write, for example, $\overline{u}_{i}$ in order to represent the set $\{u_{1}, u_{2}, \ldots, u_{i}\}$.)

Consequently, the generalised powerset construction associates to $x$ a Moore automaton consisting of $5n -1$ states, whereas the determinisation of $y$ has only one state.
Hence, the (Moore) bisimulation relation $R$ including $(\{x\}, \{y\})$ consists of $5n-1$ pairs as follows:
\begin{equation}
\label{eq:bisim-up-to-LTS}
\begin{array}{rcl}
R & = & \{(\{x\}, \{y\})\} \cup\\
&& \{(\{x\} \cup \overline{u}_{i} \cup \{v_{1}\}, \{y\}),\, (\{x\} \cup \overline{v}_{i} \cup \{u_{1}\}, \{y\}) \mid i \in \{2, \ldots, n\}\}\, \cup \\
&& \{(\{x\} \cup \overline{u}_{i}, \{y\}),\, (\{x\} \cup \overline{v}_{i}, \{y\}),\, (\{x\} \cup \overline{u}_{i} \cup \overline{v}_{i}, \{y\}) \mid i \in \{1, \ldots, n\}\}.
\end{array}
\end{equation}
%
%\begin{eqnarray}
%\label{eq:bisim-up-to-LTS}\nonumber
%R = \{(\{x\}, \{y\})\} \cup\{(\{x\} \cup \overline{u}_{i} \cup \{v_{1}\}, \{y\}),\, (\{x\} \cup \overline{v}_{i} \cup \{u_{1}\}, \{y\}) \mid i \in \{2, \ldots, n\}\}\, \cup \\
%\phantom{R=} \{(\{x\} \cup \overline{u}_{i}, \{y\}),\, (\{x\} \cup \overline{v}_{i}, \{y\}),\, (\{x\} \cup \overline{u}_{i} \cup \overline{v}_{i}, \{y\}) \mid i \in \{1, \ldots, n\}\}.\phantom{\{(\{x\}, \{y\}\;}
%\end{eqnarray}
For a better intuition, we illustrate bellow the determinisations starting from $x$ and $y$, for the case $n=3$:
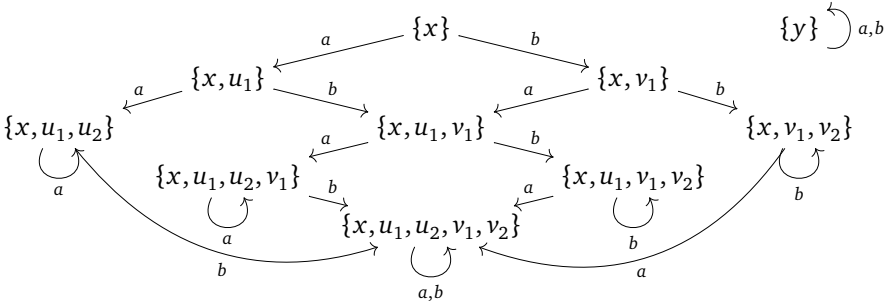
\[
\xymatrix@C=.3cm@R=.1cm{
& & \{x\}\ar[dl]_{a}\ar[dr]^{b} & &\{y\}\ar@(dr,ur)_{a,b}  \\
& \{x, u_{1}\}\ar[dl]_{a}\ar[dr]^{b} & & \{x, v_{1}\}\ar[dl]_{a}\ar[dr]^{b} &\\
\{x,u_{1},u_{2}\}\ar@(dl,dr)_{a}\ar@/^-2.5pc/[ddrr]_{b} & & \{x,u_{1},v_{1}\}\ar[dl]_{a}\ar[dr]^{b} & & \{x,v_{1},v_{2}\}\ar@(dl,dr)_{b}\ar@/^2.7pc/[ddll]^{a}\\
& \{x,u_{1},u_{2},v_{1}\}\ar@(dl,dr)_{a}\ar[dr]^{b} & & \{x,u_{1},v_{1},v_{2}\}\ar@(dl,dr)_{b}\ar[dl]_{a} &\\
&  & \{x, u_{1},u_{2},v_{1},v_{2}\}\ar@(dl,dr)_{a,b} & &
}
\]

It is easy to see that the bisimulation relating $\{x\}$ and $\{y\}$ consists of all pairs $(X, \{y\})$, with $X$ ranging over the state space of the Moore automaton derived according to the generalised powerset construction, starting with $\{x\}$.

Observe that all the pairs in $R$ in~(\ref{eq:bisim-up-to-LTS}) can be ``generated'' from $(\{x\}, \{y\})$, $(\{x\} \cup \overline{u}_{i}, \{y\})$ and $(\{x\} \cup \overline{v}_{i}, \{y\})$ by iteratively applying the rules in~(\ref{eq:congclos}). Therefore, for an arbitrary natural number $n$, the bisimulation up-to congruence stating the equivalence of $x$ and $y$ is:
\[
R^{c}  = \{(\{x\}, \{y\})\}\, \cup \{(\{x\} \cup \overline{u}_{i}, \{y\}),\, (\{x\} \cup \overline{v}_{i}, \{y\}) \mid i \in \{1, \ldots, n\}\}
\]
and consists of only $2n+1$ pairs. The latter represent exactly the states explored by {\HKC}.

\subsection{\HKC\ for must semantics}\label{ssec:HKCmust}
The coalgebraic characterisation of must testing
guarantees soundness and completeness of bisimulation up-to congruence
for the associated equivalence. Bisimulations are now relations $R\subseteq (1
+\pow(S))\times (1+\pow(S))$ on the state space $1+\pow(S)$ where
$o^{\sharp}$ and $t^{\sharp}$ are defined as in
(\ref{def:Moore-must}). Now, the congruence closure $c(R)$ of
a relation $R\subseteq (1 +\pow(S))\times (1+\pow(S))$ is defined by
the rules in~\eqref{eq:congclos} where $\sqcup$ is the join in
$(1+\pow(S))$ (rather than the union in $\pow(S)$). By simply
redefining $o^{\sharp}$, $t^{\sharp}$ and $c(R)$, the algorithm in
Figure~\ref{fig:hkc} can be used to check must equivalence and
preorder (the detailed proof is in Section~\ref{app:proofofcorrectness}).
%In particular, note that like
%for failure, the check at step \texttt{2.1} can be done with the same
%algorithm as in~\cite[Section 3.4]{BONCHI:2012:HAL-00639716:4}.

Consider, for an example, the LTS's in Section~\ref{sec:fail-must}:
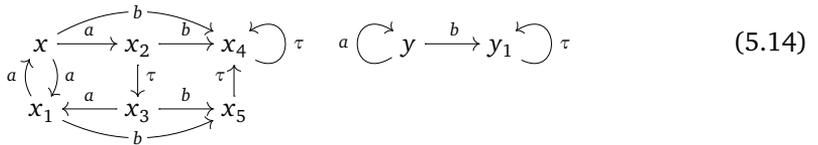
\begin{equation}
\label{eq:LTSMUST23}
\xymatrix@C=.7cm@R=.4cm{
x \ar@/^1pc/[rr]|{\;b\;} \ar[r]^{a} \ar@/^/[d]^{a} & x_2 \ar[d]^{\tau} \ar[r]^b & x_4  \ar@(dr,ur)_{\tau}\\
x_1 \ar@/^/[u]^a \ar@/_.9pc/[rr]|{\;b\;}  & x_3 \ar[l]_a \ar[r]^b & x_5 \ar[u]^{\tau} \\
}\;\;\;
\xymatrix@C=.7cm@R=.4cm{
y\ar@(dl,ul)^{a}\ar[r]^{b} & y_{1}\ar@(dr,ur)_{\tau}
}
\end{equation}

In Section~\ref{sec:fail-must} we showed that the states $x$ and $y$ are must equivalent, by identifying a bisimulation relating $\set{x}$ and $\set{y}$.
This time however, we depict by the dashed lines in~(\ref{eq:DETMUST2}) the relation
$R=\set{(\set{x},\set{y}), (\set{x_{1}, x_{2},x_{3}}, \set{y})}$
which is not a
bisimulation, but a bisimulation up-to congruence, since both
$(\top,\top)\in c(R)$ and $(\set{x,x_1},\set{y})\in c(R)$. For the
latter, observe that
\[\set{x,x_1} \; c(R) \; \set{y,x_1} \; c(R) \; \set{x_1,x_2,x_3} \; c(R) \; \set{y}.\]
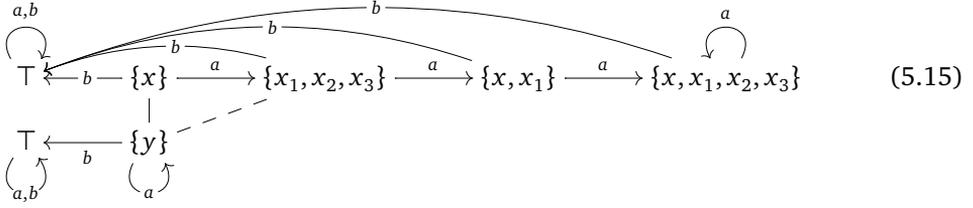
\begin{equation}\label{eq:DETMUST2}
{\xymatrix@C=1cm@R=.3cm{
{\top} \ar@(ul,ur)^-{a,b} & {\set{x}} \ar@{--}[d] \ar[l]|{\;b\;} \ar[r]^-{a} & {\set{x_{1},x_{2},x_{3}}} \ar@{--}[dl] \ar@/_{1.0pc}/[ll]|{\;b\;}\ar[r]^-{a} & {\set{x,x_1}}  \ar@/_{1.6pc}/[lll]|{\;b\;} \ar[r]^-{a} & {\{x,x_{1},x_{2},x_{3}\}}  \ar@/_{2.2pc}/[llll]|{\;b\;}\ar@(ur,ul)_{a} \\
{\top} \ar@(dl,dr)|-{a,b} & {\set{y}} \ar@(dl,dr)|{\;a\;}\ar[l]^{b}
}}
\end{equation}
It is important to remark here that \HKC\ computes this relation without the need of exploring all the reachable part of the Moore machine $(1+\pow(S),<o^{\sharp},t^{\sharp}>)$. So, amongst all the states in~\eqref{eq:DETMUST2}, \HKC\ only explores $\set{x}$, $\set{y}$ and $\set{x_{1}, x_{2},x_{3}}$.

\subsection{Correctness of \HKC\ }
\label{app:proofofcorrectness}
We provide a uniform proof of correctness of \HKC\ in Figure~\ref{fig:hkc} for language, failure and must semantics (Proposition~\ref{prop:naive:correct}).
The key step is (the analogue) of Theorem~\ref{thm:upto:c} stating that bisimulation up-to congruence is a sound proof technique. This holds for any bialgebra (see e.g. Corollary 6.6 in~\cite{DBLP:conf/sofsem/RotBR13}) and, in particular, for $(\pow(S),<o^{\sharp},t^{\sharp}>)$ (or $(1+\pow(S),<o^{\sharp},t^{\sharp}>)$) which is guaranteed to be a bialgebra by the generalised powerset construction (we refer the interested reader to~\cite{DBLP:journals/tcs/Klin11} for a nice introduction on this topic). 

  We first observe that if \hkc X Y returns \texttt{true}
  then the relation $R$ that is built before arriving to step
  \texttt{3} is a bisimulation up-to congruence. Indeed, the following proposition is
  an invariant for the loop corresponding to step~\texttt{2}:
  \begin{align*}
    R\prog c(R\cup todo)
  \end{align*}
  This invariant is preserved since at any iteration of the algorithm,
  a pair $(X',Y')$ is removed from $todo$ and inserted in $R$ after
  checking that $o^{\sharp}(X')=o^{\sharp}(Y')$ and adding $(t^{\sharp}(X')(a),t^{\sharp}(Y')(a))$ for all
  $a\in A$ in $todo$.  Since $todo$ is empty at the end of the loop,
  we eventually have $R\prog c(R)$, i.e., $R$ is a bisimulation up-to congruence.  %By   Theorem~\ref{thm:upto:c} and Proposition~\ref{prop:bisimulation-langequivalence}, $\bb{X}=\bb{Y}$.

  We now prove that if \hkc X Y returns \texttt{false}, then
  $\bb{X}\neq \bb{Y}$. Note that for all $(X',Y')$ inserted in $todo$,
  there exists a word $w\in A^\star$ such that, in the determinised NDA, $X \tr{w}X'$ and $Y \tr{w}Y'$. 
  Since $o^{\sharp}(X')\neq o^{\sharp}(Y')$, then $\bb{X}(w)=o^{\sharp}(X')\neq o^{\sharp}(Y')
  = \bb{Y}(w)$.

\subsection{Brzozowski's algorithm for failure semantics}
\label{sec:brzozowski-failure}
A variation of Brzozowski's algorithm for Moore machines is given
in~\cite{Bonchi:2012:BAA:2340820.2340823}. We could apply such
algorithm to the Moore machine $(\pow(S),<o^{\sharp},t^{\sharp}>)$
which is induced by a decorated LTS $(S,<o,t>)$, with $o$ defined as in~(\ref{eq:output-fail-2}). However, we propose a
more efficient variation that skips the first determinisation from
$(S,<o,t>)$ to $(\pow(S),<o^{\sharp},t^{\sharp}>)$.

The novel algorithm consists of the four steps described in Section~\ref{brzo}, where the procedure ``reverse and determinise'' is modified as follows: $(S,<o,t>)$ with initial state $I$ is transformed into
$(\pow(\pow(A))^S, \overline{o}_R, \overline{t}_R)$ 
where 
\[\overline{o}_{R} \colon \pow(\pow(A))^S \to \pow(\pow(A))\] and \[\overline{t}_{R} \colon \pow(\pow(A))^S \to (\pow(\pow(A))^S)^A\] are defined for all functions $\psi\in \pow(\pow(A))^S$ as
\begin{align}\label{eq:BrzoFailMoore}
  \overline{o}_{R}(\psi)  &=  \bigsqcup_{x\in I}\psi(x) &
  \overline{t}_R(\psi)(a)(x)  &=  \bigsqcup_{y \in t(x)(a)} \psi(y)
\end{align}
and the new initial state is 
$\overline{I}_{R}  = o$. 

Note that the result of this procedure is a Moore machine. Brzozowski's algorithm in Figure~\ref{fig:hkc} transforms an NDA $(S,<o,t>)$ with initial state $I$ into the minimal DA for $\bb{I}$. Analogously, our novel algorithm transforms an LTS into the minimal Moore machine for $\bb{I}$. 

Let us illustrate the minimisation procedure by means of an example.
Consider the LTS $(S, t)$ on the alphabet $A=\set{a,b,c}$ depicted
below. 
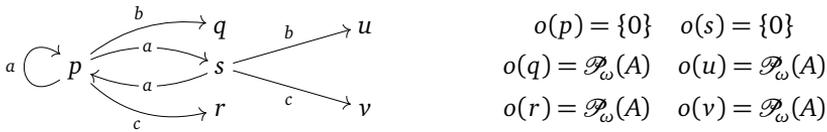
\[
%\begin{array}{c@{\qquad}c@{\quad}c}
 \xymatrix@C=1.5cm@R=.00cm{
  & q & u & o(p) =  \set{0}\quad o(s)  = \set{0} \\
  p\ar@(dl,ul)^{a}\ar@/^.6pc/[ur]^{b}\ar@/^-.8pc/[dr]_{c}\ar@/^.6pc/[r]|{\;a\;} & s\ar@/^.6pc/[l]|{\;a\;}\ar@{->}[ur]^{b}\ar@{->}[dr]_{c} &&o(q) =  \pow (A) \quad o(u)  =  \pow (A)\\
 & r & v & o(r)  =  \pow (A)\quad o(v)  =  \pow (A)
}
%& 
%\begin{array}{rcl}
%o(p) & = & \set{0}\\
%o(q) &= & \pow (A)\\
%o(r)  &= & \pow (A)
%\end{array}&
%\begin{array}{rcl}
%o(s)  &= & \set{0}\\
%o(u)  &= & \pow (A)\\
%o(v)  &=&  \pow (A)
%\end{array}
%\end{array}
\]
The function $o \colon S \to \pow (\pow (A))$ assigning to each state
$x$ the set $\Fail(t(x))$ is given on the right.
%$$o(p)  =  \set{0}\;
%o(q) =  \pow (A)\;
%o(r)  =  \pow (A)\;
%o(s)  =  \set{0}\;
%o(u)  =  \pow (A)\;
%o(v)  =  \pow (A)$$
Suppose we want to build the minimal Moore machine for the behaviour
$\bb{\set{p}}$, i.e., the set of failure pairs of $p$:
$${\Fs}(p) = \set{
 (a^{*}, \set{0}),
 \,(a^{*}b, \pow(A)),
 \,(a^{*}c, \pow(A))}.$$
 By applying our algorithm to the decorated LTS $(S,<o,t>)$, we first
 obtain the intermediate Moore machine on the left below,
 where a double arrow $\psi \Rightarrow Z$ means that the output of
 $\psi$ is the set $Z$.
 The new initial state is $\psi_1 \colon S \rightarrow \pow(\pow A)$
 which, by definition, is the output function of the original LTS
 mapping $p, s$ to $\set{0}$ and $q, r, u$ and $v$ to $\pow(A)$. The
 explicit definitions of the other functions $\psi_i$ can be easily
 computed according to the definition of $\overline{t}_R$~\eqref{eq:BrzoFailMoore}.
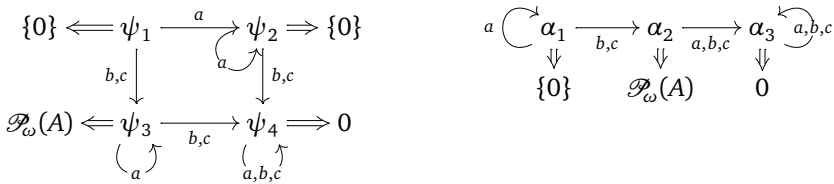
\begin{align*}
  \xymatrix@C=.4cm@R=.7cm{
    \set{0}& \psi_{1}\ar@{->}[rr]^{a}\ar@{->}[d]_{b, c}\ar@{=>}[l] && \psi_{2}\ar@{->}[d]^{b, c}\ar@<-.3ex>@(l,d)|{\;a\;}\ar@{=>}[r] &\set{0}\\
    \pow(A) & \psi_{3}\ar@(dl,dr)|{\;a\;}\ar@{->}[rr]_{b,
      c}\ar@{=>}[l] && \psi_{4}\ar@(dl,dr)|{a,b,c}\ar@{=>}[r] & 0 }
  &&
  \xymatrix@C=.5cm@R=.3cm{
    & \alpha_{1}\ar@(dl,ul)^{a}\ar@{->}[r]_{b,c}\ar@{=>}[d] & \alpha_{2}\ar@{->}[r]_-{a,b,c}\ar@{=>}[d] & \alpha_{3}\ar@(dr,ur)|{a,b,c}\ar@{=>}[d]\\
    & \set{0} & \pow(A) & 0 } 
\end{align*}
Observe that $\bb{\psi_1}$ is the ``reverse'' of $\bb{\set{p}}$. For
instance, triggering a sequence in the language denoted by $ba^{*}$ from $\psi_1$ leads to
$\psi_{3}$ with output $\pow(A)$; this is the same output we get by
executing $a^{*}b$ from $p$, according to ${
\Fs}(p)$.
Executing ``reverse and determinise'' once more (step \texttt{3}) and
taking the reachable part (step \texttt{4}), we obtain the minimal
Moore machine depicted on the right, with initial state $\alpha_1$.

The correctness of this algorithm is established in
Section~\ref{sec:proof:Brzozowski}; it builds on the coalgebraic
perspective on Brzozowski's algorithm given
in~\cite{Bonchi:2012:BAA:2340820.2340823}.

\subsection{Correctness of Brzozowski for failure semantics}\label{sec:proof:Brzozowski}
The main intuition behind Brzozowski algorithm is that the procedure \texttt{reverse and determinise} transforms a system into one having the ``reversed'' behaviour. Moreover, if the original system is \emph{reachable} (that is all the states are reachable from the initial state), then the resulting system is \emph{observable} (that is, all the states have different behaviours). Therefore, after performing the first two steps of Brzozowski's algorithm, one obtains a system which is reachable and that has a ``reversed'' behaviour. After the third step, the system has the original behaviour and, moreover, it is observable. After the fourth step, it is observable and reachable, that is, it is \emph{minimal}.

There are two key steps for our proof: (a) showing that the procedure \texttt{reverse and determinise} introduced in Section \ref{sec:brzozowski-failure} transforms a decorated LTS into a Moore machine having ``reversed'' behaviour; (b) showing that at the third step the algorithm  transforms a reachable Moore machine into an observable one.

Point (b) follows immediately from Section 5 in \cite{Bonchi:2012:BAA:2340820.2340823}, where a variation of Brzozowski's algorithm for Moore machines is introduced: 
when restricted to Moore machines, the operations of reversing and determinising in the our algorithm coincide with those in \cite{Bonchi:2012:BAA:2340820.2340823}. % algorithms coincide. %Since after the first step of the algorithm, the decorated LTS is transformed into a Moore machine, at the third step our algorithm transforms a reachable Moore machine into an observable one.

In the sequel, we prove (a) by relying on \cite{Bonchi:2012:BAA:2340820.2340823}.
Let $(S,t)$ be an LTS with the initial set of initial states $i\colon 1 \to \pow(S)$ (we prefer to use this functional notation, rather than $I\in \pow(S)$, because it is more convenient for the proof).
Let $(\pow(S),<o^{\sharp},t^{\sharp}>)$ be the corresponding Moore machine (as defined in~(\ref{def:Moore-fail})) and 
let $\bb{-}\colon \pow(S) \to (\pow(\pow(A)))^{A^\star}$ be the induced semantics map.

\medskip

By reversing and determinising as in \cite{Bonchi:2012:BAA:2340820.2340823}, we obtain the Moore machine
\[(\pow(\pow(A))^{\pow S}, <o_R^{\sharp} t^{\sharp}_R>)\]
with initial states $i_R$, defined as
\begin{equation}
\label{eq:def-reversed-Moore}
i_{R} =  o^{\sharp} \qquad \qquad
o_{R}^{\sharp}(\varphi)  =  \varphi \circ i \qquad\qquad
t_{R}^{\sharp}(\varphi)(a)(X)  =  \varphi(t^{\sharp}(X)(a)). 
\end{equation}
According to \cite{Bonchi:2012:BAA:2340820.2340823}, we know that this machine has ``reversed'' behaviour, \textit{i.e}, 
\begin{equation}
\label{eq:acc-rev-formalise}
(\forall w \in A^{*})\,.\, \bb{i_{R}}_{1}(w) = \bb{i}(w^{R})
\end{equation}
where 
\[\bb{-}_{1}\colon \pow(\pow(A))^{\pow S} \to (\pow(\pow(A)))^{A^\star}\] is the semantic map, and $w^{R}$ denotes the reverse word $w$ inductively defined as $\varepsilon^{R}  =  \varepsilon$ and 
$(aw')^{R}  =  w'^{R}a$. 

Our algorithm performs the determinisation and the ``reverse and determinise'' at once.
For a Moore machine defined as in~(\ref{eq:BrzoFailMoore},)
the map to the final coalgebra 
\[\bb{-}_{2}\colon \pow(\pow(A))^S \to (\pow(\pow(A)))^{A^\star}\]  satisfies the following lemma.

\begin{lemma}
\label{lm:lemma-psi-phi}
Let $\psi \in (\pow(\pow A))^{S}$ and $\varphi \in (\pow(\pow A))^{\pow S}$ be such that, for all $X \in \pow S$ 
\[
\varphi(X) = \sqcup_{x \in X} \psi(x).~~~~~~(\star)
\]
Then, $\bb{\psi}_{2} = \bb{\varphi}_{1}$.
\end{lemma}
\begin{proof}
 The proof is by induction on $w \in A^{*}$. For the base case, $w = \varepsilon$, we have:
$$
\bb{\psi}_{2}(\varepsilon)      =  \overline{o}_{R}(\psi)  
= \bigsqcup_{x\in i} \psi (x)  \stackrel{(\star)}=  \varphi (i) = o^{\sharp}_{R}(\varphi)    =  \bb{\varphi}_{1}(\varepsilon) 
$$

For the inductive step, consider $w \in A^{*}$ and assume that $\bb{\psi}_{2}(w) = \bb{\varphi}_{1}(w)$ 
holds for all $\psi, \varphi$ satisfying $(\star)$.

We want to prove that $\bb{\psi}_{2}(aw) = \bb{\varphi}_{1}(aw)$ holds for $a \in A$. We first define $
\overline{\varphi}_{a}(X)  =  \varphi(t^{\sharp}(X)(a))$ and $
\overline{\psi}_{a}(x)  =  \sqcup_{y \in t(x)(a)} \psi(y)$, 
where $X \in \pow S$ and $x \in S$ (which, as an intuition, will further be used when applying the induction hypothesis in our proof).

Note that $(\star)$ is satisfied by $\overline{\varphi}_{a}$ and $\overline{\psi}_{a}$:
$\overline{\varphi}_{a}(X) = \bigsqcup\limits_{x \in X} \overline{\psi}_{a}(x)$, 
because
$$
 \overline{\varphi}_{a}(X)
= \varphi(t^{\sharp}(X)(a))
= \varphi(\bigsqcup_{x \in X}t(x)(a))t
\stackrel{(\star)}= \bigsqcup_{x \in X} \bigsqcup_{y \in t(x)(a)} \psi(y)
= \bigsqcup_{x \in X} \overline{\psi}_{a}(x).
$$
At this point it is easy to see that $\bb{\psi}_{2}(aw) = \bb{\varphi}_{1}(aw)$:
\begin{align*}
 \bb{\varphi}_{1}(aw)
=~& \bb{\lambda X \,.\, \varphi(t^{\sharp}(X)(a))}_{1}(w) &\text{(by definition of $t_{R}^{\sharp}$)}\\
=~& \bb{\lambda X \,.\, \overline{\varphi}_{a}(X)}_{1}(w)\\
=~& \bb{\lambda x\,.\, \overline{\psi}_{a}(x)}_{2}(w) &\text{(by the induction hypothesis)}\\
=~& \bb{\lambda x\,.\, \sqcup_{y \in t(x)(a)} \psi(y)}_{2}(w)\\
=~& \bb{\overline{t}_{R}(\psi)(a)}_{2}(w) = \bb{\psi}_{2}(aw).  &&\ \  \qed
\end{align*}
\end{proof}

In particular, if we take $\psi= \overline{I}_{R}$ and $\varphi  =  o^{\sharp}$, we have that 
$\bb{\overline{I}_{R}}_{2} = \bb{o^{\sharp}}_{1}$.
By~(\ref{eq:acc-rev-formalise}) and the fact that and $i_{R} = o^{\sharp}$ the following holds:
\[
(\forall w \in A^{*})\,.\,\bb{o^{\sharp}}_{1}(w) = \bb{i}(w^{R}).
\]
summarising, for all $w\in A^*$, $\bb{\overline{I}_{R}}_{2}(w) = \bb{i}(w^R)$.

For an example of this fact, observe that $p$ and $\psi_1$ in Section \ref{sec:brzozowski-failure} have reversed behaviours.

\subsection{Brzozowski's algorithm for must semantics}
\label{sec:brzoz-must}
The Brzozowski algorithm introduced in Section~\ref{sec:brzozowski-failure} for failure equivalence can be used also for checking $\sim_{mst}$ and $\mst$. Now, the procedure ``reverse and determinise'' returns the Moore machine $((1+\pow(\pow(A)))^S, \overline{o}_R, \overline{t}_R)$. The initial state $\overline{I}_{R}$, the outputs 
\[\overline{o}_{R} \colon (1+\pow(\pow(A)))^S \to 1+\pow(\pow(A))\]  and the transitions 
\[\overline{t}_{R} \colon 
(1+\pow(\pow(A)))^S \to ((1+\pow(\pow(A)))^S)^A\] are defined as in~\eqref{eq:BrzoFailMoore}, plus the case 
\[\overline{t}_R(\psi)(a)(x)  =  \top \text{ if } t(x)(a) = \top,\]
by replacing $o$ and $t$ with those defined in~(\ref{def:div-LTS-dec}) and~(\ref{def:div-LTS-trans}), and by considering the join operation $\sqcup$ in $1+\pow(\pow (A))$ (rather than in $\pow(\pow (A))$).
      
In what follows, we illustrate Brzozowski's algorithm for must testing, by means of an example. Consider the divergent LTS $(S,t)$ below:
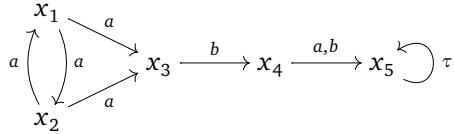
\[
\xymatrix@C=.9cm@R=.26cm{
x_{1}\ar@/^.6pc/[dd]^{a}\ar[dr]^{a} & & & &\\
& x_{3}\ar[r]^{b} & x_{4}\ar[r]^{a,b} & x_{5}\ar@(dr,ur)_{\tau}\\
x_{2}\ar@/^.6pc/[uu]^{a}\ar[ur]_{a} & & & &\\
}
\]
and $o \colon S \rightarrow 1 + \pow(\pow(A))$ the decoration function
\[
\begin{array}{rclcl}
o(x_{1}) &=& o(x_{2}) &=& \set{b}\\
o(x_{3}) &=& \set{a} &&\\
o(x_{4}) &=& \set{0} &&\\
o(x_{5}) &=& \top. &&
\end{array}
\]

Assume we want to build the minimal Moore machine for the behaviour of $x_{1}$, which is must testing equivalent with $x_{2}$. By applying our algorithm to the decorated LTS $(S, \langle o,t\rangle)$ we obtain the following intermediate Moore machine:

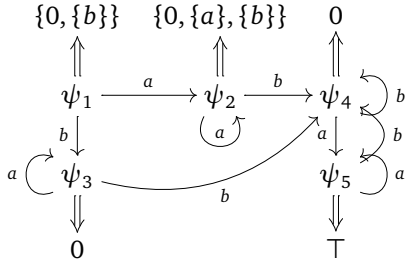
\[
\xymatrix@C=.2cm@R=.5cm{
\set{0, \set{b}} & \set{0, \set{a}, \set{b}} & 0\\
\psi_{1}\ar@{=>}[u]\ar[r]^{a}\ar[d]_{b} & \psi_{2}\ar@{=>}[u]\ar[r]^{b}\ar@(dl,dr)^{a} & \psi_{4}\ar@(dr,ur)_{b}\ar@{=>}[u]\ar[d]_{a} \\
\psi_{3}\ar@{=>}[d]\ar@/_1.5pc/[urr]_{b}\ar@(dl,ul)^{a} &  & \psi_{5}\ar@/_1.5pc/[u]_{b}\ar@(dr,ur)_{a}\ar@{=>}[d]\\
0 &  & \top
}
\]

Observe that $\bb{\psi_{1}}$ is the ``reverse'' of $\bb{\set{x_{1}}}$. For instance, each sequence $w$ in the language denoted by $ba^{*}$ determines, when triggered from $\psi_{1}$, the output $0$, which coincides with the (empty) set of actions that the automaton can fail to execute after performing $w$. Finally, we execute \texttt{reverse} and \texttt{determinise} and get the following minimal Moore automaton (with initial state $\alpha_{1}$):

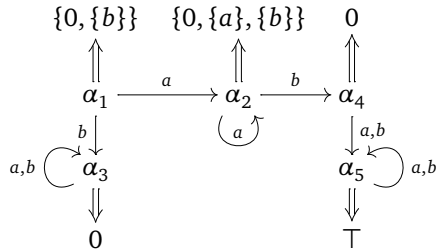
\[
\xymatrix@C=.2cm@R=.5cm{
\set{0, \set{b}} & \set{0, \set{a}, \set{b}} & 0\\
\alpha_{1}\ar@{=>}[u]\ar[r]^{a}\ar[d]_{b} & \alpha_{2}\ar@{=>}[u]\ar[r]^{b}\ar@(dl,dr)^{a} & \alpha_{4}\ar@{=>}[u]\ar[d]^{a,b}\\
\alpha_{3}\ar@{=>}[d]\ar@(dl,ul)^{a,b} & & \alpha_{5}\ar@{=>}[d]\ar@(dr,ur)_{a,b}\\
0 & & \top
}
\]

Remark that the behaviours of the must equivalent states $x_{1}$ and $x_{2}$ have been ``collapsed'' into $\alpha_{1}$.

\subsection{Correctness of Brzozowski's algorithm for must semantics}\label{sec:proof:Brzozowski-must}
In this section we show the correctness of Brzozowski's algorithm for must equivalence. The approach is similar to the one described in Section~\ref{sec:proof:Brzozowski}; the slight differences which are consequences of the divergence-sensitive nature of must semantics are summarised as follows.

Consider an LTS with divergence $(S,t\colon S\rightarrow (1+\pow S)^{A})$, with the initial set of initial states $i\colon 1 \to \pow(S)$. 
As recalled in the beginning of this chapter, the corresponding coalgebraic ingredients are extended to $1 + \pow(-)$ (see \ref{def:Moore-must}): 
the associated Moore machine has the state space in $1+\pow S$ and observations in $1+\pow(\pow A)$, whereas the
induced semantic map becomes $\bb{-}\colon 1+\pow(S) \to (1+\pow(\pow(A)))^{A^\star}$.
Consequently, the current approach considers the join operation $\sqcup$ in $1+\pow(-)$, rather than in $\pow (-)$, as for failure semantics.

\medskip

By reversing and determinising as in \cite{Bonchi:2012:BAA:2340820.2340823}, we obtain the Moore machine
\[
{\mathcal M_{R}} = ((1+\pow(\pow(A)))^{1+\pow S}, <o_R^{\sharp} t^{\sharp}_R>)
\]
for which the initial set of states $i_R$, $o_R^{\sharp}$ and $t_R^{\sharp}$ are defined as in~(\ref{eq:def-reversed-Moore}), in Section~\ref{sec:proof:Brzozowski}.

Equivalently to the statement in~(\ref{eq:acc-rev-formalise}), this machine has the ``reversed'' behaviour of the initial LTS.

The novel algorithm performing the determinisation and the ``reverse and determinise'' at once returns, for the case of must semantics, the Moore machine
\[{\mathcal{\overline{M}}_{R}} = ((1+\pow(\pow(A)))^S, \overline{o}_R, \overline{t}_R)\]
for which the corresponding initial state $\overline{I}_{R}$, the outputs $\overline{o}_{R} \colon (1+\pow(\pow(A)))^S \to 1+\pow(\pow(A))$  and the transitions $\overline{t}_{R} \colon 
(1+\pow(\pow(A)))^S \to ((1+\pow(\pow(A)))^S)^A$ are defined as in~\eqref{eq:BrzoFailMoore} (plus the case $\overline{t}_R(\psi)(a)(x)  =  \top$ if $t(x)(a) = \top$), by replacing $o$ and $t$ with those defined in~(\ref{def:div-LTS-dec}) and~(\ref{def:div-LTS-trans}), in the beginning of this chapter.
 
The fact that $\overline{\mathcal M}_{R}$ has the reverse behaviour of the original LTS follows according to a statement similar to the one in Lemma~\ref{lm:lemma-psi-phi}, by taking $\bb{-}_{2}\colon (1+\pow(\pow(A)))^{S} \to (1+\pow(\pow(A)))^{A^\star}$, $\psi= \overline{I}_{R}$ and $\varphi  =  o^{\sharp}$ (satisfying $(\star)$), and the fact that ${\mathcal M}_{R}$ has reversed behaviour:
\[(\forall w \in A^{*})\,.\, \bb{i_{R}}_{1}(w) = \bb{i}(w^{R}).\]
% Note that, for the case of must semantics, the aforementioned $\psithree $ and $\varphi$ are elements of $(1+\pow\pow A)^{S}$ and $(1+\pow\pow A)^{1+\pow S}$, respectively, and that $\varphi(\top) = \top$. Moreover, the correspondent of $\overline{\psi}_{a}$ used for the proof of Lemma~\ref{lm:lemma-psi-phi} maps $\top$ to $\top$.

To conclude, the soundness of our algorithm follows by:
\[(\forall w\in A^*)\,.\, \bb{\overline{I}_{R}}_{2}(w) = \bb{i}(w^R).\]

\section{Three families of examples}\label{sec:exp}
As discussed in the beginning of this chapter, the theoretical complexity is not
informative about the behaviour of these algorithms on concrete cases.
In this section, we compare \HKC, Brzozowski and partition
refinement~\cite{CleavelandH89} on three families of examples. First,
we need some tools to measure their behaviours. For \HKC, we take
$|R|$, the size of the produced relation $R$: indeed cycle \texttt{2}
of \HKC\ is repeated at most $1+|A|{\cdot}|R|$ times (where $|A|$ is
the size of the alphabet). For~\cite{CleavelandH89}, we consider the
size $n$ of the reachable part of determinised system: the main loop
of the partition refinement is iterated at most $n$ times. Finally,
the cost of Brzozowski algorithm is related to the size of both the
intermediate Moore machine (built after steps \texttt{1},\texttt{2})
and the minimal one (built after steps \texttt{3},\texttt{4}).

%The complexity of the previously discussed algorithms on a given
%automaton depends on distinct values:
%\begin{itemize}
%\item for partition-refinement algorithms, we have an $n\mathrm{ln} n$
%  bound, where $n$ is the size of the determinised and reachable
%  automaton.
%\item for Brzozowski's algorithm, the complexity is related to the
%  size of the intermediate automaton (i.e., the reverse, determinised,
%  and reachable one), and the size of the produced minimal realisation.
%\item for algorithms based on bisimulations up to, the complexity is
%  polynomial in the size of the produced relation
%\end{itemize}
%% In particular, for minimisation-based algorithms, the size of the
%% minimal automaton is always a lower bound, but the size of the
%% intermediate determinised automaton also plays a crucial role (whether
%% it is reversed or not).
%%
%As a consequence, the behaviour of the three kinds of algorithms can
%vary a lot on concrete examples. This is what we illustrate below on
%three examples.

% \paragraph{\HKC\ can be exponentially better than minimisation-based
%   algorithms} 
\medskip

First consider the following LTS, where $n$ is an arbitrary natural
number. After the determinisation, $\set{x}$ can reach all the states of the
shape $\set{x}\cup X_N$, where $X_N=\set{x_i \mid i \in N}$ for any $N\subseteq \set{1,\ldots,n}$.
More precisely, a trace $w \in \set{a,b}^{*}$ of length $k$ which leads $\set{x}$ to $\set{x}\cup X_N$ can be generally defined as a word whose $k-i+1$'st letter is $b$ if and only if $i \in N$.
For
instance for $n=2$, $\set{x}\tr{aa}\set{x}$, $\set{x} \tr{ab}\set{x,
  x_1}$, $\set{x} \tr{ba}\set{x, x_2}$ and $\set{x} \tr{bb} \set{x,
  x_1, x_2}$. All those states are distinguished by must testing; for
instance, $\bb{\set{x,x1,x2}}(a)=\set{a}$ while
$\bb{\set{x,x2}}(a)=\set{0}$. Therefore, the minimal Moore
machine for $\bb{\set{x}}$ has at least $2^n$ states.

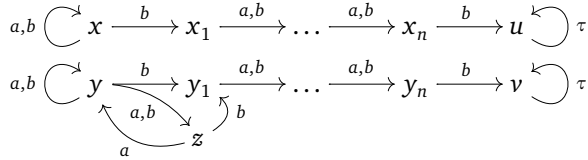
\begin{align*}
  \xymatrix@R=.3cm{
    x\ar@(dl,ul)^{a,b}\ar@{->}[r]^-{b} & x_{1}\ar@{->}[r]^-{a,b} & \dots
    \ar@{->}[r]^-{a,b} & x_{n}\ar[r]^b & u\ar@(dr,ur)_{\tau}\\
    y\ar@(dl,ul)^{a,b}\ar@{->}[r]^-{b}\ar@/^.5pc/[dr]_{a,b} &
    y_{1}\ar@{->}[r]^-{a,b} & \dots \ar@{->}[r]^-{a,b} & y_{n}\ar[r]^b & v\ar@(dr,ur)_{\tau}\\
    & z\ar@/^1pc/[ul]^{a}\ar@/^-.9pc/[u]_{b} }  
\end{align*}
One can prove that $x$ and $y$ are must equivalent by showing that relation $R=$ 
\begin{align*}
  \set{(\set{x},\set{y}), (\set{x}, \set{y , z}), (\top,\top)} \cup 
  \set{(\set{x}\cup X_N,\, \set{y , z} \cup Y_N) \mid N \subseteq \set{1,\ldots, n}}
\end{align*} 
is a bisimulation (here $Y_N=\set{y_i \mid i \in N}$). Note that $R$ contains $2^n+3$ pairs.

In order to check $\bb{\set{x}}{=}\bb{\set{y}}$, \HKC\ builds the following relation,
\begin{align*}
  R'=\set{(\set{x},\set{y}), (\set{x}, \set{y , z})} \cup \set{(\set{x ,
      x_{i}},\, \set{y , z , y_{i}}) \mid i \in \set{1,\ldots, n}}
\end{align*}
which is a bisimulation up-to and which contains only
$n+2$ pairs. It is worth to observe that $R'$ is like a ``basis'' of $R$: all the pairs $(X,Y)\in R$ can be generated by those in 
$R'$ by iteratively applying the rules in \eqref{eq:congclos}.
Therefore, \HKC\ proves $\bb{\set{x}}{=}\bb{\set{y}}$ in polynomial
time, while minimisation-based algorithms (such
as~\cite{CleavelandH89} or Brzozowski's algorithm) require exponential
time.

\smallskip For the following family of LTS's, the algorithm
from~\cite{CleavelandH89} is efficient (the LTS is already
deterministic) while Brzozowski's algorithm is not: the intermediate
Moore machine built after steps~\texttt{1,2} has exponentially many
states (for similar reasons as in the previous example, the automaton
being reversed first).
%so that constructing the intermediate automaton requires exponential time.
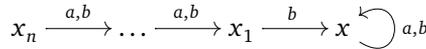
\begin{align*}
  \xymatrix@R=1em {x_n\ar[r]^{a,b}&\dots\ar[r]^{a,b}&x_1\ar[r]^{b}&x\ar@(dr,ur)_{a,b}}
\end{align*}

%Now consider the following family of LTS's, for a given natural number $n$.
%\begin{align*}
%  \xymatrix@R=1em {x_1\ar[r]^{a,b}&\dots\ar[r]^{a,b}&x_{n-1}\ar[r]^{b}&x_n\ar@(ur,dr)^{a,b}}
%\end{align*}
%We get the minimal realisation by determinisation: we simply need to
%add a transition from $x_{n-1}$ to a sink state to make the automaton
%complete. Therefore, the partition-refinment approach only requires
%$n\mathrm{ln}n$ operations.
%%
%In contrast, minimising this automaton using Brzozowski's method is
%not efficient: the determinisation of the reverse automaton has
%$2^{n+1}{-}1$ states (for similar reasons as in the previous example),
%so that constructing the intermediate automaton requires exponential
%time.

%\paragraph{Brzozowski can be exponentially better than the two other methods}
\medskip

Finally consider the family of LTS's on $A=\set{a}$, consisting
in $n$ disjoint cycles of increasing lengths. The case $n=5$ is
depicted on the left below. Suppose that we want to show that the
superposition of states $x^1_0,\dots,x^n_0$ is equivalent to $u$ given on the right.
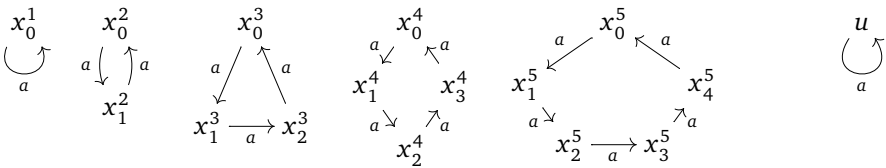
\[
\begin{array}{cccccc}
%\begin{mathpar}
  {\xymatrix @C=0pt @R=3em{%
      {x^1_0}\ar@(ld,rd)_a  \\
    }}
&
  {\xymatrix @C=0pt @R=1.4em { %
      {x^2_0}\ar@/_/[d]_a \\ 
      {x^2_1}\ar@/_/[u]_a }} %
&
  {\xymatrix @C=0pt @R=2.1em { & %
      {x^3_0}\ar[ld]_a \\ 
      {x^3_1}\ar[rr]_a && %
      {x^3_2}\ar[lu]_a }} %
&
  {\xymatrix @C=0pt @R=.6em { & %
      {x^4_0}\ar[ld]_a \\ %
      {x^4_1}\ar[rd]_a && %
      {x^4_3}\ar[lu]_a \\ &
      {x^4_2}\ar[ru]_a }} %
&
  {\xymatrix @C=0pt @R=.5em { && %
      {x^5_0}\ar[lld]_a \\ %
      {x^5_1}\ar[rd]_a &&&& %
      {x^5_4}\ar[llu]_a &\\ & %
      {x^5_2}\ar[rr]_a && %
      {x^5_3}\ar[ru]_a }} %
  \qquad\qquad
  {\xymatrix @C=0pt @R=4em{%
      {u}\ar@(ld,rd)_a  \\
    }}
%\end{mathpar}
\end{array}
\]

The states reachable from the set $\set{x^1_0,\dots x^n_0}$ in the
determinised system are of the shape $X_k=\set{x^i_{k\mod i}\mid
  i\le n}$. There are $p$ such sets, where $p={\emph lcm}[1..n]$ is the least common multiple of the first $n$ natural numbers (this number is
greater than $2^n$ for $n\ge 8$). 
With~\cite{CleavelandH89}, one would start by constructing all those
sets, and one can show that \HKC\ actually produces a relation of size
$p$. Therefore, those two methods need exponentially many steps.
On the other hand, Brzozowski's algorithm is extremely efficient on this
family of examples: the output of any state is always $\set{0}$, so
that the only reachable state in the intermediate Moore machines (built after steps \texttt{1} and \texttt{2}) 
is the function mapping all the states to $\set{0}$. Therefore we obtain the minimal
realisation immediately.

\section{Concrete tests on a synchronisation protocol}
\label{sec:concrete-tests}

We implemented the presented algorithms (Brzozowski minimisation and
\HKC) for ready, failure, and must semantics. Moreover, we tested our implementation
and compared the various algorithms, by analyzing some instances of a
multiway synchronisation protocol (MSP) due to Parrow and
Sj\"odin~\cite{ParrowS96}.

The scenario is the following: there are several clients, denoted by
$1,2,\dots$, trying to synchronise on communication channels, denoted
by $a,b\dots$. Each channel comes with a fixed subset of clients, all
of which must agree to participate for the action to take place.  For
instance, in a configuration denoted by $a(1,2),b(1,2,3)$, with three
clients and two channels; clients $1$ and $2$ have to synchronise to
perform action $a$, and the three clients have to synchronise to
perform action $b$. Parrow and Sj\"odin study protocols allowing to
schedule clients requests, so as to enforce the synchronisation
constraints. They propose an ideal and centralized scheduler as a
specification, and a distributed and more realistic scheduler. They
prove them equivalent, using a notion of equivalence called
``cs-equivalence'' which entails must-testing equivalence in the
considered case. Both schedulers are presented as finite LTS.

We computed those LTS for some small configurations, checked them for
must equivalence, and minimised the ideal scheduler
with respect to must semantics. For each configuration, we give various size
indications in Figure~\ref{fig:concrete-tests}: the first column is the
configuration; the second one gives the number of states of the
minimal Moore machine;
the third and fourth column give the number of states of the ideal
and distributed schedulers, respectively. One can notice that the
ideal schedulers are almost minimal, while the distributed ones are
huge, comparatively. The fifth column gives the number of reachable
states, after determinisation along weak transitions (i.e., the number
of states one would start with with a partition-refinement algorithm);
this number is usually smaller than the size of the distributed
scheduler since the later contains lots of intermediate states that
are removed by determinisation. The sixth column gives the size of the
intermediate automaton, after performing half of Brzozowski's
minimisation algorithm; notice that this intermediate automaton is
usually much smaller than the distributed scheduler, but also much
larger than the ideal and minimal ones. The last column gives the
number of pairs required by HKC to prove the equivalence between the
ideal and the distributed scheduler; it is systematically much less
than the size of the determinised automaton.

\begin{figure}[ht]
\centering
\[
  \begin{tabular}{|l|rrrrr|r|}
    \hline
    config.   & min. & ideal & distr.  & determ. & interm. &  HKC \\
    \hline                                                 
    % a(1)	    &   3 &     3 &      9 &     4 &      6 &    4 \\
    a(1,2)	    &   9 &     9 &     34 &    12 &     88 &   12 \\
    a(1,2,3)	    &  27 &    27 &    304 &    84 &   1110 &   82 \\
    \hline
%    a(1),b(1)	    &   5 &     6 &   1210 &   172 &     14 &   93 \\
    a(1,2),b(1)	    &  15 &    18 &   6089 &  1074 &    189 &  294 \\
    a(1,2),b(3)	    &  17 &    27 &   1057 &   303 &    436 &  225 \\
    a(1,2),b(1,2)   &  28 &    34 & 101532 & 18608 &    389 & 2236 \\
    a(1,2),b(1,3)   &  49 &    54 &  38288 & 11024 &   2568 & 5462 \\
    a(1,2),b(3,4)   &  65 &    81 &   8666 &  3230 &   7570 & 1806 \\
    a(1,2,3),b(1)   &  45 &    54 &  54090 &  8644 &   2207 & 2207 \\
    a(1,2,3),b(4)   &  53 &    81 &  12053 &  3330 &   5546 & 2116 \\
    a(1,2,3),b(1,4) &   - &   162 & 259890 &     - &      - &    - \\
    \hline
    a(1),b(2),c(3)  &   9 &    27 &   5917 &  1594 &    126 &  830 \\
    a(1),b(1),c(2)  &   9 &    18 &  37380 &  7984 &     66 & 2351 \\
    a(1),b(1),c(1)  &   9 &    11 & 149267 & 41444 &     34 & 2685 \\
    a(1,2),b(3),c(4)&  33 &    81 &  50844 & 20526 &   2176 & 6642 \\
    \hline
  \end{tabular}
\]
\caption{Concrete tests.}
\label{fig:concrete-tests}
\end{figure}

\section{Discussion}
\label{sec:concl-concur}
In Chapter~\ref{ch:dec-trace-testing} we have introduced coalgebraic characterisations of decorated trace and must testing
semantics by means of the \emph{generalised powerset construction}~\cite{fsttcs}. This allowed us to adapt proof
techniques and algorithms that have been developed for language
equivalence to must semantics. In particular, in this chapter, we showed that
\emph{bisimulations up-to congruence} (that were recently introduced in~\cite{BONCHI:2012:HAL-00639716:4} for NDA's) are sound also for must semantics. This fact guarantees the correctness of a generalisation of \HKC~\cite{BONCHI:2012:HAL-00639716:4} for checking must equivalence and preorder and suggests that the \emph{antichains}-based algorithms~\cite{tacas10,DoyenR10,CAV06} can be adapted in a similar way. We have also proposed a variation of Brzozowski's algorithm~\cite{Brzozowski} to check must semantics, by exploiting the abstract  theory in~\cite{Bonchi:2012:BAA:2340820.2340823}. Our contribution is not a simple instantiation of~\cite{Bonchi:2012:BAA:2340820.2340823}, but developing our algorithm has required some ingenuity to avoid the preliminary determinisation that would be needed to directly apply~\cite{Bonchi:2012:BAA:2340820.2340823}. We implemented these algorithms together with an interactive applet available online.

Beyond must semantics, one can use such algorithms to check the
\emph{decorated trace equivalences}~\cite{Glabbeek01} that have been
studied in~\cite{dec-tr-MFPS12}: like failure,
these are obtained by decorating the states of an LTS with a function
$o\colon S \to B$. The key of our approach is that $B$ needs to be a
semi-lattice with bottom (for must, a semi-lattice with bottom and
top); this is required by the generalised powerset construction so
that decorated LTS's can be determinised into Moore machines.

\chapter{Future work}
We provide an overview of the possible theoretical and practical further developments of the work in this thesis.

\bigskip
With respect to the contributions on generalised regular expressions modelling non-deterministic coalgebras introduced in Chapter~\ref{ch:dec-bisim}, we consider:

\smallskip
\textbf{Extensions to quantitative coalgebras.}
In the future, we would like to extend the class of systems to include quantitative coalgebras.
In~\cite{bbrs_ic}, the approach for handling non-deterministic coalgebras was extended to a large class of quantitative systems encompassing weighted automata, simple Segala, stratified and Pnueli-Zuck systems, by considering a functor type that allows the transitions of systems to take values in a monoid structure of quantitative values.

%The challenge in this respect arises from the fact that the definition of expressions for quantitative coalgebras involving the distribution monad is not as modular as for the other functors (for details see~\cite{bbrs_ic}).
%This is a consequence of the fact that the sum of two valid expressions might not be a valid expression anymore (since in distributions we require that the sum of probabilities add up to $1$). 
The challenge in this respect arises from the fact that computing bisimulation relations in a quantitative setting will involve matrix manipulations, hence requiring linear algebra techniques of which it is not clear how to implement in \CIRC.

\smallskip
\textbf{Tool enhancements and complexity studies.}
To improve usability, building a graphical interface for the tool is an obvious next step. The graphical interface should ideally allow the specification of expressions by means of  systems of equations (which are then solved internally) or even by means of an automaton, which would then be translated to an expression using Kleene's theorem.

We also would like to explore how adding more axioms than ACI to the prover (that is, each step of the bisimulation checking is performed modulo more equations) improves the performance of the tool. Our experience so far shows that by adding the axioms describing the interplay between $\emp$ and the other constructs, \textit{i.e.} $\emp \oplus \eps = \eps$, the prover works significantly faster.

We have not yet studied complexity bounds for the algorithms presented in this paper. We conjecture however that the bounds will be very similar to the already known ones for classical regular expressions~\cite{K08a,worthington}.

\bigskip

In connection with the coalgebraic handling of decorated trace, may and must testing semantics in Chapter~\ref{ch:dec-trace-testing} and Chapter~\ref{ch:algorithms}, we consider:

\smallskip
\textbf{Coalgebraic handling of other semantics.}
In the future, we want to derive a new representation of possible-futures semantics. This is motivated by the current drawback of storing for each state of the LTS's the corresponding set of traces. In this context, it might be more appropriate considering the definition of possible-futures semantics given in terms of nested bisimulations~\cite{Hennessy:1985:ALN:2455.2460}, rather than the set-theoretic one in~\cite{Glabbeek01}.

Moreover, we aim at providing coalgebraic modellings for the remaining semantics of the spectrum in~\cite{Glabbeek01}. Amongst these, we mention possible-worlds semantics, whose path-based characterisation shifts the problem of reasoning on the corresponding equivalence to a setting close to possible-futures semantics. The coalgebraic modelling of possible-futures semantics still requires an efficient handling of the traces associated with a process, as mentioned above. Orthogonally, the challenge in deriving a straightforward modelling of simulation semantics via the generalised powerset construction~\cite{gen-pow} originates from the absence of an equivalent trace-based definition.

We would also like to understand how our approach can be combined with the results in~\cite{DBLP:journals/tcs/BorealeG06} to obtain a coinductive approach 
to denotational (linear-time) semantics of different kinds of processes calculi.  The work in~\cite{DBLP:journals/tcs/BorealeG06} presents a fully abstract model of must testing for CSP by 
turning the set of processes into a (partial) Moore automaton with output on a certain semiring $K$ and 
input from a set of actions $A$. The final semantics of this automaton is then given as a 
powerserie in $K^{A^*}$. The approach can be easily extended to trace equivalence and other 
calculi, such as CCS, but no other decorated trace equivalences are further considered.
Our work is similar in spirit of the above as we also construct a Moore automaton from 
a transition system but, in general, we do not need a semiring structure, making the entire
framework much simpler. For example, for the must testing, our Moore automata
have outputs in the set $1 + \Powf(\Powf(A))$. The framework is even
simpler for the case of trace semantics, where our Moore automata have outputs in the 
two elements set $2$.

Furthermore, we think it is promising to investigate whether our approach can be extended to the testing semantics of probabilistic and non-deterministic processes\linebreak \cite{journals/corr/abs-1107-1201,DBLP:conf/pstv/YiL92,DBLP:conf/concur/Segala96}.

\smallskip
\textbf{More algorithms.} An interesting topic to investigate in the future is adapting the Brzozowski and \HKC\ algorithms to check \emph{fair testing}~\cite{Rensink:2007:FT:1223922.1224110}. In~\cite{Rensink:2007:FT:1223922.1224110}, fair testing is defined in terms of the so-called failure trees. While the corresponding coalgebraic modelling can be easily derived via the powerset construction, we do not know how to model fair testing equivalence and preorder.

We would also like to study whether Brzozowski and \HKC\ can be adapted and effectively applied to reason on decorated trace semantics of generative probabilistic systems.

\smallskip
\textbf{Rule formats for compositionality.} In the future we consider worth studying to what extent the modal characterisations of decorated trace semantics in~\cite{Glabbeek01} can be exploited 
%~\cite{Turi:1997:TMO:788019.788864,Klin:2009:BMM:1512997.1513245}
in order to develop a systematic study of their compositionality for languages defined by SOS-like rules~\cite{Plotkin81astructural} satisfying specific formats.

In this respect, we refer to the work in~\cite{Klin:2009:BMM:1512997.1513245}, where both the rule formats and decorated trace equivalences are ``massaged'' into a bialgebraic setting, by means of logical distributive laws defined in terms of notions of syntax and logical formulae.
However, applying the machinery in~\cite{Klin:2009:BMM:1512997.1513245} requires a certain amount of ingenuity for identifying the right logical behaviour. Therefore, one of the challenges (also mentioned as pointer to future work in~\cite{Klin:2009:BMM:1512997.1513245}) consists in (partially) automating the whole procedure or, at least, in gaining more insight on how this could be achieved in a rather algorithmic fashion.

\backmatter
%\bibliography{circ}

\begin{thebibliography}{10}

\bibitem[ABH$^{+}$12]{ABHKMS12}
J.~Ad{\'a}mek, F.~Bonchi, M.~H{\"u}lsbusch, B.~K{\"o}nig,
  S.~Milius, and A.~Silva.
\newblock A coalgebraic perspective on minimization and determinization.
\newblock In {\em FoSSaCS}, 58--73, 2012.

\bibitem[ABV94]{Aceto:1994:TSR:184662.184663}
L.~Aceto, B.~Bloom, and F.~Vaandrager.
\newblock Turning {SOS} rules into equations.
\newblock {\em Inf. Comput.}, 111:1--52, May 1994.

\bibitem[ACEII11]{Aceto-GSOS}
L.~Aceto, G.~Caltais, E.~I.~Goriac, and A.~Ing\'olfsd\'ottir.
\newblock Axiomatizing {GSOS} with {P}redicates.
\newblock In {\em Proc. SOS 2011}, EPTCS 62, 1--15, 2011.

\bibitem[ACGI11]{DBLP:conf/calco/AcetoCGI11}
L.~Aceto, G.~Caltais, E.~I.~Goriac, and A.~Ing{\'o}lfsd{\'o}ttir.
\newblock Preg axiomatizer - a ground bisimilarity checker for {GSOS} with
  predicates.
\newblock In {\em CALCO}, 378--385, 2011.

\bibitem[ACH$^{+}$10]{tacas10}
P.~A. Abdulla, Y.-F. Chen, L.~Hol\'{\i}k, R.~Mayr, and T.~Vojnar.
\newblock When simulation meets antichains.
\newblock In {\em Proc.\ TACAS}, volume 6015, 158--174, 2010.

\bibitem[AFV99]{sos}
L.~Aceto, W.~Fokkink, and C.~Verhoef.
\newblock Structural operational semantics.
\newblock In {\em Handbook of Process Algebra}, 197--292, 1999.

\bibitem[AH92]{Aceto:1992:TDD:147508.147527}
L.~Aceto and M.~Hennessy.
\newblock Termination, deadlock, and divergence.
\newblock {\em JACM}, 39:147--187, January 1992.

\bibitem[AK95]{ak:fix-point-set-functor}
J.~Ad{\'a}mek and V.~Koubek.
\newblock On the greatest fixed point of a set functor.
\newblock {\em Theor. Comput. Sci.}, 150(1):57--75, 1995.

\bibitem[AM89]{DBLP:conf/ctcs/AczelM89}
P.~Aczel and N.~P.~Mendler.
\newblock A final coalgebra theorem.
\newblock In {\em Cat. Theor. and Comp. Sci.}, 357--365, 1989.

\bibitem[Awo10]{awodey}
S.~Awodey.
\newblock {\em Category theory}.
\newblock Oxford Logic Guides, 2010.

\bibitem[BB87]{DBLP:journals/cn/BolognesiB87}
T.~Bolognesi and E.~Brinksma.
\newblock Introduction to the {ISO} specification language {LOTOS}.
\newblock {\em Computer Networks}, 14:25--59, 1987.

\bibitem[BBC$^{+}$12]{dec-tr-MFPS12}
F.~Bonchi, M.~M. Bonsangue, G.~Caltais, J.~J. M.~M. Rutten,
  and A.~Silva.
\newblock Final semantics for decorated traces.
\newblock {\em Electr. Notes Theor. Comput. Sci.}, 286:73--86, 2012.

\bibitem[BBRS12]{Bonchi:2012:BAA:2340820.2340823}
F.~Bonchi, M.~Bonsangue, J.~Rutten, and A.~Silva.
\newblock Brzozowski's algorithm (co)algebraically.
\newblock In {\em Logic and Program Semantics - Essays Dedicated to Dexter
  Kozen on the Occasion of His 60th Birthday}, volume 7230 of {\em Lect.
  Notes in Comput. Sci.}, 12--23, 2012.

\bibitem[BCG$^{+}$11]{sbmf}
M.~Bonsangue, G.~Caltais, E.~I.~Goriac, D.~Lucanu, J.~Rutten, and A.~Silva.
\newblock A decision procedure for bisimilarity of generalized regular
  expressions.
\newblock In {\em Proc. of the 13th Brazilian conference on Formal
  methods: foundations and applications}, SBMF'10, 226--241, 2011.

\bibitem[BCG$^{+}$13]{DBLP:journals/corr/abs-1303-1994}
M.~M. Bonsangue, G.~Caltais, E.~I.~Goriac, D.~Lucanu, J.~J. M.~M. Rutten, and A.~Silva.
\newblock Automatic equivalence proofs for non-deterministic coalgebras.
\newblock {\em Sci. Comput. Program.}, 78(9):1324--1345, 2013.

\bibitem[BdV04]{DBLP:journals/jlp/BaetenV04}
J.~C.~M. Baeten and E.~P. de~Vink.
\newblock axiomatizing GSOS with termination.
\newblock {\em J. Log. Algebr. Program.}, 60-61:323--351, 2004.

\bibitem[BFvG04]{Bloom:2004:PFD:963927.963929}
B.~Bloom, W.~Fokkink, and R.~J. van Glabbeek.
\newblock Precongruence formats for decorated trace semantics.
\newblock {\em ACM Trans. Comput. Logic}, 5:26--78, January 2004.

\bibitem[BG06]{DBLP:journals/tcs/BorealeG06}
M.~Boreale and F.~Gadducci.
\newblock Processes as formal power series: A coinductive approach to
  denotational semantics.
\newblock {\em Theor. Comput. Sci.}, 360(1-3):440--458, 2006.

\bibitem[BIM95]{Bloom:1995:BCT:200836.200876}
B.~Bloom, S.~Istrail, and A.~R. Meyer.
\newblock Bisimulation can't be traced.
\newblock {\em JACM}, 42:232--268, 1995.

\bibitem[BJM00]{Bouhoula-Jouannaud-Meseguer00}
A.~Bouhoula, J.~P.~Jouannaud, and J.~Meseguer.
\newblock Specification and proof in membership equational logic.
\newblock {\em Theor. Comput. Sci.}, 236(1-2):35--132, 2000.

\bibitem[BP13]{BONCHI:2012:HAL-00639716:4}
F.~Bonchi and D.~Pous.
\newblock Checking {NFA} equivalence with bisimulations up to congruence.
\newblock In {\em POPL}, 457--468. ACM, 2013.

\bibitem[BPK12]{BKP}
N.~Bezhanishvili, P.~Panangaden, and C.~Kupke.
\newblock Minimization via duality.
\newblock In {\em WoLLIC' 12}, 191-205, 2012.

\bibitem[BRS09]{BonsangueRS09}
M.~M. Bonsangue, J.~J. M.~M. Rutten, and A.~Silva.
\newblock An algebra for {Kripke} polynomial coalgebras.
\newblock In {\em Logic in Comput. Sci.}, 49--58, 2009.

\bibitem[Brz62]{Brzozowski}
J.~A. Brzozowski.
\newblock Canonical regular expressions and minimal state graphs for definite
  events.
\newblock In {\em Mathem. Theory of Aut.}, volume 12(6), 529--561.

\bibitem[BS12]{burris2012course}
S.~Burris and H.P. Sankappanavar.
\newblock {\em A Course in Universal Algebra}.
\newblock Dover Publications, Incorporated, 2012.

\bibitem[BSdV04]{DBLP:journals/tcs/BartelsSV04}
F.~Bartels, A.~Sokolova, and E.~P. de~Vink.
\newblock A hierarchy of probabilistic system types.
\newblock {\em Theor. Comput. Sci.}, 327(1-2):3--22, 2004.

\bibitem[CDE$^{+}$07]{DBLP:conf/maude/2007}
M.~Clavel, F.~Dur\'{a}n, S.~Eker, P.~Lincoln, N.~Mart\'{\i}-Oliet, J.~Meseguer, and C.~Talcott.
\newblock {\em All about {M}aude - a high-performance logical framework: how to
  specify, program and verify systems in rewriting logic}.
\newblock Springer-Verlag, Berlin, Heidelberg, 2007.

\bibitem[CDLT08]{DBLP:journals/topnoc/CalzolaiNLT08}
F.~Calzolai, R.~{De Nicola}, M.~Loreti, and F.~Tiezzi.
\newblock {TAPA}s: A tool for the analysis of process algebras.
\newblock {\em T. Petri Nets and Other Models of Concurrency}, 5100:54--70,
  2008.

\bibitem[CGK$^{+}$13]{DBLP:conf/tacas/CranenGKSVWW13}
S.~Cranen, J.~F.~Groote, J.~J.~A. Keiren, F.~P.~M. Stappers,
  E.~P. de~Vink, W.~Wesselink, and T.~A.~C. Willemse.
\newblock An overview of the $\mu$CRL2 toolset and its recent advances.
\newblock In {\em TACAS}, 199--213, 2013.

\bibitem[CH89]{CleavelandH89}
R.~Cleaveland and M.~Hennessy.
\newblock Testing equivalence as a bisimulation equivalence.
\newblock {\em Lect. Notes in Comput. Sci.}, vol. 407,
  11--23, 1989.

\bibitem[CHL03]{CancilaHL03}
D.~Cancila, F.~Honsell, and M.~Lenisa.
\newblock Generalized coiteration schemata.
\newblock {\em Electr. Notes in Theor. Comput. Sci.}, 82(1), 2003.

\bibitem[CPS93a]{Cleaveland:1993:CWS:151646.151648}
R.~Cleaveland, J.~Parrow, and B.~Steffen.
\newblock The Concurrency Workbench: a semantics-based tool for the
  verification of concurrent systems.
\newblock {\em ACM Trans. Program. Lang. Syst.}, 15(1):36--72, January 1993.

\bibitem[CPS93b]{cleaveland1993concurrency}
R.~Cleaveland, J.~Parrow, and B.~Steffen.
\newblock The {C}oncurrency {W}orkbench: A semantics-based tool for the
  verification of concurrent systems.
\newblock {\em ACM Trans. Program. Lang. Syst. (TOPLAS)},
  15(1):36--72, 1993.

\bibitem[CS96]{DBLP:conf/cav/CleavelandS96}
R.~Cleaveland and S.~Sims.
\newblock The {NCSU} {C}oncurrency {W}orkbench.
\newblock In {\em CAV}, volume 1102 of {\em Lect. Notes in Comput. Sci.},
  394--397, 1996.

\bibitem[DH84]{DBLP:journals/tcs/NicolaH84}
R.~{De Nicola} and M.~Hennessy.
\newblock Testing equivalences for processes.
\newblock {\em Theor. Comput. Sci.}, 34:83--133, 1984.

\bibitem[Dob08]{doberkat}
E.~E.~Doberkat.
\newblock Erratum and addendum: {E}ilenberg-{M}oore algebras for stochastic
  relations.
\newblock {\em Inf. Comput.}, 206(12):1476--1484, 2008.

\bibitem[DR10]{DoyenR10}
L.~Doyen and J.~F.~Raskin.
\newblock Antichain {A}lgorithms for {F}inite {A}utomata.
\newblock In {\em Proc.\ TACAS}, volume 6015, 2010.

\bibitem[DvGHM11]{journals/corr/abs-1107-1201}
Y.~Deng, R.~J. van Glabbeek, M.~Hennessy, and C.~Morgan.
\newblock Real-reward testing for probabilistic processes (extended abstract).
\newblock {\em Electr. Proc. Theor. Comput. Sci.}, 61--73, 2011.

\bibitem[Eng85]{DBLP:journals/tcs/Engelfriet85}
J.~Engelfriet.
\newblock Determinacy - (observation equivalence = trace equivalence).
\newblock {\em Theor. Comput. Sci.}, 36:21--25, 1985.

\bibitem[FME05]{FMT2005}
G.L. Ferrari, U.~Montanari, and E.Tuosto.
\newblock Coalgebraic minimization of {HD}-automata for the pi-calculus using
  polymorphic types.
\newblock {\em TCS}, 331(2--3):325--365, 2005.

\bibitem[GLMS11]{DBLP:conf/tacas/GaravelLMS11}
H.~Garavel, F.~Lang, R.~Mateescu, and W.~Serwe.
\newblock CADP 2010: A toolbox for the construction and analysis of distributed
  processes.
\newblock In {\em TACAS}, 372--387, 2011.

\bibitem[GLR00]{goguen-lin-rosu-2000-ase}
J.~Goguen, K.~Lin, and G.~Ro\c su.
\newblock Circular coinductive rewriting.
\newblock In {\em ASE '00: Proc. of the 15th IEEE international
  conference on Automated software engineering}, 123--132, 2000.

\bibitem[GLR10]{acca}
E.~.~Goriac, D.~Lucanu, and G.~Ro\c{s}u.
\newblock Automating coinduction with case analysis.
\newblock In {\em Proc. of the 12th international conference on Formal
  engineering methods and software engineering}, ICFEM'10, 220--236, 2010.

\bibitem[GM92]{Goguen92order-sortedalgebra}
J.~Goguen and J.~Meseguer.
\newblock Order-sorted algebra {I}: Equational deduction for multiple
  inheritance, overloading, exceptions and partial operations.
\newblock {\em Theor. Comput. Sci.}, 105(2):217--273, 1992.

\bibitem[GS02]{DBLP:journals/mscs/GummS02}
H.~P.~Gumm and T.~Schr{\"o}der.
\newblock Coalgebras of bounded type.
\newblock {\em Mathem. Struct. in Comput. Sci.}, 12(5):565--578,
  2002.

\bibitem[Hen88]{Hennessy:1988:ATP:50497}
M.~Hennessy.
\newblock {\em Algebraic theory of processes}.
\newblock MIT Press, Cambridge, MA, USA, 1988.

\bibitem[HJ98]{DBLP:journals/iandc/HermidaJ98}
C.~Hermida and B.~Jacobs.
\newblock Structural induction and coinduction in a fibrational setting.
\newblock {\em Inf. Comput.}, 145(2):107--152, 1998.

\bibitem[HJS07]{HJS07}
I.~Hasuo, B.~Jacobs, and A.~Sokolova.
\newblock Generic trace semantics via coinduction.
\newblock {\em Logical Meth. in Comput. Sci.}, 3(4), 2007.

\bibitem[HK71]{hoka71}
J.~Hopcroft and R.~Karp.
\newblock A linear algorithm for testing equivalence of finite automata.
\newblock Technical Report, Dept. of Computer Science, Cornell Univ., December
  1971.

\bibitem[HM85]{Hennessy:1985:ALN:2455.2460}
M.~Hennessy and R.~Milner.
\newblock Algebraic laws for nondeterminism and concurrency.
\newblock {\em JACM}, 32(1):137--161, January 1985.

\bibitem[HMS05]{DBLP:conf/fase/HausmannMS05}
D.~Hausmann, T.~Mossakowski, and L.~Schr{\"o}der.
\newblock Iterative {C}ircular {C}oinduction for {C}o{C}asl in
  {I}sabelle/{HOL}.
\newblock In {\em Lect.
  Notes in Comput. Sci.}, vol. 3442, 341--356, 2005.

\bibitem[Hoa78]{Hoare:1978:CSP:359576.359585}
C.~A.~R. Hoare.
\newblock Communicating sequential processes.
\newblock {\em Commun. ACM}, 21(8):666--677, August 1978.

\bibitem[Hop71]{minimisation}
J.~E. Hopcroft.
\newblock An n log n algorithm for minimizing in a finite automaton.
\newblock In {\em Proc.\ International Symposium of Theory of Machines and
  Computations}, 189--196, 1971.

\bibitem[HP85]{Harel:1989:DRS:101969.101990}
D.~Harel and A.~Pnueli.
\newblock Logics and models of concurrent systems.
\newblock On the development of reactive systems, 477--498, 1985.

\bibitem[JR97]{Jacobs97atutorial}
B.~Jacobs and J.~Rutten.
\newblock A tutorial on (co)algebras and (co)induction.
\newblock {\em EATCS Bulletin}, 62:62--222, 1997.

\bibitem[JSS12]{DBLP:conf/cmcs/JacobsSS12}
B.~Jacobs, A.~Silva, and A.~Sokolova.
\newblock Trace semantics via determinization.
\newblock In {\em CMCS}, 109--129, 2012.

\bibitem[JS90]{GPS-Jou-Smolka}
C.~C.~Jou and S.~A. Smolka.
\newblock Equivalences, congruences, and complete axiomatizations for
  probabilistic processes.
\newblock {\em CONCUR '90 Theories of
  Concurrency: Unification and Extension}, {\em Lect. Notes in Comp. Sci.}, vol. 458, 367--383, 1990.

\bibitem[Kel76]{Keller:1976:FVP:360248.360251}
R.~M. Keller.
\newblock Formal verification of parallel programs.
\newblock {\em Commun. ACM}, 19(7):371--384, July 1976.

\bibitem[Kle56]{Kleene61}
S.~Kleene.
\newblock Representation of events in nerve nets and finite automata.
\newblock {\em Automata Studies}, 3--42, 1956.

\bibitem[Kli04]{DBLP:journals/entcs/Klin04}
B.~Klin.
\newblock A coalgebraic approach to process equivalence and a coinduction
  principle for traces.
\newblock {\em Electr. Notes Theor. Comput. Sci.}, 106:201--218, 2004.

\bibitem[Kli09]{Klin:2009:BMM:1512997.1513245}
B.~Klin.
\newblock Bialgebraic methods and modal logic in structural operational
  semantics.
\newblock {\em Inf. Comput.}, 207:237--257, February 2009.

\bibitem[Kli11]{DBLP:journals/tcs/Klin11}
B.~Klin.
\newblock Bialgebras for structural operational semantics: An introduction.
\newblock {\em Theor. Comput. Sci.}, 412(38):5043--5069, 2011.

\bibitem[Koz91]{kozen91}
D.~Kozen.
\newblock A completeness theorem for {K}leene algebras and the algebra of
  regular events.
\newblock In {\em LICS}, 214--225, 1991.

\bibitem[Koz01]{kozen-nerode}
D.~Kozen.
\newblock Myhill-{N}erode relations on automatic systems and the completeness
  of {K}leene algebra.
\newblock In {\em STACS}. {\em Lect. Notes in Comput. Sci.}, vol. 2010, 27--38, 2001.

\bibitem[Koz06]{K08a}
D.~Kozen.
\newblock On the representation of Kleene algebras with tests.
\newblock In {\em MFCS}, 73--83, 2006.

\bibitem[KS83]{Kanellakis:1983:CEF:800221.806724}
P.~C. Kanellakis and S.~A. Smolka.
\newblock {CCS} expressions, finite state processes, and three problems of
  equivalence.
\newblock In {\em Proc. of the second annual ACM symposium on Principles
  of distributed computing}, PODC '83, 228--240, 1983.

\bibitem[Kur00]{k:logics-coalgebras}
A.~Kurz.
\newblock {\em Logics for Coalgebras and Applications to Computer Science}.
\newblock PhD thesis, Ludwigs-Maximilians-Universit\"at M\"unchen, 2000.

\bibitem[Len99]{Lenisa19992}
M.~Lenisa.
\newblock From set-theoretic coinduction to coalgebraic coinduction: some
  results, some problems.
\newblock {\em Electr. Notes Theor. Comput. Sci.}, 19:2--22, 1999.

\bibitem[LGCR09]{lucanu-etal-2009-calco}
D.~Lucanu, E.~I.~Goriac, G.~Caltais, and G.~Ro\c{s}u.
\newblock {CIRC}: a behavioral verification tool based on circular coinduction.
\newblock In {\em Proc. of the 3rd international conference on Algebra
  and coalgebra in computer science}, CALCO'09, 433--442, 2009.

\bibitem[LPW00]{DBLP:journals/entcs/LenisaPW00}
M.~Lenisa, J.~Power, and H.~Watanabe.
\newblock Distributivity for endofunctors, pointed and co-pointed endofunctors,
  monads and comonads.
\newblock {\em Electr. Notes Theor. Comput. Sci.}, 33:230--260, 2000.

\bibitem[Mea55]{mea55}
G.~H. Mealy.
\newblock A method to synthesizing sequential circuits.
\newblock {\em Bell System Technical Journal}, 1045--1079, 1955.

\bibitem[Mil84]{milner}
R.~Milner.
\newblock A complete inference system for a class of regular behaviours.
\newblock {\em J.~Comput.~System~Sci.}, 28(3):439--466, 1984.

\bibitem[Mil89]{milner89}
R.~Milner.
\newblock {\em Communication and concurrency}.
\newblock Prentice Hall, 1989.

\bibitem[Mon08]{Monteiro08}
L.~Monteiro.
\newblock A coalgebraic characterization of behaviours in the linear time -
  branching time spectrum.
\newblock In {\em WADT}, 251--265, 2008.

\bibitem[Moo56]{Moore56}
E.~F. Moore.
\newblock Gedanken-experiments on sequential machines.
\newblock In {\em Automata Studies}, 129--153, 1956.

\bibitem[MP81]{Park81}
D.~Michael and R.~Park.
\newblock Concurrency and automata on infinite sequences.
\newblock In {\em Theor. Comput. Sci.}, 167--183, 1981.

\bibitem[MPW92]{MilnerPW92}
R.~Milner, J.~Parrow, and D.~Walker.
\newblock A calculus of mobile processes, {I/II}.
\newblock 100(1):1--77, 1992.

\bibitem[MS73]{MS73}
A.~R. Meyer and L.~J. Stockmeyer.
\newblock Word problems requiring exponential time.
\newblock In {\em Proc.\ STOC}, 1--9, 1973.

\bibitem[Nic87]{DBLP:journals/acta/Nicola87}
R.~{de Nicola}.
\newblock Extensional equivalences for transition systems.
\newblock {\em Acta Inf.}, 24(2):211--237, 1987.

\bibitem[OH86]{DBLP:journals/acta/OlderogH86}
E.~R.~Olderog and C.~A.~R. Hoare.
\newblock Specification-oriented semantics for communicating processes.
\newblock {\em Acta Inf.}, 23(1):9--66, 1986.

\bibitem[ORT09]{derivatives-jfp09}
S.~Owens, J.~H. Reppy, and A.~Turon.
\newblock Regular-expression derivatives re-examined.
\newblock {\em J. Funct. Program.}, 19(2):173--190, 2009.

\bibitem[Phi87]{DBLP:journals/tcs/Phillips87}
I.~Phillips.
\newblock Refusal testing.
\newblock {\em Theor. Comput. Sci.}, 50:241--284, 1987.

\bibitem[Plo04]{Plotkin81astructural}
G.~D. Plotkin.
\newblock A structural approach to operational semantics.
\newblock {\em J. Log. Algebr. Program.}, 60-61:17--139, 2004.

\bibitem[Pnu85]{DBLP:conf/icalp/Pnueli85}
A.~Pnueli.
\newblock Linear and branching structures in the semantics and logics of
  reactive systems.
\newblock In {\em ICALP}, 15--32, 1985.

\bibitem[PS96]{ParrowS96}
J.~Parrow and P.~Sj{\"o}din.
\newblock Designing a multiway synchronization protocol.
\newblock {\em Comput. Communic.}, 19(14):1151--1160, 1996.

\bibitem[PT87]{Paige:87:SIAM}
R.~Paige and R.~E. Tarjan.
\newblock Three partition refinement algorithms.
\newblock {\em SIAM J. Comput.}, 16(6):973--989, 1987.

\bibitem[Rab80]{DBLP:journals/siamcomp/Rabin80}
M.~O. Rabin.
\newblock Probabilistic algorithms in finite fields.
\newblock {\em SIAM J. Comput.}, 9(2):273--280, 1980.

\bibitem[RB81]{DBLP:conf/focs/RoundsB81}
W.~C. Rounds and S.~D. Brookes.
\newblock Possible futures, acceptances, refusals, and communicating processes.
\newblock In {\em FOCS}, 140--149, 1981.

\bibitem[RBR13]{DBLP:conf/sofsem/RotBR13}
J.~Rot, M.~M. Bonsangue, and J.~J. M.~M. Rutten.
\newblock Coalgebraic bisimulation-up-to.
\newblock In {\em SOFSEM}, 369--381, 2013.

\bibitem[RL09]{rosu-lucanu-2009-calco}
G.~Ro\c{s}u and D.~Lucanu.
\newblock Circular coinduction: a proof theoretical foundation.
\newblock In {\em Proc. of the 3rd international conference on Algebra
  and coalgebra in computer science}, CALCO'09, 127--144, 2009.

\bibitem[Ros00]{rosu-thesis}
G.~Ro\c su.
\newblock {\em Hidden Logic}.
\newblock PhD thesis, University of California at San Diego, 2000.

\bibitem[RS59]{Rabin:1959:FAD:1661907.1661909}
M.~O. Rabin and D.~Scott.
\newblock Finite automata and their decision problems.
\newblock {\em IBM J. Res. Dev.}, 3(2):114--125, April 1959.

\bibitem[RTJ01]{ccsl}
J.~Rothe, H.~Tews, and B.~Jacobs.
\newblock The coalgebraic class specification language {CCSL}.
\newblock {\em J. UCS}, 7(2):175--193, 2001.

\bibitem[Rut98]{DBLP:conf/concur/Rutten98}
J.~J. M.~M. Rutten.
\newblock Automata and coinduction (an exercise in coalgebra).
\newblock In {\em CONCUR}, 194--218, 1998.

\bibitem[Rut00]{Rutten00}
J.~J. M.~M. Rutten.
\newblock Universal coalgebra: a theory of systems.
\newblock {\em Theor. Comput. Sci.}, 249(1):3--80, 2000.

\bibitem[Rut03]{DBLP:journals/tcs/Rutten03}
J.~J. M.~M. Rutten.
\newblock Behavioural differential equations: a coinductive calculus of
  streams, automata, and power series.
\newblock {\em Theor. Comput. Sci.}, 308(1-3):1--53, 2003.

\bibitem[Rut05]{streams:rutten05}
J.~J. M.~M. Rutten.
\newblock {A} coinductive calculus of streams.
\newblock {\em Mathem. Struct. in Comput. Sci.}, 15(1):93--147,
  2005.

\bibitem[RV07]{Rensink:2007:FT:1223922.1224110}
A.~Rensink and W.~Vogler.
\newblock Fair testing.
\newblock {\em Inf. Comput.}, 205(2):125--198, February 2007.

\bibitem[Sal66]{salomaa}
A.~Salomaa.
\newblock Two complete axiom systems for the algebra of regular events.
\newblock {\em JACM}, 13(1):158--169, 1966.

\bibitem[San98]{San98MFCS}
D.~Sangiorgi.
\newblock On the bisimulation proof method.
\newblock 8:447--479, 1998.

\bibitem[SBBR10]{fsttcs}
A.~Silva, F.~Bonchi, M.~M. Bonsangue, and J.~J. M.~M.
  Rutten.
\newblock {Generalizing the powerset construction, coalgebraically}.
\newblock In {\em FSTTCS 2010},
  vol. 8 of {\em LIPIcs}, 272--283, 2010.

\bibitem[SBBR11]{bbrs_ic}
A.~Silva, F.~Bonchi, M.~Bonsangue, and J.~Rutten.
\newblock Quantitative {K}leene coalgebras.
\newblock {\em Inform. and Comput.}, 209(5):822--849, 2011.

\bibitem[SBBR13]{gen-pow}
A.~Silva, F.~Bonchi, M.~M. Bonsangue, and J.~J. M.~M.
  Rutten.
\newblock Generalizing determinization from automata to coalgebras.
\newblock {\em Logical Meth. in Comput. Sci.}, 9(1), 2013.

\bibitem[SBR10]{brs_lmcs}
A.~Silva, M.~M. Bonsangue, and J.~J. M.~M. Rutten.
\newblock Non-deterministic {K}leene coalgebras.
\newblock {\em Logical Meth. in Comput. Sci.}, 6(3), 2010.

\bibitem[Seg96]{DBLP:conf/concur/Segala96}
R.~Segala.
\newblock Testing probabilistic automata.
\newblock In {\em CONCUR}, volume 1119 of {\em Lect. Notes in Comput.
  Sci.}, 299--314, 1996.

\bibitem[SR11]{sangiorgi2011advanced}
D.~Sangiorgi and J.~Rutten.
\newblock {\em Advanced Topics in Bisimulation and Coinduction}.
\newblock Cambridge Tracts in Theoretical Computer Science, 2011.

\bibitem[Sta11]{s:relating-coalgebraic-bisimulation-journal}
S.~Staton.
\newblock Relating coalgebraic notions of bisimulation.
\newblock {\em {LMCS}}, 7(1), 2011.

\bibitem[TP97]{Turi:1997:TMO:788019.788864}
D.~Turi and G.~Plotkin.
\newblock Towards a mathematical operational semantics.
\newblock In {\em Proc. 12th LICS Conf.}, 280--291, 1997.

\bibitem[TV05]{TabakovVardi05}
D.~Tabakov and M.~Vardi.
\newblock Experimental evaluation of classical automata constructions.
\newblock In {\em Proc.\ LPAR}, vol. 3835, 396--411, 2005.

\bibitem[Val95]{Valmari-fail}
A.~Valmari.
\newblock Failure-based equivalences are faster than many believe.
\newblock In {\em Struct. in Conc. Theory}, 326--340, 1995.

\bibitem[vG01a]{Glabbeek01}
R.~J. {van}~Glabbeek.
\newblock The linear time - branching time spectrum {I}. {T}he semantics of
  concrete, sequential processes.
\newblock In {\em Handbook of
  Process Algebra}, 3--99, 2001.

\bibitem[vG01b]{vanGlabbeek:2001:BTS:377786.377855}
R.~J. van Glabbeek.
\newblock Current trends in theoretical computer science.
\newblock chapter What is branching time semantics and why to use it?,
  469--479, 2001.

\bibitem[Wat95]{Wat95}
B.~W. Watson.
\newblock {\em Taxonomies and Toolkits of Regular Language Algorithms}.
\newblock PhD thesis, Eindhoven University of Technology, the Netherlands,
  1995.

\bibitem[Wat00]{Wat01}
B.~W. Watson.
\newblock Directly constructing minimal {DFA}s: Combining two algorithms by
  {B}rzozowski.
\newblock In {\em CIAA}, 311--317, 2000.

\bibitem[WDHR06]{CAV06}
M.~De Wulf, L.~Doyen, T.~A. Henzinger, and J.-F. Raskin.
\newblock Antichains: A new algorithm for checking universality of finite
  automata.
\newblock In {\em Proc.\ CAV}, volume 4144, 17--30, 2006.

\bibitem[Wor05]{worrel-final}
J.~Worrell.
\newblock On the final sequence of a finitary set functor.
\newblock {\em TCS}, 338(1-3):184--199, 2005.

\bibitem[Wor08]{worthington}
J.~Worthington.
\newblock Automatic proof generation in {K}leene algebra.
\newblock In {\em Lect. Notes in Comput. Sci.}, vol. 4988,
  382--396, 2008.

\bibitem[YL92]{DBLP:conf/pstv/YiL92}
W.~Yi and K.~G.~ Larsen.
\newblock Testing probabilistic and nondeterministic processes.
\newblock In {\em PSTV}, volume C-8 of {\em IFIP Transactions}, 47--61, 1992.

\end{thebibliography}

\chapter{Samenvatting (Dutch summary)}
\label{ch:samenvatting}

Het bestuderen van de semantiek van reactieve systemen (\emph{reactive systems})
is een belangrijke richting binnen de informatica.
Reactieve systemen voeren berekeningen uit middels interactie met hun omgeving,
en zijn over het algemeen samengesteld uit meerdere parallelle componenten
die simultaan taken uitvoeren en met elkaar communiceren.
Toepassingen bevinden zich in relatief simpele systemen als rekenmachines en
verkoopautomaten, tot programma's die mechanische apparaten zoals auto's,
metro's of ruimtevaartuigen aansturen.
Aangezien dit soort systemen veel gebruikt worden, en vaak erg complex zijn, is het
gebruik van rigoureuze methodes voor specificatie, ontwikkeling, en redenatie
over het gedrag van deze systemen een grote uitdaging.
Een mogelijke aanpak om reactieve systemen formeel te beschouwen is het gebruik
van een gemeenschappelijke taal voor de beschrijving van zowel de implementatie als
de specificatie.
In dit geval correspondeert verificatie van de implementatie met betrekking tot
de specificatie van een reactief systeem met het bewijzen van een vorm van
equivalentie/ordening tussen de beschrijvingen in de formele taal.

De doelstelling van dit proefschrift is het benutten van de krachten van een
algebra\" isch--coalgebra\" isch raamwerk voor het modelleren van reactieve systemen en het
redeneren over verschillende soorten bijbehorende semantieken op een formele
wijze.
Daarnaast richt dit proefschrift zich op het afleiden van een aantal
verificatie algoritmes die geschikt zijn voor implementatie in geautomatiseerde
systemen.

In Hoofdstuk 3 presenteren wij een beslissingsprocedure voor bisimilariteit van
een klasse van expressies die oneindige rijen (\emph{streams}), Mealy automaten,
en gelabelde transitie systemen, kan beschrijven.
Deze procedure is ge\"implementeerd in de automatische stellingbewijzer CIRC.
Hoofdstuk 4 beschrijft een uniforme coalgebra\" ische aanpak voor een collectie
van semantieken voor transitiesystemen.
Hiervoor gebruiken we een uitbreiding van de klassieke machtsverzameling
constructie.
In het bijzonder beschouwen we ``decorated trace'' equivalenties voor gelabelde
transitie-, en probabilistische systemen, en (de zogenaamde ``must'' en
``may'') ``testing''-semantieken voor divergente niet deterministische
systemen.
De coalgebra\" ische aanpak stelt ons in staat te redeneren over de eerdergenoemde
begrippen van gedrag equivalentie/ordening in termen van bisimulaties.
Verder faciliteert ons raamwerk de constructie van geverifieerde algoritmes die
niet aanwezig zijn voor bisimulariteit, zoals beschreven in Hoofdstuk 5.
In dit hoofdstuk beschrijven we een variatie van Brzozowski's algoritme om
eindige automaten te minimaliseren, en een optimalisatie van Hopcroft en Karp's
algoritme voor taal semantieken.
Beide algoritmes zijn succesvol toegepast voor het redeneren over ``decorated
trace'' en ``testing'' semantieken.
De bijbehorende implementaties kunnen online uitgeprobeerd worden:\linebreak
\url{http://perso.ens-lyon.fr/damien.pous/brz/}

\chapter{Summary}
\label{ch:summary}

One of the research areas of great importance in Computer Science is the study of the semantics of concurrent reactive systems. 
These are systems that compute by interacting with their environment, and typically consist of several parallel components, which execute simultaneously and potentially communicate with each other.
Examples of such systems range from rather simple devices such as calculators and vending machines, to programs controlling mechanical devices such as cars, subways or spaceships. In light of their widespread deployment and complexity, the application of rigorous methods for the specification, design and reasoning on the behaviour of reactive systems has always been a great challenge.
One possible approach to formally handle reactive systems is to use a ``common language" for describing both the actual implementations and their specifications. When following this technique, verifying whether an implementation and its specification describe the same behaviour reduces to proving some notion of equivalence/preorder between their corresponding descriptions over the chosen language.

The aim of this thesis is to exploit the strengths of a (co)algebraic framework in modelling reactive systems and reasoning on several types of associated semantics, in a uniform fashion.
Moreover, we derive a suite of corresponding verification algorithms suitable for implementation in automated tools.

In Chapter 3 we present a decision procedure for bisimilarity of a class of expressions defining systems such as infinite streams, deterministic automata, Mealy machines and labelled transition systems. The procedure is implemented in the automatic theorem prover CIRC.
Chapter 4 provides a uniform coalgebraic handling of a series of semantics on transition systems. This is achieved by employing a generalisation of the classical powerset construction for determinising non-deterministic automata. In particular, we deal with decorated trace equivalences for labelled transition systems and probabilistic systems and, (the so-called ``must'' and ``may'') testing semantics for divergent non-deterministic systems.
The coalgebraic approach enabled reasoning on the aforementioned notions of behavioural equivalence/preorder in terms of bisimulations. Moreover, our framework facilitated  the construction of verification algorithms which are not available for bisimilarity, as shown in Chapter 5.
There we provide a variation of Brzozowski's algorithm to minimise finite automata and an optimisation of Hopcroft and Karp's algorithm for language semantics. Both algorithms were successfully applied to reason on decorated trace and testing semantics.
The corresponding implementations can be tested online at:
\url{http://perso.ens-lyon.fr/damien.pous/brz/}.

\newpage\null\thispagestyle{empty}\newpage
\chapter{Curriculum vitae}
\label{ch:cv}

{\bf 1984}~~Born on 20 April, Suceava, Romania\\[0.5ex]

{\bf 2003 -- 2007}~~BSc in Computer Science, Alexandru Ioan Cuza University, Ia\c si, Romania\\
Final thesis title: ``The Implementation of a Programming Language in Maude, using\linebreak Denotational Semantics'', supervised by prof. dr. Dorel Lucanu\\[0.5ex]

{\bf 2007 -- 2009}~~MSc in Computer Science, Alexandru Ioan Cuza University, Ia\c si, Romania\\
Final thesis title: ``CIRC: A Behavioural Verification Tool Based on Circular Coinduction -- extensions --'', supervised by prof. dr. Dorel Lucanu\\[0.5ex]

{\bf 2010 -- 2013}~~PhD student at Reykjavik University, Iceland, and Radboud University, Nijmegen, the Netherlands\\
Final thesis title: ``Coalgebraic Tools for Bisimilarity and Decorated Trace Semantics'',\linebreak (co-)supervised by prof. dr. Jan Rutten, prof. dr. Anna Ing\'olfsd\'ottir, dr. Alexandra Silva and dr. Marcello Bonsangue

\end{document}